%% file: main.tex
\documentclass[10pt,a4paper]{amsart}
\usepackage{amssymb,amsmath}
\usepackage{graphicx}
\usepackage{bm,bbm}
\usepackage{color}
\usepackage{tikz}
\usepackage{ifthen}
\usetikzlibrary{snakes}
\usetikzlibrary{patterns}
\usepackage{pgfplots}
\usepackage{enumitem}
\usepackage{array,blkarray}
\usepackage{bigdelim}
\usepackage{listings} 
\usepackage[hidelinks]{hyperref}
\usepackage{amssymb}
\usepackage{isotope}
\usepackage[]{placeins}
\usepackage{comment}
\newtheorem{theorem}{Theorem}
\newtheorem{proposition}[theorem]{Proposition}
\newtheorem{corollary}[theorem]{Corollary}
\newtheorem{conjecture}[theorem]{Conjecture}
\newtheorem{conclusion}[theorem]{Conclusion}
\theoremstyle{definition}

\newtheorem{example}[theorem]{Example}
\theoremstyle{lemma}
\newtheorem{lemma}[theorem]{Lemma}

\theoremstyle{remark}
\newtheorem{remark}[theorem]{Remark}

\numberwithin{theorem}{section}
\numberwithin{equation}{section}
\numberwithin{table}{section}
\numberwithin{figure}{section}
\newcommand{\quotes}[1]{``#1''}
\newcommand{\ugs}{u_{\mbox{\tiny\rm gs}}}
\newcommand{\ues}{u_{\mbox{\tiny\rm es}}}
\newcommand{\ulin}{u_{\mbox{\tiny\rm lin}}}
\newcommand{\lambdags}{\lambda_{\mbox{\tiny\rm gs}}}
\newcommand{\lambdalin}{\lambda_{\mbox{\tiny\rm lin}}}
\input{def}
\begin{document}
\title[Localization and delocalization of ground states of BECs under disorder]{Localization and delocalization of ground states of Bose-Einstein condensates under disorder}
\author[]{R.~Altmann$^*$, P.~Henning$^{\dagger}$, D.~Peterseim$^*$}
\address{${}^{*}$ Department of Mathematics, University of Augsburg, Universit\"atsstr.~14, 86159 Augsburg, Germany}
\address{${}^{\dagger}$ Department of Mathematics, KTH Royal Institute of Technology, SE-100 44 Stockholm, Sweden}
\email{\{robert.altmann, daniel.peterseim\}@math.uni-augsburg.de, pathe@kth.se}
\thanks{D.~Peterseim acknowledges the support of the European Research Council through the project 865751 -- RandomMultiScales. P.~Henning acknowledges the support by the Swedish Research Council (grant 2016-03339) and the G\"oran Gustafsson foundation. \\
\indent	
This paper will be published in SIAM J. Appl. Math.}
\date{\today}
%
%
\begin{abstract}
This paper studies the localization behaviour of Bose-Einstein condensates in disorder potentials, modeled by a Gross-Pitaevskii eigenvalue problem on a bounded interval. In the regime of weak particle interaction, we are able to quantify exponential localization of the ground state, depending on statistical parameters and the strength of the potential. 
Numerical studies further show delocalization if we leave the identified parameter range, which is in agreement with experimental data. These mathematical and numerical findings allow the prediction of physically relevant regimes where localization of ground states may be observed experimentally. 
\end{abstract}
%
%
\maketitle
\setcounter{tocdepth}{2}
%
%
%
\section{Introduction}
When a dilute bosonic gas is cooled down to ultra-low temperatures close to $0K$, an extreme state of matter is formed: a Bose-Einstein condensate (BEC). A characteristic feature of such condensates is that a large fraction of the particles occupies the same quantum state and hence behaves like a single \quotes{super atom}. This allows to study certain quantum phenomena on a macroscopic observation scale, where the phenomenon of superfluidity is perhaps the most prominent one. The existence of BECs was first predicted by Bose and Einstein one century ago \cite{Bos24,Ein24}, but it took until 1995 before it could be finally experimentally confirmed~\cite{AEM95,DMA95}. Since then, there is still an increasing interest in the topic and the study of BECs and its properties became a highly active field of modern quantum physics.

In this paper, we shall investigate the behavior of BECs in a disorder potential from a mathematical perspective. The interest in studying wave phenomena in disordered media goes back to the seminal work by Anderson~\cite{And58} who discovered that electronic waves in a disorder crystal are strongly (exponentially) localized. For example, this can lead to an insulating effect for otherwise conducting materials. This so-called {\it Anderson localization} was later discovered to be a general phenomenon that can be encountered for acoustic waves, elastic waves, electromagnetic waves, and even quantum matter waves. The latter has been studied experimentally through almost non-interacting BECs~\cite{physics-paper-localization-BEC,RDF08}, where exponential localization in the sense of Anderson was observed together with a suppression of transport in an expanding BEC. For this type of \quotes{dynamical} localization, theoretical and numerical studies can be found in~\cite{CVH05,CVR06,FFG05,SPC08}. In contrast to that, the focus of this paper is to study the effect of disorder on the ground state, which is experimentally still a challenge. 

Numerical experiments that address the localization of ground states of BECs are presented, for example, in~\cite{SDK05,SDK06,AltPV18,AltP19,AltHP19ppt}. It is found, in accordance with analogous findings for the dynamical localization of BECs, that there is a sensitive interplay between the strength of disorder and the strength of particle interactions. In fact, to observe Anderson-type localization for ground states it is expected that the particle interactions need to be sufficiently weak and that the degree of disorder, both in terms of amplitude and oscillation/correlation length, needs to be sufficiently large.  Despite plenty of numerical evidence (cf.~\cite{GCP19} and the references therein for a recent overview), to the best of our knowledge there has not yet been any direct experimental observation of the localization of ground states under disorder. 

We consider the Gross-Pitaevskii eigenvalue problem (GPEVP) as a mathematical model for the ground states of BECs with weak repulsive interactions at ultra-low temperatures, where we refer to~\cite{DGP99,LSY01,PiS03,Aft06,BaC13b} and the references therein for derivations and analytical justifications. The Gross-Pitaevskii equation (GPE) is based on mean field theory and one might wonder if it is hence still applicable for BECs under strong disorder. In fact, Seiringer et al.~\cite{SYZ14} proved that Bose-Einstein condensation can take place in highly disordered potentials (with large amplitude and short oscillation length) and that the GPE is still a valid model in this regime. Based on this justification, we study localization phenomena through the (dimensionless) GPEVP in one space dimension as a model problem. The equation seeks the quantum state $u$ of the BEC under the mass normalization constraint $\int_{\R} |u|^2 \dx =1 $ such that
\begin{align}
\label{eq:GPEVP}
  - \tfrac 12\, u^{\prime\prime} + \Veps u + \kappa\, |u|^2 u = \lambda\, u.
\end{align}
Here, $\Veps$ is a disorder potential, $\kappa \ge 0$ a parameter that characterizes the strength of repulsive particle interactions, and~$\lambda$ the eigenvalue that equals a rescaled chemical potential. As mentioned above, it is well established in the literature (cf.~\cite{SDK05,SDK06,SYZ14}) that if $u$ is the ground state (i.e.~the eigenfunction to the smallest eigenvalue $\lambda$), then $u$ is expected to be exponentially localized, provided that $\Veps$ is sufficiently strong in amplitude and disorder and that $\kappa$, i.e., the particle interaction, is sufficiently small. In this paper, we are concerned with quantifying this localization in terms of statistical parameters of the potential and the size of the interaction constant $\kappa$. Previous qualitative results in this direction were derived in~\cite{SYZ14}. Our new results are consistent with these earlier findings, but they are different in the sense that we treat other types of disorder potentials and we will be able to make quantitative predictions concerning the strength of the potential relative to the oscillation length and interaction parameter.

The proof of localization is based on the corresponding linear eigenvalue problem and the observation that the nonlinearity can be interpreted as a perturbation of the given potential. Thus, results from the linear case can be transferred and extended to our more general setting of the GPEVP. 
The goal of this paper is to provide a mathematical explanation for the exponential localization of ground states under sufficiently large disorder and for the delocalization of the condensate for an increasing strength of particle interactions.

Using the abstract theory of preconditioned iterative solvers~\cite{KorY16,KorPY18}  together with exponential decay properties of the Green's function associated with the (linear) Schr\"odinger operator, it is possible to prove Anderson-type localization for the ground state of the linear problem with a sufficiently strong disorder potential~\cite{AltHP20}. 
In this paper, we generalize these results in a first step by proving the aforementioned localization for a large fraction of the lower part of the spectrum and by dealing with a broader class of potentials. This will be subject of Section~\ref{sect:linear}. In a second step, we interpret in Section~\ref{sect:GPEVP} the nonlinear GPEVP as a linear eigenvalue problem with effective potential \quotes{$\Veps + \kappa\, |\ugs|^2$}. Here, $\ugs$ denotes the ground state, for which we prove that it introduces only a sufficiently small perturbation of the potential $\Veps$, provided that $\Veps$ is sufficiently large and that $\kappa$ is sufficiently small. To be precise, we prove that if $\Veps$ is oscillating with a wave length of size $0<\eps\ll 1$, then the maximum amplitude of the disorder potential needs to be at least of order $\eps^{-2}$ and the strength of the nonlinearity must not exceed the order $\eps^{-1}$ to observe localization. We shall also formulate a conjecture saying that this scaling will change for higher space dimensions. Our findings will be support by numerical experiments in Sections \ref{subsect:numerics:higherdim}. 
Finally, we consider realistic physical values in Section~\ref{section:application-to-physical-values} and translate them into our scaling regime. With this we are able to make predictions about the localization and delocalization of ground states in different experimental configurations. 

Throughout this paper, we use the notion $a \lesssim b$ (and accordingly~$\gtrsim$) for the existence of a generic constant $c>0$, independent of the parameter~$\eps$, such that $a \le cb$. Moreover, we write $a \simeq b$ if we have $a \lesssim b$ and $a \gtrsim b$. 
%
%
\section{Physical Setting and Scaling Regime}
We shall present a motivation for the scaling regime (for $\Veps$ and $\kappa$) that we will consider in this paper. For that, we start from the GPEVP in physical units and derive a nondimensional form that depends on the oscillation length of the disorder potential. Like that, we will see how such a scaling influences the effective potential amplitude and the effective interaction constant in the corresponding nondimensional equation. The resulting scaling will be the basis for our analysis and we shall later relate it to explicit physical values in Section~\ref{section:application-to-physical-values}.

We consider a repulsive, weakly interacting Bose gas at ultra low temperatures in a disorder potential. The stationary states of such BECs are modeled by the GPEVP in physical units, where we seek the condensate's quantum state $\psi \colon\R^3 \rightarrow \R$ with corresponding eigenvalue $\mu \in \R$ such that
\begin{eqnarray*}
 -\frac{\hbar^2}{2m} \triangle \psi(\x) + V_{\mbox{\tiny\rm tr}}(\x) \hspace{2pt} \psi(\x) + \frac{g}{2}\, |\psi(\x)|^2 \psi(\x)
 &=& \mu \hspace{2pt} \psi(\x)
\end{eqnarray*}
and with the mass constraint
$$
\int_{\R^3} |\psi(\x)|^2 \hspace{2pt} \dxx = N.
$$
Here, $\hbar$ is the reduced Planck constant (in \texttt{[J$\cdot$s]}); $N$ is the number of bosons (dimensionless unit); $m$ the mass of a single boson (in \texttt{[kg]}); $a$ the scattering length (in \texttt{[m]}); $g=\frac{4 \pi \hbar^2 a}{m}$ is an atomic interaction constant (in \texttt{[J\hspace{2pt}$\cdot$\hspace{2pt}m$^3$]}); $V_{\mbox{\tiny\rm tr}}$ is an external trapping potential (in \texttt{[J]}); and the eigenvalue $\mu$ is the chemical potential of the condensate (in \texttt{[J]}). The physical unit of the wave function $\psi(\x \mbox{\rm \texttt{[m]}})$ is \texttt{[m$^{-3/2}$]}, which yields a particle density $|\psi(\x \mbox{\rm \texttt{[m]}})|^2$ measured in \texttt{[m$^{-3}$]} (particles per cubic meter). The time-dependent standing wave that describes the condensate is given by $\psi(\x) e^{-\mu \ci t / \hbar}$.

An analytical justification that the Gross-Pitaevskii model is still applicable even for strong disorder potentials was shown in~\cite{SYZ14}.
%
\subsection{Dimensionless form in $3D$}
\label{subsec-dimless-3D}
In order to introduce a non-dimensional form of the equation we can select a scaling parameter $\eps>0$ that is adjusted to the characteristic length of the condensate and which determines a targeted scaling regime. The scaling parameter is such that $\eps^2$ is measured 
in the unit~\texttt{[Hz]}.  Furthermore, we introduce a dimensionless parameter $\rho>0$ that will help us to tune the oscillation length of the nondimensional potential to $\eps$. 
With this we set
\begin{align}
\label{def-sx}
s_\x := \sqrt{ \frac{\hbar}{m } }\ \frac{\rho}{\eps}  
\qquad \mbox{in } \texttt{[m]}
\end{align}
and define the non-dimensional quantum state as
$$
u_{\mbox{\tiny\rm 3D}}(\x )
:= \sqrt{ \frac{s^3_\x}{N} }\ \psi( s_\x \x ).
$$
It is easy to verify that $u_{\mbox{\tiny\rm 3D}}$ is normalized in mass, i.e., 
$
\int_{\R^3} |u_{\mbox{\tiny\rm 3D}}(\x)|^2 \hspace{2pt} \dxx 
=1,
$
and that it solves 
\begin{eqnarray}
\label{GPE-3D}
 -\frac{1}{2} \triangle u_{\mbox{\tiny\rm 3D}}(\x ) + \frac{1}{\eps^2} \hspace{2pt} V_{\mbox{\tiny\rm 3D}}\hspace{-2pt}\left( \frac{\x}{\eps}\right) \hspace{0pt} u_{\mbox{\tiny\rm 3D}}(\x) 
 +  \eps \hspace{2pt}\kappa_{\mbox{\tiny\rm 3D}} \hspace{2pt} |u_{\mbox{\tiny\rm 3D}}(\x)|^2 u_{\mbox{\tiny\rm 3D}}(\x)
 &=& \lambda_{\mbox{\tiny\rm 3D}} \hspace{2pt}u_{\mbox{\tiny\rm 3D}}(\x),
\end{eqnarray}
where
$$
V_{\mbox{\tiny\rm 3D}}(\x) :=
 \frac{\rho^2}{\hbar} \hspace{2pt} V_{\mbox{\tiny\rm tr}}\left( 
 \x \hspace{2pt}\rho\hspace{2pt}\sqrt{\frac{\hbar}{m}} \right),
 \qquad 
 \kappa_{\mbox{\tiny\rm 3D}}  := \frac{2 \pi a N}{\rho} \sqrt{\frac{m}{\hbar}},
\qquad \mbox{and}
\qquad
\lambda_{\mbox{\tiny\rm 3D}} :=\frac{\rho^2}{\eps^2} \frac{\mu}{\hbar}.
$$
We observe that, in this scaling regime, the strength of the potential scales with $1/\eps^2$. Furthermore, if $V_{\mbox{\tiny\rm 3D}}$ is oscillating on a scale of order $\calO(1)$, then the potential in \eqref{GPE-3D}, i.e., $\frac{1}{\eps^2} \hspace{2pt} V_{\mbox{\tiny\rm 3D}}\hspace{-2pt}\left( \frac{\cdot}{\eps}\right)$, is oscillating on the $\eps$-scale. Finally, the atomic interaction constant scales with $\eps$. As we will see next, this last scaling of the interaction constant is in fact depending on the spatial dimension, whereas the regime for the potential remains unchanged.
\begin{remark}[size of $\eps$]
Formally, $\eps>0$ is not a physical parameter and can be chosen arbitrarily. Different values for $\eps$ only affect the scaling regime of the non-dimensional equation. However, in this paper we will select $\eps \ll 1$ so that the characteristic length of the rescaled wave function $u_{\mbox{\tiny\rm 3D}}$ is smaller than $1$. This will allow us to consider \eqref{GPE-3D} on the unit interval and measure the (exponential) decay of $u_{\mbox{\tiny\rm 3D}}$ in units of $\eps$.
\end{remark}

\subsection{Dimensionless form in $2D$}
\label{subsec-dimless-2D}
If the trapping potential $V_{\mbox{\tiny\rm 3D}}$ is strongly anisotropic, then the condensate can be confined into a plane or a certain space direction. In these cases, the $3D$ GPEVP~\eqref{GPE-3D} can be formally reduced to an equation in one or two space dimensions.

A typical physical experiment in $2D$ creates a disk-shaped condensate with small height. Practically this is achieved through a harmonic confinement potential with a large trapping frequency in the strong confinement direction. Without loss of generality assume that the condensate is essentially spreading in the $xy$-plane, and hence, that it is strongly confined in $z$-direction. For some large trapping frequency $\omega_z \gg 1$ (in  \texttt{[rad/s]$\hat{=}$[(2$\pi$)$^{-1}$Hz]}), a suitable potential is of the form
\begin{align*}
V_{\mbox{\tiny\rm tr}}(\x) = W_{\mbox{\tiny\rm 2D}}(x,y) + \frac{m}{2} \omega_z^2 z^2, \qquad 
\mbox{where } \x= (x,y,z)
\end{align*}
and where $W_{\mbox{\tiny\rm 2D}}$ characterizes the potential in the $xy$-plane. For the rescaled potential we consequently obtain 
\begin{align*}
V_{\mbox{\tiny\rm 3D}}(\x) = V_{\mbox{\tiny\rm 2D}}(x,y) + \frac{1}{2} \left(\frac{\omega_z}{2\pi}\right)^2 \hspace{-2pt}z^2
\qquad\mbox{with }
V_{\mbox{\tiny\rm 2D}}(\x) := \frac{\rho^2}{\hbar} \hspace{2pt} W_{\mbox{\tiny\rm 2D}}\left( 
 \x \hspace{2pt}\rho \hspace{2pt}\sqrt{\frac{\hbar}{m}} \right).
\end{align*}
In this case it can be shown that we have a separation of variables for $u_{\mbox{\tiny\rm 3D}}(\x )$ which allows to project all terms of the equation into a suitable subspace of functions that only live in the $xy$-plane and which consequently reduces the $3D$ GPEVP to a $2D$ equation (cf.~\cite{Bao14,BaC13b,BJM03} for details and analytical proofs). In our case (i.e.~for repulsive, weakly interacting BECs), the $2D$ GPEVP reads
\begin{eqnarray}
\label{GPE-2D}
 -\frac{1}{2} \triangle u_{\mbox{\tiny\rm 2D}}(\x ) + \frac{1}{\eps^2} \hspace{2pt} V_{\mbox{\tiny\rm 2D}}\hspace{-2pt}\left( \frac{\x}{\eps}\right) \hspace{0pt} u_{\mbox{\tiny\rm 2D}}(\x) 
 + \kappa_{\mbox{\tiny\rm 2D}} \hspace{2pt} |u_{\mbox{\tiny\rm 2D}}(\x)|^2 u_{\mbox{\tiny\rm 2D}}(\x)
 &=& \lambda_{\mbox{\tiny\rm 2D}} \hspace{2pt}u_{\mbox{\tiny\rm 2D}}(\x)
\end{eqnarray}
with $\x=(x,y)$, $V_{\mbox{\tiny\rm 2D}}$ as defined above, and
$$
 \kappa_{\mbox{\tiny\rm 2D}}  := a N \sqrt{\frac{m \omega_z }{\hbar}}.
$$
Note that in our scaling regime, the effective (nondimensional) trapping frequency in $z$-direction is $\frac{\omega_z \rho^2}{2 \pi \eps^2}$. This $\eps$-dependency of the trapping frequency causes the change of the scaling for the interaction constant. More precisely, we used that the full interaction constant is given by
\begin{align*}
\eps\, \kappa_{\mbox{\tiny\rm 3D}} \frac{1}{\sqrt{2\pi}}\, \sqrt{\frac{\omega_z  \rho^2}{2\pi \eps^2}}
= \kappa_{\mbox{\tiny\rm 3D}} \rho \frac{\sqrt{\omega_z}}{2\pi} =:\kappa_{\mbox{\tiny\rm 2D}}.
\end{align*}

%
\subsection{Dimensionless form in $1D$}
\label{subsec-dimless-1D}
If the strong confinement is both in $y$- and $z$-direction, a cigar-shaped condensate can be obtained. In this case, the strongly anisotropic potential trap is of the form 
\begin{align*}
V_{\mbox{\tiny\rm tr}}(\x) = W_{\mbox{\tiny\rm 1D}}(x) + \frac{m}{2} \left( \omega_y^2 y^2 + \omega_z^2 z^2 \right), 
\end{align*}
where $\omega_y, \omega_z \gg 1$ are strong trapping frequencies and $W_{\mbox{\tiny\rm 1D}}$ a potential that only acts in $x$-direction. The rescaled potential becomes
\begin{align*}
V_{\mbox{\tiny\rm 3D}}(\x) = V_{\mbox{\tiny\rm 1D}}(x) + \frac{1}{2} \left( \left(\frac{\omega_y}{2\pi}\right)^2 \hspace{-3pt} y^2 + \left(\frac{\omega_z}{2\pi}\right)^2 \hspace{-3pt}z^2 \right)
\qquad\mbox{with }
V_{\mbox{\tiny\rm 1D}}(x) := \frac{\rho^2}{\hbar} \hspace{2pt} W_{\mbox{\tiny\rm 1D}}\left( 
  x \hspace{2pt}\rho \hspace{2pt}\sqrt{\frac{\hbar}{m}} \right).
\end{align*}
Using again the aforementioned projection method as elaborated in \cite{BJM03}, it is possible to reduce the $3D$ GPEVP to a $1D$ GPEVP. In our scaling regime we obtain 
\begin{eqnarray}
\label{GPE-1D}
 -\frac{1}{2}\,  u_{\mbox{\tiny\rm 1D}}^{\prime\prime}(x ) + \frac{1}{\eps^2} \hspace{2pt} V_{\mbox{\tiny\rm 1D}}\hspace{-2pt}\left( \frac{x}{\eps}\right) \hspace{0pt} u_{\mbox{\tiny\rm 1D}}(x) 
 + \frac{\kappa_{\mbox{\tiny\rm 1D}}}{\eps} \hspace{2pt} |u_{\mbox{\tiny\rm 1D}}(x)|^2 u_{\mbox{\tiny\rm 1D}}(x)
 &=& \lambda_{\mbox{\tiny\rm 1D}} \hspace{2pt}u_{\mbox{\tiny\rm 1D}}(x),
\end{eqnarray}
where 
$$
\kappa_{\mbox{\tiny\rm 1D}} := \rho \frac{a N}{2 \pi} \sqrt{\frac{m \omega_y \omega_z}{ \hbar}}.
$$
We stress that the scaling in front of the interaction constant changed again, from $\calO(\eps)$ in $3D$ and $\calO(1)$ in $2D$ to $\calO(\eps^{-1})$ in $1D$. As before, this corresponds to the influence of the effective trapping frequencies in $y$- and $z$-direction, where we have the relation
\begin{align*}
 \eps\, \kappa_{\mbox{\tiny\rm 3D}}\, \frac{1}{2\pi}\, \sqrt{\frac{\omega_y \omega_z \rho^4}{4 \pi^2 \eps^4}}
  =  \eps^{-1} \frac{\kappa_{\mbox{\tiny\rm 3D}}}{4\pi^2} \rho^2 \sqrt{\omega_y \omega_z}
  =: \eps^{-1} \kappa_{\mbox{\tiny\rm 1D}}.
\end{align*}
In this paper, we mainly consider the $1D$ GPEVP in the scaling regime as given in \eqref{GPE-1D} and for a disorder potential $V_{\mbox{\tiny\rm 1D}}$. In this regime, we prove that the ground state is exponentially localized to a small region, where $\eps$ is selected small enough so that the condensate is confined to the unit interval. The ground state solution $u_{\mbox{\tiny\rm 1D}}$ to equation \eqref{GPE-1D} is defined as ($L^2$-normalized) eigenfunction to the smallest eigenvalue $\lambda_{\mbox{\tiny\rm 1D}}$. Equivalently, we can characterize $u_{\mbox{\tiny\rm 1D}}$ as the global minimizer of the total energy. We will later make these characterizations explicit in Section \ref{sect:GPEVP}. In Section \ref{section:application-to-physical-values} we apply our theoretical findings to realistic physical values in order to make predictions about a practical localization regime.

From now on we shall consider the one-dimensional GPEVP truncated to a bounded interval of appropriate size with homogeneous Dirichlet boundary conditions. This truncation can be physically justified by computing the Thomas--Fermi radius of the condensate, cf.~\cite{Bao14}. 
%
%
\section{Disorder Potentials and Localization Results for the Linear Case}\label{sect:linear}
In order to prove localization results for the eigenstates of the GPEVP in Section~\ref{sect:GPEVP} we need to have a closer look at the linear case first.
For~$\kappa=0$, equation~\eqref{eq:GPEVP} is known as the linear Schr\"odinger eigenvalue problem, 
\begin{align}
\label{eq:linSchroedingerEVP}
  - \tfrac 12\, u^{\prime\prime}(x) + V(x) u(x) = \lambda\, u(x)
  \qquad\text{in }D\subseteq\R 
\end{align}
with homogeneous Dirichlet boundary conditions, i.e., $u \in \V := H^1_0(D)$. 
The literature concerning Anderson-type localization of ground and excited states from a mathematical perspective is much more extensive for this linear case. In particular, it is well-known that the first eigenstates localize (in an exponential manner) under disorder. Exemplary, we mention that these effects have been analyzed in the early works~\cite{FroS83,FroMSS85,AizM93,Aiz94} and, more recently, with a landscape function  approach~\cite{FilM12,ArnDJMF16,Ste17} as well as from a multiscale point of view~\cite{AltHP20,AltP19}. In the following, we generalize the results of~\cite{AltHP20} and show that the first~$\calO(\eps^{-1})$ eigenfunctions localize in the sense of an exponential decay as a preparation for the nonlinear case that is treated in Section~\ref{sect:GPEVP}. 

Throughout this section, we will consider the weak form of the eigenvalue problem~\eqref{eq:linSchroedingerEVP}. For this, we introduce the operators~$\calA\colon \V\to\V^*$ and~$\calI\colon \V\to\V^*$ by
\begin{align}
\label{eq:calA}
  \langle \calA u, v\rangle 
  := \int_D \frac{1}{2}\, u'(x)v'(x) + V(x)\, u(x)v(x) \dx, \qquad
  \langle \calI u, v\rangle 
  := \int_D u(x)v(x) \dx.
\end{align}
The weak form of the Schr\"odinger eigenvalue problem then reads~$\calA u = \lambda \calI u$ in $\V^*$. 
%
%
\subsection{Disorder potentials}\label{sect:linear:potential}
In this section we shall specify the characteristics of the considered family of disorder potentials~$(\Veps)_{\eps>0} \subset L^{\infty}(D)$. The family consists of potentials that are rapidly oscillating on a grid with mesh width~$\eps$. To make this precise, we consider the unit interval $D:=(0,1)$ and we let $\eps>0$ denote the aforementioned (small) mesh size parameter with~$\eps\ll 1$ and, for simplicity, ~$\eps^{-1} \in \N$. With this, we begin with introducing a family of equidistant meshes that consist of small subintervals of length $\eps$, i.e.,  
\[
  \calTeps := \big\{\, \overline{T_k^{\eps}} \ |\ 1 \le k \le \eps^{-1} \big\}, \qquad
  T_k^{\eps} := ((k-1)\eps,\, k\eps).
\]
The corresponding family of potentials is assumed to be nonnegative and piecewise continuous w.r.t.~the mesh, i.e., we assume that for each member of the family we have
\[
  \Veps \in L^{\infty}(D), \qquad
  \Veps\ge 0, \qquad
  \Veps|_{T_k^{\eps}} \in C^0(T_k^{\eps}).
\]
For the subsequent analysis, we shall introduce two characteristic values for the families of potentials, namely $0<\alpha<\beta<\infty$, both independent of~$\eps$. In order to study decay properties, we will divide the interval $D$ into three relevant (disconnected) subregions. Loosely speaking, in one region the potential takes values smaller or equal to $\alpha\,\eps^{-2} |\log\eps|^{-2}$, then there is an second intermediate region, and in a third region the potential takes values larger or equal to $\beta\eps^{-2}$. To make this precise, we define two disjoint submeshes of~$\calTeps$, namely~$\calTeps_\alpha$ and~$\calTeps_\beta$, by 
\begin{align*}
  \calTeps_\alpha 
  &:= \big\{\, \overline{T^{\eps}} \in \calTeps \ |\ \sup\nolimits_{x\in T^{\eps}} \Veps(x) \le 
  \alpha\, \eps^{-2}\, |\log\eps|^{-2} 
  \big\}, \\ 
  \calTeps_\beta 
  &:= \big\{\, \overline{T^{\eps}} \in \calTeps \ |\ \inf\nolimits_{x\in T^{\eps}} \Veps(x) \ge \beta\, \eps^{-2} \big\}.
\end{align*}
Note that $\alpha \eps^{-2} |\log\eps|^{-2} < \beta\eps^{-2}$ for $\eps \ll 1$ and that the union of~$\calTeps_\alpha$ and $\calTeps_\beta$ will, in general, not add up to the entire mesh~$\calTeps$ due to the possibility of intermediate values. 
Let us define the corresponding (disconnected) sets by
\[
  \Oaeps := \bigcup\, \big\{ \overline{T^{\eps}} \in \calTeps_\alpha \big\} \qquad
  \mbox{and} \qquad
  \Obeps := \bigcup\, \big\{ \overline{T^{\eps}} \in \calTeps_\beta \big\}. 
\]
\begin{example}
We mention three representative examples of potentials, which are also illustrated in Figure~\ref{fig:potentials}. 
A special case is a periodic two-valued potential, i.e., $\Veps$ equals alternating $1$ or~$\eps^{-2}$. A reasonable choice for the characteristic parameters would be any $\alpha\in (0,1)$ and any $\beta \in (\alpha,1]$. We emphasize that eigenfunctions do not localize in the periodic setting~\cite{AltHP20}. 
Second, we may consider a two-valued Bernoulli potential, which is piecewise constant w.r.t.~$\calTeps$. For a prescribed probability~$p\in(0,1)$ we set $\Veps$ on a subinterval to $1$ and $\eps^{-2}$ otherwise. 
Reasonable values for $\alpha$ and $\beta$ can be selected as before. 
Last but not least, we may choose the value of the potential on a subinterval randomly between $0$ and some~$\eps^{-2}V_\text{max} > 0$. In this case, there are various possibilities to choose~$\alpha$ and~$\beta$. 
\end{example}
%
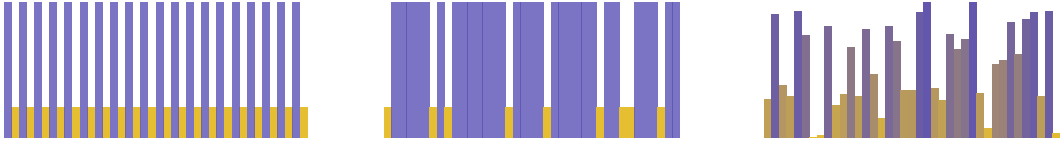
\begin{figure}
\begin{center}
\begin{tikzpicture}
  \foreach \x in {0, 0.2, ..., 3.8} {
    \fill[alphaCol] (\x+0.1, 0) -- (\x+0.2, 0) -- (\x+0.2, 0.4) -- (\x+0.1, 0.4) -- cycle;	
    \fill[betaCol, opacity=0.8] (\x, 0) -- (\x+0.1, 0) -- (\x+0.1, 1.8) -- (\x, 1.8) -- cycle;	
  }
  \foreach \x in {5, 5.1, ..., 8.9} {		
	\pgfmathparse{rnd}
	\pgfmathsetmacro{\r}{\pgfmathresult}		
	\ifthenelse{\lengthtest{\r pt < 0.6 pt}}{
	  \fill[betaCol, opacity=0.8] (\x, 0) -- (\x+0.1, 0) -- (\x+0.1, 1.8) -- (\x, 1.8) -- cycle;
	}{
	  \fill[alphaCol] (\x, 0) -- (\x+0.1, 0) -- (\x+0.1, 0.4) -- (\x, 0.4) -- cycle;
	}	
  }
  \foreach \x in {10, 10.1, ..., 13.9} {		
	\pgfmathparse{rnd}
	\pgfmathsetmacro{\y}{\pgfmathresult}
	\pgfmathparse{95*\y}
	\edef\tmp{\pgfmathresult}
	\fill[betaCol!\tmp!alphaCol] (\x, 0) -- (\x+0.1, 0) -- (\x+0.1, \y*1.8) -- (\x, \y*1.8) -- cycle;
  }
\end{tikzpicture}
\end{center}
\caption{Illustration of one-dimensional potentials. Periodic potential (left), Bernoulli potential with $p=0.4$ (middle), and fully disorder potential (right). All three potentials are discontinuous and piecewise constant.}
\label{fig:potentials}
\end{figure}
The (maximal) connectivity components of~$\Oaeps$ and $\Obeps$ will be denoted as {\em valleys} and {\em peaks}, respectively. Of particular interest are the diameter of the largest valley, which is of length~$\Laeps\eps$, i.e., 
\begin{align}
\label{def:La}
  \Laeps 
  = \sup \{\, \mbox{diam}(Q)\,\eps^{-1} \ |\ Q\subseteq \Oaeps \text{ is a closed interval} \}
\end{align}
as well as the largest interval without a peak, 
\[
  \Lbeps 
  = \sup \{\, \mbox{diam}(Q)\,\eps^{-1} \ |\ Q\subseteq \overline D\setminus\Obeps \text{ is a closed interval} \}.
\]
In general, we have $0 \le \Laeps \le \Lbeps \le \eps^{-1}$. However, for random potentials $\Veps$ that admit a certain statistical distribution of values, the size of $\Laeps$ and $\Lbeps$ is closely connected to the value of $\eps$. For disordered random potentials we can say that the smaller the value of $\eps$, the higher is the probability that $\Laeps$ and $\Lbeps$ become large and tend to infinity for $\eps \rightarrow 0$. 
\begin{example}
\label{exp:Bernoulli}
Assume that the potential $\Veps$ corresponds to a Bernoulli distribution of $\alpha$- and $\beta$-cells with probability~$p \in (0,1)$, cf.~Figure~\ref{fig:potentials} (middle). In this case, we have $\Laeps = \Lbeps$ and with high probability $\Laeps \simeq \log_{1/p}(\eps^{-1}) = \log_{p}(\eps) = \tfrac{\log \eps}{\log p}$, cf.~\cite{Mal15}. 
\end{example}
Note that, so far, the assumptions fit to any (deterministic) realization of a random potential~$\Veps$. This particularly includes the periodic case. The actual size of $\Laeps$ and $\Lbeps$ will only play a role when we start to investigate localization properties. 
%
\subsection{Norm estimates}
In this section, we derive a Friedrichs-type estimate that will be the key to finding lower bounds for the eigenvalues of the linear elliptic differential operator~$\calA_\eps$, which is the operator defined in~\eqref{eq:calA} with potential $V=\Veps$. For proving the estimate, we define a cut-off function $\eta_{\eps} \in H^1(D)$ with $0 \le \eta_{\eps}(x) \le 1$, which is based on the particular structure of the potential $\Veps$.  
More precisely, we set~$\eta_{\eps}$ as the globally continuous and piecewise linear function, which is constant~$1$ in $D\setminus\Obeps$ and vanishes in the middle of each element of~$\calTeps_\beta$, cf.~Figure~\ref{fig:cutoff}. With this, we obtain a cut-off function, which satisfies~$\Vert \nabla\eta_{\eps}\Vert_{L^\infty(D)} \le 2\,\eps^{-1}$ as well as the following Friedrichs inequality. 
%
\begin{figure}
\begin{center}
\begin{tikzpicture}[scale=0.7]
	\foreach \x in {0, 1, 4, 5, 6, 8, 12, 14} {		
		\pgfmathparse{rnd}
		\pgfmathsetmacro{\rndHeight}{\pgfmathresult}		
		\fill[alphaCol, opacity=0.5] (\x,0) -- (\x+1,0) -- (\x+1,2*\rndHeight) -- (\x,2*\rndHeight) -- cycle;
	}
	\foreach \x in {2, 3, 7, 9, 10, 11, 13} {		
		\pgfmathparse{rnd}
		\pgfmathsetmacro{\rndHeight}{\pgfmathresult}				
		\fill[betaCol, opacity=0.5] (\x,0) -- (\x+1,0) -- (\x+1,2.6+\rndHeight) -- (\x,2.6+\rndHeight) -- cycle;
	}
	\draw[thick] (0,0) -- (15,0);	
	\foreach \x in {0, ..., 15} {		
		\draw[thick] (\x, -0.2) -- (\x, 0.2);
	}
	\draw[very thick] (0, -0.2) -- (0, 0.3);
	\draw[very thick] (15, -0.2) -- (15, 0.3);	
	\draw[gray, thick, dashed] (0, 2.4) -- (15.5, 2.4);
	\node[gray] at (16.2, 2.4) {$\beta\eps^{-2}$};	
	\draw[myRed, thick] (0,1.2) -- (2,1.2) -- (2.5,0) -- (3,1.2) -- (3.5,0) -- (4,1.2) -- (7,1.2) -- (7.5,0) -- (8,1.2) -- (9,1.2) -- (9.5,0) -- (10,1.2) -- (10.5,0) -- (11,1.2) -- (11.5,0) -- (12,1.2) -- (13,1.2) -- (13.5,0) -- (14,1.2) -- (15,1.2);
	\node[myRed] at (-0.5, 1.2) {$\eta_{\eps}$};	
\end{tikzpicture}
\end{center}
\caption{Illustration of the cut-off function $\eta_{\eps}$, which is constant $1$ in~$D\setminus\Obeps$ (yellow) and vanishes in the interior of each element~$\overline{T^{\eps}}\in\calTbeps$ (purple).}
\label{fig:cutoff}
\end{figure}
\begin{lemma}
\label{lem:cutoff}
Consider a function $v\in\V= H^1_0(D)$ and the cut-off function $\eta_{\eps}$ introduced above. Then the product $\eta_{\eps} v$ satisfies the Friedrichs-type inequality 
\begin{align}
\label{eq:cutoff}
  \| \eta_{\eps} v \|_{L^2(D)} 
  \le \frac{(\Lbeps+1)\,\eps}{\pi }\ \| (\eta_{\eps} v)' \|_{L^2(D)}.
\end{align}
\end{lemma}
\begin{proof}
By the definition of $\eta_{\eps}$, the product $\eta_{\eps} v$ vanishes at the boundary of $D$ and in the center of each $\overline{T^{\eps}}\in\calTeps_\beta$. These zeros define a partition of $D$, on which we can apply the one-dimensional version of the Poincar\'e-Friedrichs inequality (or Wirtinger's inequality) of the form $\| v \|_{L^2(a,b)}\le \frac{b-a}{\pi}\,\| v' \|_{L^2(a,b)}$ for all $v\in H^1_0(a,b)$, cf.~\cite{DymMcKean72}. 
Since the diameter of each subinterval is bounded by~$(\Lbeps+1)\,\eps$, the assertion follows.  
%
\end{proof}
With the cut-off function $\eta_{\eps}$ and the previous lemma we can deduce general bounds of the $L^2$-norm in terms of the energy norm. To keep the notation short we introduce
\[
  \|v\|^2 := \|v\|^2_{L^2(D)}, \qquad
  \|v\|^2_{\Veps} := \|\sqrt{\Veps}v\|^2 = \int_D \Veps\, |v|^2 \dx, \qquad
  \|v\|^2_{\Veps,\omega} := \int_\omega \Veps\, |v|^2 \dx
\]
for any subset~$\omega\subseteq D$. The corresponding energy norm (for the linear problem) is defined as
\[
  \Vvert v \Vvert_\eps^2 
  := \langle \calA_\eps v, v\rangle
  = \tfrac 12\, \| v'\|^2 + \|v\|^2_{\Veps}, \qquad
  \Vvert v \Vvert^2_{\eps,\omega} 
  = \tfrac 12\, \| v'\|^2_{L^2(\omega)} + \|v\|^2_{\Veps,\omega}. 
\]
With this, we obtain for an arbitrary function $v\in \V$ the estimate 
\begin{eqnarray*}
  \Vert v \Vert^2
  &\le& \Vert \eta_{\eps} v \Vert^2 + \frac{\eps^2}{\beta}\, \Vert v \Vert_{\Veps,\Obeps}^2 \\
  &\overset{\eqref{eq:cutoff}}{\le}& \frac{(\Lbeps+1)^2\eps^2}{\pi^2}\, \Big( \frac{8}{\beta}\, \Vert v \Vert_{\Veps,\Obeps}^2  + 2\, \Vert v^{\prime} \Vert^2 \Big) + \frac{\eps^2}{\beta}\, \Vert v \Vert_{\Veps, \Obeps}^2 \\
  &=& \frac{4\,(\Lbeps+1)^2\eps^2}{\pi^2}\, \frac 12\, \Vert  v^{\prime} \Vert^2 
  + \frac{\eps^2}{\beta} \Big( \frac{8\,(\Lbeps+1)^2}{\pi^2} + 1 \Big) \Vert v \Vert_{\Veps, \Obeps}^2 \\ 
  &\le& \frac{4\, (\Lbeps+1)^2\eps^2}{\pi^2}\, \max\big\{ 1, \tfrac{5}{\beta} \big\}\, \Vvert v \Vvert_\eps^2, 
\end{eqnarray*}
where we have used $\frac{\pi^2}{4\,(\Lbeps+1)^2} \le 3$. 
Note that we have applied estimate~\eqref{eq:cutoff}, which introduced the parameter~$\Lbeps$. We can now formulate the following conclusion, which establishes a Friedrichs inequality for estimating the $L^2$-norm of a function by the ($\Veps$-dependent) energy norm.
\begin{conclusion}[inverse energy estimate]
With the constant~$C := \frac{4}{\pi^2} \max\{1, \frac{5}{\beta}\}$, which is independent of $\eps$, we have
\begin{align}
\label{estimate_L2norm}
  \Vert v \Vert^2
  \le C\, (\Lbeps+1)^2\, \eps^2\, \Vvert v \Vvert_\eps^2 
\end{align}
for all $v\in \V$.
\end{conclusion}
Note that a direct consequence of this estimate is that the eigenvalues of the linear Schr\"odinger operator~$\calA_\eps$ satisfy
\begin{align*}
\lambda_{\eps,j} \ge \lambda_{\eps,1} \gtrsim (\Lbeps+1)^{-2}\, \eps^{-2} \quad \mbox{for all~$j\ge 1$.}
\end{align*}
The estimate implies that if the potential takes values larger or equal to $\beta\eps^{-2}$ sufficiently often (i.e., if $\Lbeps$ is small enough), then the smallest eigenvalue is -- depending on $\Lbeps$ -- at least of order $\calO(\eps^{-p})$ for some $0<p\le 2$. 
%
%
\subsection{Abstract localization}\label{sect:linear:precond}
In this section, we take a closer look at the inverse $\calA^{-1}_{\eps}$ of the considered differential operator and prove that it almost maintains locality if $\Lbeps$ does not become too large, i.e., if $\Veps$ takes values of order $\beta\eps^{-2}$ in a significantly large subregion. In particular, we have that for any local function~$f \in L^2(D)$ the response~$\calA^{-1}_{\eps}f$ is quasi-local, i.e., exponentially decaying outside of~$\supp f$, provided that $\Lbeps$ is sufficiently small. 
This result is related to the decay properties of the Green's function associated with $\calA_{\eps}$. To prove this, we follow the arguments presented in \cite{AltHP20} and apply the theory of optimal local operator preconditioners~\cite{KorY16,KorPY18}. 
For that, we introduce an overlapping domain decomposition, which is related to the underlying $\eps$-mesh~$\calTeps$. This decomposition will be the basis for the definition of an operator preconditioner, which maintains locality. 

We denote the set of {\em interior} nodes by
\[
\calN^{\eps} := \{ z_j = j \eps \hspace{3pt} | \hspace{3pt} 0 < j < \eps^{-1}, \enspace j \in \mathbb{N} \}
\]
and define one subdomain for each node. For this, let~$\Lambda_z$ be the standard hat-function corresponding to a node $z\in\calN^{\eps}$, i.e., $\Lambda_z \in C^0(D)$ is the globally continuous and piecewise linear function with $\Lambda_z(z) = 1$ and $\Lambda_z(\tilde z) = 0$ for any other node $\tilde z\in\calN^{\eps}\setminus\{z\}$. The corresponding subdomain~$D_z^{\eps}$ is defined as the support of~$\Lambda_z$, which equals the union of the two subintervals in~$\calTeps$ which contain the node $z$. Hence we have~$\diam(D_z^{\eps}) = 2\eps$ for all $z\in\calN^{\eps}$. Note that this decomposition is independent of the particular potential~$\Veps$.

This decomposition motivates the definition of the local subspaces~$\V_z^{\eps} := H^1_0(D_z^{\eps})$ for all~$z\in\calN^{\eps}$. 
Note that all these spaces are naturally embedded in $\V=H^1_0(D)$ by the trivial continuation by zero. We also define local projections~$\Pz\colon \V\to \V_z^{\eps}$ w.r.t.~the bilinear form $a_\eps\colon \V\times\V\to \R$, 
\[
  a_\eps(u,v)
  := \langle \calA_\eps u, v \rangle
  = \int_D u'(x)v'(x) + \Veps(x)\, u(x)v(x) \dx, 
\]
by $a_\eps(\calP_{\eps,z} u, v_z) = a_\eps(u, v_z)$ for all $v_z\in\V_z^{\eps}$. Due to the mentioned embedding~$\V_z^{\eps}\hookrightarrow\V$, we can define the sum of all projections, namely $\calP_\eps\colon\V\to\V$, $\calP_\eps := \sum\nolimits_{z\in\calN^{\eps}} \Pz$. It is easily seen that the operator $\calP_\eps$ is local in the sense that it can only increase the support of a function by at most two~$\eps$-layers. Thus, \quotes{information} can only propagate in distances of order~$\eps$. 
\begin{proposition}[locality, cf.~{\cite{AltHP20}}]
\label{prop:local} 
The operator~$\calP_\eps$ maintains locality in the sense that 
\[
  \supp (\calP_\eps v) \subseteq B_{2\eps} \big( \supp v \big), \qquad
  \supp (\calP_\eps \calA_\eps^{-1}f) \subseteq B_{2\eps} \big( \supp f \big)
\]
for all $v\in\V$ and $f\in L^2(D)$. Here, $B_r(\omega) := [\omega - r, \omega + r] \cap D$ denotes the set of points, which have a distance to~$\omega$ smaller or equal to~$r$. 
\end{proposition}	
\begin{remark}
A multiple application of the operator~$\calP_\eps$ yields the locality estimate~$\supp (\calP_\eps^k v) \subseteq B_{(k+1)\eps} ( \supp v )$ for all $v\in\V$. 
\end{remark}
Based on the abstract theory for additive subspace correction methods for operator equations~\cite{KorY16} we show that~$\calP_\eps$ can be used to define an optimal preconditioner for the linear Schr\"odinger operator (if scaled accordingly). 
The precise statement reads as follows. 
\begin{proposition}[optimal preconditioner, cf.~{\cite[Th.~3.6]{AltHP20}}]
\label{prop:preconditioner}
There exists a scaling factor $\vartheta_\eps>0$ and a positive constant $\gamma_{\calP_\eps} < 1$ such that  
\[
  \Vvert \id - \vartheta_\eps {\calP_\eps}\Vvert_\eps
  := \sup_{v\in\V} \frac{\Vvert v - \vartheta_\eps\calP_\eps v\Vvert_\eps}{\Vvert v\Vvert_\eps}
	\le \gamma_{\calP_\eps} < 1.
\]
More precisely, we have $\vartheta_\eps = 1/(2+\K^{-1})$ and $\gamma_{\calP_\eps} \le 2/(2+\K^{-1})$ with constant~$\K := 4 + \tfrac{16}{\pi^{2}} \max\{1, \tfrac{5}{\beta}\} (\Lbeps+1)^2$, which which only depends on $\eps$ through~$\Lbeps$.
\end{proposition}	
%
%
%

Before we proceed, let us briefly discuss how Proposition~\ref{prop:preconditioner} can be used to infer locality. Based on $\calP_\eps$, we can define the preconditioner $\tilde{\calP}_\eps:=\vartheta_\eps \calP_\eps \calA_\eps^{-1}$ and consider the preconditioned system
$$\tilde{\calP}_\eps \calA_\eps u_{\eps} = \tilde{\calP}_\eps f$$
for a given local function~$f$. This system can be solved with the simple fixed-point iteration
\begin{align}
\label{fixed-point-iteration}
  u^{(k)}_{\eps} := \tilde{\calP}_\eps f + (\id - \tilde{\calP}_\eps \calA_\eps )\,  u^{(k-1)}_{\eps} =  \tilde{\calP}_\eps f + (\id - \vartheta_\eps \calP_\eps )\,  u^{(k-1)}_{\eps}, 
\end{align}
starting with $u^{(0)}_{\eps} = 0$. Note that the first contribution $\tilde{\calP}_\eps f=\vartheta_\eps \calP_\eps u_{\eps}$ is a local function, because we have
\begin{align*}
\tilde{\calP}_\eps f 
= \vartheta_\eps \hspace{2pt}\sum_{z \in \calN^{\eps}} \Pz (\calA_\eps^{-1} f),
\end{align*}
where $\Pz (\calA^{-1}_\eps f) \in \V_z^{\eps}$ is defined as the solution to the local problem 
$$
a_\eps(\hspace{2pt} \Pz (\calA_\eps^{-1} f) \hspace{2pt} , v_z ) 
= (f , v_z ) \qquad \mbox{for all } v_z \in \V_z^{\eps}.
$$
Hence, if $f$ has no support in the domain of $\V_z^{\eps}$, then $\Pz (\calA^{-1}_\eps f)$ is equal to zero. This ensures locality of $\tilde{\calP}_\eps f$, provided that~$f$ is a local function in the first place. Together with the locality properties of $\calP_\eps$ (cf.~Proposition \ref{prop:local} and the following remark) we conclude that the fixed-point iteration~\eqref{fixed-point-iteration} increases the support in each iteration step by only at most one $\eps$-layer around the support of the previous iterate. To be precise, every connected component of the support grows by one~$\eps$-layer at each end point, hence by a distance of at most~$2\,\eps$. 
Consequently, the generated sequence is local with $\supp(u^{(k)}_{\eps}) \subseteq B_{(k+1)\eps}(\supp f)$. 

The essential question is now, how many iteration steps~$k$ are needed so that $u^{(k)}_{\eps}$ approximates~$u_{\eps}$ up to a given accuracy? Here we can use Proposition~\ref{prop:preconditioner} together with the iteration \eqref{fixed-point-iteration} to see that
\begin{align}
  \Vvert u_{\eps} \Vvert_{\eps, D\setminus B_{(k+1)\eps}(\supp f)} 
  \le \Vvert u_{\eps} - u^{(k)}_{\eps} \Vvert_\eps 
  &= \Vvert (\id - \vartheta_\eps \calP_\eps )\,  (u_{\eps} - u_{\eps}^{(k-1)}) \Vvert_\eps \notag \\
  &= \Vvert (\id - \vartheta_\eps \calP_\eps )^k\,  u_{\eps}  \Vvert_\eps \label{support-growth}\\
  &\le \gamma_{\calP_\eps}^k \hspace{2pt} \Vvert u_{\eps} \Vvert_\eps. \notag
\end{align}
This estimate reveals that the locality of $u_{\eps}$ essentially depends on how well the operator $\calP_\eps$ is suited as a preconditioner, which is expressed through the size of the contraction factor $\gamma_{\calP_\eps} <1$. Ideally, $\gamma_{\calP_\eps}$ should not depend on $\eps$ (or at most in a weak, i.e., logarithmic, way). 
To emphasize this aspect, we shall make two examples for how the structure and the strength of~$\Veps$ influence the size of $\gamma_{\calP_\eps}$.
\begin{example}[weak potential]
In the case where the large values of the potential are scaled with~$\eps^{-1}$ instead of $\eps^{-2}$ we would have $\beta=\eps$ rather than $\beta=\calO(1)$. As a result, we have~$\K \simeq \eps^{-1} (\Lbeps+1)^2$ and the inverse $\K^{-1}$ scales as~$\eps$ (or even higher orders of~$\eps$), which in turn implies for the rate in Proposition~\ref{prop:preconditioner} that
\[
  \gamma_{\calP_\eps} 
  \simeq \frac{2}{2+\K^{-1}}
  \simeq 1 - \eps.
\]
Hence, the operator~$\calP_\eps$ is not suited as preconditioner, since~$k=\calO(\eps^{-1})$ steps would be necessary to reach a reasonable error reduction. Consequently, estimate~\eqref{support-growth} only guarantees smallness of $\Vvert u\Vvert_{\eps, D\setminus B_{r}(\supp f)}$ for some radius $r \simeq 1$, which does not allow to conclude any locality, since the domain $D\setminus B_{r}(\supp f)$ is (almost) vanishing. This example illustrates the requirement of the parameter $\beta$ to be independent of $\eps$. 
\end{example}
\begin{example}[strong potential]
\label{example-strong-potential}
Let us now assume the situation of a potential with~$\beta=5$ and a sufficient amount of peaks in the sense that $\Lbeps = \calO(1)$. Let $\delta>0$ be a small tolerance. In this case, the constant~$\K = 4 + \tfrac{16}{\pi^{2}} (\Lbeps+1)^2$ is bounded independently of~$\eps$, which also leads to an $\eps$-independent rate~$\gamma_{\calP_\eps}$. Applying again estimate \eqref{support-growth}, we obtain that we have indeed locality (exponential decay) with
$$
  \Vvert u_{\eps} \Vvert_{\eps, D\setminus B_{(k+1)\eps}(\supp f)} 
  \le \delta \hspace{2pt} \Vvert u_{\eps} \Vvert_\eps
  \qquad \mbox{for } 
  k = |\log \delta\,| / |\log \gamma_\calP| = \calO(\log(\delta^{-1})).
$$
Note that $\beta=5$ already marks the optimum in the sense that $\K$ does not further improve for larger $\beta$. 
However, even for more realistic $\Lbeps = \calO(\log_{1/p}(\eps^{-1}))$ we obtain an acceptable rate~$\gamma_{\calP_\eps}$ and the above bound for $k = \calO(\log_{1/p}(\eps^{-1}) \log(\delta^{-1}))$.
\end{example}
Motivated by Example~\ref{example-strong-potential}, we shall formulate the following statistical assumption to guarantee an exponential decay of the Green's function: 
\begin{enumerate}[label={(A\arabic*)}]
\item\label{A1} The potential is sufficiently strong in a relevant subregion, i.e.,  
$$
\Lbeps 
\simeq \log_{1/(1-p_{\beta})}(\eps^{-1}) 
= \log_{1-p_{\beta}}\eps , 
$$
for some $\eps$-independent parameter $0 < p_{\beta} < 1$.  
Here we note that we intentionally do not write $\Lbeps 
\simeq \log(\eps^{-1})$ by hiding the factor $|\log(1-p_{\beta})|^{-1}$ in the $\simeq$-notation, because $p_{\beta}$ will later play an important role. 
\end{enumerate}
We would like to emphasize that assumption~\ref{A1} equals a growth condition on~$\Lbeps$ in the following sense. The size of regions without a peak should not grow faster than logarithmically in $\eps$. We shall revisit the assumption~\ref{A1} in more detail in Section \ref{subsection-loc-first-eigenstates} below to clarify the role of $p_{\beta}$.

In the above results, we can see that the randomness of the potential (or disorder) does not play a role to observe localization for elliptic problems with a local source term. Here it is only important that~$\Veps$ is sufficiently strong (i.e.~of order $\eps^{-2}$) in a sufficiently large region (see \ref{A1}). In particular, this includes the case of a constant potential~$\Veps\equiv \beta\eps^{-2}$. 

This changes, however, when we consider the corresponding eigenvalue problem~$\calA_\eps u_{\eps} = \lambda_{\eps} \calI u_{\eps}$. To stress the crucial difference, let us consider the inverse power method, i.e., the fixed-point iteration $v^{(k+1)}_{\eps} = \calA_\eps^{-1}\calI v^{(k)}_{\eps}$ with an additional normalization step for solving the eigenvalue problem. 
If the starting value $v^{(0)}_{\eps}$ is picked in a suitable way (it needs to contain a component from the eigenfunction $u_{\eps,1}$ to the first eigenvalue, i.e., $( v^{(0)}_{\eps} , u_{\eps,1}) \not= 0$), then it is known that $v^{(k)}_{\eps}$ converges linearly to the (normalized) first eigenfunction with rate~$\lambda_{\eps,1} / \lambda_{\eps,2}$, where $\lambda_{\eps,1}$ is the first eigenvalue of the problem and $\lambda_{\eps,2}$ the second. When extending the inverse iteration to a $\tilde{\calP}_\eps$-preconditioned inverse iteration (also called PINVIT, cf.~\cite{DyaO80,BraKP95}), then each iteration step maintains locality (up to an $\eps$-layer per step) and the effective final convergence rate is of order $\calO(\tfrac{\lambda_{\eps,1}}{\lambda_{\eps,2}} + \gamma_{\calP_\eps})$, cf.~\cite{AltHP20} for details. Here we can see that in order to guarantee the localization of the first eigenspace, we do not only require that~$\gamma_{\calP_\eps}$ is sufficiently small, but also a significant spectral gap~$\lambda_{\eps,1} / \lambda_{\eps,2}$. This is exactly the point where the disorder of~$\Veps$ becomes relevant, since smallness of spectral gaps can be typically only guaranteed if the behavior of~$\Veps$ is irregular.

However, often it is not possible to ensure a relevant spectral gap after the first eigenvalue, but only after the first $K$ eigenvalues (where $K$ should be ideally not too large). Therefore it is practically (and theoretically) necessary to generalize PINVIT to a block version, which also involves that the preconditioner $\tilde{\calP}_\eps$ might be applied several times in each step. For the precise construction we refer to \cite{AltHP20}. Using the strategy sketched above, it is possible to prove the following result for the linear case.
\begin{theorem}[abstract localization result for the linear case]
\label{abstract-localization-result-eigenfunctions}
Consider a family of rapidly oscillating potentials as introduced in Section~\ref{sect:linear:potential} and assume~\ref{A1}. 
We consider the first $M_{1,\eps} \in \mathbb{N}$ eigenvalues to the linear problem $\calA_\eps u_{\eps} = \lambda_{\eps} \calI u_{\eps}$, where the eigenfunctions $u_{\eps,j} \in \V$ (corresponding to an eigenvalue~$\lambda_{\eps,j}$) are normalized in $L^2(D)$. Let $M_{2,\eps}>M_{1,\eps}$ be sufficiently large so that we have a relevant spectral gap between the $M_{1,\eps}$\!'th and $M_{2,\eps}$\!'th eigenvalue, in the sense that the gap fulfills
$$
\frac{\lambda_{\eps, M_{1,\eps}}}{\lambda_{\eps,M_{2,\eps}}} \le \frac{1}{2}. 
$$
Then the modulii of the first $M_{1,\eps}$ eigenfunctions are exponentially decaying to zero, outside of an area of size $\calO(M_{2,\eps} \hspace{2pt}\eps \hspace{2pt} \plog(\eps^{-1}))$. To be more precise, there exists a constant~$C_{\mbox{\rm \tiny spec},\eps}$, which depends on the first $M_{1,\eps}$ eigenfunctions and grows at most polynomially with $\eps^{-1}$, such that for any given tolerance $\tol >0$ and $k \gtrsim \log(\tol^{-1}) + \plog(\eps^{-1})$ there exist functions $w_{\eps,j}^{(k)}$, $1\le j \le M_{1,\eps}$, which have support on an area of size $\calO(M_{2,\eps}\, \eps\, k^2)$ and satisfy 
$$
\Vvert u_{\eps,j}  \Vvert_{\eps, D \setminus \supp w_{\eps,j}^{(k)} }   
\le \Vvert u_{\eps,j} - w_{\eps,j}^{(k)} \Vvert_\eps 
\le C_{\mbox{\rm \tiny spec}, \eps} \tol. 
$$
\end{theorem}
To better understand this abstract localization result let us consider the situation where $M_{1,\eps}=M_1$ and $M_{2,\eps}=M_2$ are independent of~$\eps$ with~$M_1<M_2=\calO(1)$ and a given tolerance~$\tol := e^{-k}$, $k>0$. Then Theorem~\ref{abstract-localization-result-eigenfunctions} states that for any eigenfunction $u_{\eps,j}$, $j\le M_1$, there exists a function $w_{\eps,j}^{(k)}$ with~$\Vvert u_{\eps,j} - w_{\eps,j}^{(k)} \Vvert_\eps \lesssim e^{-k}$. Moreover, the area of the support of $w_{\eps,j}^{(k)}$ is -- up to logarithmic terms -- bounded by~$\calO(\eps\, k^2)$. Hence, the exponential decay property of the first $M_1$ eigenfunctions follows from the fact that $u_{\eps,j}$ can be well approximated by a local function. An exponential decay in the sense of Theorem~\ref{abstract-localization-result-eigenfunctions} will be called {\em exponential localization}. 

As we can see from Theorem \ref{abstract-localization-result-eigenfunctions}, there are two important mechanisms that ensure locality of eigenstates. First, we ensure by~\ref{A1} that~$\Veps$ is sufficiently large in order to guarantee the decay of the Green's function, which leads to a small $\gamma_{\calP_\eps}$. Second, we require a sufficiently large spectral gap in the lower part of the spectrum, i.e., $M_{2,\eps}$ must not become too big so that the size of the localization region fulfills $M_{2,\eps}\, \eps\, \plog(1/\eps) \rightarrow 0$ for $\eps\rightarrow 0$. As already mentioned above, this second property is closely related to disorder as we will see in the next subsection.
%
%
\subsection{Localization of first $\calO(\eps^{-1})$ eigenstates}\label{subsection-loc-first-eigenstates}
With the help of the previously presented abstract localization result in Theorem~\ref{abstract-localization-result-eigenfunctions} we now show that the first~$\calO(\eps^{-1})$ eigenstates of the linear Schr\"odinger eigenvalue problem localize in disorder potentials. Note that there is no clear separation of localized and global eigenstates. One possible measure of localization is the $L^1$-norm of an eigenfunction. Since the eigenfunctions are normalized in~$L^2(D)$, it clearly holds $\| u_{\eps} \|_{L^1(D)} \le 1$. On the other hand, an eigenfunction that takes globally a small value cannot fulfill the condition $\| u_{\eps} \|_{L^2(D)} = 1$. Hence, a small~$L^1$-norm means that the function is concentrated in a small subdomain. 
A numerical investigation of localization is given in Figure~\ref{fig:localization}. Therein one can observe that a large amount of eigenstates have a small $L^1$-norm. As an example, for~$\eps=2^{-12}$ we observe more than~$\frac 12\, \eps^{-1}$ eigenstates with an $L^1$-norm smaller than~$0.1$. 
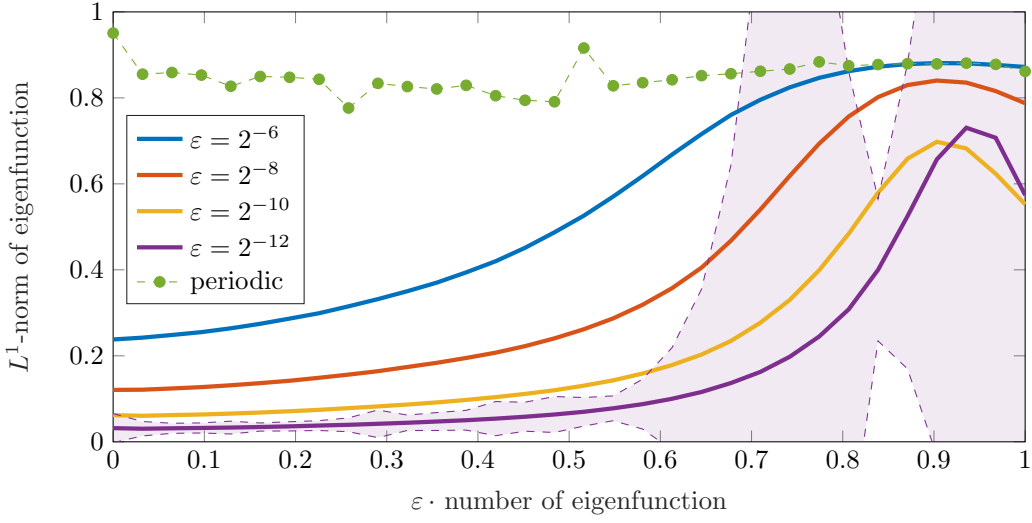
\begin{figure}
\centering
\input{pics/fig_localization_10k}
\caption{Illustration of the localization of the first $\calO(\eps^{-1})$ eigenstates for a piecewise constant potential~$\Veps$ with uniformly distributed values in the range~$[0,V_\text{max}\eps^{-2}]$ with~$V_\text{max}=5$. The plot shows the~$L^1$-norm of the first $\eps^{-1}$ eigenstates, averaged over 10.000 draws and normalized in $L^2(D)$. A small $L^1$-norm indicates localization and can be observed for a large range of eigenfunctions. 
Further, the plot includes the periodic case (no localization) and the $95\%$ confidence interval, both for $\eps=2^{-12}$. }
\label{fig:localization}
\end{figure}

Recall that a {\em valley} denotes an interval that forms a (maximal) connectivity component of~$\Oaeps$ and that its diameter is always a multiple of~$\eps$. To be precise, a (connected) interval $I_{\alpha}^{\eps} \subseteq \Oaeps$ is a valley if $\sup\nolimits_{x\in I_{\alpha}^{\eps}} \Veps(x) \le \alpha\, \eps^{-2} |\log\eps|^{-2}$ and if there is no other connected subset of $\Oaeps$ that contains $I_{\alpha}^{\eps}$. By $N_{\alpha,\ell}^{\eps}$ we denote the number of valleys with width~$\ell\eps$, $\ell \le \Laeps$. As a statistical assumption, we expect that this number fulfills
\begin{enumerate}[label={(A\arabic*)},resume]
\item\label{A2} 
$N_{\alpha,\ell}^{\eps} \simeq p_\alpha^{\ell}\, \eps^{-1}$ 
for some $\eps$-independent parameter $0 < p_{\alpha} < 1$.
\end{enumerate}
In the following, we shall add this to our general set of assumptions and note that the parameters $p_{\beta}$ in \ref{A1} and $p_{\alpha}$ in \ref{A2} are typically related. 
Assumption~\ref{A2} corresponds to a Bernoulli distribution of the potential, where the probability of a subinterval~$\overline{T^\eps}\in \calTeps$ being an element of~$\calTeps_\alpha$ equals~$p_\alpha$, cf.~Example~\ref{exp:Bernoulli}. 

On the other hand, we may define the probability of an element being part of $\calTeps_\beta$ by~$p_\beta$. Note that such a probability $p_\beta$ would correspond to the same parameter appearing in assumption~\ref{A1}. Accordingly, we define {\em non-peaks} as (maximal) connectivity components of~$D\setminus\Obeps$. 
The number of such non-peaks with width~$\ell\eps$, $\ell \le \Lbeps$ is consequently assumed to satisfy 
\begin{enumerate}[label={(A\arabic*)},resume]
\item\label{A3} 
$N_{\neg\beta,\ell}^{\eps} \lesssim (1-p_\beta)^{\ell}\, \eps^{-1}$ for the parameter~$p_{\beta}$ from \ref{A1}.  
\end{enumerate}

For the localization of eigenstates we need sufficiently many peaks, which is equivalent to $\Lbeps$ being small enough and particularly not in the range of~$\eps^{-1}$. 
Recall that this was already an assumption in the abstract localization result in Theorem~\ref{abstract-localization-result-eigenfunctions}. In fact, for Bernoulli distributions we have 
\[
  \Lbeps 
  \simeq \log_{1/(1-p_\beta)} (\eps^{-1})
  = \log_{1-p_\beta} \eps
\]
which corresponds to our previous assumption~\ref{A1}.
Further, we need sufficiently large valleys of varying size (which excludes the periodic case) in order to ensure suitable spectral gaps. For that we assume 
\begin{enumerate}[label={(A\arabic*)},resume]
\item\label{A4} $\Laeps \gtrsim \log_{p_\alpha} \eps$ for the parameter $p_{\alpha}$ as in \ref{A2}.   
\end{enumerate}
\begin{remark}
If the potential only contains the two values, then we have $p_\alpha = 1 - p_\beta$ and $\Laeps=\Lbeps$. The assumptions~\ref{A1}-\ref{A4} correspond to the statistical properties of a Bernoulli distribution of potential values and, hence, are satisfied in this case. 
\end{remark}
In order to apply the abstract localization result, we need to investigate spectral gaps. 
For this, we need to find upper and lower bounds. 
\begin{lemma}[spectral bounds from above]
\label{lem:upperBounds}
Consider a family of rapidly oscillating potentials as introduced in Section~\ref{sect:linear:potential} as well as assumptions~\ref{A2} and~\ref{A4}. 
Consider a parameter~$0<q_1\le1$ and set~$\ell_{1,\eps} := \lceil \log_{p_\alpha}(\eps^{q_1}) \rceil$. Then there exist~$M_{1,\eps}\simeq p_\alpha^{\ell_{1,\eps}-1} \eps^{-1} \simeq \eps^{q_1-1}$ disjoint (and thus orthogonal) functions~$v\in\V$ with~$\Vvert v\Vvert_\eps^2 \lesssim (\ell_{1,\eps} \eps)^{-2} \|v\|^2$. As a consequence, we have the upper eigenvalue bound
\[
  \lambda_{\eps,1}
  \le \lambda_{\eps,M_{1,\eps}}
  \lesssim  \frac{1}{(\ell_{1,\eps}\eps)^2}.
\]	
\end{lemma}
\begin{proof}
The idea of the proof is to construct shifted Laplace eigenfunctions in the largest valleys. Details are given in Appendix~\ref{app:EVbounds:upper}	
\end{proof}
\begin{lemma}[spectral bounds from below]
\label{lem:lowerBounds}
Consider a family of rapidly oscillating potentials as introduced in Section~\ref{sect:linear:potential} as well as assumptions~\ref{A1} and~\ref{A3}. 
Consider a parameter~$0<q_2<1$ and set~$\ell_{2,\eps} := \lceil \log_{1-p_\beta}(\eps^{q_2}) \rceil$. Then there exists~$M_{2,\eps}\lesssim q_2^{-1} \eps^{q_2-1}$ with 
\[
  \lambda_{\eps,M_{2,\eps}}
  \gtrsim \frac{1}{(\ell_{2,\eps}\eps)^2}.
\]	
\end{lemma}
\begin{proof}
The proof is given in Appendix~\ref{app:EVbounds:lower}. 
\end{proof}
\begin{corollary}[spectral gap and localization]
\label{cor:linearLocalizaton}
Consider a family of rapidly oscillating potentials as introduced in Section~\ref{sect:linear:potential} satisfying assumptions~\ref{A1}-\ref{A4}. Then, asymptotically, the first $\calO(\eps^{-1})$ eigenstates localize exponentially. 
\end{corollary}
\begin{proof}
For a fixed $0 < \delta \le 1$ and a gap parameter $0 < \mu < 1$, we combine the estimates of Lemmata~\ref{lem:upperBounds} and~\ref{lem:lowerBounds} with~$q_1:= \delta \le 1$ and $q_2:= \mu \delta < q_1 < 1$, where we assume that $\mu$ is small enough such that 
\[
  \big\lceil \log_{1-p_\beta}(\eps^{\mu \delta}) \big\rceil
  = \ell_{2,\eps}
  < \ell_{1,\eps} 
  = \big\lceil \log_{p_\alpha}(\eps^{\delta}) \big\rceil.
\]
This then leads to the estimate
\[
  \frac{\lambda_{\eps,M_{1,\eps}}}{\lambda_{\eps,M_{2,\eps}}}
  \lesssim \frac{\ell_{2,\eps}^2\, \eps^2}{\ell_{1,\eps}^2\, \eps^2}
  \simeq \left( \frac{\log_{1-p_\beta}(\eps^{\mu\delta})}{\log_{p_\alpha}(\eps^{\delta})} \right)^2
  = \left( \mu\ \frac{\log (p_\alpha)}{\log(1-p_\beta) } \right)^2.
\]
Hence, there is a constant $C>0$ that is independent of $\eps$, $\delta$, and $\mu$ such that
\begin{align}
\label{cor:linearLocalizaton:proof:est1}
  \frac{\lambda_{\eps,M_{1,\eps}}}{\lambda_{\eps,M_{2,\eps}}} \le C \left( \mu\  \frac{\log (p_\alpha)}{\log(1-p_\beta) } \right)^2 
\end{align}
and we see that for a sufficiently small parameter~$\mu$ we obtain a spectral gap of order $\calO(1)$. Note that the factor~$\log (p_\alpha)/ \log(1-p_\beta)$ somehow measures the amount of elements~$\overline{T^{\eps}}\in\calTeps$, which are neither part of~$\calTeps_\alpha$ nor~$\calTeps_\beta$. This factor vanishes in the special case of~$p_\alpha=1-p_\beta$.  

The detected spectral gap appears between the first~$M_{1,\eps}\simeq \eps^{\delta-1}$ and $M_{2,\eps}\lesssim (\mu \delta)^{-1} \eps^{\mu \delta-1}$ eigenvalues. Thus, the abstract localization result of Theorem~\ref{abstract-localization-result-eigenfunctions} shows the first~$M_{1,\eps}$ eigenstates essentially localize in a region, which is determined by~$M_{2,\eps}$ and the preconditioner introduced in Section~\ref{sect:linear:precond}. 
As shown in~\cite{AltHP20}, the diameter of this region is bounded (up to logarithmic terms) by~$M_{2,\eps}\,\eps \lesssim (\mu \delta)^{-1} \eps^{\mu \delta}$. Thus, although~$M_{2,\eps} \rightarrow \infty$ for $\eps \rightarrow 0$, the fraction that is asymptotically covered is of size $\calO(\eps^{\mu\delta})$, which also goes to zero. Assume that~$\mu$ is sufficiently small such that $C\, \big( \mu\, \frac{\log (p_\alpha)}{\log(1-p_\beta) } \big)^2\le \frac{1}{2}$ in \eqref{cor:linearLocalizaton:proof:est1}. Then $M_{1,\eps}\simeq \eps^{-1 + \delta}$ eigenfunctions that are exponentially decaying outside of an area of size $ \calO(\mu^{-1} \delta^{-1} \, \eps^{\mu \delta} \hspace{2pt} \plog(\eps^{-1}))$, which tends to vanish for $\eps\rightarrow 0$. Since~$\mu$ is independent of $\eps$ and $\delta$, we see that there exists some $\eps_0(\delta)>0$ so that for each $\eps \le \eps_0(\delta)$ we have $\calO(\eps^{-1 + \delta})$ exponentially localized eigenfunctions. Since $\delta$ can be selected arbitrary small, the number of  localized eigenfunctions approaches asymptotically  $\calO(\eps^{-1})$, cf.~the numerical study in Figure~\ref{fig:localization} for an illustration. 
\end{proof}
In this section, we have seen that a disorder potential leads to localized states. For this, it was crucial that $\eps\ll 1$ (oscillatory potential), that $\Lbeps$ only depends logarithmically on~$\eps$ (high amplitude potential), and that there exist potential valleys of different sizes with~$\Laeps$ sufficiently large (disorder potential). 

The next section considers the GPEVP and intends to show that the here discussed localization carries over to the nonlinear case if the parameter~$\kappa$ is not too large (i.e. at most of order $\calO(\eps^{-1})$). 
%
%
\section{Localization of Gross-Pitaevskii Ground States}\label{sect:GPEVP}
In this main part of the paper, we analyze the localization of the first eigenstates of the GPEVP in disorder potentials as introduced in Section~\ref{sect:linear:potential}. 
In particular, we prove that the ground state of the GPEVP localizes if the interaction parameter satisfies~$\kappa \lesssim \eps^{-1}$, which is according to the natural scaling that we derived in Section \ref{subsec-dimless-1D}. Numerical experiments in Section~\ref{subsect:numerics:higherdim} will then indicate that this condition is also necessary, i.e., if $\kappa$ is becoming too large, then the ground state delocalizes. 
%
\subsection{Alternative characterization of the ground state}
We consider the ground state~$\ugs$ of the GPEVP~\eqref{eq:GPEVP} with repulsive particle interactions, i.e., for $\kappa \ge 0$. 
Mathematically, the ground state ~$\ugs \in \V$ with normalized mass, i.e., $\| \ugs\| =1$, is defined as an eigenfunction to the smallest eigenvalue $\lambdags>0$ with
\begin{align}
\label{eq:GPEVP-var}
- \tfrac 12\,\ugs^{\prime\prime}  + \Veps \hspace{2pt} \ugs + \kappa\, |\ugs|^2 \ugs 
= \lambdags  \ugs. 
\end{align}
Due to the symmetric structure of the GPEVP~\eqref{eq:GPEVP-var} it is well-known that the ground state~$\ugs$ is continuous, unique up to sign, and either strictly positive or strictly negative in~$D$. In the following, we make the silent convention that we always consider the (unique) positive ground state. The above mentioned properties are general and hold for large classes of nonlinear Schr\"odinger eigenvalue problems in $\R^d$, cf.~\cite{CCM10,HenP18ppt}. 

Equivalently, the ground state $\ugs$ may be also characterized as the minimizer of the total energy. This characterization will be important for our analysis. To make this statement precise, we define the energy of a function~$v\in\V$ (in the nonlinear setting of the GPEVP) as 
\[
  E(v) 
  := \int_D \frac 12\, | v'(x) |^2 + \Veps(x)\, |v(x)|^2 + \frac{\kappa}{2}\, |v(x)|^4 \dx. 
\]
With this, the ground state $\ugs \in \V$ is the (unique) positive minimizer of the above energy in the affine space~$\{ v\in\V\ |\ \|v\| = 1 \}$, cf.~\cite{CCM10}. Thus, any normalized function $v\in \V$ satisfies~$E(\ugs) \le E(v)$. 

Last but not least, we shall exploit yet another characterization of the ground state, which will be crucial in the following analysis, as it allows us to reduce the nonlinear problem to a linear problem as studied before. In fact,~$\ugs \in \V$ (with $\| \ugs \|=1$) is equal to the ground state of the following linearized eigenvalue problem (cf.~\cite[Lem.~2]{CCM10})
\begin{align}
\label{eq:linearizedGPEVP}
  - \tfrac 12\, \ulin'' + \Veps\, \ulin + \kappa\, |\ugs|^2 \ulin = \lambdalin \ulin.
\end{align}
Note that~\eqref{eq:linearizedGPEVP} equals a linear Schr\"odinger eigenvalue problem~\eqref{eq:linSchroedingerEVP} with perturbed potential~$V_{\eps,\text{gs}} := \Veps + \kappa\, |\ugs|^2$. In order to apply the localization results of the linear case of Section~\ref{sect:linear}, we show that~$V_{\eps,\text{gs}}$ satisfies the required properties of being oscillatory with high amplitudes and different valley sizes. To put it differently, we intend to show that~$\kappa\, |\ugs|^2$ only defines a small perturbation, which does not affect the overall structure of the potential. For that we first need to derive bounds for $\ugs$ in the maximum norm. As the crucial assumption for the localization in the nonlinear case, we need to ensure that $\kappa$ does not become too large. 
%
\subsection{Boundedness of the ground state}
In this subsection, we show that the normalized ground state of the GPEVP satisfies~$|\ugs| \lesssim \eps^{-1/2}$ pointwise in~$D$. 
The proof will make use of the following result on secants. 
\begin{lemma}
\label{lem:secant}
Consider a continuously differentiable function $f\colon [a,b] \rightarrow \R$ with $f(a)=f_a$ and $f(b)=f_b$. Then, it holds that 
\begin{align*}
  \|f'\|^2_{L^2(a,b)}	
  = \int_a^b |f'(x)|^2 \dx 
  \ge \frac{(f_b-f_a)^2}{b-a}.
\end{align*}
\end{lemma}
\begin{proof}
First, we observe that for any constant $c\in\R$ we have
\[
  \int_{a}^b |f'(x) - c\,|^2 \dx 
  = \int_{a}^b |f'(x)|^2 \dx - 2c\, (f_b-f_a) + c^2(b-a).
\]
Hence, for~$\calG := \{ g \in C^{1}([a,b])\ |\ g(a)=f_a,\ g(b)=f_b\}$ we have 
\[
  \mbox{arg min}_{g\in\calG} \int_a^b |g'(x)|^2 \dx 
  =	\mbox{arg min}_{g\in\calG} \int_a^b |g'(x) - c|^2 \dx
\]
Now select $g\in\calG$ as the secant to $f$, i.e., $g(x) := s(x) = f_a + c_s (x - a)$ with~$c_s := \frac{f_b-f_a}{b-a}$. The property 
\[
  \int_{a}^b |s'(x) - c_s|^2 \dx 
  = \int_{a}^b |c_s - c_s|^2 \dx 
  = 0
\]
then implies that the integral is minimized for the secant and we have
\[
	\int_{a}^b |f'(x)|^2 \dx 
	\ge \min_{g\in\calG} \int_a^b |g'(x)|^2 \dx 
	= \int_{a}^b |s'(x)|^2 \dx 
	= c_s^2\, (b-a). \qedhere
\]
\end{proof}
Further, we will apply the well-known Sobolev embedding in one space dimension. 
\begin{lemma}[Sobolev embedding in 1D]
For $a < b$ and $u\in H^1(a,b)$ it holds that
\begin{align}
\label{eq:Sobolev}
	\| u \|_{L^{\infty}(a,b)} 
	\le \frac{\sqrt{2}}{\sqrt{b-a}} \| u \|_{H^1(a,b)}.
\end{align}
\end{lemma}
In order to obtain estimates on the~$L^\infty$-norm of the ground state, we first derive a rough estimate on the derivative of~$\ugs$. 
For this, we consider a polynomial bubble of low degree on $D$, which is normalized in $L^2(D)$, e.g., $p(x) := \sqrt{30}\, x(1-x)$. Obviously, $p\in\V$ and thus, assuming~$\max \Veps \lesssim \eps^{-2}$,   
\begin{align}
\label{eq:boundDerivativeGroundstate}
  \|\ugs'\|^2 
  \le 2\, E(\ugs) 
  \le 2\, E(p) 
  = \int_D |p'|^2 + 2\, \Veps\, |p|^2 + \kappa\, |p|^4 \dx
  \lesssim \eps^{-2} + \kappa.
\end{align}
On the other hand, we know that~$\|\ugs\| = 1$ such that the direct application of the Sobolev embedding~\eqref{eq:Sobolev} implies 
\[
  \|\ugs \|^2_{L^{\infty}(D)} 
  \le 2\, \|\ugs \|^2_{H^1(D)} 
  = 2 + 2\, \|\ugs'\|^2 
  \lesssim \eps^{-2} + \kappa.
\]
Thus, under the assumption that~$\kappa \lesssim \eps^{-2}$ we obtain the upper bound~$\|\ugs \|_{L^{\infty}(D)} \lesssim \eps^{-1}$. The following result shows that this bound can be improved significantly. 
\begin{theorem}
\label{thm:upperBoundGroundstate}
Assume~$\eps\ll 1$, $\max \Veps \lesssim \eps^{-2}$, and~$0 \le \kappa \lesssim \eps^{-2}$. Then there exists a constant $C>0$, independent of~$\eps$, such that the ground state $\ugs\in\V$ of the GPEVP can be bounded by
\begin{align*}
\|\ugs \|_{L^{\infty}(D)} \le C \, \eps^{-1/2}.
\end{align*}
\end{theorem} 
\begin{proof}
Let $C_{0}>0$ denote the hidden constant in~\eqref{eq:boundDerivativeGroundstate} so that~$\|\ugs'\|^2 \le C_{0} \hspace{2pt} \eps^{-2}$. With this, we define
\begin{align}
\label{assumption-on-delta}
\delta:=\delta(\eps)=\log(\tilde{C}_0) \hspace{2pt} |\log(\eps)|^{-1}, \qquad \mbox{where } \tilde{C}_0:= 2 \max\{2 , \sqrt{C_0}\}.
\end{align}
Note that $\delta \rightarrow 0$ for $\eps \rightarrow 0$ and
\begin{align}
\label{eps-to-delta}
\eps^{-\delta} = \tilde{C}_0.
\end{align}
We show the result by contradiction and assume that there exists a point~$x_\delta\in D$ with~$\ugs(x_\delta) \ge \eps^{-1/2-\delta}$. In this case, we define the (non-empty) set 
\[
  D_\delta 
  := \big\{ x \in D\ |\ \eps^{-1/2} \le \ugs(x) \le \eps^{-1/2-\delta} \big\}.
\]
We now show that this is an empty set and hence, by the continuity of~$\ugs$, we have a contradiction to the assumption that  $\ugs(x_\delta) \ge \eps^{-1/2-\delta}$.  

For $D_\delta $, we first observe that
\begin{align*}
  1 
  = \|\ugs\|^2
  \ge \| \ugs\|^2_{L^2(D_\delta)} 
  \ge |D_{\delta}|\, \eps^{-1}, 
\end{align*}
which implies $|D_{\delta}| \le \eps$. 
Due to the assumption~$\ugs(x_\delta) \ge \eps^{-1/2-\delta}$ there exists a subinterval of~$D_\delta$ on which~$\ugs$ actually takes the values~$\eps^{-1/2}$ and $\eps^{-1/2-\delta}$. 
Thus, by Lemma~\ref{lem:secant} we conclude that 
\[
  \|\ugs'\|^2_{L^2(D_\delta)}
  \ge \frac{(\eps^{-1/2-\delta} - \eps^{-1/2})^2}{|D_\delta|}
  = \frac{1}{\eps}\, \frac{(\eps^{-\delta} - 1)^2}{|D_\delta|}.
\]
Since $\delta>0$, we know that~$(\eps^{-\delta} - 1)^2 \gtrsim \eps^{-2\delta}$ asymptotically. More precisely, we have~$(\eps^{-\delta} - 1)^2 \ge \eps^{-2\delta} - 2\eps^{-\delta} \ge \frac 12 \eps^{-2\delta}$ for~$\delta\ge \log 4 / \log(1/\eps)$, which is guaranteed by our definition of $\delta$ in \eqref{assumption-on-delta}. 
With the upper bound on the size of~$D_\delta$, we obtain  
\[
  \|\ugs'\|^2_{L^2(D_\delta)}
  \ge \frac{\eps^{-1-2\delta}}{2\, |D_\delta|}
  \ge \frac 12\, \eps^{-2-2\delta}.
\]
Together with~$\|\ugs'\|^2 \le C_0 \hspace{2pt} \eps^{-2}$ from~\eqref{eq:boundDerivativeGroundstate}, we conclude that 
\begin{eqnarray*}
  C_0 \hspace{2pt} \eps^{-2} 
  \ge \|\ugs'\|^2
  \ge \|\ugs'\|^2_{L^2(D_\delta)}
  \ge \tfrac 12\, \eps^{-2-2\delta} \overset{\eqref{eps-to-delta}}{=}  \frac{\tilde{C}_0^2}{2} \hspace{2pt}\eps^{-2}
  \overset{\eqref{assumption-on-delta}}{\ge} 2\, C_0 \hspace{2pt} \eps^{-2}.
\end{eqnarray*}
This provides a contradiction. Thus, the set~$D_\delta$ is a null set for $\eps\ll 1$, which implies that there is no point with $\ugs(x_\delta) \ge \eps^{-1/2-\delta}$ and we have consequently for all $x\in D$
\[
\ugs(x) \le \eps^{-1/2-\delta} = \tilde{C}_0 \hspace{2pt} \eps^{-1/2}.
\qedhere
\]
\end{proof}
%
%
\subsection{Main result and proof of localization}
The previous result indicates that~$V_{\eps,\text{gs}} - \Veps = \kappa\, |\ugs|^2$ is a perturbation of order at most~$\calO(\eps^{-2})$, if $\kappa\lesssim \eps^{-1}$. Since this is the same order as the maximal values of the original potential~$\Veps$, we need to ensure that such large perturbations only occur rarely. 
\begin{lemma}
\label{lem:onlyO1subdomains}
Consider a family of rapidly oscillating potentials as introduced in Section~\ref{sect:linear:potential} with~$\max \Veps \lesssim \eps^{-2}$ and~$0\le \kappa \lesssim \eps^{-2}$. Then large values in the sense of~$|\ugs(x)| \gtrsim \eps^{-1/2}$ occur in at most~$\calO(1)$ subintervals~$\overline{T^{\eps}}\in\calTeps$ of the $\eps$-partition. 
\end{lemma}
\begin{proof}
From Theorem~\ref{thm:upperBoundGroundstate} we know that~$|\ugs| \lesssim \eps^{-1/2}$ but up to now we cannot shoot out the possibility the~$\ugs$ oscillates with high amplitudes.
We fix~
$\rho>0$
and consider a subdomain~$\overline{T^{\eps}_\rho} \in \calTeps$ on which~$\ugs$ reaches values in the full range of $\eps^{-1/2+\rho}$ and~$\eps^{-1/2}$, i.e., there exist $x_1, x_2\in T^{\eps}_\rho$ with $\ugs(x_1) \ge C\, \eps^{-1/2}$ and~$\ugs(x_2) \le C\, \eps^{-1/2+\rho}$. 
In this case, Lemma~\ref{lem:secant} implies 
\[
  \|\ugs'\|^2_{L^2(T^{\eps}_\rho)}
  \ge C^2\, \frac{(\eps^{-1/2} - \eps^{-1/2+\rho})^2}{\eps}
  \gtrsim \eps^{-2}.
\]
On the other hand, we know from~\eqref{eq:boundDerivativeGroundstate} that~$\|\ugs'\|^2 \lesssim \eps^{-2}$ such that this situation can only occur~$\calO(1)$ times. 
Further, if~$\ugs \simeq \eps^{-1/2}$ on an entire subdomain~$\overline{T^{\eps}} \in \calTeps$, then we have~$\|\ugs\|_{L^2(T^\eps)} \simeq 1$. Due to the normalization of the ground state, i.e., $\|\ugs\|=1$, also this can happen at most~$\calO(1)$ times. 
\end{proof}
We have seen that $\ugs$ can take at most values of order $\eps^{-1/2}$, and even if it takes such large values, this cannot happen too often. Knowing that large perturbations only affect a small amount of subdomains and small perturbations do not affect the statistical properties of the potential, we can now conclude that the ground state localizes in disorder potentials with high amplitudes. 
\begin{theorem}[localization of the ground state]
\label{thm:localizedGroundstate}
Consider a family of rapidly oscillating potentials as introduced in Section~\ref{sect:linear:potential} with~$\max \Veps \lesssim \eps^{-2}$. Further assume~\ref{A1}-\ref{A4} and~$\kappa \lesssim \eps^{-1}$. 
Then the ground state of the GPEVP is exponentially localized in the sense of Theorem~\ref{abstract-localization-result-eigenfunctions}. 
\end{theorem}
\begin{proof} 
We have seen that~$\ugs$ equals the ground state of the linear Schr\"odinger eigenvalue problem with the perturbed potential~$V_{\eps,\text{gs}} = \Veps + \kappa\, |\ugs|^2$, cf.~\cite[Lem.~2]{CCM10}. The idea of the proof is to show that~$V_{\eps,\text{gs}}$ still satisfies the assumptions~\ref{A1}-\ref{A4} with adjusted parameters such that the first eigenfunctions (and thus~$\ugs$) localize by Corollary~\ref{cor:linearLocalizaton}. 

By Lemma~\ref{lem:onlyO1subdomains} we know that~$\kappa\, |\ugs|^2 \gtrsim \eps^{-2}$ in at most~$\calO(1)$ subintervals of~$\calTeps$. Note that such a perturbation is significant in the sense that an element~$\overline{T^{\eps}}\in\calTeps_\alpha$ may switch to an element of~$\calTeps_\beta$. Due to the small number of such cases and~$\eps\ll 1$, this does not affect the statistical properties of the potential. 
%
On the other hand, perturbations of the order at most~$\eps^{-2+\delta}$, $\delta>0$, can be hold off by the adjustment of the characteristic parameters. Note that, in theory, even such \quotes{small} perturbations may lead to an empty set~$\calTeps_\alpha$ (for example if the original~$\Veps$ is a two-valued potential). To prevent this from happening, we define~$\beta_\text{gs} := \beta\, c_\text{gs}$ and~$\alpha_\text{gs} := \alpha\, c_\text{gs}^{-1}$ for some~$c_\text{gs} < 1$. With this, all elements in~$\calTeps_\alpha$ remain an element of~$\calTeps_{\alpha_\text{gs}}$ after a perturbation of order~$\eps^{-2+\delta}$. Since the potential is bounded by~$\alpha\, \eps^{-2}\, |\log\eps|^{-2}$ for regions corresponding to~$\calTeps_\alpha$, we can choose~$c_\text{gs}$ arbitrarily close to one such that the statistical assumptions~\ref{A1}-\ref{A4} remain valid.  
\end{proof}
\begin{remark}
The assumption~$\max \Veps \lesssim \eps^{-2}$ in Theorem~\ref{thm:localizedGroundstate} has been made for technical reasons only. For stronger potential we even expect a more significant localization behaviour of the ground state. 
\end{remark}
\begin{remark}[localization of excited states]
One can act on the assumption that the localization of eigenfunctions of the GPEVP is not restricted to the ground state alone.  
Let $\ues\in\V$ be an excited state. Then~$\ues$ is an eigenfunction of the linear eigenvalue problem 
\begin{align*}
  -\tfrac 12\, \ulin '' + \Veps\, \ulin + \kappa\, |\ues|^2 \ulin = \lambdalin \ulin.
\end{align*}
In contrast to the ground state~$\ugs$, however, we do not know whether~$\ues$ is a low-energy state or even the ground state of the linearized problem.  
Nevertheless, we can follow the argumentation from above if the excited state satisfies two properties. 
First, we need~$E(\ues) \lesssim \eps^{-2}$. 
Second, $\ues$ has to be one of the first $\calO(\eps^{-1})$ eigenstates of the linearized eigenvalue problem. 
If this is the case, then the arguments applied in Theorem~\ref{thm:upperBoundGroundstate}, Lemma~\ref{lem:onlyO1subdomains}, and Corollary~\ref{thm:localizedGroundstate} also apply for excited states. 
\end{remark}
%
%
\subsection{Conjecture in higher space dimensions}
\label{subsect:numerics:higherdim}
For dimensions~$d>1$ certain arguments such as the Sobolev embedding or Lemma~\ref{lem:secant} are not valid in the present form. Nevertheless, following the strategy of the proof of Theorem~\ref{thm:upperBoundGroundstate}, we come up with the following conjecture.
\begin{conjecture}
The ground state of the GPEVP satisfies~$\|\ugs\|_{L^\infty} \lesssim \eps^{-d/2}$.
\end{conjecture}
Assuming such an upper bound on the function values of the ground state, we conclude that~$V_{\eps,\text{gs}} - \Veps = \kappa\, |\ugs|^2$ is a perturbation of order at most~$\calO(\eps^{-2})$ as long as $\kappa \lesssim \eps^{d-2}$. 
An argumentation that such perturbations do not happen \quotes{too often} we end up with the following conjecture.
\begin{conjecture}
The ground state of the GPEVP with a disorder potential satisfying~\ref{A1}-\ref{A4} and~$\eps\ll 1$ is localized if~$\kappa\lesssim \eps^{d-2}$.
\end{conjecture}
This conjecture can be attested numerically, cf.~Figure~\ref{fig:conjecture:variantDim}. 
We consider once more the $L^1$-norm as an indication for localization. 
It can be observed that a certain plateau is reached in the predicted range of~$\kappa$, namely $\eps^{d-2}$. 
All numerical experiments are based on an equidistant mesh into intervals/squares/cubes with mesh size~$h$. We apply conforming~$Q_1$-finite elements~\cite[Sect.~3.5]{BrS08} for the spatial discretization in combination with the nonlinear eigenvalue solvers presented in~\cite{HenP18ppt,AltHP19ppt}. 
%
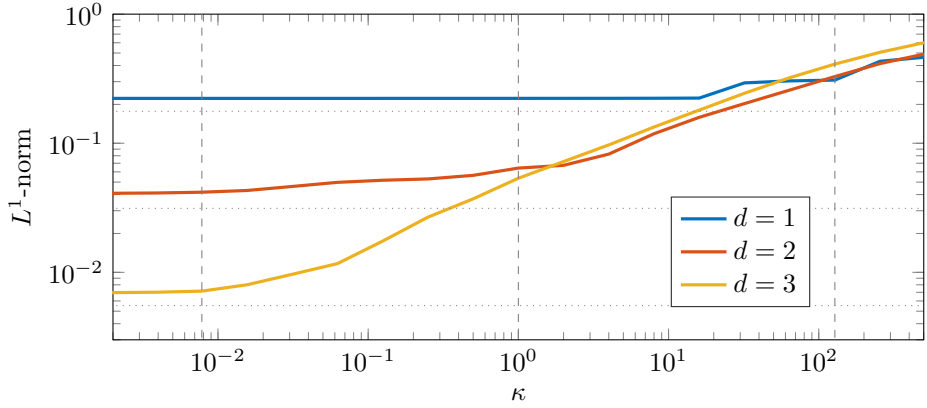
\begin{figure}
\input{pics/fig_L1vsKappa_eps7_Vmax5}
\caption{$L^1$-norms of the ground states for piecewise constant random potentials with~$\eps=2^{-7}$ and~$V_\text{max}=5$. The experiments are performed on a uniform partition with mesh size~$h=2^{-12}$ ($d=1$), $h=2^{-10}$ ($d=2$), and~$h=2^{-7}$ ($d=3$). The vertical dashed lines indicate~$\kappa = \eps^{d-2}$, whereas the horizontal dotted lines equal~$2^d\cdot \eps^{d/2}$.}
\label{fig:conjecture:variantDim}
\end{figure}

To explain the different levels of the plateau in different dimensions, we make the following consideration: For $d=1$ we have shown that the maximal value of~$\ugs$ is in the range~$\eps^{-1/2}$. Now we can construct a piecewise constant function, which is equal to~$\eps^{-1/2}$ in a small region and zero otherwise such that the~$L^2$-norm is equal to one. The corresponding $L^1$-norm then equals~$\eps^{1/2}$. 
Similar thoughts in higher dimension construct $L^2$-normalized and piecewise constant functions with an~$L^1$-norm equal to~$\eps^{d/2}$. 
This dependency of the $L^1$-norm on the dimension and on~$\eps$ can also be observed in Figure~\ref{fig:conjecture:variantEps}, which shows similar results from another perspective. 
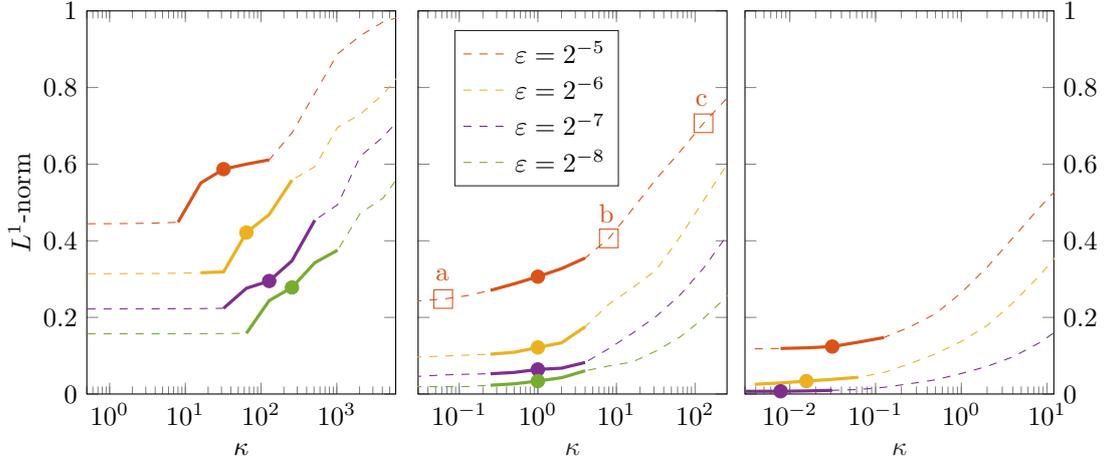
\begin{figure}
\input{pics/fig_L1vsKappa_d1_h12}\hspace{-0.2em}
\input{pics/fig_L1vsKappa_d2_h10}\hspace{0.3em}
\input{pics/fig_L1vsKappa_d3_h7}
\caption{$L^1$-norms of the ground states for piecewise constant random potentials with~$V_\text{max}=5$ and various values for~$\eps$. The dot indicates~$\kappa=\eps^{d-2}$ and the solid line the region~$[0.25\cdot \eps^{d-2}, 4\cdot \eps^{d-2}]$. The particular ground states marked with a square in the 2D case are shown in Figure~\ref{fig:conjecture:variantKap}.}
\label{fig:conjecture:variantEps}
\end{figure}
Since ground states are smooth, we cannot ask for such a perfectly localized function. Thus we expect slightly larger $L^1$ norms. In Figure~\ref{fig:conjecture:variantDim} one can see that the reached plateau is in the range~$2^d\cdot \eps^{d/2}$. 

Finally, we present a visualization of two-dimensional ground states in Figure~\ref{fig:conjecture:variantKap} showing the outcome for a single realization of a random potential with $\eps=2^{-5}$. 
\begin{figure}
\includegraphics[width=.23\textwidth]{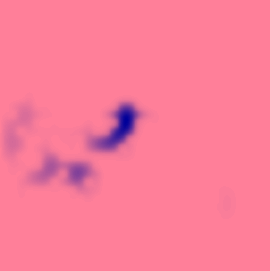}\hspace{0.2em}
\includegraphics[width=.23\textwidth]{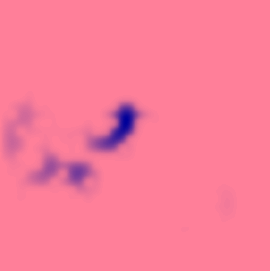}\hspace{0.2em}
\includegraphics[width=.23\textwidth]{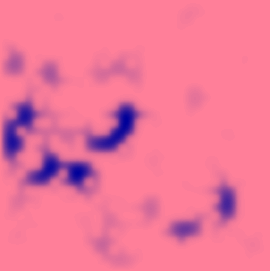}\hspace{0.2em}
\includegraphics[width=.23\textwidth]{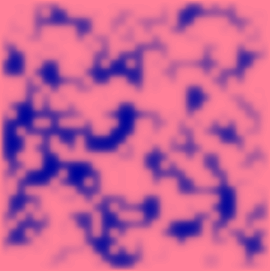}\\
{\color{mycolor2} \ \hspace{3.5cm}a\hspace{3.5cm}b\hspace{3.5cm}c }
\caption{Visualization of the ground states for $\eps=2^{-5}$, $V_\text{max}=5$, illustrating the level of localization for different interaction parameters~$\kappa = 0,\, 2^{-4},\, 2^{3},\, 2^{7}$ (from left to right), cf.~Figure~\ref{fig:conjecture:variantEps}.}
\label{fig:conjecture:variantKap}
\end{figure}
%
\section{Application to Physically Relevant Data and Mathematical Prediction}
\label{section:application-to-physical-values}

As mentioned in the introduction, a direct physical observation of Anderson-type localization for the ground states of BEC has not yet been achieved. In~\cite{SDK06} it is argued that this is because the localized density variations might not be visible in experimental measurements, due to the limitations of imaging optics. In this case, the localization has to be detected indirectly through other observables, such as a time-of-flight measurement for the velocity distribution.

In this section, we apply our localization result in Theorem~\ref{thm:localizedGroundstate} to realistic physical data and make predictions about when the localization of ground states can be expected. This might help to set up suitable practical experiments. 
We follow the descriptions of the experimental setup presented in~\cite{physics-paper-localization-BEC}, where dynamical Anderson-type localization could be observed (i.e.~the localization of an expanding BEC in a disorder potential), but where the localization of ground states is questionable, cf.~a corresponding discussion in~\cite[Sect.~2.2]{SPC08}. We will reconfirm this prediction numerically and  make own predictions for which particle number and for which strength of the potential localization is likely to happen (when keeping the remaining parameters of the experiment unchanged).

%
%
\subsection{Experimental setup and delocalization}

The setting and notation is as introduced in Section \ref{subsec-dimless-1D}. With the data provided in~\cite{physics-paper-localization-BEC} we can identify practical values for the $1D$ GPEVP~\eqref{GPE-1D} with disorder. 
First of all, the authors consider a weakly interacting Bose gas consisting of $N=1.7 \cdot 10^4$ atoms of \isotope[87]{Rb}. 
Consequently, we can specify the mass of a particle with $m = 1.441 \cdot 10^{-25} \hspace{2pt}\texttt{kg}$ and the scattering length with $a = 5.1 \cdot 10^{-9} \hspace{2pt}\texttt{m}$. The condensate is confined in $x$-direction by using a strong harmonic trapping potential in $y$- and $z$-direction with trapping frequency $\omega_y = \omega_z = 2\pi \cdot 70\hspace{2pt} \texttt{Hz}$. 
In $x$-direction, a disorder potential is imposed. The disorder is generated by passing a laser beam (with a wave length of $0.514 \hspace{2pt}\mu \texttt{m}$) through a diffusive plate. The beam is diffracted and forms a speckle pattern behind the plate. These types of potentials are called speckle potentials (or optical speckles) and have a sufficiently small correlation length, cf.~\cite{SPC08,physics-paper-localization-BEC}. In the setup of ~\cite{physics-paper-localization-BEC}, the speckle grain size of the plate (which is proportional to the correlation length $\sigma_R$) is given by $\pi \sigma_R = 0.82 \texttt{$\mu$m} = 8.2 \cdot 10^{-7} \texttt{m}$. 
Since this is an indicator for the speed of variations in the potential, we can use it as a reference value to define $\eps$. More precisely, we rescale $\pi \sigma_R$ using the default unit $\sqrt{\frac{\hbar}{m } }$ (cf.~\eqref{def-sx}) and then set $\rho$ so that we obtain a certain target oscillation length $\eps$ for our rescaled potential (on the unit interval). Recall here that $\rho$ is used to rescale the oscillation length of $V_{\mbox{\tiny\rm 1D}}$ to $1$. Hence, we set 
$$
\rho 
:= \pi \sigma_R\,  \sqrt{\frac{m}{\hbar } } 
= 8.2 \cdot 10^{-7} \cdot \sqrt{ \frac{ 1.441 \cdot 10^{-25}} { 1.05 \cdot 10^{-34} } } \approx 3.037 \cdot 10^{-2}.
$$
Note that $V_{\mbox{\tiny\rm 1D}}(x) = \frac{\rho^2}{\hbar} \hspace{2pt} W_{\mbox{\tiny\rm 1D}}\big( x \hspace{2pt}\rho \hspace{2pt}\sqrt{\frac{\hbar}{m}} \big)$, which means that if $W_{\mbox{\tiny\rm 1D}}$ is oscillating with frequency $(\pi \sigma_R)^{-1}$, then $V_{\mbox{\tiny\rm 1D}}$ is oscillating with frequency $1$. Next, we have to set the target oscillation length $\eps$ for the effective nondimensional potential $V_{\mbox{\tiny\rm 1D}}(x/\eps)$. A reasonable choice (to work on the unit interval and also for computational purposes) is $\eps := 2^{-11}$. Practically this means that if the ground state localizes, then we would observe an exponential decay in units of $\eps$. For $s_\x$ we have
\begin{align*}
s_\x = \sqrt{ \frac{\hbar}{m } }\ \frac{\rho}{\eps}  
\approx 1.678 \cdot 10^{-3}, 
\end{align*}
which means that our results on the unit interval $D=(0,1)$ will correspond to a spatial extension of $1.678 \hspace{2pt}\texttt{mm}$. It is realistic to assume that the condensate is fully confined in that region and we can prescribe zero boundary conditions on $D$ for the mathematical problem.
Finally, the maximum disorder amplitude is specified in  \cite{physics-paper-localization-BEC} with 
$$
V_R := \max_{x} W_{\mbox{\tiny\rm 1D}}(x) = 0.12 \cdot \hbar \cdot2\pi \cdot  219\hspace{2pt} \texttt{Hz}.
$$
Rescaling these numbers to the nondimensional setting, we obtain a peak amplitude of
$$
\max V_{1D} = 0.1523, \qquad
(\max V_{1D})\, \eps^{-2} \approx 6.39 \cdot 10^5. 
$$
For the interaction constant and the effective $\kappa$ we obtain 
$$
\kappa_{\mbox{\tiny\rm 1D}} 
= \rho\, \frac{a N}{2 \pi}\, \sqrt{\frac{m \omega_y \omega_z}{\hbar}}
\approx 6.824, \qquad 
\kappa \approx 6.824\, \eps^{-1} \approx 1.40 \cdot 10^4.
$$
To summarize, we can study the ground states of BEC in an experimental setup that corresponds to the data from \cite{physics-paper-localization-BEC} by considering the $1D$ GPEVP~\eqref{GPE-1D} with a potential~$\frac{1}{\eps^2} \hspace{2pt} V_{\mbox{\tiny\rm 1D}}\hspace{-2pt}\left( \frac{x}{\eps}\right)$ that oscillates on the $\eps$-scale and which takes a maximum value of about~$0.1523\, \eps^{-2}$. The particle interaction constant is $\kappa \approx 6.824\, \eps^{-1}$. Considering the relative size of $\eps$, the value of $\kappa$ appears to be at least one order of magnitude too large compared to the strength of the potential. 
Note that the given parameters mark the border of the considered scaling regime. Thus, the observed delocalization does not contradict Theorem~\ref{thm:localizedGroundstate}, which provides a scaling regime (for sufficiently small $\eps$) rather than explicit constants. 
The following numerical experiments show that relatively small adjustments of the parameters lead to localized ground states. 
%
%
\subsection{Mathematical prediction of localization}
We explore the influence of the parameters~$\eps$, $V_{\max}$, and $\kappa$ on the localization of the ground state. 
For the numerical experiments we apply again the eigenvalue solvers presented in~\cite{HenP18ppt,AltHP19ppt}. 
 
In the first experiment, we observe that the physical setup is indeed too weak to generate localization for the given particle number of $N=1.7 \cdot 10^4$. In Figure~\ref{fig:sect6:eps} it is shown that even smaller values of $\eps$ do not lead to a significant improvement of the effect. Thus, localization would only appear on rather unrealistic length scales, i.e., $\eps\le 2^{-15}$, which would correspond to a BEC that has an extension that is significantly larger than $2.68\hspace{2pt}\texttt{cm}$. 
\begin{figure}
\centering
\input{pics/gs_eps/gs_eps_1_1_1}
\input{pics/gs_eps/gs_eps_2_1_1}
\input{pics/gs_eps/gs_eps_3_1_1}
\input{pics/gs_eps/gs_eps_4_1_1}
\input{pics/gs_eps/gs_eps_5_1_1}
\caption{Illustration of ground states for sample disorder potentials with parameters~$\max V_{1D} = 0.1523$, $\kappa_{1D} = 6.824$, and variable~$\eps=2^{-9},\dots,2^{-13}$ (from left to right). Underlying mesh size $h=2^{-18}$.}
\label{fig:sect6:eps}
\end{figure}
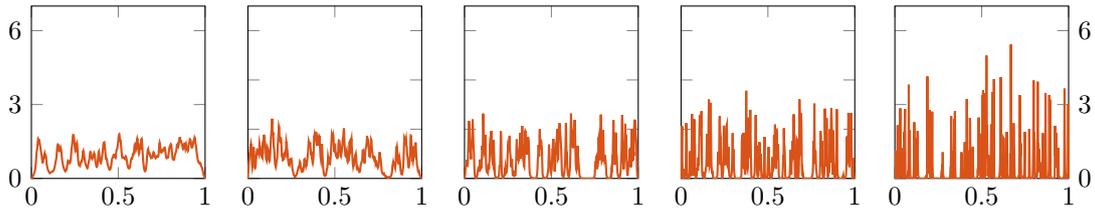

In the second experiment, we analyze the influence of~$\kappa$ and $\max V_{1D}$, i.e., the number of particles and the strength of the potential. 
Here, the results in Figure~\ref{fig:sect6:kappaVmax} indicate that both parameters intensify localization but in a different manner. The strength of the potential makes peaks steeper whereas the number of particles is more related to the number of peaks. 
Already for a disorder potential that is $4$ times stronger than the current potential combined with only an eighth of the particles, i.e., when~$\max V_{1D}$ and~$\kappa_{1D}$ are almost the same, localization gets visible. 
In this sense, the given physical values of $V_{\max}$ and $\kappa$ mark the transition phase between local and global ground states, which is in accordance with Theorem~\ref{thm:localizedGroundstate}. 
Note, however, that there is no sharp phase transition between localization and delocalization. 
\begin{figure}
\centering
\input{pics/square/GroundState_1_1_1}
\input{pics/square/GroundState_1_1_2}
\input{pics/square/GroundState_1_1_3}
\input{pics/square/GroundState_1_1_4}\\
\input{pics/square/GroundState_1_2_1}
\input{pics/square/GroundState_1_2_2}
\input{pics/square/GroundState_1_2_3}
\input{pics/square/GroundState_1_2_4}\\
\input{pics/square/GroundState_1_3_1}
\input{pics/square/GroundState_1_3_2}
\input{pics/square/GroundState_1_3_3}
\input{pics/square/GroundState_1_3_4}\\
\input{pics/square/GroundState_1_4_1}
\input{pics/square/GroundState_1_4_2}
\input{pics/square/GroundState_1_4_3}
\input{pics/square/GroundState_1_4_4}
\caption{Ground states for sample disorder potentials with parameters~$\max V_{1D} = 0.1523\cdot c_V$, $\kappa_{1D} = 6.824\cdot c_\kappa$, and fixed~$\eps=2^{-11}$. Underlying mesh size $h=2^{-18}$.} 
\label{fig:sect6:kappaVmax}
\end{figure}
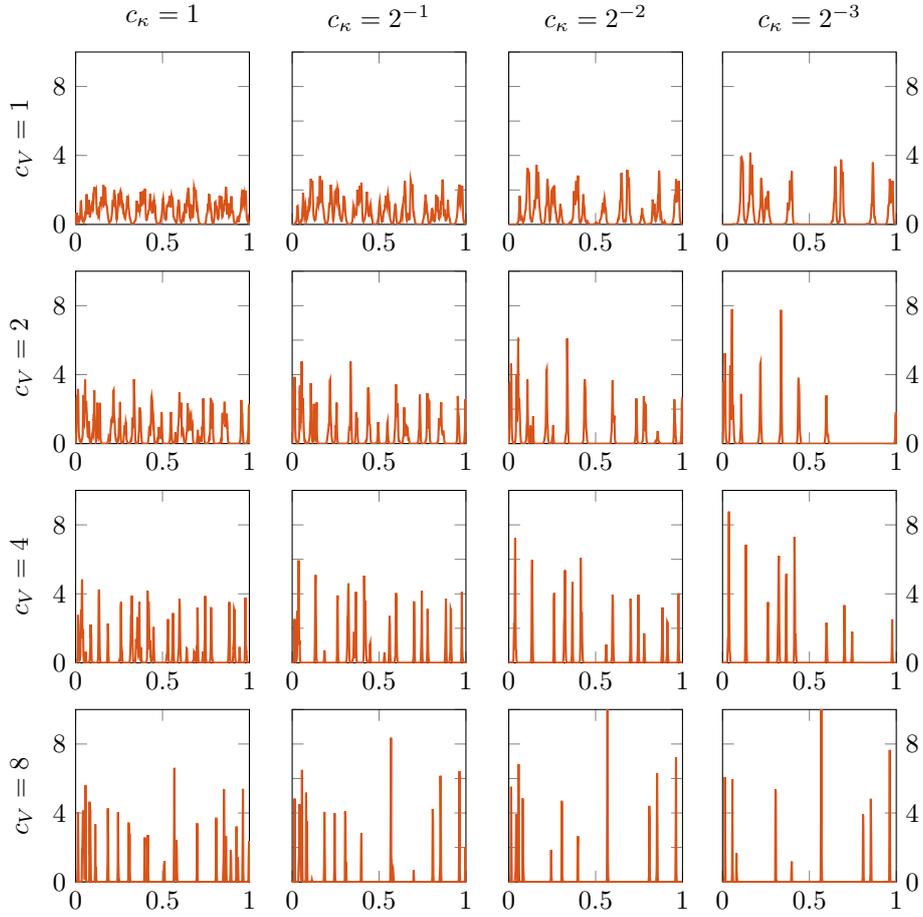
We emphasize that these parameters would not lead to localization for a periodic potential. To see this, we refer to Figure~\ref{fig:sect6:L1periodic} where we consider once more the $L^1$-norm of the ground state as an indicator for localization.  
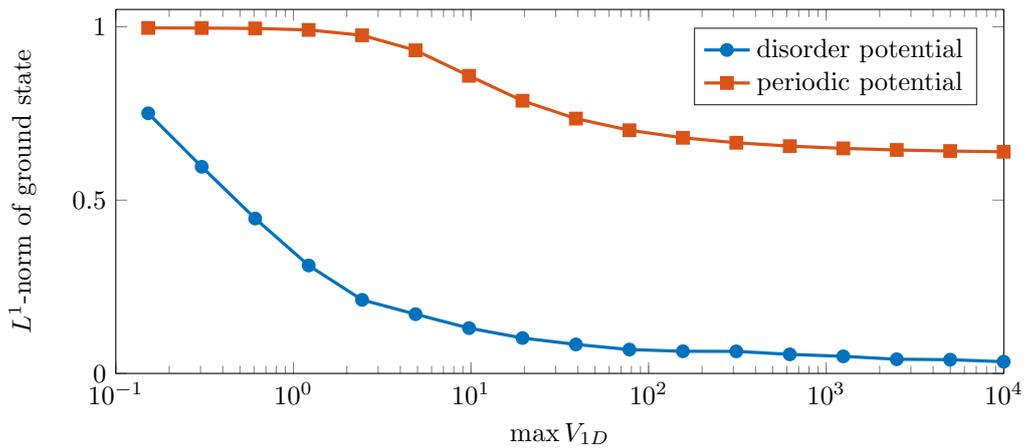
\begin{figure}
\centering
\input{pics/fig_L1norm_maxV1D}
\caption{$L^1$-norm of ground states for sample disorder and random potential for fixed~$\eps=2^{-11}$, $\kappa_{1D}=6.824$, and varying $\max V_{1D}$. Underlying mesh size $h=2^{-16}$.} 
\label{fig:sect6:L1periodic}
\end{figure}

Finally, in Figure~\ref{fig:sect6:L1kappaVmax} we study the influence of $\max V_{1D}$ and $\kappa_{1D}$ on the localization for fixed~$\eps=2^{-7}$. In the left graph of Figure~\ref{fig:sect6:L1kappaVmax} we fix the strength of the potential with three different values and decrease the size of $\kappa_{1D}$ step by step. We observe that this leads to a slow but visibly increasing localization effect for the ground state (measured in the $L^1$-norm). When $\kappa_{1D}$ drops below a critical value, the localization stagnates and reaches a plateau. It is not possible to go beyond that localization level without touching the parameters of the potential. Hence, we can reconfirm that the weaker the particle interactions, the more pronounced is the localization. This observation is in accordance with previous findings.

In the right graph of Figure~\ref{fig:sect6:L1kappaVmax} the converse situation is depicted. Here we fix the parameter $\kappa_{1D}$ with three different values (where the largest is according to physical setup) and we increase the strength of the potential. We observe that increasing the strength of the potential, seems to have a stronger effect on the localization than decreasing the strength of particle interactions $\kappa_{1D}$. In our graphs it appears that doubling the strength of the potential leads to a reduction of the $L^1$-norm by a factor of around $1/4$. Hence, the critical amplitude of the potential is inverse proportional to the degree of location expressed through the $L^1$-norm.

These findings show that in order to expect localization of ground states in the setting of \cite{physics-paper-localization-BEC}, the strength of the disorder potential needs to be increased at least by a factor $4$ and the number of particles needs to be reduced by a factor $8$. This finding is again inline with the predictions of Theorem~\ref{thm:localizedGroundstate}. Potentials with an even smaller oscillation length are expected to further pronounce the localization effect. In terms of physical realizability, we expect that a suitable experimental setup can be realized, even though it is certainly challenging. Here, 
superlattices 
are a good alternative for corresponding physical experiments. 
Superlattices are pseudorandom potentials that are generated by the superposition of several optical lattices of different wavelengths and which can be additionally adjusted at different angles to increase the disorder, cf.~\cite{DZS03,RoB03,SDK06}.

In conclusion, our numerical experiments confirm the theoretical predictions about the localization and delocalization of ground states of BECs. Furthermore, by applying our identified physical data, we expect that an experimental observation of the localization may be possible. 

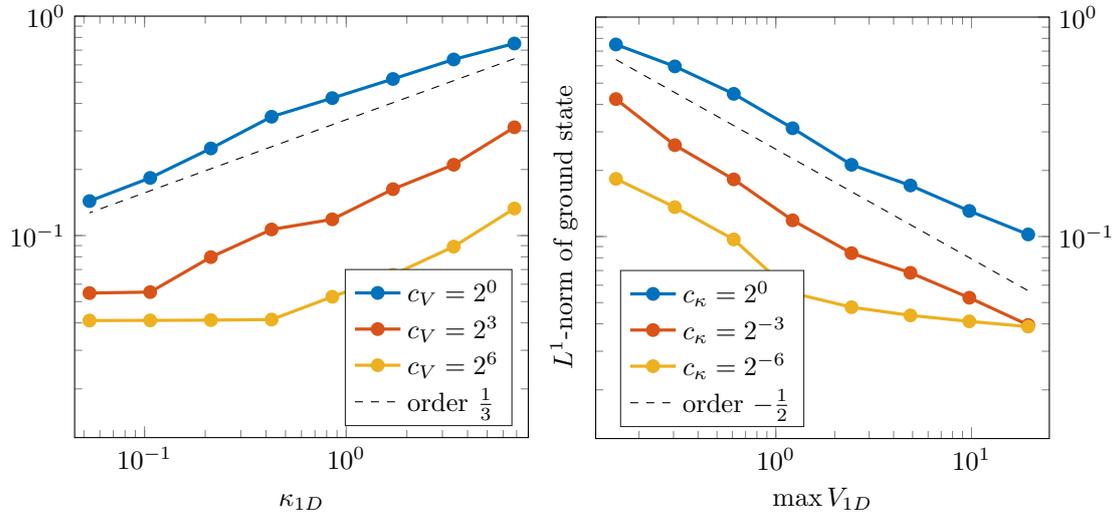
\begin{figure}
\centering
\input{pics/fig_L1norm_kappaV_left} 
\input{pics/fig_L1norm_kappaV_right}
\caption{$L^1$-norm of ground states for sample disorder potentials for~$\eps=2^{-7}$. 
Left: over $\kappa_{1D}$ for various~$\max V_{1D}=0.1523\cdot c_V$. 
Right: over $\max V_{1D}$ for various $\kappa_{1D}=6.824\cdot c_\kappa$. Underlying mesh size $h=2^{-17}$.}  
\label{fig:sect6:L1kappaVmax}
\end{figure}
%
%
\section*{Acknowledgment}
We would like to thank the anonymous reviewers for their constructive and helpful suggestions, which considerably improved the readability of the paper. 
%
%
\newcommand{\etalchar}[1]{$^{#1}$}
\def\cprime{$'$}

%
%
\appendix
\section{Eigenvalue Bounds}\label{app:EVbounds}
This appendix collects the technical proofs of upper and lower eigenvalue bounds. 
%
\subsection{Proof of Lemma~\ref{lem:upperBounds} (upper bound)}\label{app:EVbounds:upper}
In each valley of length~$j\eps$ with $\ell_{1,\eps} \le j \le \Laeps$ we consider the first eigenfunction of the Laplacian with homogeneous Dirichlet boundary conditions. It is well-known that the corresponding eigenvalue equals~$\pi^2/(j\eps)^2$. Extending this function by zero, we obtain a function~$v\in\V$, which satisfies
\[
  \Vvert v\Vvert_\eps^2
  = \int_{\Oaeps} |v'|^2 + \Veps\, |v|^2 \dx
  \le \frac{\pi^2}{(j\eps)^2} \|v\|^2 + \frac{\alpha}{ |\log\eps|^{2}\, \eps^2 } \|v\|^2.
\] 
Hence, we conclude that~$\Vvert v\Vvert_\eps^2 \lesssim \max\{ \ell_{1,\eps}^{-2}, |\log\eps|^{-2} \}\, \eps^{-2} \|v\|^2 \simeq (\ell_{1,\eps} \eps)^{-2} \|v\|^2$. Since all these functions have a support in different valleys, they are all disjoint and thus, orthogonal. It remains to count the number of valleys. With $N_{\alpha,j}^{\eps} \simeq p_\alpha^{j}\eps^{-1}$ we obtain in total 
\[
  M_{1,\eps}
  := \sum_{j=\ell_{1,\eps}}^{\Laeps} N_{\alpha,j}^{\eps}
  \simeq \frac{1}{\eps} \sum_{j=\ell_{1,\eps}}^{\Laeps} p_\alpha^{j}
  \simeq \frac{p_\alpha^{\ell_{1,\eps}-1}}{\eps}
\]
valleys, which are larger or equal to~$\ell_{1,\eps}\eps$. Thus, we have constructed~$M_{1,\eps}$ orthogonal functions with an energy bounded by~$\calO((\ell_{1,\eps}\eps)^{-2})$, which implies the corresponding upper bound for~$\lambda_{\eps,M_{1,\eps}}$. 
%
\subsection{Proof of Lemma~\ref{lem:lowerBounds} (lower bound)}\label{app:EVbounds:lower}
The proof of the stated lower bound is based on the max-min principle for eigenvalues in Hilbert spaces~\cite[Ch.~1]{WeiS72}. Thus, we need to construct an $M_{2,\eps}$-dimensional subspace of~$\V$ such that any function in its complement satisfies an energy bound of the form~$\|v\|^2 \lesssim (\ell_{2,\eps}\eps)^2 \Vvert v\Vvert_\eps^2$.  

In contrast to the upper bound shown in the previous subsection, the construction here is based on non-peaks rather than valleys. This is crucial as the proof rests upon the fact that these intervals are surrounded either by peaks or the boundary. 
We consider the non-peaks of width~$j\eps$ with $\ell_{2,\eps} \le j \le \Lbeps$ and define a uniform partition of these intervals with a mesh size~$h \simeq \ell_{2,\eps} \eps$. On this partition we consider the nodal interpolation operator~$\Pi$, which interpolates in the nodes and linearly goes to zero outside of the non-peaks within the width~$\eps/2$. We emphasize that this surrounding strip is by construction part of~$\Obeps$. The image of this interpolation operator equals a finite element space of dimension 
\[
  M_{2,\eps}
  \simeq \sum_{j=\ell_{2,\eps}}^{\Lbeps} N_{\neg\beta,j}^{\eps} \Big( \big\lfloor j / \ell_{2,\eps} \big\rfloor +2\Big)
  \lesssim \frac{1}{\eps} \frac{\Lbeps}{\ell_{2,\eps}} \sum_{j=\ell_{2,\eps}}^{\Lbeps} (1-p_\beta)^{j}
  \simeq \frac{1}{\eps} \frac{\Lbeps}{\ell_{2,\eps}} (1-p_\beta)^{\ell_{2,\eps}-1}.
\]
Using $\Lbeps \simeq \log_{1-p_\beta} \eps$ from assumption~\ref{A1} and~$\ell_{2,\eps} \simeq \log_{1-p_\beta}(\eps^{q_2})$, this yields~$\Lb/\ell_{2,\eps} \simeq q_2^{-1}$, which is independent of $\eps$. Hence, we have $M_{2,\eps} \lesssim q_2^{-1} \eps^{q_2-1}$. Further, one can show similarly as in~\cite{AltHP20} that the constructed nodal interpolation operator satisfies~$\Vert v - \Pi v\Vert \lesssim \ell_{2,\eps}\,\eps\, \Vvert v \Vvert_\eps$ for all $v\in\V$. 
Thus, all functions~$v\in\ker\Pi\subseteq \V$, which is the complement of an~$M_{2,\eps}$-dimensional subspace, satisfy 
\[
  \| v \|
  = \| v - \Pi v\| 
  \lesssim \ell_{2,\eps}\, \eps\, \Vvert v \Vvert_\eps.
\]
The max-min principle then implies~$\lambda_{\eps,M_{2,\eps}} \ge \min_{v\in\ker\Pi} \Vvert v \Vvert_\eps^2 / \| v\|^2 \gtrsim (\ell_{2,\eps}\eps)^{-2}$. 
\end{document}

%% file: def.tex
\definecolor{myBlue}{RGB}{113,104,238} 
\definecolor{myGreen}{RGB}{154,205,50} 
\definecolor{myGreen2}{RGB}{114,175,30} 
\definecolor{myRed}{RGB}{180,50,50}  
\definecolor{myOrange}{RGB}{225,92,22} 
\definecolor{lgray}{RGB}{200,200,200} 
\definecolor{llgray}{RGB}{155,155,155} 
\definecolor{alphaCol}{RGB}{230,190,50} 
\definecolor{betaCol}{RGB}{90,80,180} 
\definecolor{mycolor1}{rgb}{0.00000,0.44700,0.74100}%
\definecolor{mycolor2}{rgb}{0.85000,0.32500,0.09800}%
\definecolor{mycolor3}{rgb}{0.92900,0.69400,0.12500}%
\definecolor{mycolor4}{rgb}{0.49400,0.18400,0.55600}%
\definecolor{mycolor5}{rgb}{0.46600,0.67400,0.18800}%
\definecolor{mycolor6}{rgb}{0.30100,0.74500,0.93300}%
\definecolor{mycolor7}{rgb}{0.63500,0.07800,0.18400}%
\newcommand\N{\mathbb N}

\newcommand\R{\mathbb R}

\DeclareMathOperator{\supp}{supp}
\DeclareMathOperator{\id}{id}

\DeclareMathOperator{\diam}{diam}

\DeclareMathOperator{\tol}{tol}

\DeclareMathOperator{\plog}{polylog}

\newcommand\calA{\mathcal A}

\newcommand\calG{\mathcal G}

\newcommand\calI{\mathcal I}

\newcommand\calN{\mathcal N}
\newcommand\calO{\mathcal O}
\newcommand\calP{\mathcal P}

\newcommand\calT{\mathcal T}
\newcommand\calV{\mathcal V}

\newcommand{\calTeps}{\mathcal{T}^{\eps}}
\newcommand{\Veps}{V_{\eps}}

\newcommand\eps{\varepsilon}
\newcommand{\x}{\mathbf{x}}
\newcommand{\ci}{\mathrm{i}} 

\def\dx{\,\text{d}x}
\def\dxx{\,\text{d}\x}

\newcommand\V{\calV}   
   
\newcommand\Vvert{{|\hskip-0.9pt|\hskip-0.9pt|}}

\newcommand\Oaeps{{D_{\alpha}^{\eps}}}  
\newcommand\Obeps{{D_{\beta}^{\eps}}}  
\newcommand\calTbeps{{\calT_\beta^{\eps}}}

\newcommand\Pz{\calP_{\eps,z}}

\newcommand\Lb{L_{\neg\beta}}
\newcommand\Laeps{L_\alpha^{\eps}}
\newcommand\Lbeps{L_{\neg\beta}^{\eps}}
\newcommand\K{K_\eps}

%% file: pics/fig_localization_10k.tex
%
%
%
\usepgfplotslibrary{fillbetween}
\begin{tikzpicture}

\begin{axis}[%
width=0.8\textwidth,
height=0.38\textwidth,
at={(0.758in,0.481in)},
scale only axis,
xmin=0,
xmax=1,
xlabel style={font=\color{white!15!black}},
xlabel={$\eps\, \cdot $ number of eigenfunction},
ymin=0,
ymax=1.0,
ylabel style={font=\color{white!15!black}},
ylabel={$L^1$-norm of eigenfunction},
axis background/.style={fill=white},
title style={font=\bfseries},
legend style={at={(0.21,0.32)}, anchor=south east, legend cell align=left, align=left, draw=white!15!black}
]


\addplot [color=mycolor1, ultra thick]
  table[row sep=crcr]{%
     0    0.2380\\
0.0323    0.2423\\
0.0645    0.2485\\
0.0968    0.2551\\
0.1290    0.2640\\
0.1613    0.2745\\
0.1935    0.2866\\
0.2258    0.2991\\
0.2581    0.3151\\
0.2903    0.3319\\
0.3226    0.3503\\
0.3548    0.3702\\
0.3871    0.3941\\
0.4194    0.4200\\
0.4516    0.4512\\
0.4839    0.4873\\
0.5161    0.5259\\
0.5484    0.5706\\
0.5806    0.6182\\
0.6129    0.6682\\
0.6452    0.7159\\
0.6774    0.7603\\
0.7097    0.7956\\
0.7419    0.8244\\
0.7742    0.8466\\
0.8065    0.8621\\
0.8387    0.8725\\
0.8710    0.8786\\
0.9032    0.8809\\
0.9355    0.8804\\
0.9677    0.8767\\
1.0000    0.8719\\
};
\addlegendentry{$\eps = 2^{-6}$}


\addplot [color=mycolor2, ultra thick]
  table[row sep=crcr]{%
     0    0.1207\\
0.0323    0.1212\\
0.0645    0.1241\\
0.0968    0.1272\\
0.1290    0.1315\\
0.1613    0.1365\\
0.1935    0.1420\\
0.2258    0.1489\\
0.2581    0.1563\\
0.2903    0.1642\\
0.3226    0.1740\\
0.3548    0.1837\\
0.3871    0.1950\\
0.4194    0.2073\\
0.4516    0.2225\\
0.4839    0.2406\\
0.5161    0.2618\\
0.5484    0.2872\\
0.5806    0.3188\\
0.6129    0.3574\\
0.6452    0.4052\\
0.6774    0.4678\\
0.7097    0.5403\\
0.7419    0.6186\\
0.7742    0.6936\\
0.8065    0.7566\\
0.8387    0.8018\\
0.8710    0.8298\\
0.9032    0.8403\\
0.9355    0.8355\\
0.9677    0.8158\\
1.0000    0.7876\\
};
\addlegendentry{$\eps = 2^{-8}$}


\addplot [color=mycolor3, ultra thick]
  table[row sep=crcr]{%
     0    0.0619\\
0.0323    0.0605\\
0.0645    0.0622\\
0.0968    0.0635\\
0.1290    0.0656\\
0.1613    0.0681\\
0.1935    0.0711\\
0.2258    0.0743\\
0.2581    0.0783\\
0.2903    0.0823\\
0.3226    0.0867\\
0.3548    0.0916\\
0.3871    0.0974\\
0.4194    0.1039\\
0.4516    0.1114\\
0.4839    0.1199\\
0.5161    0.1306\\
0.5484    0.1430\\
0.5806    0.1589\\
0.6129    0.1786\\
0.6452    0.2030\\
0.6774    0.2346\\
0.7097    0.2771\\
0.7419    0.3300\\
0.7742    0.3993\\
0.8065    0.4839\\
0.8387    0.5794\\
0.8710    0.6585\\
0.9032    0.6979\\
0.9355    0.6823\\
0.9677    0.6234\\
1.0000    0.5524\\
};
\addlegendentry{$\eps = 2^{-10}$}


\addplot [color=mycolor4, ultra thick]
table[row sep=crcr]{%
     0    0.0318\\
0.0323    0.0304\\
0.0645    0.0314\\
0.0968    0.0322\\
0.1290    0.0331\\
0.1613    0.0344\\
0.1935    0.0359\\
0.2258    0.0378\\
0.2581    0.0395\\
0.2903    0.0419\\
0.3226    0.0443\\
0.3548    0.0472\\
0.3871    0.0502\\
0.4194    0.0539\\
0.4516    0.0583\\
0.4839    0.0635\\
0.5161    0.0698\\
0.5484    0.0779\\
0.5806    0.0876\\
0.6129    0.1001\\
0.6452    0.1159\\
0.6774    0.1367\\
0.7097    0.1626\\
0.7419    0.1978\\
0.7742    0.2450\\
0.8065    0.3081\\
0.8387    0.3999\\
0.8710    0.5247\\
0.9032    0.6570\\
0.9355    0.7305\\
0.9677    0.7072\\
1.0000    0.5740\\
};
\addlegendentry{$\eps = 2^{-12}$}

\addplot [color=mycolor5, mark=*, mark options={solid, mycolor5}, dashed]
  table[row sep=crcr]{%
     0    0.9507\\
0.0323    0.8551\\
0.0645    0.8590\\
0.0968    0.8530\\
0.1290    0.8270\\
0.1613    0.8498\\
0.1935    0.8479\\
0.2258    0.8433\\
0.2581    0.7764\\
0.2903    0.8338\\
0.3226    0.8263\\
0.3548    0.8207\\
0.3871    0.8292\\
0.4194    0.8051\\
0.4516    0.7944\\
0.4839    0.7909\\
0.5161    0.9160\\
0.5484    0.8281\\
0.5806    0.8353\\
0.6129    0.8421\\
0.6452    0.8517\\
0.6774    0.8561\\
0.7097    0.8619\\
0.7419    0.8671\\
0.7742    0.8838\\
0.8065    0.8748\\
0.8387    0.8775\\
0.8710    0.8796\\
0.9032    0.8790\\
0.9355    0.8808\\
0.9677    0.8779\\
1.0000    0.8616\\
};
\addlegendentry{periodic} 


\addplot [name path=conf1, color=mycolor4, dashed] 
table[row sep=crcr]{%
     0    0.0656\\
0.0323    0.0472\\
0.0645    0.0431\\
0.0968    0.0439\\
0.1290    0.0480\\
0.1613    0.0438\\
0.1935    0.0467\\
0.2258    0.0492\\
0.2581    0.0556\\
0.2903    0.0741\\
0.3226    0.0621\\
0.3548    0.0680\\
0.3871    0.0732\\
0.4194    0.0936\\
0.4516    0.0921\\
0.4839    0.1053\\
0.5161    0.1030\\
0.5484    0.1068\\
0.5806    0.1457\\
0.6129    0.2191\\
0.6452    0.3552\\
0.6774    0.6437\\
0.7097    1.2107\\
0.7419    1.5047\\
0.7742    1.2976\\
0.8065    0.8617\\
0.8387    0.5652\\
0.8710    0.8798\\
0.9032    1.3469\\
0.9355    1.6161\\
0.9677    1.5760\\
1.0000    1.1883\\
};

\addplot [name path=conf2, color=mycolor4, dashed] 
table[row sep=crcr]{%
     0   -0.0019\\
0.0323    0.0137\\
0.0645    0.0197\\
0.0968    0.0205\\
0.1290    0.0182\\
0.1613    0.0249\\
0.1935    0.0252\\
0.2258    0.0263\\
0.2581    0.0235\\
0.2903    0.0097\\
0.3226    0.0264\\
0.3548    0.0264\\
0.3871    0.0272\\
0.4194    0.0142\\
0.4516    0.0246\\
0.4839    0.0217\\
0.5161    0.0366\\
0.5484    0.0491\\
0.5806    0.0296\\
0.6129   -0.0190\\
0.6452   -0.1233\\
0.6774   -0.3704\\
0.7097   -0.8854\\
0.7419   -1.1091\\
0.7742   -0.8076\\
0.8065   -0.2455\\
0.8387    0.2346\\
0.8710    0.1696\\
0.9032   -0.0329\\
0.9355   -0.1552\\
0.9677   -0.1615\\
1.0000   -0.0403\\
};

\addplot[opacity=0.1, color=mycolor4] 
fill between[ 
of = conf1 and conf2
];

\end{axis}
\end{tikzpicture}%

%% file: pics/fig_L1vsKappa_eps7_Vmax5.tex
%
%
%
\begin{tikzpicture}

\begin{axis}[%
width=4.2in,
height=1.7in,
scale only axis,
xmode=log,
xmin=0.002,
xmax=500,
xminorticks=true,
xlabel={$\kappa$},
ymode=log,
ymin=3e-3,
ymax=1,
ylabel={$L^1$-norm},
axis background/.style={fill=white},
legend style={at={(0.86,0.44)},legend cell align=left, align=left, draw=white!15!black}
]

\addplot [color=mycolor1, very thick]
  table[row sep=crcr]{%
0.0009765625	0.222588591094026\\
0.001953125	0.222588663748425\\
0.00390625	0.222588809071331\\
0.0078125	0.222589099773648\\
0.015625	0.222589681404861\\
0.03125	0.222590845578004\\
0.0625	0.222593177602777\\
0.125	0.222593096862282\\
0.25	0.222602217703719\\
0.5	0.222620675403263\\
1	0.222653024514419\\
2	0.222718345010396\\
4	0.222847380279926\\
8	0.22310480970155\\
16	0.223723053513942\\
32	0.293323695100309\\
64	0.304051487351453\\
128	0.308373120630515\\
256	0.431615507489752\\
512	0.464118527078518\\
1024	0.537836760015258\\
};
\addlegendentry{$d=1$}


\addplot [color=mycolor2, very thick] 
table[row sep=crcr]{%
	0.0009765625	0.0408281741247965\\
	0.001953125	0.0409574452793695\\
	0.00390625	0.0411719140338841\\
	0.0078125	0.0417531211769939\\
	0.015625	0.0430569890561614\\
	0.03125	0.0462501082507296\\
	0.0625	0.0497783570805933\\
	0.125	0.0516681825983036\\
	0.25	0.0528890926657542\\
	0.5	0.0563889889325521\\
	1	0.0641974021420421\\
	2	0.06743161799967\\
	4	0.0823719000985463\\
	8	0.11845235111267\\
	16	0.158744369883681\\
	32	0.202775401991801\\
	64	0.257995666758781\\
	128	0.328549897395589\\
	256	0.414005013963527\\
	512	0.492719325846861\\
	1024	0.58575408223647\\
};
\addlegendentry{$d=2$}


\addplot [color=mycolor3, very thick]
table[row sep=crcr]{%
0.0009765625	0.00694034090955334\\
0.001953125	0.00696431333216911\\
0.00390625	0.00701179770608896\\
0.0078125	0.00715136858990861\\
0.015625	0.00800905172881247\\
0.03125	0.00965450257234612\\
0.0625	0.0116837306586576\\
0.125	0.017460312977961\\
0.25	0.0267518899359518\\
0.5	0.0370706747805829\\
1	0.0535601353999174\\
2	0.0721243301255077\\
4	0.0970244619362574\\
8	0.13353035269092\\
16	0.180568811936034\\
32	0.243726157500803\\
64	0.321744468378011\\
128	0.410211005581435\\
256	0.506476419503453\\
512	0.603677859391261\\
1024	0.695500526780848\\
2048	0.77668613844392\\
4096	0.845024426732138\\
};
\addlegendentry{$d=3$}


\addplot [gray, dashed]
table[row sep=crcr]{%
	0.0078125	1e-3\\
	0.0078125	1\\
};

\addplot [gray, dashed]
table[row sep=crcr]{%
	1	1e-3\\
	1	1\\
};

\addplot [gray, dashed]
table[row sep=crcr]{%
	128	1e-3\\	
	128	1\\
};


\addplot [gray, dotted]
table[row sep=crcr]{%
	0.0001	0.176776696\\
	1000	0.176776696\\
};

\addplot [gray, dotted]
table[row sep=crcr]{%
	0.0001	0.0313\\
	1000	0.0313\\
};

\addplot [gray, dotted]
table[row sep=crcr]{%
	0.0001	0.005524272\\
	1000	0.005524272\\
};

\end{axis}
\end{tikzpicture}%

%% file: pics/fig_L1vsKappa_d1_h12.tex
%
%
%
\begin{tikzpicture}

\begin{axis}[%
width=1.6in,
height=2.0in,
scale only axis,
xmode=log,
xmin=0.5,
xmax=6000,
xminorticks=true,
xlabel={$\kappa$},
ymin=0,
ymax=1,
ylabel={$L^1$-norm},
ylabel style={below=0.03in},
axis background/.style={fill=white},
legend style={at={(0.25,0.94)},legend cell align=left, align=left, draw=white!15!black}
]

\addplot [color=mycolor2, dashed]
  table[row sep=crcr]{%
0.00390625	0.444126777968124\\
0.0078125	0.444129105230081\\
0.015625	0.444133762904426\\
0.03125	0.444143090944928\\
0.0625	0.444161798516957\\
0.125	0.444199425515199\\
0.25	0.444275575668121\\
0.5	0.444431877433018\\
1	0.44464836921724\\
2	0.445188983001309\\
4	0.446145800795596\\
8	0.448280941106303\\
16	0.550817552604907\\
32	0.587052117853949\\
64	0.599988134782899\\
128	0.611237220981565\\
256	0.681597948926358\\
512	0.784683004288552\\
1024	0.887633427557532\\
2048	0.935463831454436\\
4096	0.971285986747348\\
8192	0.989220723051418\\
};

\addplot [color=mycolor3, dashed]
  table[row sep=crcr]{%
0.00390625	0.31397485481121\\
0.0078125	0.313975556752139\\
0.015625	0.313976960623319\\
0.03125	0.313979768323712\\
0.0625	0.313985383562355\\
0.125	0.313996613436509\\
0.25	0.314019071140586\\
0.5	0.314063981381496\\
1	0.314138087613479\\
2	0.314312943254933\\
4	0.314661487853716\\
8	0.315306338478171\\
16	0.316550526489147\\
32	0.318765367979245\\
64	0.421479062533549\\
128	0.468444770211407\\
256	0.558431892978755\\
512	0.593309265998994\\
1024	0.696106421208684\\
2048	0.730081609378348\\
4096	0.782449622736561\\
8192	0.852132262167548\\
};

\addplot [color=mycolor4, dashed]
  table[row sep=crcr]{%
0.00390625	0.222039586249761\\
0.0078125	0.222039833398917\\
0.015625	0.222040327700673\\
0.03125	0.222041316318124\\
0.0625	0.222043293610119\\
0.125	0.222047248433167\\
0.25	0.222055159122061\\
0.5	0.222070985380946\\
1	0.222099369402242\\
2	0.222161953286925\\
4	0.222283534220099\\
8	0.222523733306941\\
16	0.222988183289711\\
32	0.223872419293258\\
64	0.275969945709941\\
128	0.295022042593267\\
256	0.347732666965117\\
512	0.453613637759153\\
1024	0.494018414899505\\
2048	0.620784434935672\\
4096	0.67104824852821\\
8192	0.736185843671375\\
};

\addplot [color=mycolor5, dashed]
  table[row sep=crcr]{%
0.00390625	0.157157549378742\\
0.0078125	0.157157636639139\\
0.015625	0.157157811159431\\
0.03125	0.15715816019802\\
0.0625	0.157158858267305\\
0.125	0.157160254375056\\
0.25	0.157163046473417\\
0.5	0.157168630251205\\
1	0.157178782052326\\
2	0.157200949424007\\
4	0.157245233792954\\
8	0.157332187926618\\
16	0.157502624296351\\
32	0.157833985033928\\
64	0.158460239874375\\
128	0.243836932434611\\
256	0.278192496983051\\
512	0.342615414854185\\
1024	0.375325783346177\\
2048	0.47210270655533\\
4096	0.511535588705188\\
8192	0.59433516573923\\
};


\addplot [color=mycolor2, very thick]
table[row sep=crcr]{%
8	0.448280941106303\\
16	0.550817552604907\\
32	0.587052117853949\\
64	0.599988134782899\\
128	0.611237220981565\\
};

\addplot [color=mycolor3, very thick]
table[row sep=crcr]{%
16	0.316550526489147\\
32	0.318765367979245\\
64	0.421479062533549\\
128	0.468444770211407\\
256	0.558431892978755\\
};

\addplot [color=mycolor4, very thick]
table[row sep=crcr]{%
32	0.223872419293258\\
64	0.275969945709941\\
128	0.295022042593267\\
256	0.347732666965117\\
512	0.453613637759153\\
};

\addplot [color=mycolor5, very thick]
table[row sep=crcr]{%
64	0.158460239874375\\
128	0.243836932434611\\
256	0.278192496983051\\
512	0.342615414854185\\
1024	0.375325783346177\\
};


\addplot [color=mycolor2, mark=*, mark size=2.5pt]
table[row sep=crcr]{%
32	0.587052117853949\\
};

\addplot [color=mycolor3, mark=*, mark size=2.5pt]
table[row sep=crcr]{%
64	0.421479062533549\\
};

\addplot [color=mycolor4, mark=*, mark size=2.5pt]
table[row sep=crcr]{%
128	0.295022042593267\\
};

\addplot [color=mycolor5, mark=*, mark size=2.5pt]
table[row sep=crcr]{%
256	0.278192496983051\\
};

%
%
%
%

\end{axis}
\end{tikzpicture}%

%% file: pics/fig_L1vsKappa_d2_h10.tex
%
%
%
\begin{tikzpicture}

\begin{axis}[%
width=1.6in,
height=2.0in,
scale only axis,
xmode=log,
xmin=0.03,
xmax=255,
xminorticks=true,
xlabel style={font=\color{white!15!black}},
xlabel={$\kappa$},
ymin=0,
ymax=1,
yticklabels = {},
ylabel style={font=\color{white!15!black}},
axis background/.style={fill=white},
legend style={at={(.65,0.95)},legend cell align=left, align=left, draw=white!15!black}
]

\addplot [color=mycolor2, dashed]
  table[row sep=crcr]{%
	0.0009765625	0.238530083308983\\
	0.001953125	0.238602634991944\\
	0.00390625	0.238747474431654\\
	0.0078125	0.239036099402366\\
	0.015625	0.24219352924545\\
	0.03125	0.243270039155871\\
	0.0625	0.247892403241705\\
	0.125	0.256685855776303\\
	0.25	0.270522695476369\\
	0.5	0.287746460050304\\
	1	0.306261406740995\\
	2	0.327337786298677\\
	4	0.354919250649841\\
	8	0.406152864995499\\
	16	0.482395413543121\\
	32	0.561069517024603\\
	64	0.632102824529787\\
	128	0.706613660283701\\
	256	0.772757585393674\\
	512	0.832094602587523\\
	1024	0.886903554139841\\
};
\addlegendentry{$\eps=2^{-5}$}

\addplot [color=mycolor3, dashed]
  table[row sep=crcr]{%
	0.0009765625	0.091824490822977\\
	0.001953125	0.0921051296468588\\
	0.00390625	0.0924436265646691\\
	0.0078125	0.0933164464878822\\
	0.015625	0.0945720375656281\\
	0.03125	0.0967780498769261\\
	0.0625	0.0989755595725956\\
	0.125	0.101610212960691\\
	0.25	0.10400341577929\\
	0.5	0.108987063732336\\
	1	0.121557058219855\\
	2	0.133687011815912\\
	4	0.175445388146077\\
	8	0.234417719951503\\
	16	0.27706320852101\\
	32	0.323356562429689\\
	64	0.409310696436868\\
	128	0.504609771443874\\
	256	0.600207857897595\\
	512	0.686213235396395\\
	1024	0.764977999136142\\
};
\addlegendentry{$\eps=2^{-6}$}

\addplot [color=mycolor4, dashed]
  table[row sep=crcr]{%
	0.0009765625	0.0408281741247965\\
	0.001953125	0.0409574452793695\\
	0.00390625	0.0411719140338841\\
	0.0078125	0.0417531211769939\\
	0.015625	0.0430569890561614\\
	0.03125	0.0462501082507296\\
	0.0625	0.0497783570805933\\
	0.125	0.0516681825983036\\
	0.25	0.0528890926657542\\
	0.5	0.0563889889325521\\
	1	0.0641974021420421\\
	2	0.06743161799967\\
	4	0.0823719000985463\\
	8	0.11845235111267\\
	16	0.158744369883681\\
	32	0.202775401991801\\
	64	0.257995666758781\\
	128	0.328549897395589\\
	256	0.414005013963527\\
	512	0.492719325846861\\
	1024	0.58575408223647\\
};
\addlegendentry{$\eps=2^{-7}$}

\addplot [color=mycolor5, dashed]
  table[row sep=crcr]{%
	0.0009765625	0.0191095081881223\\
	0.001953125	0.0191130103070752\\
	0.00390625	0.0191119043642156\\
	0.0078125	0.0191180873640937\\
	0.015625	0.0191315202289676\\
	0.03125	0.0191549403176395\\
	0.0625	0.0192092086714283\\
	0.125	0.0193781969535472\\
	0.25	0.0230229425078705\\
	0.5	0.026610704306073\\
	1	0.0337434516431123\\
	2	0.0425744118572329\\
	4	0.0606720629956803\\
	8	0.073117908011248\\
	16	0.0836536752048233\\
	32	0.112346214563422\\
	64	0.148479139077326\\
	128	0.197259632880821\\
	256	0.255189954521249\\
	512	0.32603057829341\\
	1024	0.407875149346344\\
};
\addlegendentry{$\eps=2^{-8}$}


\addplot [color=mycolor2, very thick]
table[row sep=crcr]{%
	0.25	0.270522695476369\\
0.5	0.287746460050304\\
1	0.306261406740995\\
2	0.327337786298677\\
4	0.354919250649841\\
};

\addplot [color=mycolor3, very thick]
table[row sep=crcr]{%
	0.25	0.10400341577929\\
0.5	0.108987063732336\\
1	0.121557058219855\\
2	0.133687011815912\\
4	0.175445388146077\\
};

\addplot [color=mycolor4, very thick]
table[row sep=crcr]{%
	0.25	0.0528890926657542\\
0.5	0.0563889889325521\\
1	0.0641974021420421\\
2	0.06743161799967\\
4	0.0823719000985463\\
};

\addplot [color=mycolor5, very thick]
table[row sep=crcr]{%
	0.25	0.0230229425078705\\
0.5	0.026610704306073\\
1	0.0337434516431123\\
2	0.0425744118572329\\
4	0.0606720629956803\\
};


\addplot [color=mycolor2, mark=*, mark size=2.5pt]
table[row sep=crcr]{%
1	0.306261406740995\\
};

\addplot [color=mycolor2, mark=square, mark size=3.5pt, only marks]
table[row sep=crcr]{%
0.0625	0.247892403241705\\
8	0.406152864995499\\
128	0.706613660283701\\
};

\addplot [color=mycolor3, mark=*, mark size=2.5pt]
table[row sep=crcr]{%
1	0.121557058219855\\
};

\addplot [color=mycolor4, mark=*, mark size=2.5pt]
table[row sep=crcr]{%
1	0.0641974021420421\\
};

\addplot [color=mycolor5, mark=*, mark size=2.5pt]
table[row sep=crcr]{%
1	0.0337434516431123\\
};

\end{axis}

\node[mycolor2] at (0.316, 1.58) {a};	
\node[mycolor2] at (2.48, 2.4) {b};	
\node[mycolor2] at (3.73, 3.9) {c};	
\end{tikzpicture}%

%% file: pics/fig_L1vsKappa_d3_h7.tex
%
%
%
\begin{tikzpicture}

\begin{axis}[%
width=1.6in,
height=2.0in,
at={(0.758in,0.481in)},
scale only axis,
xmode=log,
xmin=0.003,
xmax=12,
xminorticks=true,
xlabel style={font=\color{white!15!black}},
xlabel={$\kappa$},
ymin=0,
ymax=1,
yticklabel pos=right,
ylabel style={font=\color{white!15!black}},
axis background/.style={fill=white},
legend style={at={(0.25,0.94)},legend cell align=left, align=left, draw=white!15!black}
]


\addplot [color=mycolor2, dashed]
  table[row sep=crcr]{%
0.00390625	0.118064538253251\\
0.0078125	0.118911953601086\\
0.015625	0.120592686265185\\
0.03125	0.123902640405407\\
0.0625	0.135190572982736\\
0.125	0.147458005052191\\
0.25	0.174974335147918\\
0.5	0.209490893626072\\
1	0.264495859068696\\
2	0.331081890846303\\
4	0.407571250864064\\
8	0.484516880668689\\
16	0.555796913748613\\
32	0.633350341177488\\
64	0.707596491914789\\
128	0.770304468532542\\
256	0.827235624327313\\
};

\addplot [color=mycolor3, dashed]
  table[row sep=crcr]{%
0.00390625	0.0254543635866193\\
0.0078125	0.0288922944296693\\
0.015625	0.0338009486726643\\
0.03125	0.0383811892585833\\
0.0625	0.0436032322163382\\
0.125	0.056256719041736\\
0.25	0.0790032871854309\\
0.5	0.104759867773795\\
1	0.137324046997375\\
2	0.178639570655686\\
4	0.237857575571794\\
8	0.306701170491741\\
16	0.384637165958594\\
32	0.465689880333161\\
64	0.551632967622099\\
128	0.635679911715181\\
256	0.714065095363752\\
};

\addplot [color=mycolor4, dashed]
table[row sep=crcr]{%
0.0009765625	0.00694034090955334\\
0.001953125	0.00696431333216911\\
0.00390625	0.00701179770608896\\
0.0078125	0.00715136858990861\\
0.015625	0.00800905172881247\\
0.03125	0.00965450257234612\\
0.0625	0.0116837306586576\\
0.125	0.017460312977961\\
0.25	0.0267518899359518\\
0.5	0.0370706747805829\\
1	0.0535601353999174\\
2	0.0721243301255077\\
4	0.0970244619362574\\
8	0.132774026406686\\
16	0.178528931000715\\
};



\addplot [color=mycolor2, very thick]
table[row sep=crcr]{%
0.0078125	0.118911953601086\\
0.015625	0.120592686265185\\
0.03125	0.123902640405407\\
0.0625	0.135190572982736\\
0.125	0.147458005052191\\
};

\addplot [color=mycolor3, very thick]
table[row sep=crcr]{%
0.00390625	0.0254543635866193\\
0.0078125	0.0288922944296693\\
0.015625	0.0338009486726643\\
0.03125	0.0383811892585833\\
0.0625	0.0436032322163382\\
};

\addplot [color=mycolor4, very thick]
table[row sep=crcr]{%
	0.001953125	0.00696431333216911\\
	0.00390625	0.00701179770608896\\
	0.0078125	0.00715136858990861\\
	0.015625	0.00800905172881247\\
	0.03125	0.00965450257234612\\
};



\addplot [color=mycolor2, mark=*, mark size=2.5pt]
table[row sep=crcr]{%
0.03125	0.123902640405407\\
};

\addplot [color=mycolor3, mark=*, mark size=2.5pt]
table[row sep=crcr]{%
0.015625	0.0338009486726643\\
};

\addplot [color=mycolor4, mark=*, mark size=2.5pt]
table[row sep=crcr]{%
	0.0078125	0.00715136858990861\\
};

\end{axis}
\end{tikzpicture}%

%% file: pics/gs_eps/gs_eps_1_1_1.tex
%
%
%
\begin{tikzpicture}
\begin{axis}[%
width=0.9in,
height=0.9in,
scale only axis,
xmin=0,
xmax=1,
xtick={0,0.5,1},
ymin=0,
ymax=7,
yminorticks=true,
ytick={0,3,6},
axis background/.style={fill=white},
]
\addplot [color=mycolor2, thick]
  table[row sep=crcr]{%
0	0\\
0.001953125	0.0162968264483504\\
0.00390625	0.0332036562850577\\
0.005859375	0.052201398055673\\
0.0078125	0.075170761648831\\
0.009765625	0.0989952853652475\\
0.01171875	0.127163849164334\\
0.013671875	0.169858516061155\\
0.015625	0.210524166075665\\
0.017578125	0.24682500101037\\
0.01953125	0.293941129853424\\
0.021484375	0.361518943471703\\
0.0234375	0.464718978057499\\
0.025390625	0.604443011514064\\
0.02734375	0.797255139321274\\
0.029296875	0.988453436976515\\
0.03125	1.11468449715803\\
0.033203125	1.20579785204245\\
0.03515625	1.35420048699472\\
0.037109375	1.52633109249464\\
0.0390625	1.60636718736177\\
0.041015625	1.58242440096181\\
0.04296875	1.48489725187453\\
0.044921875	1.465653435392\\
0.046875	1.48987093406903\\
0.048828125	1.39979751771949\\
0.05078125	1.30379391898082\\
0.052734375	1.22052151695335\\
0.0546875	1.11935469541375\\
0.056640625	1.0377355798899\\
0.05859375	0.954608290220468\\
0.060546875	0.839278122543066\\
0.0625	0.713573206392048\\
0.064453125	0.646495962745564\\
0.06640625	0.65165187091835\\
0.068359375	0.71800043035885\\
0.0703125	0.802880716656897\\
0.072265625	0.88269733137291\\
0.07421875	0.976400971766321\\
0.076171875	1.02507312336994\\
0.078125	0.957630387704345\\
0.080078125	0.903903226016766\\
0.08203125	0.929248082678784\\
0.083984375	0.911909250010559\\
0.0859375	0.774580565581677\\
0.087890625	0.610267898500601\\
0.08984375	0.510581981169514\\
0.091796875	0.454065347060731\\
0.09375	0.398577858056017\\
0.095703125	0.34186240764115\\
0.09765625	0.312162323407253\\
0.099609375	0.28802421011632\\
0.1015625	0.257196630417516\\
0.103515625	0.230167802756616\\
0.10546875	0.212604246413889\\
0.107421875	0.217377065727476\\
0.109375	0.231712589232467\\
0.111328125	0.245054614622204\\
0.11328125	0.248749172423604\\
0.115234375	0.257807307924349\\
0.1171875	0.269637190757833\\
0.119140625	0.268465149794168\\
0.12109375	0.272206607478671\\
0.123046875	0.29476591065264\\
0.125	0.312702681880557\\
0.126953125	0.316027516571566\\
0.12890625	0.335018489848996\\
0.130859375	0.383056435941834\\
0.1328125	0.432230598394799\\
0.134765625	0.491235978910633\\
0.13671875	0.56240840823923\\
0.138671875	0.647175109127963\\
0.140625	0.723420443975961\\
0.142578125	0.772944918195635\\
0.14453125	0.885354266712797\\
0.146484375	1.05594849212284\\
0.1484375	1.24217061600427\\
0.150390625	1.3544189038745\\
0.15234375	1.38592953331347\\
0.154296875	1.37795665774989\\
0.15625	1.29155764308792\\
0.158203125	1.15927707184053\\
0.16015625	1.11074785218931\\
0.162109375	1.10732228077703\\
0.1640625	1.12592645966277\\
0.166015625	1.21367456025148\\
0.16796875	1.26060599329414\\
0.169921875	1.21477854576294\\
0.171875	1.12428898544175\\
0.173828125	0.997081281350296\\
0.17578125	0.835155497641413\\
0.177734375	0.701675144395448\\
0.1796875	0.635529741204269\\
0.181640625	0.565048376452848\\
0.18359375	0.485689119771994\\
0.185546875	0.454652896794956\\
0.1875	0.412901120956568\\
0.189453125	0.364231420322548\\
0.19140625	0.32950913501939\\
0.193359375	0.300208166364268\\
0.1953125	0.285971496090635\\
0.197265625	0.296260953824165\\
0.19921875	0.316152399254788\\
0.201171875	0.347384472128655\\
0.203125	0.394192832378225\\
0.205078125	0.459165883713161\\
0.20703125	0.564432478620932\\
0.208984375	0.675906325778937\\
0.2109375	0.727385417356779\\
0.212890625	0.717698593620555\\
0.21484375	0.699912194156134\\
0.216796875	0.690077993384416\\
0.21875	0.651524762332783\\
0.220703125	0.621235641018791\\
0.22265625	0.639312269436258\\
0.224609375	0.684499299024567\\
0.2265625	0.770375908230769\\
0.228515625	0.934679453506639\\
0.23046875	1.13407257286128\\
0.232421875	1.2498434240847\\
0.234375	1.36604440633633\\
0.236328125	1.53062200823648\\
0.23828125	1.67972085601507\\
0.240234375	1.76560201743349\\
0.2421875	1.76026341739151\\
0.244140625	1.70899967696697\\
0.24609375	1.56740871867486\\
0.248046875	1.41402879477357\\
0.25	1.28096259335899\\
0.251953125	1.15575316076133\\
0.25390625	1.10224435192385\\
0.255859375	1.12319010030029\\
0.2578125	1.2587646211114\\
0.259765625	1.40766496168516\\
0.26171875	1.44720816008044\\
0.263671875	1.37461555060753\\
0.265625	1.17956718837369\\
0.267578125	0.923432341465747\\
0.26953125	0.71973396393451\\
0.271484375	0.599285967539612\\
0.2734375	0.538890058814449\\
0.275390625	0.531970145049767\\
0.27734375	0.518917380447376\\
0.279296875	0.505484495988408\\
0.28125	0.551127868359609\\
0.283203125	0.611337043628875\\
0.28515625	0.656058859674405\\
0.287109375	0.706489859534956\\
0.2890625	0.725202825219514\\
0.291015625	0.714500134991697\\
0.29296875	0.76173904330868\\
0.294921875	0.920418806766836\\
0.296875	1.08685900572525\\
0.298828125	1.12217398688674\\
0.30078125	1.04714944413387\\
0.302734375	0.9647906295518\\
0.3046875	0.873270964808737\\
0.306640625	0.784031888400805\\
0.30859375	0.704430350055243\\
0.310546875	0.614820551105368\\
0.3125	0.561621355364111\\
0.314453125	0.531042004252\\
0.31640625	0.493248246363782\\
0.318359375	0.456497455892666\\
0.3203125	0.429193235571886\\
0.322265625	0.413797622255254\\
0.32421875	0.437627046898232\\
0.326171875	0.51658632523021\\
0.328125	0.625560892477774\\
0.330078125	0.685504657221076\\
0.33203125	0.709034519194789\\
0.333984375	0.776388597766946\\
0.3359375	0.921473536095435\\
0.337890625	1.05678864753247\\
0.33984375	1.0811891294882\\
0.341796875	1.06955576078082\\
0.34375	1.01599149916829\\
0.345703125	0.915257879360697\\
0.34765625	0.82660193378329\\
0.349609375	0.796676470886979\\
0.3515625	0.830909253378484\\
0.353515625	0.871501730972556\\
0.35546875	0.877726008648251\\
0.357421875	0.877089743050153\\
0.359375	0.944994957444516\\
0.361328125	1.01901399014663\\
0.36328125	0.985042679114034\\
0.365234375	0.910592752254842\\
0.3671875	0.795362548615897\\
0.369140625	0.66522314463004\\
0.37109375	0.590477532732257\\
0.373046875	0.522464729813949\\
0.375	0.470325673003583\\
0.376953125	0.466445870835667\\
0.37890625	0.507521580516592\\
0.380859375	0.596646710736292\\
0.3828125	0.752721690595471\\
0.384765625	0.941191743376924\\
0.38671875	1.07589707371823\\
0.388671875	1.10331140063269\\
0.390625	1.01092995763475\\
0.392578125	0.842375100546513\\
0.39453125	0.715796672992429\\
0.396484375	0.664436854366518\\
0.3984375	0.625096979072975\\
0.400390625	0.617717616841884\\
0.40234375	0.691879881446218\\
0.404296875	0.796484348545229\\
0.40625	0.874101956228489\\
0.408203125	0.971840980368984\\
0.41015625	1.08048873176922\\
0.412109375	1.23690070095426\\
0.4140625	1.40191974490762\\
0.416015625	1.44029926390382\\
0.41796875	1.45880077460331\\
0.419921875	1.46797820469175\\
0.421875	1.36857211593928\\
0.423828125	1.20315048384269\\
0.42578125	1.04667252070632\\
0.427734375	0.896337954583024\\
0.4296875	0.794562919590159\\
0.431640625	0.791703611198086\\
0.43359375	0.796343545230079\\
0.435546875	0.721039108563337\\
0.4375	0.617494784324596\\
0.439453125	0.555429001113786\\
0.44140625	0.504596088491538\\
0.443359375	0.467755471195395\\
0.4453125	0.427073358802001\\
0.447265625	0.374821800691639\\
0.44921875	0.354545007821882\\
0.451171875	0.357433906190921\\
0.453125	0.36822879789716\\
0.455078125	0.390766023314143\\
0.45703125	0.447075052052221\\
0.458984375	0.520410007525921\\
0.4609375	0.622238233849332\\
0.462890625	0.728615652191161\\
0.46484375	0.809562952984431\\
0.466796875	0.832963189087248\\
0.46875	0.844408595465466\\
0.470703125	0.902688633877776\\
0.47265625	0.958878945408813\\
0.474609375	0.956891956560329\\
0.4765625	0.937172136746805\\
0.478515625	0.945666525375772\\
0.48046875	0.919513765845716\\
0.482421875	0.859468822881539\\
0.484375	0.843075035198382\\
0.486328125	0.944904884144915\\
0.48828125	1.07457935522714\\
0.490234375	1.15474756004949\\
0.4921875	1.19067129098193\\
0.494140625	1.20911971926437\\
0.49609375	1.29848053118187\\
0.498046875	1.39788974902799\\
0.5	1.52882353530454\\
0.501953125	1.68204357301923\\
0.50390625	1.77103482328531\\
0.505859375	1.78953189878064\\
0.5078125	1.6982798177216\\
0.509765625	1.5589559018603\\
0.51171875	1.52552963151137\\
0.513671875	1.52988836567103\\
0.515625	1.3839029991265\\
0.517578125	1.19913109026003\\
0.51953125	1.03750617183744\\
0.521484375	0.906064580966872\\
0.5234375	0.837146736501185\\
0.525390625	0.841063120409827\\
0.52734375	0.864806380822028\\
0.529296875	0.853777534296676\\
0.53125	0.847169852493074\\
0.533203125	0.84043649200046\\
0.53515625	0.818454237377819\\
0.537109375	0.791018886922931\\
0.5390625	0.784879600812893\\
0.541015625	0.836520226901744\\
0.54296875	0.894113740259554\\
0.544921875	0.877782090745837\\
0.546875	0.81295897110734\\
0.548828125	0.740034684778102\\
0.55078125	0.644627252164986\\
0.552734375	0.565757879189973\\
0.5546875	0.509456638707758\\
0.556640625	0.446799983646825\\
0.55859375	0.419801033466144\\
0.560546875	0.406780973387949\\
0.5625	0.403992426329592\\
0.564453125	0.41751240405334\\
0.56640625	0.446064976665106\\
0.568359375	0.519470343781016\\
0.5703125	0.613797263114057\\
0.572265625	0.715865631158167\\
0.57421875	0.795514032956832\\
0.576171875	0.812672721060296\\
0.578125	0.800569334068874\\
0.580078125	0.811977021780766\\
0.58203125	0.888527198988352\\
0.583984375	0.981090738399101\\
0.5859375	1.00081825574191\\
0.587890625	0.985000162089134\\
0.58984375	1.03509319014147\\
0.591796875	1.18455728257961\\
0.59375	1.34298805654217\\
0.595703125	1.37942335273961\\
0.59765625	1.33652056124664\\
0.599609375	1.25250568188288\\
0.6015625	1.16456129575316\\
0.603515625	1.2125987692625\\
0.60546875	1.38875240801999\\
0.607421875	1.55141821596011\\
0.609375	1.62364217630065\\
0.611328125	1.57293394278523\\
0.61328125	1.38951403543776\\
0.615234375	1.20348830508078\\
0.6171875	1.07881642570399\\
0.619140625	1.00420347853706\\
0.62109375	1.04374128985405\\
0.623046875	1.18537265142424\\
0.625	1.27852929743603\\
0.626953125	1.2368460613053\\
0.62890625	1.23344328448846\\
0.630859375	1.40423972762341\\
0.6328125	1.58148370052961\\
0.634765625	1.60227379605123\\
0.63671875	1.48373142375036\\
0.638671875	1.26421006712574\\
0.640625	1.00554912327306\\
0.642578125	0.779144850969001\\
0.64453125	0.621933098626062\\
0.646484375	0.503893647505593\\
0.6484375	0.41257596686054\\
0.650390625	0.349979956921645\\
0.65234375	0.306750948821762\\
0.654296875	0.290413408133149\\
0.65625	0.304249102641596\\
0.658203125	0.353195342124466\\
0.66015625	0.43917115884269\\
0.662109375	0.553639743656997\\
0.6640625	0.708328086484619\\
0.666015625	0.910274721988337\\
0.66796875	1.06643145165044\\
0.669921875	1.15573299814029\\
0.671875	1.19034210512644\\
0.673828125	1.10983267183623\\
0.67578125	1.04165607401555\\
0.677734375	1.04700342414421\\
0.6796875	1.03822150016587\\
0.681640625	0.981430762361108\\
0.68359375	0.916612028330359\\
0.685546875	0.908948989076147\\
0.6875	0.871151499326701\\
0.689453125	0.827052823090376\\
0.69140625	0.837417911602073\\
0.693359375	0.839125798023445\\
0.6953125	0.840759812192889\\
0.697265625	0.874914521565687\\
0.69921875	0.962841202070388\\
0.701171875	1.03732192483711\\
0.703125	1.06918224153403\\
0.705078125	1.05732655321871\\
0.70703125	1.01732370803781\\
0.708984375	0.950179464748819\\
0.7109375	0.891738876836102\\
0.712890625	0.846750180499282\\
0.71484375	0.819640064653063\\
0.716796875	0.829740962470209\\
0.71875	0.848230227205225\\
0.720703125	0.886571209795079\\
0.72265625	0.908331412320872\\
0.724609375	0.913307086733154\\
0.7265625	0.894330263060687\\
0.728515625	0.884092205368617\\
0.73046875	0.903640023133804\\
0.732421875	0.989159423002575\\
0.734375	1.18425231918317\\
0.736328125	1.37545590659573\\
0.73828125	1.44383233986672\\
0.740234375	1.37833878595057\\
0.7421875	1.22968530164321\\
0.744140625	1.11840470664726\\
0.74609375	1.0594066356367\\
0.748046875	1.05475303392502\\
0.75	1.103610511778\\
0.751953125	1.12954541228079\\
0.75390625	1.1706546740387\\
0.755859375	1.17546248135025\\
0.7578125	1.12291757967773\\
0.759765625	1.13404451246235\\
0.76171875	1.1578433687475\\
0.763671875	1.19905315120169\\
0.765625	1.2533176210257\\
0.767578125	1.25633973925141\\
0.76953125	1.19693049883453\\
0.771484375	1.08629329310245\\
0.7734375	0.967291133928341\\
0.775390625	0.861559622186067\\
0.77734375	0.833517076395168\\
0.779296875	0.848714510367029\\
0.78125	0.839237797790039\\
0.783203125	0.864721010928389\\
0.78515625	0.874824681075193\\
0.787109375	0.868845486973189\\
0.7890625	0.95372184815864\\
0.791015625	1.0589045558246\\
0.79296875	1.16265763759616\\
0.794921875	1.27250511040634\\
0.796875	1.29373358753193\\
0.798828125	1.27796839546075\\
0.80078125	1.33661511636106\\
0.802734375	1.39389345612348\\
0.8046875	1.37873068078865\\
0.806640625	1.2782846460967\\
0.80859375	1.10819153481165\\
0.810546875	0.935550241978818\\
0.8125	0.824311344991573\\
0.814453125	0.777693153723298\\
0.81640625	0.767824667859975\\
0.818359375	0.761549412283084\\
0.8203125	0.754825399190239\\
0.822265625	0.798788965877505\\
0.82421875	0.863497629050343\\
0.826171875	0.943673522566417\\
0.828125	1.00480556738779\\
0.830078125	1.03696309308646\\
0.83203125	1.08612700993544\\
0.833984375	1.13194997958018\\
0.8359375	1.24203083412403\\
0.837890625	1.35060032900452\\
0.83984375	1.38002168876494\\
0.841796875	1.3975267368127\\
0.84375	1.34953690279519\\
0.845703125	1.30649819889512\\
0.84765625	1.35514060787302\\
0.849609375	1.4917608883026\\
0.8515625	1.6201409869268\\
0.853515625	1.64908329640739\\
0.85546875	1.64915454594741\\
0.857421875	1.56224707904665\\
0.859375	1.50705758744994\\
0.861328125	1.50770344936898\\
0.86328125	1.4499751064022\\
0.865234375	1.41154236112456\\
0.8671875	1.39425631444775\\
0.869140625	1.3163478101136\\
0.87109375	1.23694793974362\\
0.873046875	1.1945458709813\\
0.875	1.21120193980145\\
0.876953125	1.24113931898391\\
0.87890625	1.30485273187754\\
0.880859375	1.42182256737257\\
0.8828125	1.42298096486613\\
0.884765625	1.30667327106766\\
0.88671875	1.22007521722742\\
0.888671875	1.28674079514863\\
0.890625	1.35811213984207\\
0.892578125	1.2774443820122\\
0.89453125	1.22168574922867\\
0.896484375	1.21694034376832\\
0.8984375	1.14943212031062\\
0.900390625	1.10917222579845\\
0.90234375	1.13611570540813\\
0.904296875	1.26789500736488\\
0.90625	1.41645550415562\\
0.908203125	1.42490251246112\\
0.91015625	1.48627017472925\\
0.912109375	1.59608427623161\\
0.9140625	1.5866492334156\\
0.916015625	1.55316614074151\\
0.91796875	1.50865635191172\\
0.919921875	1.38529816879686\\
0.921875	1.27000588636385\\
0.923828125	1.25269859648716\\
0.92578125	1.26492705043318\\
0.927734375	1.30000996005586\\
0.9296875	1.36896062751511\\
0.931640625	1.45021564543107\\
0.93359375	1.52241934849808\\
0.935546875	1.54538542566378\\
0.9375	1.59460937109205\\
0.939453125	1.66423164436322\\
0.94140625	1.63341409733692\\
0.943359375	1.50016597933041\\
0.9453125	1.39296096854733\\
0.947265625	1.33225190474662\\
0.94921875	1.21648134360466\\
0.951171875	1.07398249451805\\
0.953125	0.977145394789108\\
0.955078125	0.882191590820117\\
0.95703125	0.787646303750184\\
0.958984375	0.719179596898099\\
0.9609375	0.680351492709555\\
0.962890625	0.680046585041059\\
0.96484375	0.654632006746276\\
0.966796875	0.607218773338983\\
0.96875	0.597647952107396\\
0.970703125	0.608580704689725\\
0.97265625	0.62128197222789\\
0.974609375	0.598017495109787\\
0.9765625	0.526409828328538\\
0.978515625	0.475158538710084\\
0.98046875	0.455851452521361\\
0.982421875	0.425973305641747\\
0.984375	0.374740314860744\\
0.986328125	0.314033634061757\\
0.98828125	0.247240556865671\\
0.990234375	0.182588011619137\\
0.9921875	0.127451907748408\\
0.994140625	0.087166576835966\\
0.99609375	0.0548207474953839\\
0.998046875	0.0267140661470384\\
1	0\\
};

\end{axis}
\end{tikzpicture}%

%% file: pics/gs_eps/gs_eps_2_1_1.tex
%
%
\begin{tikzpicture}
\begin{axis}[%
width=0.9in,
height=0.9in,
scale only axis,
xmin=0,
xmax=1,
xtick={0,0.5,1},
ymin=0,
ymax=7,
yminorticks=true,
yticklabels = {},
axis background/.style={fill=white},
]
\addplot [color=mycolor2, thick]
  table[row sep=crcr]{%
0	0\\
0.001953125	0.260353247861612\\
0.00390625	0.530018784258725\\
0.005859375	0.687185295799298\\
0.0078125	0.798531385802848\\
0.009765625	1.03164582855659\\
0.01171875	1.08434783322266\\
0.013671875	1.28773818058341\\
0.015625	1.30611836202585\\
0.017578125	1.40684437562316\\
0.01953125	1.32707813799836\\
0.021484375	0.908726627355834\\
0.0234375	0.573817168839666\\
0.025390625	0.436840014844712\\
0.02734375	0.513103816443675\\
0.029296875	0.788432632529226\\
0.03125	1.13442972586095\\
0.033203125	1.26707310245551\\
0.03515625	1.0754494532601\\
0.037109375	0.755872239572572\\
0.0390625	0.50091512912344\\
0.041015625	0.416910960372447\\
0.04296875	0.433039850735485\\
0.044921875	0.583694170642766\\
0.046875	0.775302002406753\\
0.048828125	0.950510390030491\\
0.05078125	1.08522977180053\\
0.052734375	1.166830569151\\
0.0546875	1.15512162059526\\
0.056640625	1.15698853785474\\
0.05859375	1.06832260125473\\
0.060546875	1.00743228283037\\
0.0625	0.96674572623492\\
0.064453125	1.17551855503483\\
0.06640625	1.42032714816033\\
0.068359375	1.33266436926508\\
0.0703125	1.15736089645678\\
0.072265625	1.02233267422477\\
0.07421875	0.981809946157371\\
0.076171875	0.93493058236786\\
0.078125	0.903798186079574\\
0.080078125	1.09712955914998\\
0.08203125	1.2140521658705\\
0.083984375	1.10246259924878\\
0.0859375	1.00713569253387\\
0.087890625	1.18183286239723\\
0.08984375	1.4996203452925\\
0.091796875	1.77610036717972\\
0.09375	1.65839611454957\\
0.095703125	1.48599814449822\\
0.09765625	1.48821327902141\\
0.099609375	1.38520979851876\\
0.1015625	1.25232871820532\\
0.103515625	0.850729030947443\\
0.10546875	0.766656109447956\\
0.107421875	0.72818090771902\\
0.109375	0.887571117675859\\
0.111328125	1.05144367994369\\
0.11328125	1.24299996564813\\
0.115234375	1.19897360966615\\
0.1171875	1.13158005189641\\
0.119140625	0.890831489157612\\
0.12109375	0.743569117898373\\
0.123046875	0.701528520864475\\
0.125	0.739492804403106\\
0.126953125	0.922772663977374\\
0.12890625	1.07162066049215\\
0.130859375	1.41214325394533\\
0.1328125	1.70163719740113\\
0.134765625	1.87153292157237\\
0.13671875	2.13127075931706\\
0.138671875	2.4237544798746\\
0.140625	2.03568683720804\\
0.142578125	1.75666372099718\\
0.14453125	1.87142286561501\\
0.146484375	1.6263844927493\\
0.1484375	1.27481054598652\\
0.150390625	1.23249455192508\\
0.15234375	1.02312385786505\\
0.154296875	0.741568706649776\\
0.15625	0.558243332165871\\
0.158203125	0.518081459279098\\
0.16015625	0.462845943550242\\
0.162109375	0.523379530548509\\
0.1640625	0.647010605487056\\
0.166015625	0.754352515570055\\
0.16796875	0.791304745954193\\
0.169921875	1.09735789188714\\
0.171875	1.37564028838031\\
0.173828125	1.59157804160607\\
0.17578125	1.82276789222694\\
0.177734375	2.12513571711187\\
0.1796875	2.07139849616685\\
0.181640625	1.72148578210258\\
0.18359375	1.53075848354452\\
0.185546875	1.50978941897473\\
0.1875	1.64076347724774\\
0.189453125	1.63704964918361\\
0.19140625	1.57017774578917\\
0.193359375	1.20585538645121\\
0.1953125	0.923061953029064\\
0.197265625	0.870680225801487\\
0.19921875	0.950623132306484\\
0.201171875	0.909601083094849\\
0.203125	0.760464129048161\\
0.205078125	0.776624678039708\\
0.20703125	0.787707333660498\\
0.208984375	0.883381110916131\\
0.2109375	1.02914300397953\\
0.212890625	1.05974713468899\\
0.21484375	1.01927303439623\\
0.216796875	0.982342291584834\\
0.21875	1.14457987008621\\
0.220703125	1.49778505493086\\
0.22265625	1.41457542149067\\
0.224609375	1.22123768999618\\
0.2265625	1.17425393865929\\
0.228515625	1.00162703471704\\
0.23046875	0.772363070636924\\
0.232421875	0.669986353741925\\
0.234375	0.598702009336926\\
0.236328125	0.541917244062224\\
0.23828125	0.526940199775157\\
0.240234375	0.622397652815653\\
0.2421875	0.770355146249986\\
0.244140625	0.775051836480146\\
0.24609375	0.789715776974273\\
0.248046875	0.7755565780858\\
0.25	0.788393462947675\\
0.251953125	0.698094501013849\\
0.25390625	0.549286662970317\\
0.255859375	0.39110880928435\\
0.2578125	0.328906749306474\\
0.259765625	0.266096707415099\\
0.26171875	0.209449499932498\\
0.263671875	0.146760650840227\\
0.265625	0.120106078786211\\
0.267578125	0.0959224049307621\\
0.26953125	0.0652469152831933\\
0.271484375	0.0670321984137151\\
0.2734375	0.088146685007118\\
0.275390625	0.0953587068423787\\
0.27734375	0.117420559497359\\
0.279296875	0.161444369704547\\
0.28125	0.186371204843275\\
0.283203125	0.234214475874026\\
0.28515625	0.305280704390068\\
0.287109375	0.338373646848018\\
0.2890625	0.393815964826821\\
0.291015625	0.45836982435132\\
0.29296875	0.454142254011295\\
0.294921875	0.496660186705747\\
0.296875	0.557150197864839\\
0.298828125	0.673602132019783\\
0.30078125	0.641922699731688\\
0.302734375	0.561529283570855\\
0.3046875	0.425111768568414\\
0.306640625	0.347458947448867\\
0.30859375	0.311487572781884\\
0.310546875	0.343241506732638\\
0.3125	0.363181243361294\\
0.314453125	0.29174720707445\\
0.31640625	0.288984409408805\\
0.318359375	0.34060624003787\\
0.3203125	0.423242059094377\\
0.322265625	0.5247499294809\\
0.32421875	0.773572220613858\\
0.326171875	0.785724534718643\\
0.328125	0.668157809045034\\
0.330078125	0.682706624780646\\
0.33203125	0.562422077845268\\
0.333984375	0.496111787785188\\
0.3359375	0.550763250464191\\
0.337890625	0.548786784711135\\
0.33984375	0.634642277142817\\
0.341796875	0.771319609198309\\
0.34375	0.922681208600535\\
0.345703125	1.05812506615269\\
0.34765625	1.31911364152476\\
0.349609375	1.50986165876357\\
0.3515625	1.4545289755668\\
0.353515625	1.44873735544639\\
0.35546875	1.303553185411\\
0.357421875	1.46340411356669\\
0.359375	1.30477836533714\\
0.361328125	1.14480500601748\\
0.36328125	0.992909420777517\\
0.365234375	0.94753551892764\\
0.3671875	1.08273887313036\\
0.369140625	0.866184975321594\\
0.37109375	0.672104160371802\\
0.373046875	0.616415873290473\\
0.375	0.690882361914966\\
0.376953125	0.85081450670972\\
0.37890625	1.25823666209243\\
0.380859375	1.66723754294176\\
0.3828125	1.7532242343156\\
0.384765625	1.63990737023312\\
0.38671875	1.54548239625399\\
0.388671875	1.86146390768819\\
0.390625	1.86022031363035\\
0.392578125	1.60738966254474\\
0.39453125	1.42802698753478\\
0.396484375	1.42670047846992\\
0.3984375	1.32494893169866\\
0.400390625	1.59691195920712\\
0.40234375	1.53290851518667\\
0.404296875	1.40113077740265\\
0.40625	1.41224678111562\\
0.408203125	1.55729523991276\\
0.41015625	1.72771056938621\\
0.412109375	1.70398036628352\\
0.4140625	1.37697676354762\\
0.416015625	1.07151244129129\\
0.41796875	0.85304990525486\\
0.419921875	0.759652924171891\\
0.421875	0.928921765613737\\
0.423828125	1.17364763548862\\
0.42578125	1.39832616535151\\
0.427734375	1.17557836857348\\
0.4296875	0.811370761425665\\
0.431640625	0.545938302836098\\
0.43359375	0.37030284405951\\
0.435546875	0.257152457952055\\
0.4375	0.17597382688884\\
0.439453125	0.135726176558705\\
0.44140625	0.109336966620715\\
0.443359375	0.0985743598096885\\
0.4453125	0.102485745276258\\
0.447265625	0.121931778780104\\
0.44921875	0.167180081789128\\
0.451171875	0.223774104026332\\
0.453125	0.33384326005459\\
0.455078125	0.400686731531655\\
0.45703125	0.44243958553191\\
0.458984375	0.619897588179476\\
0.4609375	0.874818388542587\\
0.462890625	0.787784611809549\\
0.46484375	0.657298727450609\\
0.466796875	0.634833264910773\\
0.46875	0.695033178878272\\
0.470703125	0.613251131713636\\
0.47265625	0.59900212269487\\
0.474609375	0.754413115844621\\
0.4765625	0.928204624191231\\
0.478515625	1.18310787614209\\
0.48046875	1.50442357624621\\
0.482421875	1.74806494608286\\
0.484375	1.68407656627825\\
0.486328125	1.67895425943103\\
0.48828125	1.71163449867973\\
0.490234375	1.77519966086964\\
0.4921875	1.67976562969496\\
0.494140625	1.31129211303315\\
0.49609375	1.10987062208812\\
0.498046875	1.09589994880939\\
0.5	1.29453934277611\\
0.501953125	1.3187894990207\\
0.50390625	1.07957996980082\\
0.505859375	0.965113312497263\\
0.5078125	0.979954423855411\\
0.509765625	0.94328523959875\\
0.51171875	0.853778574078232\\
0.513671875	0.823555327769166\\
0.515625	0.823742786551437\\
0.517578125	0.892670612755608\\
0.51953125	1.11812645298771\\
0.521484375	1.43964347851899\\
0.5234375	1.53988623358245\\
0.525390625	1.53037911176583\\
0.52734375	1.45454932194237\\
0.529296875	1.33066095202637\\
0.53125	1.09437894799864\\
0.533203125	0.921324422538247\\
0.53515625	0.8944549541943\\
0.537109375	0.993185129678619\\
0.5390625	1.06495495965503\\
0.541015625	1.1665112232377\\
0.54296875	1.32316530591698\\
0.544921875	1.52882149870003\\
0.546875	1.64592910791406\\
0.548828125	1.71176996047948\\
0.55078125	1.94817912221648\\
0.552734375	2.03191816802385\\
0.5546875	1.65209495104893\\
0.556640625	1.29533436835439\\
0.55859375	1.14996821514486\\
0.560546875	0.928671753106918\\
0.5625	0.903505785623389\\
0.564453125	0.928943492984505\\
0.56640625	1.05955998573577\\
0.568359375	1.11393803542631\\
0.5703125	0.891910747499986\\
0.572265625	0.710002077618988\\
0.57421875	0.668128993473034\\
0.576171875	0.706375731774381\\
0.578125	0.79482940491555\\
0.580078125	0.780723751111806\\
0.58203125	1.01580852440012\\
0.583984375	1.44952256727399\\
0.5859375	1.48588812140484\\
0.587890625	1.39479345893824\\
0.58984375	1.15419814014656\\
0.591796875	1.21899430607489\\
0.59375	1.53554646527689\\
0.595703125	1.60487645791663\\
0.59765625	1.30830116246987\\
0.599609375	1.07595185171677\\
0.6015625	0.810195107303883\\
0.603515625	0.597987463010489\\
0.60546875	0.488887676644867\\
0.607421875	0.561797216988861\\
0.609375	0.547143899280578\\
0.611328125	0.578925962441556\\
0.61328125	0.726583824240895\\
0.615234375	0.843535239859906\\
0.6171875	0.934447100640213\\
0.619140625	0.945930434801005\\
0.62109375	0.996610172069554\\
0.623046875	0.868156813498343\\
0.625	0.625763009727636\\
0.626953125	0.414659661379946\\
0.62890625	0.374773117468485\\
0.630859375	0.389816504694123\\
0.6328125	0.431999215734589\\
0.634765625	0.392296562427483\\
0.63671875	0.355332212677394\\
0.638671875	0.317113800524331\\
0.640625	0.256905257676222\\
0.642578125	0.258739605013149\\
0.64453125	0.287281157348854\\
0.646484375	0.378905980492896\\
0.6484375	0.529228287676454\\
0.650390625	0.581253685359793\\
0.65234375	0.619067621868433\\
0.654296875	0.684144629927563\\
0.65625	0.701618231613834\\
0.658203125	0.807188141958182\\
0.66015625	1.12956612593011\\
0.662109375	1.14224247801689\\
0.6640625	0.940740252153642\\
0.666015625	0.982924391647065\\
0.66796875	0.882278561093036\\
0.669921875	0.639530377228133\\
0.671875	0.63683654059511\\
0.673828125	0.864305705564959\\
0.67578125	1.11392518729277\\
0.677734375	1.22084594805425\\
0.6796875	1.2651560081671\\
0.681640625	1.35311612383477\\
0.68359375	1.46407819494737\\
0.685546875	1.73007588199991\\
0.6875	1.63781076643656\\
0.689453125	1.59801229624218\\
0.69140625	1.57298149232804\\
0.693359375	1.36904694019751\\
0.6953125	1.18781857822025\\
0.697265625	1.10208793697264\\
0.69921875	0.897216631029346\\
0.701171875	0.885620417137209\\
0.703125	1.07006999870178\\
0.705078125	1.26493335056848\\
0.70703125	1.3294358868827\\
0.708984375	1.58135300016727\\
0.7109375	1.54410627267536\\
0.712890625	1.28101671229899\\
0.71484375	1.06518097086815\\
0.716796875	1.0521268509138\\
0.71875	0.870033293772013\\
0.720703125	0.860828360957061\\
0.72265625	0.854964428964099\\
0.724609375	0.809727732805807\\
0.7265625	0.867106315740826\\
0.728515625	0.875616482001547\\
0.73046875	0.844805358774307\\
0.732421875	0.770330045823473\\
0.734375	0.630568904871871\\
0.736328125	0.68365267616578\\
0.73828125	0.688170960735694\\
0.740234375	0.61557178262461\\
0.7421875	0.516272350740189\\
0.744140625	0.487661050895018\\
0.74609375	0.473666862714767\\
0.748046875	0.43819940228581\\
0.75	0.471587056346369\\
0.751953125	0.524438481824334\\
0.75390625	0.722765253071923\\
0.755859375	0.914162708777017\\
0.7578125	1.06395650328667\\
0.759765625	0.870597011129576\\
0.76171875	0.764121258414559\\
0.763671875	0.780508278995192\\
0.765625	0.675851270006497\\
0.767578125	0.510693484530792\\
0.76953125	0.331282666950376\\
0.771484375	0.265136590304684\\
0.7734375	0.218797793005176\\
0.775390625	0.188427894553296\\
0.77734375	0.18202206664576\\
0.779296875	0.167564629611963\\
0.78125	0.177791607031043\\
0.783203125	0.198734292549391\\
0.78515625	0.160041286684763\\
0.787109375	0.11437250875258\\
0.7890625	0.0875955327734978\\
0.791015625	0.0827325798995811\\
0.79296875	0.0762786927969889\\
0.794921875	0.0877697422651836\\
0.796875	0.0832352686084392\\
0.798828125	0.0712995829802836\\
0.80078125	0.0560568631245403\\
0.802734375	0.0444749104772493\\
0.8046875	0.0429686710709044\\
0.806640625	0.0377635602365471\\
0.80859375	0.0431612214948183\\
0.810546875	0.0563836464410625\\
0.8125	0.0562893578703566\\
0.814453125	0.0476906442186611\\
0.81640625	0.045586153356285\\
0.818359375	0.0486769394411379\\
0.8203125	0.0620666840098006\\
0.822265625	0.0905789754181731\\
0.82421875	0.109063346449099\\
0.826171875	0.148330201969944\\
0.828125	0.220486441255536\\
0.830078125	0.349231700602984\\
0.83203125	0.521512290980312\\
0.833984375	0.598945706827108\\
0.8359375	0.666476147134597\\
0.837890625	0.59725222368251\\
0.83984375	0.614702445087086\\
0.841796875	0.585174879221357\\
0.84375	0.550800359783348\\
0.845703125	0.463518661026134\\
0.84765625	0.418357756226623\\
0.849609375	0.541796923292562\\
0.8515625	0.87446540674003\\
0.853515625	1.23590245660079\\
0.85546875	1.22927378301225\\
0.857421875	1.2908982659228\\
0.859375	1.12267006480566\\
0.861328125	0.853990345104085\\
0.86328125	0.704153352650391\\
0.865234375	0.618523347902258\\
0.8671875	0.608104871821919\\
0.869140625	0.811994579899773\\
0.87109375	0.878370415648502\\
0.873046875	0.787701765277617\\
0.875	0.688582735974103\\
0.876953125	0.488038316762042\\
0.87890625	0.338777496937828\\
0.880859375	0.244244105029336\\
0.8828125	0.202840386930961\\
0.884765625	0.166069465839238\\
0.88671875	0.161564302438722\\
0.888671875	0.208922920302012\\
0.890625	0.291242402474641\\
0.892578125	0.404722728109715\\
0.89453125	0.613815998374765\\
0.896484375	0.955582718452101\\
0.8984375	1.35787096556135\\
0.900390625	1.66303337948497\\
0.90234375	1.62640781766979\\
0.904296875	1.53383500676307\\
0.90625	1.56927260915856\\
0.908203125	1.51395004060155\\
0.91015625	1.59026629826462\\
0.912109375	1.36386280579325\\
0.9140625	1.29012234225652\\
0.916015625	0.984526242856211\\
0.91796875	0.907442107577506\\
0.919921875	0.813932956466138\\
0.921875	0.589244038150369\\
0.923828125	0.464435375237986\\
0.92578125	0.454581494185951\\
0.927734375	0.523104810818113\\
0.9296875	0.618951773599008\\
0.931640625	0.710222912877355\\
0.93359375	0.974449258580021\\
0.935546875	1.47287759954434\\
0.9375	1.7371562744299\\
0.939453125	1.72498413512036\\
0.94140625	1.36329497078718\\
0.943359375	1.1505578346346\\
0.9453125	0.960767815363359\\
0.947265625	0.746298882588971\\
0.94921875	0.689246152035108\\
0.951171875	0.895195875576991\\
0.953125	0.97794541881802\\
0.955078125	1.27261016002125\\
0.95703125	1.27181150354302\\
0.958984375	0.986620078771504\\
0.9609375	0.948652454888916\\
0.962890625	1.16550281220997\\
0.96484375	1.45901843881217\\
0.966796875	1.20739811015519\\
0.96875	0.970074118919078\\
0.970703125	0.795334710502753\\
0.97265625	0.556173963984276\\
0.974609375	0.48290116217128\\
0.9765625	0.446103490132323\\
0.978515625	0.313912115485659\\
0.98046875	0.265471896283354\\
0.982421875	0.266508229954317\\
0.984375	0.284903981982204\\
0.986328125	0.196022965755783\\
0.98828125	0.1472071590043\\
0.990234375	0.113366970780826\\
0.9921875	0.0873721265958313\\
0.994140625	0.0712587217829977\\
0.99609375	0.0539531747324658\\
0.998046875	0.0279950134143353\\
1	0\\
};

\end{axis}
\end{tikzpicture}%

%% file: pics/gs_eps/gs_eps_3_1_1.tex
%
%
\begin{tikzpicture}
\begin{axis}[%
width=0.9in,
height=0.9in,
scale only axis,
xmin=0,
xmax=1,
xtick={0,0.5,1},
ymin=0,
ymax=7,
yminorticks=true,
yticklabels = {},
axis background/.style={fill=white},
]
\addplot [color=mycolor2, thick]
  table[row sep=crcr]{%
0	0\\
0.001953125	0.0879717964934148\\
0.00390625	0.17742641671557\\
0.005859375	0.409487459185145\\
0.0078125	0.481104882579999\\
0.009765625	0.514760696104643\\
0.01171875	0.561706979321425\\
0.013671875	0.76002588236755\\
0.015625	0.562014428748546\\
0.017578125	0.541467585094235\\
0.01953125	0.635558898434837\\
0.021484375	1.23987730347696\\
0.0234375	1.78274987740778\\
0.025390625	2.3414391650713\\
0.02734375	2.18979127141289\\
0.029296875	1.59354362626738\\
0.03125	1.15267336561535\\
0.033203125	1.09455582243515\\
0.03515625	1.03319712144422\\
0.037109375	0.782361540655904\\
0.0390625	1.28154668331502\\
0.041015625	1.28382240013685\\
0.04296875	1.63930070588599\\
0.044921875	1.92096969473337\\
0.046875	2.39418255427017\\
0.048828125	2.02808756304549\\
0.05078125	1.54489490301629\\
0.052734375	1.19698345430851\\
0.0546875	0.950028394058817\\
0.056640625	0.485757574509297\\
0.05859375	0.192749761868595\\
0.060546875	0.0800120471727064\\
0.0625	0.0396684472740786\\
0.064453125	0.0222450529466048\\
0.06640625	0.0351542955463043\\
0.068359375	0.038592111957436\\
0.0703125	0.052224506927384\\
0.072265625	0.0992046068202652\\
0.07421875	0.18872012951474\\
0.076171875	0.288470624054765\\
0.078125	0.383710349820859\\
0.080078125	0.435245514920933\\
0.08203125	0.448889447048479\\
0.083984375	0.280674370248685\\
0.0859375	0.422098643935927\\
0.087890625	0.670419725283195\\
0.08984375	0.769380797755811\\
0.091796875	0.534057187148183\\
0.09375	0.636946286167157\\
0.095703125	0.419920511126061\\
0.09765625	0.44612306385857\\
0.099609375	0.858003027854966\\
0.1015625	1.32524363527356\\
0.103515625	1.16215949941116\\
0.10546875	1.36820333400088\\
0.107421875	2.63165721483261\\
0.109375	2.19465396746092\\
0.111328125	1.52094434207544\\
0.11328125	1.02890434337789\\
0.115234375	1.07698695624369\\
0.1171875	1.21966715610147\\
0.119140625	1.77088247544146\\
0.12109375	1.48318905358856\\
0.123046875	0.880391932366347\\
0.125	0.409236152414396\\
0.126953125	0.293855009429727\\
0.12890625	0.262012460314538\\
0.130859375	0.293057496070356\\
0.1328125	0.299979941730444\\
0.134765625	0.258719060532901\\
0.13671875	0.411910411097653\\
0.138671875	0.571833739931423\\
0.140625	0.376869028028244\\
0.142578125	0.618537657825875\\
0.14453125	0.908945042517624\\
0.146484375	1.33317793995409\\
0.1484375	1.26143383873449\\
0.150390625	1.40581813544448\\
0.15234375	1.64312632216259\\
0.154296875	1.93139295851104\\
0.15625	1.28395253760591\\
0.158203125	0.886104925714149\\
0.16015625	0.703540476666356\\
0.162109375	0.65237192133742\\
0.1640625	0.582738111313663\\
0.166015625	1.24836252895797\\
0.16796875	1.57213392475708\\
0.169921875	1.07341782753276\\
0.171875	0.732803237286613\\
0.173828125	0.327870443173695\\
0.17578125	0.111406636306067\\
0.177734375	0.050514960538388\\
0.1796875	0.0312806907079257\\
0.181640625	0.0254367234481203\\
0.18359375	0.037499553194468\\
0.185546875	0.0354294484871675\\
0.1875	0.0207495219797904\\
0.189453125	0.0125318213961014\\
0.19140625	0.0179057582852515\\
0.193359375	0.0374299636831792\\
0.1953125	0.0644366959308532\\
0.197265625	0.111735441476739\\
0.19921875	0.219119749796796\\
0.201171875	0.504606534383559\\
0.203125	0.523360504208017\\
0.205078125	0.567693698558465\\
0.20703125	1.15414928127147\\
0.208984375	1.91114126841064\\
0.2109375	1.71174637936576\\
0.212890625	0.670641435847921\\
0.21484375	0.240346524080865\\
0.216796875	0.0927410210184736\\
0.21875	0.0515549780163824\\
0.220703125	0.0213916983563073\\
0.22265625	0.0175606379013076\\
0.224609375	0.0392193068874087\\
0.2265625	0.0642780663617495\\
0.228515625	0.118028357455342\\
0.23046875	0.145328024763794\\
0.232421875	0.316562082604199\\
0.234375	0.430214308339041\\
0.236328125	0.261668102413436\\
0.23828125	0.208953632461587\\
0.240234375	0.229020271344705\\
0.2421875	0.256024871615348\\
0.244140625	0.376050290896751\\
0.24609375	0.562768038262074\\
0.248046875	0.996315146843807\\
0.25	1.31515550255354\\
0.251953125	1.25757346594791\\
0.25390625	1.36268144934763\\
0.255859375	1.39953733522257\\
0.2578125	1.60990582819067\\
0.259765625	1.73489524837564\\
0.26171875	1.14150125378195\\
0.263671875	0.981206901353827\\
0.265625	0.835003887688713\\
0.267578125	0.646958336267583\\
0.26953125	0.464569237123544\\
0.271484375	0.603016646934064\\
0.2734375	0.785619666308141\\
0.275390625	0.388689150167769\\
0.27734375	0.150778759512062\\
0.279296875	0.121041655157017\\
0.28125	0.113683675417166\\
0.283203125	0.224284040633516\\
0.28515625	0.278735227136828\\
0.287109375	0.278218804875281\\
0.2890625	0.193583465116265\\
0.291015625	0.118146323191989\\
0.29296875	0.12821358168748\\
0.294921875	0.369752821213189\\
0.296875	0.668997261225584\\
0.298828125	1.26408542961186\\
0.30078125	1.58459936496147\\
0.302734375	1.24599037606594\\
0.3046875	1.17691888336947\\
0.306640625	0.811522963610996\\
0.30859375	0.775860087031091\\
0.310546875	0.927655632448866\\
0.3125	0.773206990886683\\
0.314453125	0.824681756266991\\
0.31640625	0.969376869280203\\
0.318359375	1.37075875378465\\
0.3203125	1.85302825740321\\
0.322265625	1.16154213268237\\
0.32421875	0.804806986096578\\
0.326171875	1.69842238487602\\
0.328125	1.58989231728096\\
0.330078125	1.59617335793729\\
0.33203125	2.13213101504956\\
0.333984375	1.37156342155577\\
0.3359375	1.24689044179922\\
0.337890625	0.720874699271131\\
0.33984375	0.683120789767224\\
0.341796875	0.510302315959401\\
0.34375	0.472442657793175\\
0.345703125	0.864982254555038\\
0.34765625	1.86141293211841\\
0.349609375	1.93003342434823\\
0.3515625	1.90013920713239\\
0.353515625	1.37146674378151\\
0.35546875	1.5524874084444\\
0.357421875	1.06906880917618\\
0.359375	0.966282813970956\\
0.361328125	1.26899257404561\\
0.36328125	1.06862016286627\\
0.365234375	1.69113580227132\\
0.3671875	1.60023120913112\\
0.369140625	1.54994343703022\\
0.37109375	1.17581624460179\\
0.373046875	0.54243761196988\\
0.375	0.227753395252498\\
0.376953125	0.175188890011758\\
0.37890625	0.0993622033480536\\
0.380859375	0.0746097022650153\\
0.3828125	0.0650874523345902\\
0.384765625	0.0394518956371417\\
0.38671875	0.0328788924980372\\
0.388671875	0.0205370289994631\\
0.390625	0.0200426375374146\\
0.392578125	0.0154700438633401\\
0.39453125	0.0172444456225607\\
0.396484375	0.0156515007559233\\
0.3984375	0.0130755592627681\\
0.400390625	0.0077808883931098\\
0.40234375	0.00481842079232574\\
0.404296875	0.00366092638009333\\
0.40625	0.00347748601627162\\
0.408203125	0.00402058728288808\\
0.41015625	0.00776294772488933\\
0.412109375	0.017797703602656\\
0.4140625	0.0434127627661589\\
0.416015625	0.111419126636167\\
0.41796875	0.183046647335214\\
0.419921875	0.39110877634177\\
0.421875	0.779772755318719\\
0.423828125	0.97396401980517\\
0.42578125	1.35596022918001\\
0.427734375	1.25969316688388\\
0.4296875	0.89684545111881\\
0.431640625	1.69192582800453\\
0.43359375	2.00836013136934\\
0.435546875	1.1863123339887\\
0.4375	1.56260665230076\\
0.439453125	1.52546110166139\\
0.44140625	1.52135234493543\\
0.443359375	1.34463843065721\\
0.4453125	1.22760048280955\\
0.447265625	0.691794065878653\\
0.44921875	0.2060377593111\\
0.451171875	0.0924584751686397\\
0.453125	0.0468839019890208\\
0.455078125	0.027398512360235\\
0.45703125	0.0323292769938379\\
0.458984375	0.0274237528525258\\
0.4609375	0.0141851906869353\\
0.462890625	0.00886136394822765\\
0.46484375	0.0064688192275632\\
0.466796875	0.0113607298166627\\
0.46875	0.0154141038192529\\
0.470703125	0.0142549177904499\\
0.47265625	0.0325698436845258\\
0.474609375	0.073618715579571\\
0.4765625	0.110746022976728\\
0.478515625	0.210876512061625\\
0.48046875	0.259603178653726\\
0.482421875	0.409847150833757\\
0.484375	0.826925451300121\\
0.486328125	1.37697126632094\\
0.48828125	1.13108602788392\\
0.490234375	0.722157270782369\\
0.4921875	0.680350385980897\\
0.494140625	0.617201047371682\\
0.49609375	0.238493540124889\\
0.498046875	0.200400999216083\\
0.5	0.287066036500356\\
0.501953125	0.369188500313161\\
0.50390625	0.700929295209739\\
0.505859375	1.15514882932766\\
0.5078125	1.15368052796868\\
0.509765625	1.48241590192988\\
0.51171875	1.60231184487696\\
0.513671875	1.39670918151598\\
0.515625	0.85578833541538\\
0.517578125	0.392597694584282\\
0.51953125	0.283540222974494\\
0.521484375	0.557903469249779\\
0.5234375	1.11353377936568\\
0.525390625	1.20372451179605\\
0.52734375	0.990377596448537\\
0.529296875	1.14483894340893\\
0.53125	1.63703315094372\\
0.533203125	2.01524534119324\\
0.53515625	1.29387996545319\\
0.537109375	0.755276137100001\\
0.5390625	0.783465128045894\\
0.541015625	1.47861823972985\\
0.54296875	2.00023413848324\\
0.544921875	2.07981721175735\\
0.546875	2.2439393758914\\
0.548828125	1.99908709546963\\
0.55078125	1.85114586296251\\
0.552734375	1.68437179500342\\
0.5546875	1.44366008127663\\
0.556640625	1.18902713885656\\
0.55859375	0.747451306080873\\
0.560546875	0.289233678413359\\
0.5625	0.121647843678191\\
0.564453125	0.109420927360302\\
0.56640625	0.0809584853213539\\
0.568359375	0.0582705253039389\\
0.5703125	0.0657676033589349\\
0.572265625	0.155041867636045\\
0.57421875	0.22663755939218\\
0.576171875	0.363322571590961\\
0.578125	0.473296678270039\\
0.580078125	0.931049035344673\\
0.58203125	1.54668289914133\\
0.583984375	2.24546503585789\\
0.5859375	1.51238205550076\\
0.587890625	0.90428524021842\\
0.58984375	0.933082452796329\\
0.591796875	0.48351063725006\\
0.59375	0.225327425235795\\
0.595703125	0.177606039434752\\
0.59765625	0.0970768266830767\\
0.599609375	0.066073245398691\\
0.6015625	0.0768276408121294\\
0.603515625	0.145780481486333\\
0.60546875	0.316455203444649\\
0.607421875	0.447762745966651\\
0.609375	1.02331058874701\\
0.611328125	1.9935002573101\\
0.61328125	2.63432598978741\\
0.615234375	2.08899779840603\\
0.6171875	2.00298314874642\\
0.619140625	1.95452504292226\\
0.62109375	1.38837823952328\\
0.623046875	1.26207729020986\\
0.625	1.26087610813542\\
0.626953125	1.40059745323421\\
0.62890625	1.25934082409865\\
0.630859375	1.52391964026819\\
0.6328125	2.2667752843548\\
0.634765625	2.37796532834103\\
0.63671875	1.78999018458324\\
0.638671875	1.13423330338375\\
0.640625	0.539461565120687\\
0.642578125	0.392802581335751\\
0.64453125	0.338970999618006\\
0.646484375	0.42745848962084\\
0.6484375	0.579632726344868\\
0.650390625	1.12975286764296\\
0.65234375	1.329709821692\\
0.654296875	0.890139576219286\\
0.65625	0.698286971171717\\
0.658203125	0.484880602258861\\
0.66015625	0.298613150743463\\
0.662109375	0.132061459817155\\
0.6640625	0.0912206306544802\\
0.666015625	0.0755236141791676\\
0.66796875	0.0427953397768073\\
0.669921875	0.018600165692717\\
0.671875	0.00965479617312928\\
0.673828125	0.00517314296918325\\
0.67578125	0.00318640025829414\\
0.677734375	0.00182450256600485\\
0.6796875	0.00193699379511854\\
0.681640625	0.00284138474561827\\
0.68359375	0.0050748486993676\\
0.685546875	0.00501313295823519\\
0.6875	0.00432996587607419\\
0.689453125	0.00306907563126478\\
0.69140625	0.00218773888017305\\
0.693359375	0.00135957606857022\\
0.6953125	0.000750719365686102\\
0.697265625	0.000215115284219505\\
0.69921875	0.000101437393564523\\
0.701171875	5.51307187503154e-05\\
0.703125	3.03379877106922e-05\\
0.705078125	3.11822174823228e-05\\
0.70703125	2.18835947222091e-05\\
0.708984375	8.87559839830849e-06\\
0.7109375	5.26852445394003e-06\\
0.712890625	4.84382089045425e-06\\
0.71484375	8.38365771380223e-06\\
0.716796875	1.58875526571389e-05\\
0.71875	2.75382706521063e-05\\
0.720703125	5.46789007719958e-05\\
0.72265625	7.34738668126747e-05\\
0.724609375	9.20062993186642e-05\\
0.7265625	0.000224651536157335\\
0.728515625	0.000408445395342371\\
0.73046875	0.000377411530898875\\
0.732421875	0.000422382395229534\\
0.734375	0.00122886013503732\\
0.736328125	0.00199737725673174\\
0.73828125	0.00329877568723088\\
0.740234375	0.00676028016886245\\
0.7421875	0.0106925597935167\\
0.744140625	0.0120507847564927\\
0.74609375	0.0321368423431224\\
0.748046875	0.0621769963409607\\
0.75	0.177135413374814\\
0.751953125	0.486513845405666\\
0.75390625	0.594398128413502\\
0.755859375	0.573693675397524\\
0.7578125	0.856702608798828\\
0.759765625	0.893707745480508\\
0.76171875	0.656330874196246\\
0.763671875	0.382490742630283\\
0.765625	0.510476901016742\\
0.767578125	0.518011783659924\\
0.76953125	0.700781764848977\\
0.771484375	1.59944618043498\\
0.7734375	1.67994445239498\\
0.775390625	1.29579080390789\\
0.77734375	0.743318040190226\\
0.779296875	1.32337246717068\\
0.78125	1.30776663587945\\
0.783203125	1.60516567960882\\
0.78515625	2.5938584718869\\
0.787109375	1.77947620081366\\
0.7890625	1.08513107646221\\
0.791015625	0.700878378645082\\
0.79296875	0.739625614641369\\
0.794921875	1.11795101306718\\
0.796875	1.34563403624564\\
0.798828125	1.95338515536663\\
0.80078125	1.09431449494953\\
0.802734375	0.844582861656008\\
0.8046875	0.746332209673137\\
0.806640625	0.814185743269462\\
0.80859375	0.939354783963642\\
0.810546875	0.776188877894582\\
0.8125	0.714845897347317\\
0.814453125	0.487856773517963\\
0.81640625	0.175809058894053\\
0.818359375	0.0931409265684595\\
0.8203125	0.0502509214346388\\
0.822265625	0.0339867595097912\\
0.82421875	0.0143649455103063\\
0.826171875	0.00611183587849595\\
0.828125	0.00460217978723296\\
0.830078125	0.00370484683506038\\
0.83203125	0.00202421106251263\\
0.833984375	0.00266565632218056\\
0.8359375	0.00737846500078569\\
0.837890625	0.016800619394318\\
0.83984375	0.033275950391406\\
0.841796875	0.100682057330331\\
0.84375	0.201267683690766\\
0.845703125	0.324593454940346\\
0.84765625	0.251643621193425\\
0.849609375	0.428456421363489\\
0.8515625	0.931120518167502\\
0.853515625	0.939099240536147\\
0.85546875	1.76609045356294\\
0.857421875	2.37443884171767\\
0.859375	1.62823798421097\\
0.861328125	1.01676251372478\\
0.86328125	0.959559879379367\\
0.865234375	0.554358235180524\\
0.8671875	0.642829577878678\\
0.869140625	1.15160398525219\\
0.87109375	1.86222657224682\\
0.873046875	1.11102600357639\\
0.875	0.983426618281858\\
0.876953125	0.865255032762257\\
0.87890625	0.436143748239948\\
0.880859375	0.237899031055534\\
0.8828125	0.155690991613267\\
0.884765625	0.0821525205710644\\
0.88671875	0.0587141705598727\\
0.888671875	0.0383101787082841\\
0.890625	0.0352481997586785\\
0.892578125	0.0402420379190868\\
0.89453125	0.0636737470216908\\
0.896484375	0.148354038545493\\
0.8984375	0.265042409198464\\
0.900390625	0.508037915893659\\
0.90234375	0.7696756322828\\
0.904296875	1.24423895937835\\
0.90625	2.34667347243861\\
0.908203125	2.57771403298792\\
0.91015625	1.60316851010349\\
0.912109375	1.21660837383452\\
0.9140625	0.677923897111141\\
0.916015625	0.373214924297688\\
0.91796875	0.605180281702846\\
0.919921875	0.68702159049709\\
0.921875	0.964218179661642\\
0.923828125	1.67656721933172\\
0.92578125	2.22280888327673\\
0.927734375	1.63092062108844\\
0.9296875	1.58741256717606\\
0.931640625	1.74910311754589\\
0.93359375	0.70544428235904\\
0.935546875	0.299258158933633\\
0.9375	0.146562409373153\\
0.939453125	0.141834860002288\\
0.94140625	0.327462989567575\\
0.943359375	0.662735962556187\\
0.9453125	0.497391552422479\\
0.947265625	0.678638210177001\\
0.94921875	0.857570819198946\\
0.951171875	1.17647011287634\\
0.953125	1.19834279288243\\
0.955078125	0.667588937182269\\
0.95703125	0.453806037935252\\
0.958984375	0.588798148374071\\
0.9609375	1.0645730954944\\
0.962890625	2.01413161862544\\
0.96484375	1.92151781144497\\
0.966796875	1.19436587207082\\
0.96875	0.73862419686025\\
0.970703125	0.751720763464986\\
0.97265625	1.0790243220475\\
0.974609375	0.918605227010804\\
0.9765625	2.08172219719558\\
0.978515625	2.22681765614511\\
0.98046875	2.38782070900822\\
0.982421875	2.03994524417464\\
0.984375	2.08382907929692\\
0.986328125	1.71081216279435\\
0.98828125	1.57393695154004\\
0.990234375	1.24619372972134\\
0.9921875	0.452433861126685\\
0.994140625	0.166828645535542\\
0.99609375	0.0524535350868333\\
0.998046875	0.016930629504839\\
1	0\\
};

\end{axis}
\end{tikzpicture}%

%% file: pics/gs_eps/gs_eps_4_1_1.tex
%
%
\begin{tikzpicture}
\begin{axis}[%
width=0.9in,
height=0.9in,
scale only axis,
xmin=0,
xmax=1,
xtick={0,0.5,1},
ymin=0,
ymax=7,
yminorticks=true,
yticklabels = {},
axis background/.style={fill=white},
]
\addplot [color=mycolor2, thick]
  table[row sep=crcr]{%
0	0\\
0.001953125	0.170357508087224\\
0.00390625	1.4708649867915\\
0.005859375	2.12387017572826\\
0.0078125	0.558128264258692\\
0.009765625	0.268187136373908\\
0.01171875	0.0595538707149265\\
0.013671875	0.0137949069442864\\
0.015625	0.0279565480048641\\
0.017578125	0.155063274124171\\
0.01953125	0.380784024169302\\
0.021484375	0.260831392958267\\
0.0234375	0.0602290089118091\\
0.025390625	0.119709619202597\\
0.02734375	0.440195960558088\\
0.029296875	2.24341912387454\\
0.03125	2.07058293658764\\
0.033203125	1.44117746698889\\
0.03515625	0.189125402552276\\
0.037109375	0.113712233056573\\
0.0390625	0.0294450539408063\\
0.041015625	0.00822291391588893\\
0.04296875	0.0115537121679197\\
0.044921875	0.0197647689387522\\
0.046875	0.0344789968300677\\
0.048828125	0.200791043639738\\
0.05078125	0.533114982723781\\
0.052734375	1.10337416153589\\
0.0546875	0.264757966104321\\
0.056640625	0.0955419441599091\\
0.05859375	0.36151208283183\\
0.060546875	0.781218538011435\\
0.0625	2.61931282147591\\
0.064453125	2.17194489304611\\
0.06640625	0.557529078230752\\
0.068359375	0.242785308733123\\
0.0703125	0.110037124278125\\
0.072265625	0.450680996309851\\
0.07421875	1.49960037952055\\
0.076171875	2.39900384956486\\
0.078125	1.74591524859224\\
0.080078125	1.32725759662908\\
0.08203125	0.105338181249341\\
0.083984375	0.100851125649344\\
0.0859375	0.048600987902413\\
0.087890625	0.0700633012883855\\
0.08984375	0.111647892197865\\
0.091796875	0.368283145933966\\
0.09375	0.362441316438454\\
0.095703125	0.433542278227195\\
0.09765625	0.322799544095419\\
0.099609375	0.851612098008143\\
0.1015625	0.535241946918011\\
0.103515625	0.962295367192807\\
0.10546875	2.01632183964458\\
0.107421875	2.57970042703802\\
0.109375	1.75577809251647\\
0.111328125	1.8822289803441\\
0.11328125	2.21514989901909\\
0.115234375	2.07059812988281\\
0.1171875	2.728690653617\\
0.119140625	1.17353991615684\\
0.12109375	0.276030806390721\\
0.123046875	0.0919884100312198\\
0.125	0.0800976807919856\\
0.126953125	0.1653778922537\\
0.12890625	0.275295977829753\\
0.130859375	0.0959373596882533\\
0.1328125	0.0901874763065975\\
0.134765625	0.0381159899103694\\
0.13671875	0.10254494014844\\
0.138671875	0.214083121880599\\
0.140625	0.269373924213757\\
0.142578125	0.419838989212692\\
0.14453125	0.471177884404067\\
0.146484375	0.834029930324854\\
0.1484375	1.81113260162051\\
0.150390625	1.13632589677616\\
0.15234375	0.501975412393275\\
0.154296875	0.633546870340113\\
0.15625	0.664785826634015\\
0.158203125	2.44177860242916\\
0.16015625	3.20409922796074\\
0.162109375	1.62739493143755\\
0.1640625	0.499754070956477\\
0.166015625	0.553706015982763\\
0.16796875	1.75530572105273\\
0.169921875	3.07156530601967\\
0.171875	2.08862299290694\\
0.173828125	0.850563242878995\\
0.17578125	0.358051841490005\\
0.177734375	0.280065034823827\\
0.1796875	0.302731727083209\\
0.181640625	0.269591839166331\\
0.18359375	0.072988117689094\\
0.185546875	0.0294640509978714\\
0.1875	0.0158077332079189\\
0.189453125	0.00418311336819447\\
0.19140625	0.000893924469038959\\
0.193359375	0.000164383369814749\\
0.1953125	0.000109676569923692\\
0.197265625	0.000292999645748307\\
0.19921875	0.00080488450062263\\
0.201171875	0.0014658743258595\\
0.203125	0.0056934868408864\\
0.205078125	0.0210864009790746\\
0.20703125	0.0576916294589069\\
0.208984375	0.140046666811598\\
0.2109375	0.515155139178965\\
0.212890625	1.085446696608\\
0.21484375	0.907991020535746\\
0.216796875	0.639550585299696\\
0.21875	0.627617275538848\\
0.220703125	2.28185458605029\\
0.22265625	2.54095371009514\\
0.224609375	1.62032901075974\\
0.2265625	1.68444375216474\\
0.228515625	1.52219401647261\\
0.23046875	0.626217120458791\\
0.232421875	0.456574604636191\\
0.234375	0.0666345374198071\\
0.236328125	0.016093145843291\\
0.23828125	0.0431370023434535\\
0.240234375	0.170653203631174\\
0.2421875	0.280419805659275\\
0.244140625	0.895309695434102\\
0.24609375	2.14159051924785\\
0.248046875	2.0562660725412\\
0.25	1.10194474193047\\
0.251953125	0.61599467689374\\
0.25390625	0.542263158011107\\
0.255859375	1.37597220512125\\
0.2578125	0.926502350657395\\
0.259765625	1.02334902341627\\
0.26171875	2.13520460814029\\
0.263671875	2.19237117441516\\
0.265625	0.723025982070371\\
0.267578125	0.493882890665317\\
0.26953125	0.456238867952282\\
0.271484375	1.1774368524706\\
0.2734375	1.28671445000737\\
0.275390625	0.528884135352363\\
0.27734375	0.124106456156272\\
0.279296875	0.0778423601138282\\
0.28125	0.0348078038245348\\
0.283203125	0.0121463165628698\\
0.28515625	0.00907185907659698\\
0.287109375	0.0417821274452235\\
0.2890625	0.100901826631144\\
0.291015625	0.215424487700223\\
0.29296875	0.340296087196275\\
0.294921875	1.14403432162112\\
0.296875	1.83493379251225\\
0.298828125	1.12633224668493\\
0.30078125	0.812102921854378\\
0.302734375	1.17769992084904\\
0.3046875	0.409479702196369\\
0.306640625	0.191198601930799\\
0.30859375	0.0704964864336268\\
0.310546875	0.0267710487588997\\
0.3125	0.00724669183610865\\
0.314453125	0.00583496507883864\\
0.31640625	0.00231035665982227\\
0.318359375	0.000427795100472534\\
0.3203125	0.000131851791471592\\
0.322265625	3.43867530643693e-05\\
0.32421875	1.32996052332908e-05\\
0.326171875	3.85439133702298e-06\\
0.328125	5.1184458903023e-06\\
0.330078125	1.43854780045612e-05\\
0.33203125	8.52971491792452e-05\\
0.333984375	0.000374411719932872\\
0.3359375	0.000697219826982968\\
0.337890625	0.00331059476984701\\
0.33984375	0.00515996251897749\\
0.341796875	0.00539978242586902\\
0.34375	0.0124782936278592\\
0.345703125	0.0313408391298685\\
0.34765625	0.221012708104809\\
0.349609375	1.49664636230256\\
0.3515625	1.78170281887795\\
0.353515625	1.23090866483957\\
0.35546875	0.402491569762927\\
0.357421875	0.478758130938279\\
0.359375	1.85582591996077\\
0.361328125	0.783111060563941\\
0.36328125	0.441511343187493\\
0.365234375	0.36759487574393\\
0.3671875	0.342810042063886\\
0.369140625	0.171746976745283\\
0.37109375	0.387527424066217\\
0.373046875	1.03902853991153\\
0.375	3.56054020521958\\
0.376953125	0.962818776372552\\
0.37890625	0.719340400225739\\
0.380859375	0.443020787234819\\
0.3828125	0.379171187255727\\
0.384765625	0.568649621907301\\
0.38671875	1.98666020462376\\
0.388671875	2.46971644735544\\
0.390625	1.11474260120363\\
0.392578125	0.834507622378005\\
0.39453125	1.11591438509504\\
0.396484375	1.704488055757\\
0.3984375	2.78594329727863\\
0.400390625	1.62273052894683\\
0.40234375	0.818795751491111\\
0.404296875	0.14274828660453\\
0.40625	0.0226035329727276\\
0.408203125	0.00571307348764092\\
0.41015625	0.00209667302486189\\
0.412109375	0.00403473891726597\\
0.4140625	0.00982262000316378\\
0.416015625	0.0174057092816875\\
0.41796875	0.0379271466863845\\
0.419921875	0.0971519892568481\\
0.421875	0.148234918029246\\
0.423828125	0.261989848177006\\
0.42578125	1.10913082165176\\
0.427734375	1.49592191286622\\
0.4296875	2.57256826818111\\
0.431640625	1.01482929411759\\
0.43359375	1.10566521536368\\
0.435546875	1.33957641794548\\
0.4375	0.297517870411484\\
0.439453125	0.11156121836749\\
0.44140625	0.0780130649421049\\
0.443359375	0.111679507572062\\
0.4453125	0.151671515628744\\
0.447265625	0.557674284058257\\
0.44921875	0.940192712799083\\
0.451171875	1.47603124577953\\
0.453125	1.09169677827505\\
0.455078125	0.462433131678232\\
0.45703125	0.343558478242061\\
0.458984375	0.989452376180647\\
0.4609375	0.400111063332382\\
0.462890625	0.0525482537005972\\
0.46484375	0.0102547494401141\\
0.466796875	0.00616832717597645\\
0.46875	0.00314787919051817\\
0.470703125	0.00539831942210028\\
0.47265625	0.0259322358356287\\
0.474609375	0.111217322708999\\
0.4765625	0.326815846681015\\
0.478515625	1.07699276675856\\
0.48046875	1.68433741806737\\
0.482421875	0.210090458854768\\
0.484375	0.0366503678361357\\
0.486328125	0.00419327775284475\\
0.48828125	0.00343441991147068\\
0.490234375	0.00746995028482479\\
0.4921875	0.0271518632925646\\
0.494140625	0.0320199239478011\\
0.49609375	0.0837978652943428\\
0.498046875	0.293215592195772\\
0.5	0.825744636800109\\
0.501953125	0.980354284508627\\
0.50390625	1.27700534442474\\
0.505859375	0.809592857653232\\
0.5078125	0.431443516561071\\
0.509765625	0.144833570898458\\
0.51171875	0.148769328403441\\
0.513671875	0.0223382124527006\\
0.515625	0.00790509029398193\\
0.517578125	0.0157285640160252\\
0.51953125	0.0171682994867411\\
0.521484375	0.047508241190117\\
0.5234375	0.330523866724777\\
0.525390625	0.857654468779813\\
0.52734375	0.574010369349001\\
0.529296875	0.343152638322371\\
0.53125	0.399475804072394\\
0.533203125	0.732733408509991\\
0.53515625	1.74953391025574\\
0.537109375	0.745937845758359\\
0.5390625	1.57137016188349\\
0.541015625	2.18817324605099\\
0.54296875	1.10122948661388\\
0.544921875	0.490905509314104\\
0.546875	0.392659077159645\\
0.548828125	1.95547370748666\\
0.55078125	1.97017403049763\\
0.552734375	1.40987745836023\\
0.5546875	0.75901975069135\\
0.556640625	0.874626539424841\\
0.55859375	1.40592190617194\\
0.560546875	0.946378941241713\\
0.5625	0.149868977928595\\
0.564453125	0.0481125769255475\\
0.56640625	0.0208817585927394\\
0.568359375	0.00721999218425885\\
0.5703125	0.00259657989505609\\
0.572265625	0.00119332395289217\\
0.57421875	0.000673485651154872\\
0.576171875	0.000293033412281201\\
0.578125	0.000606405241680677\\
0.580078125	0.00255615828033279\\
0.58203125	0.0173657605782084\\
0.583984375	0.0875944224483114\\
0.5859375	0.449123956898496\\
0.587890625	0.942680156178502\\
0.58984375	0.771330994344947\\
0.591796875	2.03418041451144\\
0.59375	1.42198852844487\\
0.595703125	1.28860321328689\\
0.59765625	0.920204807124755\\
0.599609375	0.136746835085452\\
0.6015625	0.0245801781875466\\
0.603515625	0.00984837910805057\\
0.60546875	0.00321323924296941\\
0.607421875	0.000626128465541137\\
0.609375	0.000436529218507963\\
0.611328125	0.000560790035018705\\
0.61328125	0.00221198192684431\\
0.615234375	0.00669879566865086\\
0.6171875	0.0279352961750001\\
0.619140625	0.0401014574086551\\
0.62109375	0.0262909857096182\\
0.623046875	0.0199318049169998\\
0.625	0.100542741333642\\
0.626953125	0.380414966325836\\
0.62890625	0.749399312347392\\
0.630859375	1.16699855461131\\
0.6328125	0.57356666492894\\
0.634765625	0.16103329687262\\
0.63671875	0.686407120914206\\
0.638671875	0.718389698139805\\
0.640625	1.51073529202855\\
0.642578125	1.22612408476115\\
0.64453125	0.758361971048758\\
0.646484375	2.22200526289689\\
0.6484375	2.60206830152528\\
0.650390625	2.01886189077278\\
0.65234375	1.37959768747635\\
0.654296875	0.615134566754224\\
0.65625	0.0794081180843626\\
0.658203125	0.0386168300650256\\
0.66015625	0.033494424832468\\
0.662109375	0.142096261896143\\
0.6640625	0.402695092847672\\
0.666015625	1.59493302417975\\
0.66796875	1.55807048894072\\
0.669921875	0.747795526002847\\
0.671875	0.146848437810615\\
0.673828125	0.0520068342890086\\
0.67578125	0.0522590835073108\\
0.677734375	0.278905483008992\\
0.6796875	1.13746077762003\\
0.681640625	3.20904652070769\\
0.68359375	2.32175278265877\\
0.685546875	1.45073612321073\\
0.6875	1.1699085016763\\
0.689453125	1.24778288739187\\
0.69140625	2.6412291964419\\
0.693359375	2.00323070850987\\
0.6953125	0.474360518077505\\
0.697265625	0.144094637797923\\
0.69921875	0.111758498796306\\
0.701171875	0.224192545081869\\
0.703125	0.0937720445124936\\
0.705078125	0.0318725241958992\\
0.70703125	0.00712838140827306\\
0.708984375	0.0026273289882384\\
0.7109375	0.000481602052985423\\
0.712890625	0.000117804632607565\\
0.71484375	2.31676495340686e-05\\
0.716796875	1.48404706680399e-05\\
0.71875	8.01294371654705e-06\\
0.720703125	2.96674501399649e-06\\
0.72265625	8.25003290404994e-07\\
0.724609375	9.2049895849279e-08\\
0.7265625	2.0630582237525e-08\\
0.728515625	1.21158062649496e-08\\
0.73046875	8.3327262781354e-09\\
0.732421875	3.94119617776049e-08\\
0.734375	1.72568081188358e-07\\
0.736328125	5.77759270889804e-07\\
0.73828125	2.50788040919409e-06\\
0.740234375	5.60781011069007e-06\\
0.7421875	2.81504922336929e-05\\
0.744140625	9.52610616323282e-05\\
0.74609375	0.000386831271732015\\
0.748046875	0.00158180322384677\\
0.75	0.0104099489504001\\
0.751953125	0.0415907323282669\\
0.75390625	0.0865835866870481\\
0.755859375	0.0955928267829436\\
0.7578125	0.21899070673201\\
0.759765625	0.549171860388554\\
0.76171875	1.5280169696413\\
0.763671875	0.658466538487477\\
0.765625	1.31732400356568\\
0.767578125	3.03947178961054\\
0.76953125	1.11127453272692\\
0.771484375	0.59196748778661\\
0.7734375	1.14069216977529\\
0.775390625	0.882375351260917\\
0.77734375	0.369129942324353\\
0.779296875	0.29655733500498\\
0.78125	0.172121662585702\\
0.783203125	0.0357944610312359\\
0.78515625	0.00808254359196984\\
0.787109375	0.0019027672774036\\
0.7890625	0.000508477291904353\\
0.791015625	0.000333854040700465\\
0.79296875	0.00225870648807937\\
0.794921875	0.0122141877308822\\
0.796875	0.0864232167832158\\
0.798828125	0.151920398740918\\
0.80078125	0.798059620031913\\
0.802734375	1.92895311965552\\
0.8046875	1.73087972055014\\
0.806640625	0.801175916186467\\
0.80859375	0.482363024488683\\
0.810546875	1.52779142382763\\
0.8125	0.670627538661617\\
0.814453125	0.243952949538052\\
0.81640625	0.0682965778195237\\
0.818359375	0.0206361321932429\\
0.8203125	0.0146026405576274\\
0.822265625	0.0411137932141997\\
0.82421875	0.212577824601719\\
0.826171875	0.765954322007616\\
0.828125	2.8336820406033\\
0.830078125	1.81247744446309\\
0.83203125	0.63471663472702\\
0.833984375	0.47513056217043\\
0.8359375	0.933639373419552\\
0.837890625	2.44072029285041\\
0.83984375	1.63383109252843\\
0.841796875	0.675163288946142\\
0.84375	0.0902271712630186\\
0.845703125	0.0629014776758362\\
0.84765625	0.0978475202632153\\
0.849609375	0.232948352233827\\
0.8515625	0.391771514257012\\
0.853515625	0.255760941414452\\
0.85546875	1.04655774543499\\
0.857421875	2.00100037525866\\
0.859375	0.65529188342614\\
0.861328125	0.970627585713745\\
0.86328125	2.38688892575505\\
0.865234375	2.85041049684392\\
0.8671875	2.36262668988666\\
0.869140625	1.62127309724707\\
0.87109375	0.434205209951273\\
0.873046875	0.168418225697733\\
0.875	0.988072856740057\\
0.876953125	1.94570602468593\\
0.87890625	0.712389119860405\\
0.880859375	0.226277992192139\\
0.8828125	0.0778867275486004\\
0.884765625	0.159720310606226\\
0.88671875	0.575691672206219\\
0.888671875	0.252598334707137\\
0.890625	0.166287184056221\\
0.892578125	0.444614307126108\\
0.89453125	1.19326496888525\\
0.896484375	2.80736190092942\\
0.8984375	1.2649254827769\\
0.900390625	0.262640232178434\\
0.90234375	0.107043938582259\\
0.904296875	0.0685933799531768\\
0.90625	0.0497359399789705\\
0.908203125	0.0204532773469659\\
0.91015625	0.00356020358957742\\
0.912109375	0.00103103740035515\\
0.9140625	0.00020311205339298\\
0.916015625	5.23895809509594e-05\\
0.91796875	3.18761713457414e-05\\
0.919921875	4.29190569788598e-05\\
0.921875	0.000254891158745404\\
0.923828125	0.00144870014246117\\
0.92578125	0.00634928105024813\\
0.927734375	0.0110535648678208\\
0.9296875	0.0425744064093391\\
0.931640625	0.0860269104299268\\
0.93359375	0.138545326561838\\
0.935546875	0.0287271327770352\\
0.9375	0.00720774163937228\\
0.939453125	0.00202772674444982\\
0.94140625	0.00553464124762086\\
0.943359375	0.0115988089067032\\
0.9453125	0.0285979363849189\\
0.947265625	0.0656034418501175\\
0.94921875	0.136683076329333\\
0.951171875	0.341863398311613\\
0.953125	1.2625990236724\\
0.955078125	1.70436116443524\\
0.95703125	1.01658148546892\\
0.958984375	0.618064590175791\\
0.9609375	1.12399041165569\\
0.962890625	2.64921599580538\\
0.96484375	2.40399224109046\\
0.966796875	0.638408317915549\\
0.96875	0.460525678092105\\
0.970703125	1.3026707848544\\
0.97265625	1.78033138853096\\
0.974609375	1.01536909382663\\
0.9765625	1.77632025293515\\
0.978515625	2.66011168198683\\
0.98046875	2.12076470110806\\
0.982421875	0.717432175466796\\
0.984375	0.0962260804426375\\
0.986328125	0.0485099695012936\\
0.98828125	0.0864683348755551\\
0.990234375	0.452547774425892\\
0.9921875	0.839404558430196\\
0.994140625	0.248161576349072\\
0.99609375	0.0620097575240722\\
0.998046875	0.022018306475861\\
1	0\\
};

\end{axis}
\end{tikzpicture}%

%% file: pics/gs_eps/gs_eps_5_1_1.tex
%
%
\begin{tikzpicture}
\begin{axis}[%
width=0.9in,
height=0.9in,
scale only axis,
xmin=0,
xmax=1,
xtick={0,0.5,1},
ymin=0,
ymax=7,
yminorticks=true,
ytick={0,3,6},
yticklabel pos=right,
axis background/.style={fill=white},
]
\addplot [color=mycolor2, thick]
  table[row sep=crcr]{%
0	0\\
0.001953125	1.12748956243994\\
0.00390625	0.377038705760942\\
0.005859375	0.0244623665880817\\
0.0078125	0.000713588558223856\\
0.009765625	0.00370719512061953\\
0.01171875	0.0107204561125789\\
0.013671875	0.395067833724778\\
0.015625	2.15465007200899\\
0.017578125	0.139951194131227\\
0.01953125	0.0124522498411968\\
0.021484375	0.000851281755116697\\
0.0234375	0.000178661503482938\\
0.025390625	0.00359656357745922\\
0.02734375	0.00729664242950969\\
0.029296875	0.248506178621721\\
0.03125	2.84701083097096\\
0.033203125	0.480524982452608\\
0.03515625	0.0526793678323671\\
0.037109375	0.915548066059803\\
0.0390625	1.2300258068975\\
0.041015625	0.0633018580695109\\
0.04296875	0.00865535238336912\\
0.044921875	0.00230699093645137\\
0.046875	0.0105330334578016\\
0.048828125	0.0417558839882778\\
0.05078125	0.178826733929771\\
0.052734375	1.88687124733248\\
0.0546875	1.58659636473085\\
0.056640625	2.12718506005525\\
0.05859375	2.81529070587438\\
0.060546875	0.21091782616963\\
0.0625	0.00831684030218285\\
0.064453125	0.00323104270846853\\
0.06640625	0.000395314907588887\\
0.068359375	0.00075167013724531\\
0.0703125	0.00407367582895495\\
0.072265625	0.0204020755114085\\
0.07421875	0.138776880085844\\
0.076171875	0.103095545539184\\
0.078125	0.584385868123716\\
0.080078125	3.80802985893081\\
0.08203125	0.438143934669334\\
0.083984375	1.76829235315234\\
0.0859375	1.96542911768943\\
0.087890625	0.204514113861334\\
0.08984375	0.0364697874074638\\
0.091796875	0.00509751317050751\\
0.09375	0.000579765530626021\\
0.095703125	2.28510191364373e-05\\
0.09765625	1.1529111206973e-06\\
0.099609375	9.06147710943535e-06\\
0.1015625	0.000113365552708656\\
0.103515625	0.00176705642869937\\
0.10546875	0.0276598384263835\\
0.107421875	0.0906741197245842\\
0.109375	0.441932892167312\\
0.111328125	2.27969317819601\\
0.11328125	1.15994580058434\\
0.115234375	0.249585161815911\\
0.1171875	0.0177255417683635\\
0.119140625	0.0133831388895547\\
0.12109375	0.148676591588758\\
0.123046875	1.64658159531401\\
0.125	0.583611118701415\\
0.126953125	0.136745880819941\\
0.12890625	0.312989380019933\\
0.130859375	1.92341038664417\\
0.1328125	0.506751597661577\\
0.134765625	0.0562089884872655\\
0.13671875	0.0378235634848041\\
0.138671875	0.00243171493086919\\
0.140625	0.000342743473286617\\
0.142578125	1.44741327881521e-05\\
0.14453125	9.50195155075736e-07\\
0.146484375	1.30556214311861e-07\\
0.1484375	3.0245295971217e-07\\
0.150390625	5.54185648577443e-08\\
0.15234375	1.51055768965951e-08\\
0.154296875	1.47487362121004e-09\\
0.15625	9.2396943681952e-11\\
0.158203125	1.51431654976811e-11\\
0.16015625	1.20028639646955e-11\\
0.162109375	2.02969686057605e-10\\
0.1640625	4.01040400031374e-09\\
0.166015625	1.23488118681137e-07\\
0.16796875	1.61697766211374e-06\\
0.169921875	2.0342165329952e-05\\
0.171875	0.000146764407078908\\
0.173828125	0.00451717942760837\\
0.17578125	0.0407152643709239\\
0.177734375	0.00600588056230173\\
0.1796875	0.0139461447547538\\
0.181640625	0.00630457628141167\\
0.18359375	0.021756881669993\\
0.185546875	0.298564119147169\\
0.1875	4.14144047929724\\
0.189453125	0.48450103843469\\
0.19140625	0.186243662077572\\
0.193359375	1.86547056704177\\
0.1953125	0.872841791189211\\
0.197265625	0.747659908778785\\
0.19921875	2.77901456628776\\
0.201171875	0.584693756748275\\
0.203125	0.012018297304858\\
0.205078125	0.000516938881756061\\
0.20703125	0.000940142882893929\\
0.208984375	0.00753994026819488\\
0.2109375	0.0578668764272632\\
0.212890625	0.80773843082432\\
0.21484375	2.71377088365227\\
0.216796875	0.504298966734324\\
0.21875	0.0979091425306965\\
0.220703125	0.00871806581821041\\
0.22265625	0.000655096688622402\\
0.224609375	0.00022442678352259\\
0.2265625	0.000154942291423392\\
0.228515625	1.10512079491842e-05\\
0.23046875	4.6258540170925e-06\\
0.232421875	2.07109527012655e-07\\
0.234375	1.3662586767975e-06\\
0.236328125	4.7856247319789e-05\\
0.23828125	0.000883006337983247\\
0.240234375	0.00644704533702784\\
0.2421875	0.000108032157164056\\
0.244140625	2.19492995508225e-06\\
0.24609375	3.39789219835499e-07\\
0.248046875	1.87500432371014e-08\\
0.25	7.88212968780288e-09\\
0.251953125	1.55529652186779e-08\\
0.25390625	2.56099468970474e-08\\
0.255859375	2.77576025768627e-07\\
0.2578125	3.35661280356505e-06\\
0.259765625	3.05430177744199e-05\\
0.26171875	0.00115432976415033\\
0.263671875	0.00440728136322923\\
0.265625	0.0211822800184338\\
0.267578125	0.197688227604371\\
0.26953125	0.300866617720967\\
0.271484375	0.214305160920021\\
0.2734375	0.166748940755291\\
0.275390625	1.06024268069296\\
0.27734375	0.17003094879498\\
0.279296875	0.119971342761075\\
0.28125	0.00887089886273291\\
0.283203125	0.000741505596735378\\
0.28515625	5.39697974328225e-05\\
0.287109375	6.7189390801486e-06\\
0.2890625	3.01527453420819e-07\\
0.291015625	1.0022446544962e-08\\
0.29296875	1.7591784114747e-07\\
0.294921875	5.75964439220904e-07\\
0.296875	9.19028172322691e-07\\
0.298828125	3.69712486074123e-07\\
0.30078125	7.46210986244599e-09\\
0.302734375	8.21508202501378e-10\\
0.3046875	6.44358680592206e-09\\
0.306640625	1.51693757815202e-07\\
0.30859375	3.01856776038769e-06\\
0.310546875	1.10564906141194e-05\\
0.3125	0.000162319077162277\\
0.314453125	0.00189176108401859\\
0.31640625	0.00375901895682801\\
0.318359375	0.0979175695274352\\
0.3203125	0.462694176357684\\
0.322265625	0.5258238732678\\
0.32421875	2.52027280913173\\
0.326171875	0.976487052725158\\
0.328125	0.0390704813133758\\
0.330078125	0.0128781719224976\\
0.33203125	0.237310423164152\\
0.333984375	1.09084072283054\\
0.3359375	0.0668569874078706\\
0.337890625	0.0350689838771766\\
0.33984375	1.17855893925139\\
0.341796875	2.40576675427724\\
0.34375	0.852823475689565\\
0.345703125	2.72286121347833\\
0.34765625	0.397851091322963\\
0.349609375	0.0252456484015835\\
0.3515625	0.00763048912764621\\
0.353515625	0.000368946711969598\\
0.35546875	1.70415948468865e-05\\
0.357421875	5.41927465337324e-07\\
0.359375	8.57987593315176e-08\\
0.361328125	5.79469989527848e-09\\
0.36328125	9.92866381238515e-11\\
0.365234375	4.85213065494709e-11\\
0.3671875	1.54472261369826e-09\\
0.369140625	3.52211269861993e-08\\
0.37109375	6.60574417284596e-07\\
0.373046875	1.33365155306539e-05\\
0.375	0.000532701947909349\\
0.376953125	0.00741518461284049\\
0.37890625	0.0489460273609432\\
0.380859375	0.510572774769173\\
0.3828125	1.0636722193986\\
0.384765625	0.968258882011137\\
0.38671875	0.239403788393145\\
0.388671875	0.0464825386098881\\
0.390625	0.00971244086767208\\
0.392578125	0.000316824441425559\\
0.39453125	2.85205463739477e-05\\
0.396484375	0.000695544726418704\\
0.3984375	0.035911746715399\\
0.400390625	0.559039976516915\\
0.40234375	1.47027379034743\\
0.404296875	0.104330663277063\\
0.40625	0.0570223324522367\\
0.408203125	0.00404899062931317\\
0.41015625	0.00574263552823622\\
0.412109375	0.166868591189299\\
0.4140625	3.22463150113676\\
0.416015625	0.550569340716247\\
0.41796875	0.657143722718039\\
0.419921875	1.12929555313226\\
0.421875	0.0464342858191305\\
0.423828125	0.010638295203851\\
0.42578125	0.0643164053700528\\
0.427734375	0.663604816377152\\
0.4296875	0.400712825334738\\
0.431640625	2.44987223773652\\
0.43359375	2.29967117570442\\
0.435546875	0.307184567013132\\
0.4375	0.122183923345681\\
0.439453125	0.07666449408205\\
0.44140625	0.00418840541640955\\
0.443359375	0.0186157967716429\\
0.4453125	0.0643625083252132\\
0.447265625	1.07229658073434\\
0.44921875	1.28998866934946\\
0.451171875	0.063319625965774\\
0.453125	0.0126597741779491\\
0.455078125	0.000606383933155022\\
0.45703125	2.30064099456448e-05\\
0.458984375	1.20641249330278e-06\\
0.4609375	2.47745892019585e-08\\
0.462890625	1.42550165692605e-08\\
0.46484375	1.43845094260814e-07\\
0.466796875	8.91078065359299e-07\\
0.46875	4.53584604608305e-06\\
0.470703125	0.00020183256433564\\
0.47265625	0.00195989167546097\\
0.474609375	0.018369301580509\\
0.4765625	0.274677762660019\\
0.478515625	0.249878750009574\\
0.48046875	0.993332906645347\\
0.482421875	2.47236680694594\\
0.484375	0.211748861960975\\
0.486328125	0.793694207464453\\
0.48828125	1.35760245306886\\
0.490234375	1.82303694938052\\
0.4921875	0.0633659826538098\\
0.494140625	0.00227534240180708\\
0.49609375	0.00234709002350948\\
0.498046875	0.00563114194872448\\
0.5	0.0415848771312512\\
0.501953125	0.689392810050984\\
0.50390625	3.38609446548964\\
0.505859375	0.911426515168245\\
0.5078125	3.58007504281661\\
0.509765625	0.128063671603625\\
0.51171875	0.0257358373331566\\
0.513671875	0.0152985661088254\\
0.515625	0.00277336247702926\\
0.517578125	0.030708404630521\\
0.51953125	0.120530003519242\\
0.521484375	3.47125274584987\\
0.5234375	0.802365795139806\\
0.525390625	1.49694269958607\\
0.52734375	4.98953725514077\\
0.529296875	1.68154950147948\\
0.53125	2.16769438239459\\
0.533203125	0.586626217131289\\
0.53515625	0.100120513385937\\
0.537109375	0.401126955507315\\
0.5390625	0.936007949681456\\
0.541015625	0.0486053725798129\\
0.54296875	0.00177367818410801\\
0.544921875	0.0062805616757151\\
0.546875	0.0380656885993846\\
0.548828125	0.181802462662975\\
0.55078125	0.429489634645437\\
0.552734375	1.11954351350842\\
0.5546875	2.19255601792521\\
0.556640625	0.625992641736526\\
0.55859375	0.021001887677505\\
0.560546875	0.229796801030653\\
0.5625	3.50607500530864\\
0.564453125	0.703052502878718\\
0.56640625	0.0880516481889564\\
0.568359375	0.396051330619429\\
0.5703125	4.03431597896917\\
0.572265625	0.0642425090375881\\
0.57421875	0.00448548845551175\\
0.576171875	0.000524304568053784\\
0.578125	0.000116979550207903\\
0.580078125	1.32310087748219e-05\\
0.58203125	4.46378154238441e-07\\
0.583984375	3.46569532185502e-08\\
0.5859375	9.31306919840905e-09\\
0.587890625	3.84121308864123e-08\\
0.58984375	5.48986497128709e-07\\
0.591796875	3.60762917841024e-06\\
0.59375	8.6502858770239e-05\\
0.595703125	0.000457763318394604\\
0.59765625	0.000481101173452617\\
0.599609375	0.0155041860367271\\
0.6015625	0.0923894274937784\\
0.603515625	0.254076427750841\\
0.60546875	0.548076779743775\\
0.607421875	1.76583810656963\\
0.609375	4.11102531575023\\
0.611328125	1.13397883330689\\
0.61328125	0.316418797407739\\
0.615234375	0.0853989264313001\\
0.6171875	0.00966744155375056\\
0.619140625	0.011525134293486\\
0.62109375	0.0134746700581092\\
0.623046875	0.0777361950241902\\
0.625	1.64512349658768\\
0.626953125	1.61803228639969\\
0.62890625	0.877800812774581\\
0.630859375	0.0270025939997735\\
0.6328125	0.00407295567731133\\
0.634765625	0.001082936079933\\
0.63671875	0.006924906576257\\
0.638671875	0.251341592775794\\
0.640625	0.972119784924761\\
0.642578125	0.512674521043974\\
0.64453125	1.88766507136832\\
0.646484375	0.197471280467274\\
0.6484375	0.0245888306232557\\
0.650390625	0.00457891879339601\\
0.65234375	0.000761894731637399\\
0.654296875	2.71946342286154e-05\\
0.65625	0.000171413491507827\\
0.658203125	0.00803607193970429\\
0.66015625	0.0161826751024506\\
0.662109375	0.16759308571877\\
0.6640625	1.00586025920568\\
0.666015625	0.377464284812173\\
0.66796875	5.43763238817407\\
0.669921875	2.2437016018593\\
0.671875	0.10425064363924\\
0.673828125	0.00359294366918007\\
0.67578125	0.000390023904853958\\
0.677734375	0.000704064127644492\\
0.6796875	0.0139372329131856\\
0.681640625	0.373740442337514\\
0.68359375	2.23927850157001\\
0.685546875	1.18479593641926\\
0.6875	1.17386901004625\\
0.689453125	0.151294022863989\\
0.69140625	0.00625810105865569\\
0.693359375	0.00122342425627167\\
0.6953125	0.000139768083144884\\
0.697265625	7.27653246377341e-06\\
0.69921875	3.89508690952565e-07\\
0.701171875	1.18134752150278e-07\\
0.703125	2.77344382078283e-07\\
0.705078125	4.61439238426421e-06\\
0.70703125	0.000252623094058259\\
0.708984375	0.00243898059683148\\
0.7109375	0.0390645175762716\\
0.712890625	0.086461714683576\\
0.71484375	0.0822688514388956\\
0.716796875	1.07075678191816\\
0.71875	3.38092802779899\\
0.720703125	2.97424489732935\\
0.72265625	0.878112327165017\\
0.724609375	1.4318866076885\\
0.7265625	0.0704103213557078\\
0.728515625	0.00239172464171626\\
0.73046875	0.00112385029104788\\
0.732421875	0.000113440625689882\\
0.734375	0.00110956283110837\\
0.736328125	0.00460957294105497\\
0.73828125	0.0261104707549176\\
0.740234375	0.0630387500563719\\
0.7421875	0.0192137588608101\\
0.744140625	0.0607896916914918\\
0.74609375	0.739751597275547\\
0.748046875	1.89456843646727\\
0.75	0.110136702762049\\
0.751953125	0.00365490167273353\\
0.75390625	0.000188282669360058\\
0.755859375	3.45463105407392e-05\\
0.7578125	2.46107152972105e-06\\
0.759765625	1.27292061282933e-07\\
0.76171875	1.2653200375824e-07\\
0.763671875	4.86781354845613e-06\\
0.765625	5.93928051393121e-05\\
0.767578125	0.000555684921813936\\
0.76953125	0.0180822567597053\\
0.771484375	0.153527404094983\\
0.7734375	2.01005394602648\\
0.775390625	0.796463402595027\\
0.77734375	0.0760553013205927\\
0.779296875	0.00172177372817206\\
0.78125	0.000496171525971206\\
0.783203125	0.000167039006280446\\
0.78515625	2.60584563418547e-05\\
0.787109375	0.000624561419838845\\
0.7890625	0.000356713501317029\\
0.791015625	0.0022696754493841\\
0.79296875	0.0186432039845219\\
0.794921875	0.520201341747599\\
0.796875	1.36202369835468\\
0.798828125	3.97677841562677\\
0.80078125	1.02878701512887\\
0.802734375	1.88014939404617\\
0.8046875	1.44970574061058\\
0.806640625	0.0721669976111685\\
0.80859375	0.0601566661307254\\
0.810546875	0.0163692333823782\\
0.8125	0.00118261999006928\\
0.814453125	0.000338661261299505\\
0.81640625	0.0114654395472696\\
0.818359375	0.120254185798071\\
0.8203125	0.32026323930808\\
0.822265625	3.92885227301599\\
0.82421875	0.34580446651656\\
0.826171875	1.43830119551091\\
0.828125	0.378338714141465\\
0.830078125	1.19712509799317\\
0.83203125	1.12063489041058\\
0.833984375	0.256981293251442\\
0.8359375	0.651616114940069\\
0.837890625	0.0444599823723685\\
0.83984375	0.000827044150771354\\
0.841796875	1.64256965468018e-05\\
0.84375	2.71759396347417e-05\\
0.845703125	0.000271026193651232\\
0.84765625	0.00596269529042372\\
0.849609375	0.271665790244279\\
0.8515625	1.8373455848992\\
0.853515625	0.184870738964919\\
0.85546875	0.0839673228376474\\
0.857421875	0.0054290098056763\\
0.859375	0.0157389065096221\\
0.861328125	0.0124127546441515\\
0.86328125	0.210059709482057\\
0.865234375	0.353226297844971\\
0.8671875	3.45318436576956\\
0.869140625	1.08100252702217\\
0.87109375	0.025276182503541\\
0.873046875	0.00119882594538496\\
0.875	8.67280496601733e-05\\
0.876953125	6.51256708504833e-05\\
0.87890625	0.000183543285260909\\
0.880859375	0.00375786555133063\\
0.8828125	0.0235523641459348\\
0.884765625	0.148092610930255\\
0.88671875	0.356886202043422\\
0.888671875	3.38953344819353\\
0.890625	1.59069441676898\\
0.892578125	0.55788094930001\\
0.89453125	3.19981552161088\\
0.896484375	0.12489258978194\\
0.8984375	0.0063124347969675\\
0.900390625	0.0002670169885675\\
0.90234375	5.41706976085373e-06\\
0.904296875	6.795759606966e-07\\
0.90625	1.12059060308695e-06\\
0.908203125	6.25092413693178e-06\\
0.91015625	2.22739491322636e-05\\
0.912109375	0.000189205902851122\\
0.9140625	0.000805125697299438\\
0.916015625	0.00220546984155218\\
0.91796875	0.0370213940091068\\
0.919921875	0.0725727733926952\\
0.921875	2.13333353428586\\
0.923828125	0.92001683290198\\
0.92578125	1.10885465352257\\
0.927734375	1.84430751799628\\
0.9296875	1.38790426357238\\
0.931640625	0.446662704951736\\
0.93359375	0.643980014323771\\
0.935546875	0.952648484237419\\
0.9375	0.214561598256386\\
0.939453125	0.0240948217080908\\
0.94140625	1.24847376783766\\
0.943359375	0.12049972854298\\
0.9453125	0.0739046610977698\\
0.947265625	0.00205390205798157\\
0.94921875	0.000168515029036906\\
0.951171875	1.24075239066032e-05\\
0.953125	6.96781797729324e-07\\
0.955078125	1.52899647347667e-07\\
0.95703125	3.88796760913605e-09\\
0.958984375	5.00940760192065e-10\\
0.9609375	3.08560639754029e-09\\
0.962890625	1.22981498310661e-07\\
0.96484375	5.61164801166695e-07\\
0.966796875	3.98732176402427e-06\\
0.96875	8.51318586671987e-06\\
0.970703125	0.00072144378480532\\
0.97265625	0.0227499763587306\\
0.974609375	0.249342978253446\\
0.9765625	3.64836017303486\\
0.978515625	1.68930580819865\\
0.98046875	0.559141907859066\\
0.982421875	0.0478789116554522\\
0.984375	0.00300294643851649\\
0.986328125	0.000101272355090062\\
0.98828125	5.02779921253725e-05\\
0.990234375	0.000101308358172208\\
0.9921875	0.00214901000361355\\
0.994140625	0.0372118446758114\\
0.99609375	0.243694450667642\\
0.998046875	3.01314304412011\\
1	0\\
};
\end{axis}
\end{tikzpicture}%

%% file: pics/square/GroundState_1_1_1.tex
%
%
\begin{tikzpicture}

\begin{axis}[%
width=0.9in,
height=0.9in,
scale only axis,
xmin=0,
xmax=1,
xtick={0,0.5,1},
xlabel={$c_\kappa = 1$},
xlabel style={below=-1.4in},
ymin=0,
ymax=10,
yminorticks=true,
ytick={0,4,8},
ylabel={$c_V = 1$},
ylabel style={below=0.1in},
axis background/.style={fill=white},
]
\addplot [color=mycolor2, thick]
table[row sep=crcr]{%
0	0\\
0.001953125	0.167972972566256\\
0.00390625	0.508122550088351\\
0.005859375	0.678863938965472\\
0.0078125	0.424417754852016\\
0.009765625	0.324007071967434\\
0.01171875	0.180334316285269\\
0.013671875	0.130767208507824\\
0.015625	0.158434407087513\\
0.017578125	0.272003584122497\\
0.01953125	0.412247047141441\\
0.021484375	0.391695463651774\\
0.0234375	0.296698653858257\\
0.025390625	0.405135240827414\\
0.02734375	0.709229353417595\\
0.029296875	1.31640993791106\\
0.03125	1.33359976953648\\
0.033203125	0.957956414644595\\
0.03515625	0.488577329040131\\
0.037109375	0.399475232464936\\
0.0390625	0.240526002171306\\
0.041015625	0.174754813031272\\
0.04296875	0.183427971557608\\
0.044921875	0.170656083548297\\
0.046875	0.145563910660511\\
0.048828125	0.230756768752582\\
0.05078125	0.330619925937448\\
0.052734375	0.457627423037242\\
0.0546875	0.368926232313746\\
0.056640625	0.423039460988648\\
0.05859375	0.791141126889112\\
0.060546875	1.12212646748882\\
0.0625	1.75049076115615\\
0.064453125	1.58826572343323\\
0.06640625	1.02935080768071\\
0.068359375	0.780031972737413\\
0.0703125	0.581169069573582\\
0.072265625	0.826390656662518\\
0.07421875	1.21446159946191\\
0.076171875	1.41398704828766\\
0.078125	1.18147051332919\\
0.080078125	0.847824374527099\\
0.08203125	0.422285210368846\\
0.083984375	0.621145767226421\\
0.0859375	0.640920024111639\\
0.087890625	0.769505716392139\\
0.08984375	0.824340224219386\\
0.091796875	1.06798515274103\\
0.09375	1.16848855443684\\
0.095703125	1.24207743750144\\
0.09765625	1.09398573028946\\
0.099609375	1.33772669436053\\
0.1015625	1.19394140642593\\
0.103515625	1.44806756354706\\
0.10546875	1.92018094853848\\
0.107421875	2.16005261128163\\
0.109375	1.92817633661987\\
0.111328125	1.93683035021824\\
0.11328125	2.02292506081281\\
0.115234375	2.00268011727955\\
0.1171875	2.04845142131447\\
0.119140625	1.38207367363901\\
0.12109375	0.851927195611468\\
0.123046875	0.606443440514257\\
0.125	0.586245677823889\\
0.126953125	0.731275302725449\\
0.12890625	0.784030720189171\\
0.130859375	0.50591241132802\\
0.1328125	0.554162689802836\\
0.134765625	0.486045390684988\\
0.13671875	0.724580377923125\\
0.138671875	1.01165177609109\\
0.140625	1.08342514905604\\
0.142578125	1.2168782097639\\
0.14453125	1.19473378948827\\
0.146484375	1.32011736008597\\
0.1484375	1.6070459790571\\
0.150390625	1.3565134029749\\
0.15234375	1.17460643130301\\
0.154296875	1.29265662731826\\
0.15625	1.30282422552665\\
0.158203125	2.00688179508481\\
0.16015625	2.28000601928068\\
0.162109375	1.71793805973574\\
0.1640625	1.1561336151426\\
0.166015625	1.19611796224669\\
0.16796875	1.70848381167839\\
0.169921875	2.15369937108049\\
0.171875	1.87740494119958\\
0.173828125	1.36586117982924\\
0.17578125	1.12694777738809\\
0.177734375	1.12285018234343\\
0.1796875	1.08792591133288\\
0.181640625	0.990736812793043\\
0.18359375	0.667169031309441\\
0.185546875	0.559389551123495\\
0.1875	0.449624283753297\\
0.189453125	0.264361443575827\\
0.19140625	0.146767655841879\\
0.193359375	0.0991522424934223\\
0.1953125	0.105708899343459\\
0.197265625	0.137181267423754\\
0.19921875	0.169546336965836\\
0.201171875	0.183520846857007\\
0.203125	0.232284510029895\\
0.205078125	0.32769246258829\\
0.20703125	0.491788961946763\\
0.208984375	0.639667073479168\\
0.2109375	0.975224717186186\\
0.212890625	1.32310145248844\\
0.21484375	1.30553905079184\\
0.216796875	1.18900222990999\\
0.21875	1.15582887099593\\
0.220703125	1.80206541270108\\
0.22265625	1.98453625296743\\
0.224609375	1.75493975249918\\
0.2265625	1.70020979003537\\
0.228515625	1.45480440269459\\
0.23046875	0.976849442596029\\
0.232421875	0.706569278978059\\
0.234375	0.35889083600093\\
0.236328125	0.277636637716566\\
0.23828125	0.39336878401662\\
0.240234375	0.634531039347649\\
0.2421875	0.807113652221501\\
0.244140625	1.17899782430205\\
0.24609375	1.73245013678287\\
0.248046875	1.74953645741371\\
0.25	1.48977907268511\\
0.251953125	1.30594618214347\\
0.25390625	1.23686398609975\\
0.255859375	1.5710797694066\\
0.2578125	1.44200671076119\\
0.259765625	1.50120384110992\\
0.26171875	1.85572196008442\\
0.263671875	1.78584768922784\\
0.265625	1.28175444451448\\
0.267578125	1.10904704130301\\
0.26953125	0.999208763264712\\
0.271484375	1.19725506312661\\
0.2734375	1.1514183719091\\
0.275390625	0.807229302143351\\
0.27734375	0.58779738239062\\
0.279296875	0.578940038778653\\
0.28125	0.435577097288836\\
0.283203125	0.370050758775465\\
0.28515625	0.336592114545497\\
0.287109375	0.472936842131428\\
0.2890625	0.62618128540443\\
0.291015625	0.808256908740866\\
0.29296875	0.893284340200784\\
0.294921875	1.36962233954263\\
0.296875	1.65578868749311\\
0.298828125	1.44752047561371\\
0.30078125	1.31623832861236\\
0.302734375	1.40239928944175\\
0.3046875	1.04528277315024\\
0.306640625	0.877520777295298\\
0.30859375	0.65400777851778\\
0.310546875	0.48190490047843\\
0.3125	0.34553579383575\\
0.314453125	0.327809588978928\\
0.31640625	0.224859219607787\\
0.318359375	0.12513387849384\\
0.3203125	0.0868533483389469\\
0.322265625	0.0542348900373422\\
0.32421875	0.0410449990177858\\
0.326171875	0.0273556462239382\\
0.328125	0.0246433949310893\\
0.330078125	0.0242555390768677\\
0.33203125	0.0394605728471835\\
0.333984375	0.0731910626633801\\
0.3359375	0.09126254159421\\
0.337890625	0.136932889142637\\
0.33984375	0.183835691905521\\
0.341796875	0.188399237835439\\
0.34375	0.233325678936403\\
0.345703125	0.258978576848997\\
0.34765625	0.480121986844887\\
0.349609375	1.07949241121504\\
0.3515625	1.36413547464264\\
0.353515625	1.21491040392384\\
0.35546875	0.847239614065818\\
0.357421875	0.820734875970507\\
0.359375	1.17136109750741\\
0.361328125	0.949427509603199\\
0.36328125	0.90311872833316\\
0.365234375	0.88480893706108\\
0.3671875	0.864005189298839\\
0.369140625	0.748299401131557\\
0.37109375	0.904222568400139\\
0.373046875	1.21758569003422\\
0.375	1.90516954508998\\
0.376953125	1.3253189704604\\
0.37890625	1.31441130188666\\
0.380859375	1.20642977720095\\
0.3828125	1.15421073053013\\
0.384765625	1.19309537344378\\
0.38671875	1.73703293064468\\
0.388671875	1.93500034597776\\
0.390625	1.5198611631948\\
0.392578125	1.42160444897163\\
0.39453125	1.52761386672741\\
0.396484375	1.77800386533351\\
0.3984375	2.06344063496995\\
0.400390625	1.61802604992404\\
0.40234375	1.09002202523633\\
0.404296875	0.512786958977239\\
0.40625	0.276936505041666\\
0.408203125	0.250934148750839\\
0.41015625	0.263413801700001\\
0.412109375	0.380678453596981\\
0.4140625	0.518877833605425\\
0.416015625	0.572426140850557\\
0.41796875	0.604262671337445\\
0.419921875	0.790228493783\\
0.421875	0.826606369093683\\
0.423828125	0.841974049062958\\
0.42578125	1.30453338685566\\
0.427734375	1.49297438003817\\
0.4296875	1.78758661849413\\
0.431640625	1.35344804663514\\
0.43359375	1.33211066010846\\
0.435546875	1.28392156583692\\
0.4375	0.80145697460873\\
0.439453125	0.688288062738606\\
0.44140625	0.684427318972306\\
0.443359375	0.69954891254244\\
0.4453125	0.742196037500032\\
0.447265625	1.06703553208119\\
0.44921875	1.29622667251709\\
0.451171875	1.42238239300855\\
0.453125	1.22573772025527\\
0.455078125	0.908940629410608\\
0.45703125	0.704879508755469\\
0.458984375	0.804708143798355\\
0.4609375	0.514706102879104\\
0.462890625	0.224018605814389\\
0.46484375	0.135576035736608\\
0.466796875	0.113680653773875\\
0.46875	0.0717171266456891\\
0.470703125	0.0407166994272282\\
0.47265625	0.0352358042173662\\
0.474609375	0.0381552355451781\\
0.4765625	0.050229222770732\\
0.478515625	0.065029522336932\\
0.48046875	0.0708539501412964\\
0.482421875	0.0378545534029782\\
0.484375	0.0401254174535313\\
0.486328125	0.0682423226314562\\
0.48828125	0.153479767026884\\
0.490234375	0.250372784171885\\
0.4921875	0.424367638416126\\
0.494140625	0.46061092071055\\
0.49609375	0.589643060392751\\
0.498046875	0.830306828757118\\
0.5	1.21714203311195\\
0.501953125	1.40176121692279\\
0.50390625	1.51816787414072\\
0.505859375	1.29549618537997\\
0.5078125	1.04514544330177\\
0.509765625	0.719823312579355\\
0.51171875	0.63286848963\\
0.513671875	0.348256794368998\\
0.515625	0.313703701445932\\
0.517578125	0.405528587244112\\
0.51953125	0.369852247720824\\
0.521484375	0.348897713769549\\
0.5234375	0.696008213605651\\
0.525390625	1.09316372498196\\
0.52734375	1.10891085400871\\
0.529296875	1.03584028937968\\
0.53125	1.04752705169166\\
0.533203125	1.11152529542027\\
0.53515625	1.32285296711546\\
0.537109375	1.09039132173902\\
0.5390625	1.48808054680892\\
0.541015625	1.65641673533951\\
0.54296875	1.34423225921096\\
0.544921875	1.09862427843367\\
0.546875	0.98005057504872\\
0.548828125	1.53976455403672\\
0.55078125	1.67809681822882\\
0.552734375	1.54133652704844\\
0.5546875	1.30608294488713\\
0.556640625	1.27999762824611\\
0.55859375	1.34283059496329\\
0.560546875	1.03011172058934\\
0.5625	0.552574491870955\\
0.564453125	0.421805394670148\\
0.56640625	0.371979689015118\\
0.568359375	0.30472099971145\\
0.5703125	0.261979428434946\\
0.572265625	0.234785325856068\\
0.57421875	0.195569903931153\\
0.576171875	0.128431538747463\\
0.578125	0.103080720535711\\
0.580078125	0.102040328324553\\
0.58203125	0.177127098680117\\
0.583984375	0.351103294544015\\
0.5859375	0.694905946142818\\
0.587890625	0.993419767473554\\
0.58984375	1.00237837322982\\
0.591796875	1.38325818121393\\
0.59375	1.38719935460902\\
0.595703125	1.32712489741346\\
0.59765625	1.01219998761544\\
0.599609375	0.43428012924007\\
0.6015625	0.245719714412842\\
0.603515625	0.199876223013392\\
0.60546875	0.133530831498615\\
0.607421875	0.0911287913194012\\
0.609375	0.0987510968235909\\
0.611328125	0.11586169305352\\
0.61328125	0.192065732182316\\
0.615234375	0.285820320966764\\
0.6171875	0.447929279421917\\
0.619140625	0.557302773697974\\
0.62109375	0.508486276470956\\
0.623046875	0.393340603831385\\
0.625	0.533064733560729\\
0.626953125	0.797870958568701\\
0.62890625	1.07649499660377\\
0.630859375	1.22234313313158\\
0.6328125	0.905895850241141\\
0.634765625	0.646789400780937\\
0.63671875	1.01980919370379\\
0.638671875	1.14698684022663\\
0.640625	1.43353402038047\\
0.642578125	1.3306919772896\\
0.64453125	1.27622908942895\\
0.646484375	1.88953243780249\\
0.6484375	2.1061341241722\\
0.650390625	1.89857271228802\\
0.65234375	1.53711904066873\\
0.654296875	0.982498410670385\\
0.65625	0.502220917461468\\
0.658203125	0.420216020195719\\
0.66015625	0.351305775199876\\
0.662109375	0.507379699992066\\
0.6640625	0.690061462132406\\
0.666015625	1.09921755499805\\
0.66796875	1.08458452777948\\
0.669921875	0.831347243234621\\
0.671875	0.517759641184653\\
0.673828125	0.42796898350575\\
0.67578125	0.454119109082322\\
0.677734375	0.764230018706301\\
0.6796875	1.29915352372719\\
0.681640625	2.14867351358702\\
0.68359375	2.05383829153868\\
0.685546875	1.7510618467973\\
0.6875	1.66551206484561\\
0.689453125	1.66446799108104\\
0.69140625	1.98713545946762\\
0.693359375	1.68983749441849\\
0.6953125	1.00277910153128\\
0.697265625	0.769079855601036\\
0.69921875	0.76182655252929\\
0.701171875	0.890755687998412\\
0.703125	0.670898980243173\\
0.705078125	0.433551475438349\\
0.70703125	0.255594643007606\\
0.708984375	0.162859133671098\\
0.7109375	0.0875019561747447\\
0.712890625	0.0548093517992679\\
0.71484375	0.0396185309354514\\
0.716796875	0.0430133633116594\\
0.71875	0.0368300729352672\\
0.720703125	0.0245626867776806\\
0.72265625	0.0128932217383595\\
0.724609375	0.00581574637049343\\
0.7265625	0.00364574033330538\\
0.728515625	0.00281683961922032\\
0.73046875	0.00191141565895316\\
0.732421875	0.00222425534888743\\
0.734375	0.00339009630165362\\
0.736328125	0.00528546944765376\\
0.73828125	0.00847719118790834\\
0.740234375	0.012549936236434\\
0.7421875	0.0209323821768789\\
0.744140625	0.0353060437248364\\
0.74609375	0.0597887708536658\\
0.748046875	0.096841213219225\\
0.75	0.203924811810406\\
0.751953125	0.422966638276448\\
0.75390625	0.581071035736784\\
0.755859375	0.610886565197485\\
0.7578125	0.678494308338096\\
0.759765625	0.82773631559487\\
0.76171875	1.19269601068048\\
0.763671875	1.0145710053023\\
0.765625	1.34979900544922\\
0.767578125	1.77318071644664\\
0.76953125	1.29126820495181\\
0.771484375	1.09347599571346\\
0.7734375	1.25952843480936\\
0.775390625	1.12066446233613\\
0.77734375	0.960234183549232\\
0.779296875	0.904306567030492\\
0.78125	0.684212056552968\\
0.783203125	0.348136140785221\\
0.78515625	0.200593307713713\\
0.787109375	0.112217404816034\\
0.7890625	0.0751861454114256\\
0.791015625	0.055279828857175\\
0.79296875	0.0755987285020526\\
0.794921875	0.14669191054599\\
0.796875	0.33513370310979\\
0.798828125	0.478198823255511\\
0.80078125	0.863710953863115\\
0.802734375	1.28199820708792\\
0.8046875	1.30820849611673\\
0.806640625	1.04556118457663\\
0.80859375	0.848755793909628\\
0.810546875	1.07577465367968\\
0.8125	0.844911260494272\\
0.814453125	0.58702627602232\\
0.81640625	0.41697799201417\\
0.818359375	0.332951850165376\\
0.8203125	0.313087159417852\\
0.822265625	0.369474313833057\\
0.82421875	0.634285739591201\\
0.826171875	1.04377212722808\\
0.828125	1.67454614018331\\
0.830078125	1.50418025354119\\
0.83203125	1.05175548391109\\
0.833984375	1.01520297429096\\
0.8359375	1.2144906677373\\
0.837890625	1.63173430587225\\
0.83984375	1.34496995558857\\
0.841796875	0.794925671552195\\
0.84375	0.431265351555039\\
0.845703125	0.46173910686301\\
0.84765625	0.604879716151747\\
0.849609375	0.799269741937467\\
0.8515625	0.906728610443677\\
0.853515625	0.755018005556565\\
0.85546875	1.05832682547238\\
0.857421875	1.41715674952131\\
0.859375	1.13837316001129\\
0.861328125	1.36049383712607\\
0.86328125	1.84994135429958\\
0.865234375	2.17113993032837\\
0.8671875	1.98749590690206\\
0.869140625	1.54306854596893\\
0.87109375	0.832345555047828\\
0.873046875	0.567840063965751\\
0.875	0.936930718913232\\
0.876953125	1.22543068855357\\
0.87890625	0.893008942481021\\
0.880859375	0.598228710368499\\
0.8828125	0.457929031424223\\
0.884765625	0.507081527349486\\
0.88671875	0.753976416539326\\
0.888671875	0.612042194440457\\
0.890625	0.641679552590858\\
0.892578125	0.883087354364976\\
0.89453125	1.22187102614292\\
0.896484375	1.63008142104669\\
0.8984375	1.16894969541185\\
0.900390625	0.81317734198373\\
0.90234375	0.739429750108241\\
0.904296875	0.746316410478673\\
0.90625	0.689204749399857\\
0.908203125	0.46799936643519\\
0.91015625	0.239804005145467\\
0.912109375	0.147203964540342\\
0.9140625	0.0863786943250686\\
0.916015625	0.0543147043674603\\
0.91796875	0.0457548575038262\\
0.919921875	0.0361604137511574\\
0.921875	0.0469807972987345\\
0.923828125	0.0920258794774308\\
0.92578125	0.162064091381216\\
0.927734375	0.203915238268673\\
0.9296875	0.315324056891293\\
0.931640625	0.394968010538678\\
0.93359375	0.416835496160862\\
0.935546875	0.240418884186967\\
0.9375	0.175009605419399\\
0.939453125	0.175772200879133\\
0.94140625	0.298448594259521\\
0.943359375	0.420121122471829\\
0.9453125	0.545118732601914\\
0.947265625	0.658362735488366\\
0.94921875	0.738928916829164\\
0.951171875	0.882402365719429\\
0.953125	1.30432387786794\\
0.955078125	1.54998634146538\\
0.95703125	1.39049700154188\\
0.958984375	1.25006470818515\\
0.9609375	1.46375672552254\\
0.962890625	1.99456639942806\\
0.96484375	1.83824859538698\\
0.966796875	1.20533608401032\\
0.96875	1.12085913483157\\
0.970703125	1.49806310688987\\
0.97265625	1.68607087671101\\
0.974609375	1.48874466950381\\
0.9765625	1.74784342978264\\
0.978515625	1.94518306347673\\
0.98046875	1.63463629021271\\
0.982421875	0.89313396216661\\
0.984375	0.431402791302469\\
0.986328125	0.314424868743119\\
0.98828125	0.245821286176624\\
0.990234375	0.337472001829155\\
0.9921875	0.394163632551871\\
0.994140625	0.246979480556503\\
0.99609375	0.140689041703822\\
0.998046875	0.0707655595919883\\
1	0\\
};
\end{axis}
\end{tikzpicture}%

%% file: pics/square/GroundState_1_1_2.tex
%
%
\begin{tikzpicture}

\begin{axis}[%
width=0.9in,
height=0.9in,
scale only axis,
xmin=0,
xmax=1,
xtick={0,0.5,1},
xlabel={$c_\kappa = 2^{-1}$},
xlabel style={below=-1.4in},
ymin=0,
ymax=10,
yminorticks=true,
yticklabels = {},
axis background/.style={fill=white},
]
\addplot [color=mycolor2, thick]
table[row sep=crcr]{%
0	0\\
0.001953125	0.00291038556190987\\
0.00390625	0.00910442004084681\\
0.005859375	0.0127336892778346\\
0.0078125	0.00906055100051919\\
0.009765625	0.00956005816830627\\
0.01171875	0.00955289577543926\\
0.013671875	0.0165612068293106\\
0.015625	0.0341367253593196\\
0.017578125	0.0745730546861477\\
0.01953125	0.127167376687207\\
0.021484375	0.141320794559109\\
0.0234375	0.144320611736218\\
0.025390625	0.265349284591857\\
0.02734375	0.527225250842841\\
0.029296875	1.053068139873\\
0.03125	1.07011079896866\\
0.033203125	0.736167929052415\\
0.03515625	0.340117813131705\\
0.037109375	0.245921154138878\\
0.0390625	0.132203918863961\\
0.041015625	0.0797407610029433\\
0.04296875	0.072543620610425\\
0.044921875	0.063329815194675\\
0.046875	0.0560451430853849\\
0.048828125	0.09753152564084\\
0.05078125	0.156471255144666\\
0.052734375	0.237998820566596\\
0.0546875	0.230599577873122\\
0.056640625	0.338645227049553\\
0.05859375	0.725073007034918\\
0.060546875	1.11648029199329\\
0.0625	1.82988506324935\\
0.064453125	1.64307085224626\\
0.06640625	0.992004978171144\\
0.068359375	0.685703104799687\\
0.0703125	0.461205170562825\\
0.072265625	0.628903588978182\\
0.07421875	0.935496282432949\\
0.076171875	1.0934928080007\\
0.078125	0.885956260650499\\
0.080078125	0.607167519702783\\
0.08203125	0.268746094183063\\
0.083984375	0.361393702680915\\
0.0859375	0.375632963739466\\
0.087890625	0.479287392723596\\
0.08984375	0.562010773403463\\
0.091796875	0.800678665687031\\
0.09375	0.936744898887126\\
0.095703125	1.055480906508\\
0.09765625	0.991333844176425\\
0.099609375	1.30444221490285\\
0.1015625	1.23868137609615\\
0.103515625	1.63506542932396\\
0.10546875	2.29336435930417\\
0.107421875	2.6427847634105\\
0.109375	2.36142502374507\\
0.111328125	2.368700236038\\
0.11328125	2.47903735441267\\
0.115234375	2.44445790106479\\
0.1171875	2.45465494975586\\
0.119140625	1.57723136015302\\
0.12109375	0.872947705831552\\
0.123046875	0.526104170435555\\
0.125	0.393983621550335\\
0.126953125	0.425864353303151\\
0.12890625	0.426697392381342\\
0.130859375	0.267969354221912\\
0.1328125	0.30248796478755\\
0.134765625	0.291718263977144\\
0.13671875	0.500295916813869\\
0.138671875	0.759517259410991\\
0.140625	0.87805824155536\\
0.142578125	1.05059083690085\\
0.14453125	1.09908288403321\\
0.146484375	1.28595568154249\\
0.1484375	1.62173167269929\\
0.150390625	1.37555092312432\\
0.15234375	1.20858107514716\\
0.154296875	1.37990569649551\\
0.15625	1.48544993753184\\
0.158203125	2.42598373799513\\
0.16015625	2.7992787949312\\
0.162109375	2.06143834611992\\
0.1640625	1.33544179911157\\
0.166015625	1.36756823427132\\
0.16796875	2.00686598767697\\
0.169921875	2.56276323571974\\
0.171875	2.18058392462229\\
0.173828125	1.49688085145776\\
0.17578125	1.12479574628084\\
0.177734375	1.01452235308945\\
0.1796875	0.904464774011077\\
0.181640625	0.766179185405587\\
0.18359375	0.465957787613786\\
0.185546875	0.344051591675287\\
0.1875	0.250983660695454\\
0.189453125	0.135800958273897\\
0.19140625	0.0682638379738585\\
0.193359375	0.039671806213715\\
0.1953125	0.0378325752836548\\
0.197265625	0.0490456765007618\\
0.19921875	0.0655667240344895\\
0.201171875	0.0807701850905656\\
0.203125	0.121723994773443\\
0.205078125	0.200226033673814\\
0.20703125	0.334846623148151\\
0.208984375	0.490502802085071\\
0.2109375	0.825401064794246\\
0.212890625	1.18964727978459\\
0.21484375	1.22051730444845\\
0.216796875	1.16620736567008\\
0.21875	1.23054401805282\\
0.220703125	2.04547867995618\\
0.22265625	2.30485253931219\\
0.224609375	2.02708462497839\\
0.2265625	1.92437731479381\\
0.228515625	1.58645931808733\\
0.23046875	0.989373122781121\\
0.232421875	0.671925464615372\\
0.234375	0.316074095474933\\
0.236328125	0.227554733664803\\
0.23828125	0.324729065733227\\
0.240234375	0.556227435471317\\
0.2421875	0.768787730707147\\
0.244140625	1.22517215416423\\
0.24609375	1.88520804442456\\
0.248046875	1.91717940403559\\
0.25	1.60096950191029\\
0.251953125	1.37782891430904\\
0.25390625	1.3041402579517\\
0.255859375	1.68748542744967\\
0.2578125	1.56657397858896\\
0.259765625	1.66448926751817\\
0.26171875	2.09665964060205\\
0.263671875	1.98966294718847\\
0.265625	1.33448015895608\\
0.267578125	1.05818911155739\\
0.26953125	0.858171109245393\\
0.271484375	0.955774453808516\\
0.2734375	0.863754839932533\\
0.275390625	0.561419251061433\\
0.27734375	0.358465435629476\\
0.279296875	0.314936541104011\\
0.28125	0.223167795855952\\
0.283203125	0.188183064050927\\
0.28515625	0.186054430322669\\
0.287109375	0.299032912740958\\
0.2890625	0.443278365122964\\
0.291015625	0.634549434205378\\
0.29296875	0.788911879141253\\
0.294921875	1.31523089772784\\
0.296875	1.64791861447639\\
0.298828125	1.42493297833898\\
0.30078125	1.24741135647288\\
0.302734375	1.2779795937883\\
0.3046875	0.886258253874882\\
0.306640625	0.675173650432033\\
0.30859375	0.454007804673221\\
0.310546875	0.298030302255292\\
0.3125	0.185276491765773\\
0.314453125	0.153374133982923\\
0.31640625	0.0956863523466025\\
0.318359375	0.0481340493360099\\
0.3203125	0.029703035294154\\
0.322265625	0.0166005023555567\\
0.32421875	0.011075049469353\\
0.326171875	0.0065271072017469\\
0.328125	0.00529288168255673\\
0.330078125	0.00510591234839193\\
0.33203125	0.0087802856500478\\
0.333984375	0.0174663476353982\\
0.3359375	0.0241667366271659\\
0.337890625	0.0414349096764447\\
0.33984375	0.0618144530606955\\
0.341796875	0.075464320727299\\
0.34375	0.110671122021689\\
0.345703125	0.151691030250148\\
0.34765625	0.325095048161148\\
0.349609375	0.779954778436036\\
0.3515625	1.00520224238909\\
0.353515625	0.876340740836771\\
0.35546875	0.575979123156399\\
0.357421875	0.524938684547394\\
0.359375	0.738407001071087\\
0.361328125	0.590509607599583\\
0.36328125	0.56230353062132\\
0.365234375	0.567410726492952\\
0.3671875	0.586278885700008\\
0.369140625	0.573096161338633\\
0.37109375	0.808192405490596\\
0.373046875	1.22445264243809\\
0.375	2.0154848811806\\
0.376953125	1.38382350035206\\
0.37890625	1.34495538515367\\
0.380859375	1.2242927882108\\
0.3828125	1.18856075741419\\
0.384765625	1.28926626956962\\
0.38671875	1.97535561160193\\
0.388671875	2.23269981985937\\
0.390625	1.72999069407312\\
0.392578125	1.59954910437745\\
0.39453125	1.7383788226083\\
0.396484375	2.06961063937877\\
0.3984375	2.42548552869801\\
0.400390625	1.85329740568212\\
0.40234375	1.18584495458818\\
0.404296875	0.517559969749175\\
0.40625	0.23879883481046\\
0.408203125	0.163681461644367\\
0.41015625	0.137042463921869\\
0.412109375	0.177126134590811\\
0.4140625	0.243403993902334\\
0.416015625	0.288247349594637\\
0.41796875	0.348400356841229\\
0.419921875	0.51198868787652\\
0.421875	0.612140703328281\\
0.423828125	0.72972417904719\\
0.42578125	1.26466194589249\\
0.427734375	1.55025878827843\\
0.4296875	1.89580580786749\\
0.431640625	1.39002277149771\\
0.43359375	1.30560137247935\\
0.435546875	1.20231797237082\\
0.4375	0.686629879503184\\
0.439453125	0.52514035671825\\
0.44140625	0.482332485803538\\
0.443359375	0.483597491872778\\
0.4453125	0.532954532311313\\
0.447265625	0.816429676388333\\
0.44921875	1.03267589559702\\
0.451171875	1.14779787858977\\
0.453125	0.957175795768073\\
0.455078125	0.651049243135747\\
0.45703125	0.448843525736764\\
0.458984375	0.467639387793151\\
0.4609375	0.280921845408836\\
0.462890625	0.113438367025726\\
0.46484375	0.0605504577947015\\
0.466796875	0.0441951201081749\\
0.46875	0.0251404930356871\\
0.470703125	0.0119049625097773\\
0.47265625	0.00778463299934716\\
0.474609375	0.00630191335611014\\
0.4765625	0.00722880041880392\\
0.478515625	0.00919372910764202\\
0.48046875	0.0107362893501596\\
0.482421875	0.00784171963058226\\
0.484375	0.0131056337765412\\
0.486328125	0.0285887562220528\\
0.48828125	0.070676966958339\\
0.490234375	0.125114024089184\\
0.4921875	0.230109083251236\\
0.494140625	0.281963158939606\\
0.49609375	0.410755995495527\\
0.498046875	0.650948048645169\\
0.5	1.03139889464153\\
0.501953125	1.24690783250578\\
0.50390625	1.37171859314704\\
0.505859375	1.13787725887333\\
0.5078125	0.8598490914476\\
0.509765625	0.537680697294135\\
0.51171875	0.424833663368013\\
0.513671875	0.206027972887748\\
0.515625	0.157904963332188\\
0.517578125	0.190253366915714\\
0.51953125	0.178684336489649\\
0.521484375	0.193625595977584\\
0.5234375	0.429135333962462\\
0.525390625	0.714083622252031\\
0.52734375	0.758756884379515\\
0.529296875	0.747885779113063\\
0.53125	0.804516688801299\\
0.533203125	0.916692566945056\\
0.53515625	1.16211145135352\\
0.537109375	1.01873661074402\\
0.5390625	1.46360372767635\\
0.541015625	1.66041746609061\\
0.54296875	1.32974329322909\\
0.544921875	1.07017106188064\\
0.546875	0.972926792851392\\
0.548828125	1.59772670415995\\
0.55078125	1.7613397510083\\
0.552734375	1.59039045987124\\
0.5546875	1.30178588360202\\
0.556640625	1.2320478722405\\
0.55859375	1.24769511799311\\
0.560546875	0.906245838054239\\
0.5625	0.435971643095448\\
0.564453125	0.284887037435609\\
0.56640625	0.211273373536014\\
0.568359375	0.145969469171771\\
0.5703125	0.105911276037004\\
0.572265625	0.0828135075470333\\
0.57421875	0.0642610521931417\\
0.576171875	0.0424878527050137\\
0.578125	0.0391614087201953\\
0.580078125	0.0497154379017679\\
0.58203125	0.103110938948132\\
0.583984375	0.224098980444154\\
0.5859375	0.475198344445616\\
0.587890625	0.721991568413398\\
0.58984375	0.777419405065522\\
0.591796875	1.12950115578237\\
0.59375	1.15526839290993\\
0.595703125	1.09600770904159\\
0.59765625	0.806522432483389\\
0.599609375	0.325687886603864\\
0.6015625	0.163559598574312\\
0.603515625	0.113469502886273\\
0.60546875	0.0655864216306827\\
0.607421875	0.0350228256628773\\
0.609375	0.0279701039435561\\
0.611328125	0.025163133698577\\
0.61328125	0.0388443764125129\\
0.615234375	0.059916134392209\\
0.6171875	0.101581478466308\\
0.619140625	0.137460756275819\\
0.62109375	0.146587214214726\\
0.623046875	0.14429569123196\\
0.625	0.256359027099711\\
0.626953125	0.448325627978834\\
0.62890625	0.657164422006781\\
0.630859375	0.784883647334532\\
0.6328125	0.602567863839636\\
0.634765625	0.482125890077813\\
0.63671875	0.864333288878528\\
0.638671875	1.05570061894786\\
0.640625	1.41604224870896\\
0.642578125	1.37343633172451\\
0.64453125	1.41067887172122\\
0.646484375	2.20893793086428\\
0.6484375	2.51203283930005\\
0.650390625	2.23511809890904\\
0.65234375	1.74030858328525\\
0.654296875	1.04697982840404\\
0.65625	0.470074665277584\\
0.658203125	0.32889995010448\\
0.66015625	0.220876055068604\\
0.662109375	0.278831520101131\\
0.6640625	0.373422818379106\\
0.666015625	0.610032389495287\\
0.66796875	0.604705535837384\\
0.669921875	0.460612586116193\\
0.671875	0.302295943159086\\
0.673828125	0.295462853351508\\
0.67578125	0.402450316597131\\
0.677734375	0.802046761386536\\
0.6796875	1.5014039224529\\
0.681640625	2.59139447549696\\
0.68359375	2.49207607562772\\
0.685546875	2.09628258344692\\
0.6875	1.96213865852847\\
0.689453125	1.94886170313796\\
0.69140625	2.3217770964665\\
0.693359375	1.91688206481254\\
0.6953125	1.04113420249172\\
0.697265625	0.683727270627793\\
0.69921875	0.575448660995062\\
0.701171875	0.606362290294574\\
0.703125	0.420673567186999\\
0.705078125	0.251051952574987\\
0.70703125	0.134521944062075\\
0.708984375	0.0786953250937043\\
0.7109375	0.0381993233766247\\
0.712890625	0.0208656148979017\\
0.71484375	0.0121760615442903\\
0.716796875	0.0107368605889341\\
0.71875	0.0081408352755767\\
0.720703125	0.00497941258825331\\
0.72265625	0.00244672146765615\\
0.724609375	0.00101308732661504\\
0.7265625	0.000576543586225171\\
0.728515625	0.000419496129974524\\
0.73046875	0.000292434880178312\\
0.732421875	0.000383634828619783\\
0.734375	0.000660032705038816\\
0.736328125	0.00114282010535208\\
0.73828125	0.0020577798892986\\
0.740234375	0.00335170549483628\\
0.7421875	0.00628642488950371\\
0.744140625	0.0115807596440214\\
0.74609375	0.0214091321711598\\
0.748046875	0.0384150572256046\\
0.75	0.0874689601087872\\
0.751953125	0.192422417166938\\
0.75390625	0.287265764234009\\
0.755859375	0.334733813615075\\
0.7578125	0.436434039502042\\
0.759765625	0.612851224132995\\
0.76171875	0.964564530694015\\
0.763671875	0.885718387339323\\
0.765625	1.27267165974347\\
0.767578125	1.73209215067352\\
0.76953125	1.22327878782087\\
0.771484375	0.971610074349153\\
0.7734375	1.0696606561205\\
0.775390625	0.905180439945835\\
0.77734375	0.713189049792304\\
0.779296875	0.621265468167567\\
0.78125	0.439562049671393\\
0.783203125	0.209217245481829\\
0.78515625	0.11075159674083\\
0.787109375	0.0571321395488106\\
0.7890625	0.0351862721759286\\
0.791015625	0.0253450488069002\\
0.79296875	0.0370272030597447\\
0.794921875	0.07730477735591\\
0.796875	0.188212050469094\\
0.798828125	0.288793972758432\\
0.80078125	0.56512633530278\\
0.802734375	0.884352095567755\\
0.8046875	0.89820743988931\\
0.806640625	0.674094343445652\\
0.80859375	0.498060147099861\\
0.810546875	0.592119024583485\\
0.8125	0.440009680531529\\
0.814453125	0.290858185194912\\
0.81640625	0.203096104692202\\
0.818359375	0.175194507010211\\
0.8203125	0.197628990696006\\
0.822265625	0.287433617133183\\
0.82421875	0.565788934051145\\
0.826171875	1.01338582864613\\
0.828125	1.71532273820863\\
0.830078125	1.53444210574647\\
0.83203125	1.03248374560288\\
0.833984375	0.967829386930027\\
0.8359375	1.17262809506113\\
0.837890625	1.59632508004811\\
0.83984375	1.28351597664927\\
0.841796875	0.719878543829522\\
0.84375	0.342577820128523\\
0.845703125	0.312233928455571\\
0.84765625	0.393545705579945\\
0.849609375	0.535320863666719\\
0.8515625	0.642121034137052\\
0.853515625	0.598320834782826\\
0.85546875	0.950234096890977\\
0.857421875	1.35037675150421\\
0.859375	1.14674813883594\\
0.861328125	1.48218051920489\\
0.86328125	2.16423343450002\\
0.865234375	2.60467969522014\\
0.8671875	2.35568631534408\\
0.869140625	1.75326919593909\\
0.87109375	0.870166272436917\\
0.873046875	0.48466974918072\\
0.875	0.690013798563094\\
0.876953125	0.86370773421562\\
0.87890625	0.591763439747438\\
0.880859375	0.364130806490415\\
0.8828125	0.250631519885545\\
0.884765625	0.262833031802289\\
0.88671875	0.398421511848073\\
0.888671875	0.347855017827327\\
0.890625	0.42657674601321\\
0.892578125	0.676289801332437\\
0.89453125	1.02479491464071\\
0.896484375	1.42095707610859\\
0.8984375	0.982827250414813\\
0.900390625	0.599425329919077\\
0.90234375	0.47164364287938\\
0.904296875	0.418837702620576\\
0.90625	0.351984873295051\\
0.908203125	0.222064256931872\\
0.91015625	0.104967646301579\\
0.912109375	0.0585506920710883\\
0.9140625	0.0302520827006027\\
0.916015625	0.0159857298457487\\
0.91796875	0.0101159403860286\\
0.919921875	0.00523755093321755\\
0.921875	0.00353408520043192\\
0.923828125	0.0050233950086403\\
0.92578125	0.00836829203608259\\
0.927734375	0.0110912166309883\\
0.9296875	0.0187066172744382\\
0.931640625	0.0259991238340444\\
0.93359375	0.0316691388229258\\
0.935546875	0.0252299709529965\\
0.9375	0.0325014220893889\\
0.939453125	0.0573571226947855\\
0.94140625	0.128170513467925\\
0.943359375	0.207874365720718\\
0.9453125	0.30788293719624\\
0.947265625	0.424891959490714\\
0.94921875	0.554387634380365\\
0.951171875	0.763572295898179\\
0.953125	1.2489728145158\\
0.955078125	1.55188633683958\\
0.95703125	1.42480755773272\\
0.958984375	1.3263512998079\\
0.9609375	1.64079090390464\\
0.962890625	2.31223436273599\\
0.96484375	2.11996364067993\\
0.966796875	1.33904290677813\\
0.96875	1.21630678775622\\
0.970703125	1.64680984136414\\
0.97265625	1.88007866642025\\
0.974609375	1.67168922472305\\
0.9765625	1.99662496227366\\
0.978515625	2.23300868481707\\
0.98046875	1.83455995628869\\
0.982421875	0.957114322554201\\
0.984375	0.412600744855472\\
0.986328125	0.247435285486071\\
0.98828125	0.138809454599692\\
0.990234375	0.139188081361793\\
0.9921875	0.142535369963605\\
0.994140625	0.0824236771384519\\
0.99609375	0.043644574937554\\
0.998046875	0.0208334100516944\\
1	0\\
};
\end{axis}
\end{tikzpicture}%

%% file: pics/square/GroundState_1_1_3.tex
%
%
\begin{tikzpicture}

\begin{axis}[%
width=0.9in,
height=0.9in,
scale only axis,
xmin=0,
xmax=1,
xtick={0,0.5,1},
xlabel={$c_\kappa = 2^{-2}$},
xlabel style={below=-1.4in},
ymin=0,
ymax=10,
yminorticks=true,
yticklabels = {},
axis background/.style={fill=white},
]
\addplot [color=mycolor2, thick]
table[row sep=crcr]{%
0	0\\
0.001953125	3.03727993802809e-06\\
0.00390625	9.76991308161248e-06\\
0.005859375	1.43750390992113e-05\\
0.0078125	1.18238704193251e-05\\
0.009765625	1.58969119201661e-05\\
0.01171875	2.02468399345222e-05\\
0.013671875	4.1713093811714e-05\\
0.015625	9.26767104523237e-05\\
0.017578125	0.000213126650137567\\
0.01953125	0.000381586218363843\\
0.021484375	0.000461244788351394\\
0.0234375	0.000542582097103354\\
0.025390625	0.00111616557413588\\
0.02734375	0.00234247383569586\\
0.029296875	0.00488567561952168\\
0.03125	0.00506840434005195\\
0.033203125	0.00363506724673231\\
0.03515625	0.00211490114482344\\
0.037109375	0.00250112546638488\\
0.0390625	0.00274268703668348\\
0.041015625	0.00497487074209694\\
0.04296875	0.00943999312065498\\
0.044921875	0.0148732614422446\\
0.046875	0.0212079302925775\\
0.048828125	0.0481874838219839\\
0.05078125	0.0884045909407231\\
0.052734375	0.145205500193185\\
0.0546875	0.159103529516844\\
0.056640625	0.266016130882006\\
0.05859375	0.608611235554433\\
0.060546875	0.98339242738388\\
0.0625	1.66193162925293\\
0.064453125	1.47356822666777\\
0.06640625	0.829979808438356\\
0.068359375	0.509004897741652\\
0.0703125	0.267734772189118\\
0.072265625	0.277576381434617\\
0.07421875	0.36407353814834\\
0.076171875	0.403640248987979\\
0.078125	0.314264798039204\\
0.080078125	0.211046793022659\\
0.08203125	0.0968409918130552\\
0.083984375	0.143980578707503\\
0.0859375	0.167992242060729\\
0.087890625	0.244090775736942\\
0.08984375	0.323544777249699\\
0.091796875	0.512385247877847\\
0.09375	0.651900777869025\\
0.095703125	0.803403970364833\\
0.09765625	0.842681012438032\\
0.099609375	1.23281529345703\\
0.1015625	1.27770800962687\\
0.103515625	1.86223042258419\\
0.10546875	2.75985104779988\\
0.107421875	3.25704987455708\\
0.109375	2.92176017979902\\
0.111328125	2.93231079579797\\
0.11328125	3.07491324700221\\
0.115234375	3.02187169686386\\
0.1171875	2.99266717192216\\
0.119140625	1.86364372880111\\
0.12109375	0.969007069390401\\
0.123046875	0.52998527011434\\
0.125	0.327551646282299\\
0.126953125	0.302105105291154\\
0.12890625	0.273628621524774\\
0.130859375	0.160046432454106\\
0.1328125	0.168834569337754\\
0.134765625	0.164911846554338\\
0.13671875	0.298733359856503\\
0.138671875	0.477091879309506\\
0.140625	0.58910220416579\\
0.142578125	0.754289403294651\\
0.14453125	0.856221817839371\\
0.146484375	1.07843337665282\\
0.1484375	1.43144357758648\\
0.150390625	1.25791149932485\\
0.15234375	1.18143898870651\\
0.154296875	1.44940102153391\\
0.15625	1.7037304197947\\
0.158203125	2.95139628895207\\
0.16015625	3.45750312107339\\
0.162109375	2.5032135510792\\
0.1640625	1.57191057634457\\
0.166015625	1.58773791654714\\
0.16796875	2.36789532229579\\
0.169921875	3.04606778260731\\
0.171875	2.53554570383446\\
0.173828125	1.65729077811883\\
0.17578125	1.14541263330307\\
0.177734375	0.934464632914637\\
0.1796875	0.763111401792707\\
0.181640625	0.606074737863549\\
0.18359375	0.343682877757531\\
0.185546875	0.233153917353345\\
0.1875	0.159338756313589\\
0.189453125	0.0818422270867332\\
0.19140625	0.0388372563340673\\
0.193359375	0.020648651152803\\
0.1953125	0.0182134627615354\\
0.197265625	0.0234647062411995\\
0.19921875	0.0329406235993772\\
0.201171875	0.0438532658815689\\
0.203125	0.0726347052694498\\
0.205078125	0.129376036224898\\
0.20703125	0.229919807610636\\
0.208984375	0.362100916235307\\
0.2109375	0.649142478887664\\
0.212890625	0.98225286848961\\
0.21484375	1.06134119698802\\
0.216796875	1.0962279521917\\
0.21875	1.28327708260164\\
0.220703125	2.27358691457631\\
0.22265625	2.62089505304361\\
0.224609375	2.29830194296923\\
0.2265625	2.14528589278939\\
0.228515625	1.71815089501059\\
0.23046875	1.01512263205865\\
0.232421875	0.658494552757031\\
0.234375	0.293151463520889\\
0.236328125	0.197839963794198\\
0.23828125	0.278984339336098\\
0.240234375	0.495152482349953\\
0.2421875	0.72347988633453\\
0.244140625	1.22200314638906\\
0.24609375	1.9403314361536\\
0.248046875	1.98163926546936\\
0.25	1.62788087440792\\
0.251953125	1.3807094922187\\
0.25390625	1.30907927368805\\
0.255859375	1.72424973507632\\
0.2578125	1.62447969923693\\
0.259765625	1.76662733843659\\
0.26171875	2.2643181406395\\
0.263671875	2.12254266968895\\
0.265625	1.33998279514038\\
0.267578125	0.973956275040049\\
0.26953125	0.696890557046966\\
0.271484375	0.697178797623447\\
0.2734375	0.584187156081007\\
0.275390625	0.356997892033549\\
0.27734375	0.207326943191294\\
0.279296875	0.1658908165464\\
0.28125	0.110944072055225\\
0.283203125	0.0907638462769084\\
0.28515625	0.0929071107319222\\
0.287109375	0.160223563583719\\
0.2890625	0.253709513666875\\
0.291015625	0.388184704107069\\
0.29296875	0.522711072691089\\
0.294921875	0.92095493359066\\
0.296875	1.1805790050815\\
0.298828125	1.00518855493857\\
0.30078125	0.842429686922927\\
0.302734375	0.828299858147324\\
0.3046875	0.542484326636253\\
0.306640625	0.385125714536243\\
0.30859375	0.241966265098556\\
0.310546875	0.148183459453602\\
0.3125	0.0848936513917017\\
0.314453125	0.0646977794658514\\
0.31640625	0.0380767656274232\\
0.318359375	0.0180527385627588\\
0.3203125	0.0103893973972994\\
0.322265625	0.00541194793820449\\
0.32421875	0.00327526608450222\\
0.326171875	0.00165346369616374\\
0.328125	0.000993388868279378\\
0.330078125	0.000585516433167668\\
0.33203125	0.000677269622250339\\
0.333984375	0.00122345364938669\\
0.3359375	0.00171379950984829\\
0.337890625	0.00309847089632824\\
0.33984375	0.00488592027978198\\
0.341796875	0.00654665428469724\\
0.34375	0.0104540166250655\\
0.345703125	0.0158444795550542\\
0.34765625	0.0363743825781332\\
0.349609375	0.0908060155325028\\
0.3515625	0.121012035542649\\
0.353515625	0.110848236172908\\
0.35546875	0.0826513026509134\\
0.357421875	0.0966726222481903\\
0.359375	0.163638476273068\\
0.361328125	0.154889820281382\\
0.36328125	0.191954724993745\\
0.365234375	0.241832589569529\\
0.3671875	0.308427892953624\\
0.369140625	0.388935367085492\\
0.37109375	0.665829431572488\\
0.373046875	1.12178606173848\\
0.375	1.92404070551311\\
0.376953125	1.31520738814661\\
0.37890625	1.27157414014661\\
0.380859375	1.170863413781\\
0.3828125	1.17854189373698\\
0.384765625	1.36314727725987\\
0.38671875	2.19974751802047\\
0.388671875	2.52583809319523\\
0.390625	1.94367218698814\\
0.392578125	1.78893311666145\\
0.39453125	1.97181678415296\\
0.396484375	2.39769593000276\\
0.3984375	2.83656182826402\\
0.400390625	2.12908752748293\\
0.40234375	1.31650930633017\\
0.404296875	0.550869656510042\\
0.40625	0.235486605070425\\
0.408203125	0.137543999201476\\
0.41015625	0.0944054871679832\\
0.412109375	0.102974260719392\\
0.4140625	0.135141738363086\\
0.416015625	0.164436089838\\
0.41796875	0.215146936935521\\
0.419921875	0.340507495030057\\
0.421875	0.44673745217034\\
0.423828125	0.591431185128011\\
0.42578125	1.09993055541733\\
0.427734375	1.41270564915492\\
0.4296875	1.74690714647232\\
0.431640625	1.23408901572513\\
0.43359375	1.09271567256023\\
0.435546875	0.953927686766417\\
0.4375	0.493837789171051\\
0.439453125	0.313502207044174\\
0.44140625	0.223226055844421\\
0.443359375	0.161548893133839\\
0.4453125	0.134795067177419\\
0.447265625	0.170677836574486\\
0.44921875	0.199923621717176\\
0.451171875	0.211542915078097\\
0.453125	0.168286356150275\\
0.455078125	0.10643975735633\\
0.45703125	0.0666186698727387\\
0.458984375	0.063950904465912\\
0.4609375	0.0366566489060722\\
0.462890625	0.0141504526458721\\
0.46484375	0.00705509218864354\\
0.466796875	0.00477176946011546\\
0.46875	0.00257185029611761\\
0.470703125	0.00112272606178365\\
0.47265625	0.000634706570231403\\
0.474609375	0.000406241913275377\\
0.4765625	0.00037939056037135\\
0.478515625	0.000419291191104158\\
0.48046875	0.000478262726820358\\
0.482421875	0.000373549178438684\\
0.484375	0.000686511570914274\\
0.486328125	0.00159199817728568\\
0.48828125	0.00409382396007246\\
0.490234375	0.00759762425462494\\
0.4921875	0.0147063957765729\\
0.494140625	0.0194652676030553\\
0.49609375	0.030673428878375\\
0.498046875	0.0522263140708448\\
0.5	0.0870672676617602\\
0.501953125	0.109245344227569\\
0.50390625	0.122014811671524\\
0.505859375	0.0996826414247477\\
0.5078125	0.0728135864566724\\
0.509765625	0.0441161555090053\\
0.51171875	0.0348934278104959\\
0.513671875	0.0191704170851822\\
0.515625	0.0204185577227251\\
0.517578125	0.0329891554924428\\
0.51953125	0.0399330411253991\\
0.521484375	0.0579241304259606\\
0.5234375	0.145112490067574\\
0.525390625	0.256404038282687\\
0.52734375	0.293588736272789\\
0.529296875	0.322546345059848\\
0.53125	0.390289709719161\\
0.533203125	0.502749712118398\\
0.53515625	0.703785242758368\\
0.537109375	0.684376316642241\\
0.5390625	1.05720428612555\\
0.541015625	1.24206114405804\\
0.54296875	1.00355582025251\\
0.544921875	0.821916895732505\\
0.546875	0.78659682595294\\
0.548828125	1.36574732167976\\
0.55078125	1.52436068372128\\
0.552734375	1.35180967362603\\
0.5546875	1.06345392347824\\
0.556640625	0.964439503024034\\
0.55859375	0.94085487616729\\
0.560546875	0.656848738063713\\
0.5625	0.29753551228333\\
0.564453125	0.179195969727773\\
0.56640625	0.121004231153832\\
0.568359375	0.0756387418189449\\
0.5703125	0.0485996589036563\\
0.572265625	0.0329604514002225\\
0.57421875	0.0219004157746364\\
0.576171875	0.0110347900244366\\
0.578125	0.00555620754341166\\
0.580078125	0.00226129084347589\\
0.58203125	0.00106195390159956\\
0.583984375	0.000766911404033009\\
0.5859375	0.000966620830249115\\
0.587890625	0.00122699933266579\\
0.58984375	0.00121700974746577\\
0.591796875	0.00171128685394849\\
0.59375	0.00171917751558646\\
0.595703125	0.00160781945968758\\
0.59765625	0.00116987577670025\\
0.599609375	0.000492962748497984\\
0.6015625	0.000336582216163129\\
0.603515625	0.000436907292289174\\
0.60546875	0.000530465252507237\\
0.607421875	0.000852971764654136\\
0.609375	0.00177334519901151\\
0.611328125	0.0033554781176946\\
0.61328125	0.00704777695565321\\
0.615234375	0.012571308762113\\
0.6171875	0.0234977370210704\\
0.619140625	0.0343161613498198\\
0.62109375	0.0413814134894389\\
0.623046875	0.0473946259262248\\
0.625	0.0959659093804722\\
0.626953125	0.181196855578872\\
0.62890625	0.282035110826819\\
0.630859375	0.361649172344278\\
0.6328125	0.312219615605345\\
0.634765625	0.323800299402683\\
0.63671875	0.690996507968151\\
0.638671875	0.925027273465406\\
0.640625	1.34031197741743\\
0.642578125	1.37888251656378\\
0.64453125	1.5480288301306\\
0.646484375	2.56776068520872\\
0.6484375	2.98114846519674\\
0.650390625	2.63132995243131\\
0.65234375	1.992399796625\\
0.654296875	1.15220925037196\\
0.65625	0.477915860832415\\
0.658203125	0.292806146626917\\
0.66015625	0.151730879507178\\
0.662109375	0.139583768127555\\
0.6640625	0.155178469310758\\
0.666015625	0.240333443934536\\
0.66796875	0.24215381450097\\
0.669921875	0.20331109847208\\
0.671875	0.170520388181207\\
0.673828125	0.230096348406031\\
0.67578125	0.402079189731056\\
0.677734375	0.897105699913416\\
0.6796875	1.78211795853352\\
0.681640625	3.1613981484535\\
0.68359375	3.04905316194244\\
0.685546875	2.53234977400832\\
0.6875	2.33098151188\\
0.689453125	2.29479501790067\\
0.69140625	2.72205084005786\\
0.693359375	2.20122389888373\\
0.6953125	1.13076732954496\\
0.697265625	0.670871005934434\\
0.69921875	0.497355567855384\\
0.701171875	0.477402487525222\\
0.703125	0.310529663252534\\
0.705078125	0.175540804530568\\
0.70703125	0.0886523922642604\\
0.708984375	0.0492637983569344\\
0.7109375	0.0226422403998718\\
0.712890625	0.0115840764600951\\
0.71484375	0.00613101141651742\\
0.716796875	0.00484588868356316\\
0.71875	0.00341739982411743\\
0.720703125	0.00197961993131985\\
0.72265625	0.000931985019922463\\
0.724609375	0.000363639096176076\\
0.7265625	0.000189254291652651\\
0.728515625	0.000122918088182644\\
0.73046875	7.07640220553118e-05\\
0.732421875	7.78729215574998e-05\\
0.734375	0.000129162862419353\\
0.736328125	0.000230104148683769\\
0.73828125	0.000437799099634055\\
0.740234375	0.000752299415758824\\
0.7421875	0.00150823790112778\\
0.744140625	0.00292692295102515\\
0.74609375	0.00570581113617129\\
0.748046875	0.0108857019546212\\
0.75	0.0259985459782685\\
0.751953125	0.059483138266905\\
0.75390625	0.0942136198047994\\
0.755859375	0.118524716959173\\
0.7578125	0.172585723265313\\
0.759765625	0.265260233056227\\
0.76171875	0.443985318272439\\
0.763671875	0.43745748289431\\
0.765625	0.671896999813269\\
0.767578125	0.939610663573566\\
0.76953125	0.640181092946336\\
0.771484375	0.466156301505915\\
0.7734375	0.476591020595963\\
0.775390625	0.379825078237097\\
0.77734375	0.275626764034186\\
0.779296875	0.222949306926808\\
0.78125	0.14929187667672\\
0.783203125	0.0677848132217483\\
0.78515625	0.033690465464747\\
0.787109375	0.0159212760245415\\
0.7890625	0.00793246489960961\\
0.791015625	0.00359095835118535\\
0.79296875	0.00192465150701225\\
0.794921875	0.0022160838000133\\
0.796875	0.0046621607129921\\
0.798828125	0.00720399182235466\\
0.80078125	0.0146793718978509\\
0.802734375	0.0240416776355161\\
0.8046875	0.0254766293904755\\
0.806640625	0.021159665986319\\
0.80859375	0.019528171376338\\
0.810546875	0.0291207521763486\\
0.8125	0.0281471076873348\\
0.814453125	0.0282896426675836\\
0.81640625	0.0380680157014023\\
0.818359375	0.0623428411947594\\
0.8203125	0.108724818306623\\
0.822265625	0.199938836411722\\
0.82421875	0.435707720212436\\
0.826171875	0.825695864626225\\
0.828125	1.44818096819931\\
0.830078125	1.28537064989657\\
0.83203125	0.829306200885702\\
0.833984375	0.738319531064019\\
0.8359375	0.883749640544011\\
0.837890625	1.20392008923356\\
0.83984375	0.948899168210954\\
0.841796875	0.515527093311099\\
0.84375	0.229955698411124\\
0.845703125	0.194386499683164\\
0.84765625	0.243995982815163\\
0.849609375	0.34745397300167\\
0.8515625	0.446386365652867\\
0.853515625	0.470986551756879\\
0.85546875	0.834644490425132\\
0.857421875	1.25436106815914\\
0.859375	1.14858391453417\\
0.861328125	1.62258108630291\\
0.86328125	2.53485389926756\\
0.865234375	3.12064827930904\\
0.8671875	2.79865274833617\\
0.869140625	2.01731470642148\\
0.87109375	0.940067541290068\\
0.873046875	0.426274008126757\\
0.875	0.471861844697585\\
0.876953125	0.533310701471365\\
0.87890625	0.336808823281771\\
0.880859375	0.184353979269354\\
0.8828125	0.097666331193364\\
0.884765625	0.0672247096311516\\
0.88671875	0.0760795414123732\\
0.888671875	0.054742455910213\\
0.890625	0.0548680031747243\\
0.892578125	0.0805617844980913\\
0.89453125	0.121827332959613\\
0.896484375	0.169896740812717\\
0.8984375	0.114294102835486\\
0.900390625	0.0638241916343281\\
0.90234375	0.0451794385372533\\
0.904296875	0.036208460685798\\
0.90625	0.0281784500523638\\
0.908203125	0.0168615510357652\\
0.91015625	0.00758827210229395\\
0.912109375	0.00402220849701082\\
0.9140625	0.00197827898727781\\
0.916015625	0.00101531836052092\\
0.91796875	0.000660699091300089\\
0.919921875	0.000416982877848751\\
0.921875	0.000477083422548031\\
0.923828125	0.000977239707554158\\
0.92578125	0.00192430800992372\\
0.927734375	0.00285765837264811\\
0.9296875	0.00525885348328258\\
0.931640625	0.00790888319531777\\
0.93359375	0.0106160292863472\\
0.935546875	0.010065389339492\\
0.9375	0.0156571884980429\\
0.939453125	0.0311388062633152\\
0.94140625	0.0739758455714705\\
0.943359375	0.126782933143047\\
0.9453125	0.201006878279971\\
0.947265625	0.299769854628314\\
0.94921875	0.428694301282407\\
0.951171875	0.643918104281282\\
0.953125	1.12297923852501\\
0.955078125	1.44839202718709\\
0.95703125	1.37911492561038\\
0.958984375	1.35955958237263\\
0.9609375	1.79646241110805\\
0.962890625	2.61953073256065\\
0.96484375	2.39768686323946\\
0.966796875	1.47013500424639\\
0.96875	1.30505488235148\\
0.970703125	1.77835111166236\\
0.97265625	2.05377471911017\\
0.974609375	1.84357936434895\\
0.9765625	2.24004134544594\\
0.978515625	2.51980290236174\\
0.98046875	2.03946324470706\\
0.982421875	1.03535463827095\\
0.984375	0.423572304461269\\
0.986328125	0.23564449202555\\
0.98828125	0.116511300762165\\
0.990234375	0.0992486568345451\\
0.9921875	0.0932041192886403\\
0.994140625	0.0511098765711143\\
0.99609375	0.0258006417981257\\
0.998046875	0.0118922106397245\\
1	0\\
};
\end{axis}
\end{tikzpicture}%

%% file: pics/square/GroundState_1_1_4.tex
%
%
\begin{tikzpicture}

\begin{axis}[%
width=0.9in,
height=0.9in,
scale only axis,
xmin=0,
xmax=1,
xtick={0,0.5,1},
xlabel={$c_\kappa = 2^{-3}$},
xlabel style={below=-1.4in},
ymin=0,
ymax=10,
yminorticks=true,
ytick={0,4,8},
yticklabel pos=right,
axis background/.style={fill=white},
]
\addplot [color=mycolor2, thick]
table[row sep=crcr]{%
0	0\\
0.001953125	2.85994314191888e-09\\
0.00390625	9.38748041225225e-09\\
0.005859375	1.43198800537364e-08\\
0.0078125	1.29298793469967e-08\\
0.009765625	1.97158313728774e-08\\
0.01171875	2.79484719018934e-08\\
0.013671875	6.19343099015793e-08\\
0.015625	1.43328777071413e-07\\
0.017578125	3.41370665746531e-07\\
0.01953125	6.33797154931881e-07\\
0.021484375	8.14657362140556e-07\\
0.0234375	1.0525590647901e-06\\
0.025390625	2.3205305651061e-06\\
0.02734375	5.06450774625744e-06\\
0.029296875	1.09937587404224e-05\\
0.03125	1.20655247149677e-05\\
0.033203125	1.00948432991495e-05\\
0.03515625	9.12987049124673e-06\\
0.037109375	1.62917270849549e-05\\
0.0390625	2.30298305629765e-05\\
0.041015625	4.89345990204214e-05\\
0.04296875	9.8849467027341e-05\\
0.044921875	0.000165160950191486\\
0.046875	0.000250365866117789\\
0.048828125	0.000595824935506545\\
0.05078125	0.00114456375436889\\
0.052734375	0.00196055060330691\\
0.0546875	0.00231932418513539\\
0.056640625	0.00419105326050136\\
0.05859375	0.0100090545881593\\
0.060546875	0.0167850555351983\\
0.0625	0.0291553585730271\\
0.064453125	0.0258910857624639\\
0.06640625	0.0144948984928004\\
0.068359375	0.0092577927897957\\
0.0703125	0.00602930603337871\\
0.072265625	0.00868762931950444\\
0.07421875	0.0142217508917722\\
0.076171875	0.0186197847472699\\
0.078125	0.0182416511063914\\
0.080078125	0.0169790144760097\\
0.08203125	0.0193074820426835\\
0.083984375	0.0516628183704658\\
0.0859375	0.0775614744527476\\
0.087890625	0.133091428548866\\
0.08984375	0.197865737865337\\
0.091796875	0.341873830060329\\
0.09375	0.468259954365717\\
0.095703125	0.6280335389513\\
0.09765625	0.731939714150555\\
0.099609375	1.17616031935901\\
0.1015625	1.32207733539717\\
0.103515625	2.10241690787299\\
0.10546875	3.26678320411868\\
0.107421875	3.94052677826152\\
0.109375	3.55174349620289\\
0.111328125	3.56926208263739\\
0.11328125	3.75034575861533\\
0.115234375	3.67342913097488\\
0.1171875	3.59234357308499\\
0.119140625	2.18193110839154\\
0.12109375	1.08560800878271\\
0.123046875	0.559054784599317\\
0.125	0.305885346970697\\
0.126953125	0.248935063085466\\
0.12890625	0.203963369236705\\
0.130859375	0.107638901651362\\
0.1328125	0.0946541542035999\\
0.134765625	0.0839913643176583\\
0.13671875	0.147294519661132\\
0.138671875	0.240997184246808\\
0.140625	0.31418511847109\\
0.142578125	0.429823084967955\\
0.14453125	0.532835887540865\\
0.146484375	0.729182144518023\\
0.1484375	1.03658448895044\\
0.150390625	0.985508124617491\\
0.15234375	1.0606499029205\\
0.154296875	1.45328299632253\\
0.15625	1.90351890934041\\
0.158203125	3.50346316082211\\
0.16015625	4.16645530684087\\
0.162109375	2.96841464689266\\
0.1640625	1.80495694403495\\
0.166015625	1.78398187675591\\
0.16796875	2.68353569059118\\
0.169921875	3.46770391621296\\
0.171875	2.82818365462369\\
0.173828125	1.77374484945264\\
0.17578125	1.1441164033887\\
0.177734375	0.856126622811822\\
0.1796875	0.645979226290178\\
0.181640625	0.484812494970216\\
0.18359375	0.260224041637296\\
0.185546875	0.165327386761484\\
0.1875	0.107260053177278\\
0.189453125	0.0528509104273293\\
0.19140625	0.0239322868520823\\
0.193359375	0.0117061708532001\\
0.1953125	0.00932669432734769\\
0.197265625	0.0115187421143977\\
0.19921875	0.016490139161052\\
0.201171875	0.0231021752477914\\
0.203125	0.0408157586413743\\
0.205078125	0.0768959926127372\\
0.20703125	0.143046369756787\\
0.208984375	0.238347540744161\\
0.2109375	0.44962647017388\\
0.212890625	0.712524897160972\\
0.21484375	0.819310771845655\\
0.216796875	0.930040945832675\\
0.21875	1.21256087107241\\
0.220703125	2.27698658858389\\
0.22265625	2.68083365207742\\
0.224609375	2.34310572751675\\
0.2265625	2.14987437134862\\
0.228515625	1.67606239494853\\
0.23046875	0.944102347571145\\
0.232421875	0.587030669137104\\
0.234375	0.245668521337414\\
0.236328125	0.146774268443104\\
0.23828125	0.18804205769141\\
0.240234375	0.332186908133041\\
0.2421875	0.502295474847625\\
0.244140625	0.885838591142363\\
0.24609375	1.44292831451339\\
0.248046875	1.48069230359945\\
0.25	1.20393691576014\\
0.251953125	1.0176604399033\\
0.25390625	0.982466165201996\\
0.255859375	1.33556760610381\\
0.2578125	1.2950312203095\\
0.259765625	1.46504897054029\\
0.26171875	1.92741829275302\\
0.263671875	1.79550339379693\\
0.265625	1.08052695819645\\
0.267578125	0.732709807672319\\
0.26953125	0.472039516990912\\
0.271484375	0.426017017599797\\
0.2734375	0.331365670317307\\
0.275390625	0.191226385952844\\
0.27734375	0.100849048232401\\
0.279296875	0.070490167235287\\
0.28125	0.0395017601841434\\
0.283203125	0.0208332975504364\\
0.28515625	0.0102137164385516\\
0.287109375	0.00620998751815647\\
0.2890625	0.00433137865463335\\
0.291015625	0.003621125957102\\
0.29296875	0.00297924550527789\\
0.294921875	0.00407343353034215\\
0.296875	0.00469938865485727\\
0.298828125	0.00370331265625834\\
0.30078125	0.00282336768578464\\
0.302734375	0.00258794431624945\\
0.3046875	0.00159570238204242\\
0.306640625	0.00106374298413308\\
0.30859375	0.00063224711933894\\
0.310546875	0.000367126463967907\\
0.3125	0.000198047101379172\\
0.314453125	0.000141814114704074\\
0.31640625	7.99051606411575e-05\\
0.318359375	3.64978359554934e-05\\
0.3203125	2.0522783299652e-05\\
0.322265625	1.11432429701922e-05\\
0.32421875	8.38668656724931e-06\\
0.326171875	7.5544840167447e-06\\
0.328125	1.1589486411957e-05\\
0.330078125	1.98545135683678e-05\\
0.33203125	4.57596054122384e-05\\
0.333984375	0.000103296613769678\\
0.3359375	0.000162505396338822\\
0.337890625	0.000320566283897605\\
0.33984375	0.000533544259280918\\
0.341796875	0.000770696052115645\\
0.34375	0.00131127713420319\\
0.345703125	0.00213554163519615\\
0.34765625	0.00515251686371883\\
0.349609375	0.0133076291067208\\
0.3515625	0.0185833178975072\\
0.353515625	0.0187108114994627\\
0.35546875	0.0170341553009968\\
0.357421875	0.0256990094510822\\
0.359375	0.0497687853846805\\
0.361328125	0.0530259847948084\\
0.36328125	0.0766147441133262\\
0.365234375	0.108028138072078\\
0.3671875	0.15268605941344\\
0.369140625	0.21642092226893\\
0.37109375	0.403346742686753\\
0.373046875	0.718460348737732\\
0.375	1.26890561446517\\
0.376953125	0.881686698309536\\
0.37890625	0.886357980663801\\
0.380859375	0.876337666740695\\
0.3828125	0.971003303767852\\
0.384765625	1.25468379503716\\
0.38671875	2.1678976611479\\
0.388671875	2.54631313702051\\
0.390625	1.96382221686868\\
0.392578125	1.82178822951253\\
0.39453125	2.06153959757525\\
0.396484375	2.57673286375371\\
0.3984375	3.08630864958389\\
0.400390625	2.28331945078783\\
0.40234375	1.37193847479438\\
0.404296875	0.555228005544542\\
0.40625	0.223861939830278\\
0.408203125	0.11285484049124\\
0.41015625	0.0574949659736184\\
0.412109375	0.0357246364076394\\
0.4140625	0.0265982701621327\\
0.416015625	0.0185728207396802\\
0.41796875	0.0121790444644229\\
0.419921875	0.0122824027257941\\
0.421875	0.0118398238284477\\
0.423828125	0.0130549041434399\\
0.42578125	0.0230329353753267\\
0.427734375	0.0294829289164819\\
0.4296875	0.0362191799887201\\
0.431640625	0.0245647872217848\\
0.43359375	0.0205181589765609\\
0.435546875	0.0171658762589365\\
0.4375	0.00844397120520809\\
0.439453125	0.00494222899102262\\
0.44140625	0.00313943860214555\\
0.443359375	0.00188494199568321\\
0.4453125	0.00122498416875714\\
0.447265625	0.00116503816791826\\
0.44921875	0.00115021929952778\\
0.451171875	0.00108615908260921\\
0.453125	0.000808668484813601\\
0.455078125	0.000479684553530766\\
0.45703125	0.000278015238359745\\
0.458984375	0.000249198674031845\\
0.4609375	0.000137506229309696\\
0.462890625	5.13141474966857e-05\\
0.46484375	2.43630500081553e-05\\
0.466796875	1.55938526721147e-05\\
0.46875	8.08637698619078e-06\\
0.470703125	3.35140690005072e-06\\
0.47265625	1.72487808898348e-06\\
0.474609375	9.03315415594462e-07\\
0.4765625	6.24087725431616e-07\\
0.478515625	4.39248443897255e-07\\
0.48046875	3.54113309586428e-07\\
0.482421875	1.26603504137406e-07\\
0.484375	5.12748001293473e-08\\
0.486328125	1.83940991727353e-08\\
0.48828125	1.02639686695157e-08\\
0.490234375	6.0307137165308e-09\\
0.4921875	5.25782400734057e-09\\
0.494140625	3.55657047713014e-09\\
0.49609375	3.44508083962602e-09\\
0.498046875	4.58433510374238e-09\\
0.5	7.13108861304372e-09\\
0.501953125	8.91821556345318e-09\\
0.50390625	1.02821488772072e-08\\
0.505859375	9.04374792416337e-09\\
0.5078125	7.80387638301339e-09\\
0.509765625	6.75151913847461e-09\\
0.51171875	8.9604349286179e-09\\
0.513671875	1.02577579317898e-08\\
0.515625	2.04467935106296e-08\\
0.517578125	4.21274647155685e-08\\
0.51953125	5.79624894681498e-08\\
0.521484375	9.34480586431825e-08\\
0.5234375	2.43831495799096e-07\\
0.525390625	4.42736495339398e-07\\
0.52734375	5.28961752813637e-07\\
0.529296875	6.19099349240172e-07\\
0.53125	8.00084796189333e-07\\
0.533203125	1.10178971137932e-06\\
0.53515625	1.62966041327843e-06\\
0.537109375	1.69065782784565e-06\\
0.5390625	2.7345978793724e-06\\
0.541015625	3.31973847725048e-06\\
0.54296875	2.77262395616368e-06\\
0.544921875	2.41537941852416e-06\\
0.546875	2.52569650912421e-06\\
0.548828125	4.67447061317863e-06\\
0.55078125	5.29708117797453e-06\\
0.552734375	4.65799211725059e-06\\
0.5546875	3.57193794347954e-06\\
0.556640625	3.14488280349209e-06\\
0.55859375	2.99198123668963e-06\\
0.560546875	2.04393576365413e-06\\
0.5625	9.01540066268299e-07\\
0.564453125	5.26524008734839e-07\\
0.56640625	3.46678124879669e-07\\
0.568359375	2.16276894313215e-07\\
0.5703125	1.4671038616802e-07\\
0.572265625	1.17247132256924e-07\\
0.57421875	1.06628535493982e-07\\
0.576171875	1.04245637454016e-07\\
0.578125	1.63654752872637e-07\\
0.580078125	3.07584578255805e-07\\
0.58203125	7.6245830798853e-07\\
0.583984375	1.81814310807073e-06\\
0.5859375	4.18346970142795e-06\\
0.587890625	6.97761636341651e-06\\
0.58984375	8.59472810463251e-06\\
0.591796875	1.42252560973427e-05\\
0.59375	1.69691396710264e-05\\
0.595703125	1.90970524987255e-05\\
0.59765625	1.85200542577083e-05\\
0.599609375	1.58010647748345e-05\\
0.6015625	3.2118826877493e-05\\
0.603515625	7.51844281081512e-05\\
0.60546875	0.000117472435210277\\
0.607421875	0.000220254869500004\\
0.609375	0.000490818207029526\\
0.611328125	0.000981086769783552\\
0.61328125	0.00214383738803926\\
0.615234375	0.00399361588073356\\
0.6171875	0.00783436388668841\\
0.619140625	0.0120244938579866\\
0.62109375	0.0157462507347627\\
0.623046875	0.0198349941343153\\
0.625	0.0433931440625409\\
0.626953125	0.0862663574984579\\
0.62890625	0.140998652401601\\
0.630859375	0.193174340688836\\
0.6328125	0.186363907866282\\
0.634765625	0.233485470236733\\
0.63671875	0.550628232292112\\
0.638671875	0.784400196487535\\
0.640625	1.20892342080539\\
0.642578125	1.31956184270174\\
0.64453125	1.62244913670065\\
0.646484375	2.8394003113135\\
0.6484375	3.36272700250486\\
0.650390625	2.94896621453239\\
0.65234375	2.18069393131468\\
0.654296875	1.22405293392204\\
0.65625	0.484875852480437\\
0.658203125	0.278937109341038\\
0.66015625	0.129425127531704\\
0.662109375	0.0996994316646963\\
0.6640625	0.0958168042767464\\
0.666015625	0.141096557779591\\
0.66796875	0.14533274628517\\
0.669921875	0.135460691930851\\
0.671875	0.137658183627399\\
0.673828125	0.220499059764371\\
0.67578125	0.428074579500881\\
0.677734375	1.00747777617152\\
0.6796875	2.07968908792603\\
0.681640625	3.765395079311\\
0.68359375	3.63277621648115\\
0.685546875	2.9719263262642\\
0.6875	2.67781712334071\\
0.689453125	2.59990830722612\\
0.69140625	3.06006993012399\\
0.693359375	2.43069700883578\\
0.6953125	1.19825487710298\\
0.697265625	0.662249154959374\\
0.69921875	0.44706390454278\\
0.701171875	0.397796015689287\\
0.703125	0.245481481962667\\
0.705078125	0.132876757316375\\
0.70703125	0.0640656877894906\\
0.708984375	0.0341921656009628\\
0.7109375	0.0150774898528823\\
0.712890625	0.0073582984259403\\
0.71484375	0.00364266476079696\\
0.716796875	0.00266041860583122\\
0.71875	0.00177436598909988\\
0.720703125	0.000984470868064675\\
0.72265625	0.000447958847015156\\
0.724609375	0.000166929812827997\\
0.7265625	8.06601732292066e-05\\
0.728515625	4.64111592429533e-05\\
0.73046875	1.86650992770785e-05\\
0.732421875	8.42792942340126e-06\\
0.734375	4.30365873799656e-06\\
0.736328125	2.41109898146458e-06\\
0.73828125	1.14976954032789e-06\\
0.740234375	6.39259298684464e-07\\
0.7421875	2.75716297094826e-07\\
0.744140625	1.37585491152038e-07\\
0.74609375	7.47229150778215e-08\\
0.748046875	2.85329212941271e-08\\
0.75	1.2426703809067e-08\\
0.751953125	1.04291781097556e-08\\
0.75390625	8.33026385138492e-09\\
0.755859375	6.43308983602158e-09\\
0.7578125	5.80220850160541e-09\\
0.759765625	7.02882749037176e-09\\
0.76171875	1.10655524334282e-08\\
0.763671875	1.10579155999617e-08\\
0.765625	1.77526019061125e-08\\
0.767578125	2.61976325841589e-08\\
0.76953125	1.97186134369335e-08\\
0.771484375	1.84698673001563e-08\\
0.7734375	2.47567758384842e-08\\
0.775390625	2.59696099613187e-08\\
0.77734375	2.99153391338627e-08\\
0.779296875	4.07873335440948e-08\\
0.78125	5.01494514697643e-08\\
0.783203125	6.10581776148831e-08\\
0.78515625	1.24282336506868e-07\\
0.787109375	2.41029534195e-07\\
0.7890625	5.66433419483152e-07\\
0.791015625	1.07975246136018e-06\\
0.79296875	2.906998069695e-06\\
0.794921875	7.25322170969566e-06\\
0.796875	1.93262526197218e-05\\
0.798828125	3.2578542546647e-05\\
0.80078125	7.07980455747469e-05\\
0.802734375	0.0001223628007554\\
0.8046875	0.000138796274886783\\
0.806640625	0.00013507002491492\\
0.80859375	0.000158293455653843\\
0.810546875	0.000278898187149646\\
0.8125	0.000314384557505588\\
0.814453125	0.000376757622361638\\
0.81640625	0.000598781226505433\\
0.818359375	0.00108802078430955\\
0.8203125	0.00203063607230504\\
0.822265625	0.00393734429486577\\
0.82421875	0.00893519570851471\\
0.826171875	0.0175762079511997\\
0.828125	0.0318755481535407\\
0.830078125	0.0289019879355446\\
0.83203125	0.0197950538583108\\
0.833984375	0.0206730044350581\\
0.8359375	0.028611472140782\\
0.837890625	0.042872177784746\\
0.83984375	0.0378065913698909\\
0.841796875	0.0261578654446393\\
0.84375	0.0272692352365882\\
0.845703125	0.0602661496756944\\
0.84765625	0.112200634854735\\
0.849609375	0.199523816193327\\
0.8515625	0.293024674180434\\
0.853515625	0.360141178883289\\
0.85546875	0.707387123311783\\
0.857421875	1.12208123261518\\
0.859375	1.11708858175613\\
0.861328125	1.72299453204052\\
0.86328125	2.86271903658898\\
0.865234375	3.59868825034116\\
0.8671875	3.20809827760686\\
0.869140625	2.25780215861312\\
0.87109375	1.0132592961354\\
0.873046875	0.411437394108645\\
0.875	0.382036684367677\\
0.876953125	0.396339388542951\\
0.87890625	0.237178383123578\\
0.880859375	0.12301912112422\\
0.8828125	0.0592183062336459\\
0.884765625	0.0333586237429592\\
0.88671875	0.0300491527952196\\
0.888671875	0.0168552260988765\\
0.890625	0.0100270251803132\\
0.892578125	0.00855579068910687\\
0.89453125	0.00952395273089616\\
0.896484375	0.011704349396252\\
0.8984375	0.00746886876392869\\
0.900390625	0.00387718460521693\\
0.90234375	0.00252811384775767\\
0.904296875	0.00186757155395341\\
0.90625	0.00136452179635135\\
0.908203125	0.000782230219504752\\
0.91015625	0.000339272467677831\\
0.912109375	0.000174345382745762\\
0.9140625	8.6089404075155e-05\\
0.916015625	4.98611607632845e-05\\
0.91796875	4.74685263924667e-05\\
0.919921875	5.49406802230891e-05\\
0.921875	0.00011080682529336\\
0.923828125	0.000272130942321149\\
0.92578125	0.000576680674730133\\
0.927734375	0.00091089006806995\\
0.9296875	0.00176946153192835\\
0.931640625	0.00281020741905837\\
0.93359375	0.00403755843710574\\
0.935546875	0.00428060311370325\\
0.9375	0.00738785799601894\\
0.939453125	0.0156983060421317\\
0.94140625	0.0388339675662839\\
0.943359375	0.0694202313718392\\
0.9453125	0.116028332807447\\
0.947265625	0.183685604488103\\
0.94921875	0.281640553300596\\
0.951171875	0.451821104754551\\
0.953125	0.829990313634318\\
0.955078125	1.11549390584561\\
0.95703125	1.12499339585764\\
0.958984375	1.20959752211437\\
0.9609375	1.73460721013703\\
0.962890625	2.62688609249958\\
0.96484375	2.40514703139635\\
0.966796875	1.4332222721965\\
0.96875	1.23910380090791\\
0.970703125	1.69187268521754\\
0.97265625	1.97574948614417\\
0.974609375	1.79730926612004\\
0.9765625	2.22883128052454\\
0.978515625	2.52620329516418\\
0.98046875	2.01953680079303\\
0.982421875	1.00269845189729\\
0.984375	0.394632249850401\\
0.986328125	0.20870912077436\\
0.98828125	0.0955290327760273\\
0.990234375	0.0728896300508044\\
0.9921875	0.0640443310249998\\
0.994140625	0.0336668517896192\\
0.99609375	0.01635301590248\\
0.998046875	0.00732319664622936\\
1	0\\
};
\end{axis}
\end{tikzpicture}%

%% file: pics/square/GroundState_1_2_1.tex
%
%
\begin{tikzpicture}

\begin{axis}[%
width=0.9in,
height=0.9in,
scale only axis,
xmin=0,
xmax=1,
xtick={0,0.5,1},
ymin=0,
ymax=10,
yminorticks=true,
ytick={0,4,8},
ylabel={$c_V = 2$},
ylabel style={below=0.1in},
axis background/.style={fill=white},
]
\addplot [color=mycolor2, thick]
table[row sep=crcr]{%
0	0\\
0.001953125	0.233861814186483\\
0.00390625	0.809851870737817\\
0.005859375	2.62004664500696\\
0.0078125	2.68344058755622\\
0.009765625	2.32407035321637\\
0.01171875	1.18112039414275\\
0.013671875	3.16837210188682\\
0.015625	2.59230387125985\\
0.017578125	0.784300836949149\\
0.01953125	0.38311791827967\\
0.021484375	0.326113965405671\\
0.0234375	0.197702196031862\\
0.025390625	0.20190197251888\\
0.02734375	0.118459825368229\\
0.029296875	0.0516884531203537\\
0.03125	0.051269858574329\\
0.033203125	0.10110148756389\\
0.03515625	0.199854205979241\\
0.037109375	0.248319192185822\\
0.0390625	0.369717236079505\\
0.041015625	0.591490389127131\\
0.04296875	2.48964037479176\\
0.044921875	2.79020437998084\\
0.046875	1.16487118152502\\
0.048828125	1.11066100649073\\
0.05078125	1.66918603623424\\
0.052734375	2.9745252198875\\
0.0546875	3.70658398348751\\
0.056640625	1.96081918894258\\
0.05859375	1.79194808579607\\
0.060546875	2.42643318657614\\
0.0625	2.11520876464989\\
0.064453125	1.71369504761823\\
0.06640625	1.07348455155848\\
0.068359375	0.716277006651539\\
0.0703125	0.404523176454442\\
0.072265625	0.390358196470021\\
0.07421875	0.535874682213902\\
0.076171875	1.28310559546488\\
0.078125	0.866877960294018\\
0.080078125	0.412171376500681\\
0.08203125	0.140879565564505\\
0.083984375	0.0540975017653905\\
0.0859375	0.0396444869427361\\
0.087890625	0.0706992160556579\\
0.08984375	0.152723712195285\\
0.091796875	0.258773179485657\\
0.09375	0.392034510383761\\
0.095703125	0.392394538529127\\
0.09765625	0.587312700303074\\
0.099609375	0.51803075064506\\
0.1015625	0.860307275871137\\
0.103515625	1.41553120649297\\
0.10546875	1.35645807591166\\
0.107421875	3.08616754165491\\
0.109375	1.95720680458388\\
0.111328125	1.11462426324608\\
0.11328125	0.931945184758695\\
0.115234375	0.308969034001433\\
0.1171875	0.121000229791692\\
0.119140625	0.171310238195665\\
0.12109375	0.442898488784029\\
0.123046875	1.05542056486086\\
0.125	2.37180771858487\\
0.126953125	1.52193199510648\\
0.12890625	0.929770789298453\\
0.130859375	0.449270444398189\\
0.1328125	0.360569406790996\\
0.134765625	0.466892383815796\\
0.13671875	0.675634294719684\\
0.138671875	1.41412502533533\\
0.140625	2.36095412005548\\
0.142578125	0.619819724331721\\
0.14453125	0.166684465561095\\
0.146484375	0.0856342142271929\\
0.1484375	0.0392381981348363\\
0.150390625	0.026299694623405\\
0.15234375	0.0143597962819369\\
0.154296875	0.014333862944514\\
0.15625	0.00685314291050189\\
0.158203125	0.00460054865739958\\
0.16015625	0.0018603267902105\\
0.162109375	0.000611081631179907\\
0.1640625	0.000281926240036208\\
0.166015625	0.000235455729201335\\
0.16796875	0.000233908088028252\\
0.169921875	0.000525207605563673\\
0.171875	0.000652342436936588\\
0.173828125	0.00103656727818413\\
0.17578125	0.00212269703792853\\
0.177734375	0.00366603853776644\\
0.1796875	0.00969549376510452\\
0.181640625	0.0221299754678624\\
0.18359375	0.0315702204721842\\
0.185546875	0.107396252480985\\
0.1875	0.214851853093038\\
0.189453125	0.317468477194041\\
0.19140625	0.486746210793082\\
0.193359375	0.469737552980829\\
0.1953125	0.204049717343256\\
0.197265625	0.15656046740279\\
0.19921875	0.360765631897512\\
0.201171875	0.792449346223836\\
0.203125	1.00810780234441\\
0.205078125	1.29947287532217\\
0.20703125	1.38281160168183\\
0.208984375	1.33182950758529\\
0.2109375	1.49763614075934\\
0.212890625	1.23281810606477\\
0.21484375	1.91055610733499\\
0.216796875	3.04605065052477\\
0.21875	3.09300248825296\\
0.220703125	1.81580942966363\\
0.22265625	1.45307368345403\\
0.224609375	1.34824126722509\\
0.2265625	0.607785439022166\\
0.228515625	0.541737080964651\\
0.23046875	0.298539483103439\\
0.232421875	0.179966422762484\\
0.234375	0.122896710307075\\
0.236328125	0.119241903235284\\
0.23828125	0.151347312058496\\
0.240234375	0.20070055858367\\
0.2421875	0.133386014246085\\
0.244140625	0.228334683036982\\
0.24609375	0.25097987551681\\
0.248046875	0.677751111951382\\
0.25	1.4798601072777\\
0.251953125	1.17880389983151\\
0.25390625	1.58647509690559\\
0.255859375	2.41688881880542\\
0.2578125	0.933999418957006\\
0.259765625	0.303451657327283\\
0.26171875	0.107867720977697\\
0.263671875	0.0656446196412727\\
0.265625	0.0424180081379145\\
0.267578125	0.026540698676191\\
0.26953125	0.0230136765026268\\
0.271484375	0.0272387816933891\\
0.2734375	0.0485293342224975\\
0.275390625	0.153359832333933\\
0.27734375	0.460150849132284\\
0.279296875	0.71682157735783\\
0.28125	1.0523428755113\\
0.283203125	0.804439284498194\\
0.28515625	0.715839018201827\\
0.287109375	0.786587398986719\\
0.2890625	1.2185804523248\\
0.291015625	1.10310349199754\\
0.29296875	0.258370992195673\\
0.294921875	0.158473992871208\\
0.296875	0.112444463054909\\
0.298828125	0.0734456756628488\\
0.30078125	0.054856313847138\\
0.302734375	0.0409132651535705\\
0.3046875	0.0171968012868259\\
0.306640625	0.00646814095323717\\
0.30859375	0.00355477832385593\\
0.310546875	0.00585782863331099\\
0.3125	0.0154444174059306\\
0.314453125	0.0640651347623037\\
0.31640625	0.189005841768697\\
0.318359375	0.237344199201331\\
0.3203125	0.21472139518153\\
0.322265625	0.314528534272999\\
0.32421875	0.496090966701551\\
0.326171875	1.18209959061226\\
0.328125	1.50828173848732\\
0.330078125	1.35604898691068\\
0.33203125	0.897882328191562\\
0.333984375	1.30337133868776\\
0.3359375	3.7252612074072\\
0.337890625	2.95623744669871\\
0.33984375	1.96841376119475\\
0.341796875	0.908957749502034\\
0.34375	0.245342560386808\\
0.345703125	0.0739725250141826\\
0.34765625	0.033428414238692\\
0.349609375	0.0191343261669914\\
0.3515625	0.0125756116671949\\
0.353515625	0.00990905344873087\\
0.35546875	0.0158855657849194\\
0.357421875	0.0395937057252058\\
0.359375	0.089972460125659\\
0.361328125	0.229134001120022\\
0.36328125	0.64114435432364\\
0.365234375	0.803419377687474\\
0.3671875	1.84063194055687\\
0.369140625	2.08654706168282\\
0.37109375	1.47833746164667\\
0.373046875	1.89592503351165\\
0.375	1.41196265241542\\
0.376953125	0.653878542922885\\
0.37890625	0.393044512558389\\
0.380859375	0.107346011353402\\
0.3828125	0.0615160619305271\\
0.384765625	0.0400838294351746\\
0.38671875	0.0320350082500688\\
0.388671875	0.0190424439519259\\
0.390625	0.0100627199413933\\
0.392578125	0.0036617567498\\
0.39453125	0.00179072156829369\\
0.396484375	0.000857648036181426\\
0.3984375	0.000403428073934698\\
0.400390625	0.000413025858003213\\
0.40234375	0.000529727186138843\\
0.404296875	0.000704845774233021\\
0.40625	0.000713407658226943\\
0.408203125	0.000914871166624091\\
0.41015625	0.00245494168621212\\
0.412109375	0.00877586621664733\\
0.4140625	0.0307873046642601\\
0.416015625	0.0619952266339644\\
0.41796875	0.123865269359318\\
0.419921875	0.35805606665061\\
0.421875	0.688510027357353\\
0.423828125	0.960515097223016\\
0.42578125	0.979814006588242\\
0.427734375	0.612822665890151\\
0.4296875	0.441424136146987\\
0.431640625	0.83327339269231\\
0.43359375	1.35875485819677\\
0.435546875	2.60396710681865\\
0.4375	2.73383860097089\\
0.439453125	2.54873812396549\\
0.44140625	2.5130148056518\\
0.443359375	1.59810660645823\\
0.4453125	1.30483653884824\\
0.447265625	1.44301163882908\\
0.44921875	1.22181265261651\\
0.451171875	0.500134288750555\\
0.453125	0.187415197238087\\
0.455078125	0.0556751668628345\\
0.45703125	0.0386427994841672\\
0.458984375	0.0433913365263014\\
0.4609375	0.0717467288418559\\
0.462890625	0.0496379108418209\\
0.46484375	0.0317011742643851\\
0.466796875	0.0699177883857246\\
0.46875	0.157508327748441\\
0.470703125	0.305029798964135\\
0.47265625	0.318867520736629\\
0.474609375	0.233662861720294\\
0.4765625	0.408989883411101\\
0.478515625	0.960973839417897\\
0.48046875	0.818410999168546\\
0.482421875	0.446753231360004\\
0.484375	0.343915612526581\\
0.486328125	0.343161468815841\\
0.48828125	0.329050180580435\\
0.490234375	0.461262495256984\\
0.4921875	0.909903718796314\\
0.494140625	1.8322462639618\\
0.49609375	1.64689595197351\\
0.498046875	0.696051296456668\\
0.5	0.279965438688652\\
0.501953125	0.0898277897473144\\
0.50390625	0.0311518679687381\\
0.505859375	0.0164918423381494\\
0.5078125	0.00497246183887454\\
0.509765625	0.00316474575471659\\
0.51171875	0.00263579840948376\\
0.513671875	0.00149837404178482\\
0.515625	0.000519232420861919\\
0.517578125	0.000232070655487263\\
0.51953125	0.000103813008968669\\
0.521484375	9.78532625211205e-05\\
0.5234375	8.79832901476017e-05\\
0.525390625	0.000178784982621367\\
0.52734375	0.000661747572580036\\
0.529296875	0.00246667715545641\\
0.53125	0.00354375780882122\\
0.533203125	0.00539259416949676\\
0.53515625	0.0107600686920732\\
0.537109375	0.0308590495461052\\
0.5390625	0.0930425904860282\\
0.541015625	0.13683840743989\\
0.54296875	0.339967545361278\\
0.544921875	0.785708524765585\\
0.546875	1.68882280264556\\
0.548828125	1.83438972454464\\
0.55078125	1.02991488900362\\
0.552734375	0.560124854710665\\
0.5546875	0.243458953998347\\
0.556640625	0.0830775301311753\\
0.55859375	0.0318522336194805\\
0.560546875	0.048917539262947\\
0.5625	0.0842049064695866\\
0.564453125	0.0888818276198722\\
0.56640625	0.0608928363671778\\
0.568359375	0.037594426342527\\
0.5703125	0.0721004422106245\\
0.572265625	0.322945630735189\\
0.57421875	0.648056939565755\\
0.576171875	0.447400453787232\\
0.578125	0.355312327761714\\
0.580078125	0.157336204658643\\
0.58203125	0.104810729002448\\
0.583984375	0.0994845325129364\\
0.5859375	0.134221942802615\\
0.587890625	0.411792459107011\\
0.58984375	0.992432851316836\\
0.591796875	1.03934244917642\\
0.59375	1.44654406306143\\
0.595703125	2.22794943418771\\
0.59765625	2.96316124881142\\
0.599609375	1.97422094966203\\
0.6015625	1.52898434090171\\
0.603515625	1.75890310211307\\
0.60546875	1.96609282499636\\
0.607421875	2.36811794769488\\
0.609375	1.44760448939361\\
0.611328125	0.531343648864096\\
0.61328125	0.577251007575898\\
0.615234375	0.889714807993714\\
0.6171875	0.853461028099033\\
0.619140625	0.702538828344214\\
0.62109375	0.446544499891305\\
0.623046875	0.209159839159419\\
0.625	0.0933054176300713\\
0.626953125	0.0530031809526709\\
0.62890625	0.0221472428274146\\
0.630859375	0.0339756997791315\\
0.6328125	0.0940009144992073\\
0.634765625	0.161152406645924\\
0.63671875	0.411052831425655\\
0.638671875	0.635708899225597\\
0.640625	0.657714589924366\\
0.642578125	1.17630158483623\\
0.64453125	2.34566251641859\\
0.646484375	0.886053994947603\\
0.6484375	0.724539538547028\\
0.650390625	1.11977236115511\\
0.65234375	1.70141745770469\\
0.654296875	1.33493744482443\\
0.65625	0.781127453108149\\
0.658203125	1.6507957022991\\
0.66015625	1.56702679093263\\
0.662109375	1.5631580171568\\
0.6640625	1.63624416695072\\
0.666015625	1.15024120129738\\
0.66796875	0.762735021961339\\
0.669921875	1.11123664687598\\
0.671875	1.15872539153544\\
0.673828125	0.55603798702723\\
0.67578125	0.205291146260617\\
0.677734375	0.05366622956074\\
0.6796875	0.01539418033866\\
0.681640625	0.00434499763687601\\
0.68359375	0.00120872970301334\\
0.685546875	0.000542924359476356\\
0.6875	0.00104406126669463\\
0.689453125	0.00131295158232182\\
0.69140625	0.00342072790773888\\
0.693359375	0.00770735509726256\\
0.6953125	0.0222332789449888\\
0.697265625	0.0579491737277471\\
0.69921875	0.247617508488495\\
0.701171875	0.451648190909925\\
0.703125	0.798387028792342\\
0.705078125	0.401810488197031\\
0.70703125	0.25161897946114\\
0.708984375	0.144927478166823\\
0.7109375	0.213911105867654\\
0.712890625	0.117255011358643\\
0.71484375	0.123977288518872\\
0.716796875	0.297291589846827\\
0.71875	0.34023735400167\\
0.720703125	0.175906790602411\\
0.72265625	0.152570217495321\\
0.724609375	0.152405948374073\\
0.7265625	0.514970108797599\\
0.728515625	0.986695988487097\\
0.73046875	0.776563870981902\\
0.732421875	1.48596932806376\\
0.734375	2.61990053118018\\
0.736328125	1.9677157923834\\
0.73828125	1.12663849077694\\
0.740234375	0.28571285814259\\
0.7421875	0.105111573505706\\
0.744140625	0.0274851855160741\\
0.74609375	0.0172623207480053\\
0.748046875	0.00792243762267195\\
0.75	0.0036601953430794\\
0.751953125	0.00166371189327934\\
0.75390625	0.00258045316993312\\
0.755859375	0.00838051091460069\\
0.7578125	0.00699229929969236\\
0.759765625	0.0135477890726461\\
0.76171875	0.042285791884838\\
0.763671875	0.0841339620558361\\
0.765625	0.127517762881851\\
0.767578125	0.348078737164523\\
0.76953125	0.425651539440543\\
0.771484375	0.374188467644362\\
0.7734375	0.640490571350791\\
0.775390625	1.74366551964095\\
0.77734375	2.64868027302124\\
0.779296875	1.96223825624514\\
0.78125	1.82935675717578\\
0.783203125	1.34163336640097\\
0.78515625	0.881297757795786\\
0.787109375	2.3224879295641\\
0.7890625	2.26958863862082\\
0.791015625	0.692414609339763\\
0.79296875	0.298385483214777\\
0.794921875	0.119272568689132\\
0.796875	0.0505694366098167\\
0.798828125	0.0238758272856562\\
0.80078125	0.00646528968302335\\
0.802734375	0.00167687345859515\\
0.8046875	0.000625820350220927\\
0.806640625	0.000385865407385963\\
0.80859375	0.00044152983759065\\
0.810546875	0.000622997154275204\\
0.8125	0.0018058413049978\\
0.814453125	0.00257923997303811\\
0.81640625	0.0047988691166354\\
0.818359375	0.00340304187436288\\
0.8203125	0.00794341505640487\\
0.822265625	0.0148931495900858\\
0.82421875	0.0268289372683186\\
0.826171875	0.0233070943063818\\
0.828125	0.0419420482850318\\
0.830078125	0.049484165337861\\
0.83203125	0.0448041395368675\\
0.833984375	0.0854929526373504\\
0.8359375	0.195100333212369\\
0.837890625	0.277845123208731\\
0.83984375	0.23077886637788\\
0.841796875	0.686580413643876\\
0.84375	1.61778639306339\\
0.845703125	1.78159087037699\\
0.84765625	1.31852956666197\\
0.849609375	1.3947440472479\\
0.8515625	1.66899089625527\\
0.853515625	1.30895875853915\\
0.85546875	1.9326209091988\\
0.857421875	2.42425998769316\\
0.859375	1.63480683847851\\
0.861328125	0.974936738164348\\
0.86328125	1.14717747460877\\
0.865234375	1.07703855828261\\
0.8671875	1.15544510188208\\
0.869140625	1.73289645836627\\
0.87109375	1.37824298923434\\
0.873046875	1.15839600841296\\
0.875	0.526522272874795\\
0.876953125	0.176229858408081\\
0.87890625	0.0630742412379914\\
0.880859375	0.0169174726116952\\
0.8828125	0.0242397940374506\\
0.884765625	0.0186004941290871\\
0.88671875	0.00654574794558924\\
0.888671875	0.00759300715837643\\
0.890625	0.00717800116421142\\
0.892578125	0.00268620199066984\\
0.89453125	0.000706852515324752\\
0.896484375	0.000265671562185603\\
0.8984375	0.000155129786837626\\
0.900390625	6.77800237089134e-05\\
0.90234375	3.44738921814015e-05\\
0.904296875	1.26166086927518e-05\\
0.90625	8.05051227967966e-06\\
0.908203125	7.82692761112577e-06\\
0.91015625	3.26157957606453e-06\\
0.912109375	9.83704156262494e-07\\
0.9140625	1.07755904702346e-06\\
0.916015625	1.75567108960773e-06\\
0.91796875	3.17716573632377e-06\\
0.919921875	7.67764517282741e-06\\
0.921875	1.50946741770369e-05\\
0.923828125	3.42915725724398e-05\\
0.92578125	0.000140915837047859\\
0.927734375	0.000223897103235188\\
0.9296875	0.000403991077869329\\
0.931640625	0.000630237663277792\\
0.93359375	0.00103338677706782\\
0.935546875	0.0012329296487184\\
0.9375	0.000952379259925622\\
0.939453125	0.00232675771733565\\
0.94140625	0.0119760923113528\\
0.943359375	0.0327127891137021\\
0.9453125	0.107282158472079\\
0.947265625	0.175949090480826\\
0.94921875	0.462116178395722\\
0.951171875	1.34426005985054\\
0.953125	2.53036495953774\\
0.955078125	1.85724804768523\\
0.95703125	1.78127409668502\\
0.958984375	1.74351199181418\\
0.9609375	0.865615630466008\\
0.962890625	0.597720025221492\\
0.96484375	0.149173321521329\\
0.966796875	0.0574420937969284\\
0.96875	0.0278900283751257\\
0.970703125	0.00902425024986438\\
0.97265625	0.00448419186294355\\
0.974609375	0.00459824318468225\\
0.9765625	0.00495329130636116\\
0.978515625	0.00317123664322951\\
0.98046875	0.00480738124246933\\
0.982421875	0.00915077280565653\\
0.984375	0.026343259722216\\
0.986328125	0.0674953258448427\\
0.98828125	0.157004730561952\\
0.990234375	0.182221526539452\\
0.9921875	0.461470597217512\\
0.994140625	0.824656490751977\\
0.99609375	2.01431646628326\\
0.998046875	2.29709444680101\\
1	0\\
};
\end{axis}
\end{tikzpicture}%

%% file: pics/square/GroundState_1_2_2.tex
%
%
\begin{tikzpicture}

\begin{axis}[%
width=0.9in,
height=0.9in,
scale only axis,
xmin=0,
xmax=1,
xtick={0,0.5,1},
ymin=0,
ymax=10,
yminorticks=true,
yticklabels = {},
axis background/.style={fill=white},
]
\addplot [color=mycolor2, thick]
table[row sep=crcr]{%
0	0\\
0.001953125	0.23918496961853\\
0.00390625	0.892634942130427\\
0.005859375	3.0373813084821\\
0.0078125	3.16183255818638\\
0.009765625	2.70871511088038\\
0.01171875	1.3723636900932\\
0.013671875	3.86735635008463\\
0.015625	3.10761297622442\\
0.017578125	0.863236375472291\\
0.01953125	0.349253842948345\\
0.021484375	0.229872086127887\\
0.0234375	0.111690197814644\\
0.025390625	0.0828105518947763\\
0.02734375	0.0436119919886641\\
0.029296875	0.0183504509806562\\
0.03125	0.021259181329817\\
0.033203125	0.0498356477567561\\
0.03515625	0.112190640556145\\
0.037109375	0.176049498389495\\
0.0390625	0.316269219855686\\
0.041015625	0.654724330101438\\
0.04296875	3.00000141135158\\
0.044921875	3.36470447986477\\
0.046875	1.28620854139088\\
0.048828125	1.18473527869426\\
0.05078125	1.95572888305082\\
0.052734375	3.76245263697017\\
0.0546875	4.76943925602019\\
0.056640625	2.36814981330837\\
0.05859375	1.98415747124506\\
0.060546875	2.65080667821066\\
0.0625	2.2163779883075\\
0.064453125	1.6704183698073\\
0.06640625	0.890062182695936\\
0.068359375	0.502629894997895\\
0.0703125	0.19960883405994\\
0.072265625	0.103776246792976\\
0.07421875	0.0694714257694146\\
0.076171875	0.127333476310673\\
0.078125	0.0786877771112459\\
0.080078125	0.0339861713102022\\
0.08203125	0.0106684617537998\\
0.083984375	0.00417136350454394\\
0.0859375	0.00434865806393361\\
0.087890625	0.0100376798179878\\
0.08984375	0.0256768226367593\\
0.091796875	0.0515537396240937\\
0.09375	0.0913107650026382\\
0.095703125	0.122027007973086\\
0.09765625	0.240211835761927\\
0.099609375	0.298925594252919\\
0.1015625	0.61111354835457\\
0.103515625	1.22649677415811\\
0.10546875	1.37823488741818\\
0.107421875	3.51030247921028\\
0.109375	2.10645289552892\\
0.111328125	1.0423274757698\\
0.11328125	0.773622573445116\\
0.115234375	0.237386200946371\\
0.1171875	0.0858135490839531\\
0.119140625	0.127884888260162\\
0.12109375	0.366367295823284\\
0.123046875	0.964936457897585\\
0.125	2.3095576543768\\
0.126953125	1.41896425353939\\
0.12890625	0.778485670328965\\
0.130859375	0.322156530903745\\
0.1328125	0.23956991183507\\
0.134765625	0.337663622535592\\
0.13671875	0.568654272311276\\
0.138671875	1.37133965043029\\
0.140625	2.38410283097619\\
0.142578125	0.59603980046098\\
0.14453125	0.145119036832015\\
0.146484375	0.0655412284239164\\
0.1484375	0.0250554070567493\\
0.150390625	0.0141955698372507\\
0.15234375	0.00585489123961661\\
0.154296875	0.00466651206029505\\
0.15625	0.00192853859821223\\
0.158203125	0.00110266602247584\\
0.16015625	0.000403673464070871\\
0.162109375	0.000117228367665416\\
0.1640625	4.24426026011443e-05\\
0.166015625	2.08167469858618e-05\\
0.16796875	8.25260671907062e-06\\
0.169921875	6.73391001948724e-06\\
0.171875	5.53850815155235e-06\\
0.173828125	7.2886460741323e-06\\
0.17578125	1.5680615940417e-05\\
0.177734375	3.0557027303254e-05\\
0.1796875	9.22363643676721e-05\\
0.181640625	0.000236636457450282\\
0.18359375	0.000389485999679\\
0.185546875	0.00152580419317951\\
0.1875	0.00341196532320806\\
0.189453125	0.0059548536175274\\
0.19140625	0.0108806055036919\\
0.193359375	0.0132386132440523\\
0.1953125	0.0120684204303214\\
0.197265625	0.031189228391002\\
0.19921875	0.106583881720732\\
0.201171875	0.267662053685145\\
0.203125	0.390532940917624\\
0.205078125	0.595126949879788\\
0.20703125	0.733469056172386\\
0.208984375	0.826622914084659\\
0.2109375	1.10351201838829\\
0.212890625	1.10346043967767\\
0.21484375	2.11655258452896\\
0.216796875	3.64664796601424\\
0.21875	3.70469211154248\\
0.220703125	2.00353845469358\\
0.22265625	1.32298167971406\\
0.224609375	1.08653706145709\\
0.2265625	0.432922914101769\\
0.228515625	0.298044714564989\\
0.23046875	0.139209511161869\\
0.232421875	0.0673697211792598\\
0.234375	0.0346574543849172\\
0.236328125	0.0249797167129666\\
0.23828125	0.0271895820110089\\
0.240234375	0.0385543838011901\\
0.2421875	0.0388427892030046\\
0.244140625	0.0916183265102525\\
0.24609375	0.138600003206424\\
0.248046875	0.442815131371897\\
0.25	1.06333814852421\\
0.251953125	0.933253870932236\\
0.25390625	1.46447636845034\\
0.255859375	2.37611623855138\\
0.2578125	0.871019327090246\\
0.259765625	0.258378183847236\\
0.26171875	0.0813641385842204\\
0.263671875	0.0389832809332965\\
0.265625	0.021031317889352\\
0.267578125	0.00917302246793553\\
0.26953125	0.00434848559288389\\
0.271484375	0.00148042999045434\\
0.2734375	0.000567671754000185\\
0.275390625	0.000226755517731836\\
0.27734375	0.000366038560252971\\
0.279296875	0.000519045011759811\\
0.28125	0.000775526181090205\\
0.283203125	0.000652972435711085\\
0.28515625	0.000665620758680229\\
0.287109375	0.00098727330753497\\
0.2890625	0.00177769615310401\\
0.291015625	0.00165134674291839\\
0.29296875	0.000359599501248861\\
0.294921875	0.000186785333025347\\
0.296875	0.00011327905936772\\
0.298828125	6.5212340485719e-05\\
0.30078125	4.75710648178653e-05\\
0.302734375	4.20234246430429e-05\\
0.3046875	3.42356713125222e-05\\
0.306640625	7.52461652873109e-05\\
0.30859375	0.000219449516849848\\
0.310546875	0.000719158482169605\\
0.3125	0.00232728285023806\\
0.314453125	0.0105632926694444\\
0.31640625	0.0338456171276228\\
0.318359375	0.049381977576574\\
0.3203125	0.0616323908239832\\
0.322265625	0.123884567893182\\
0.32421875	0.248830324553077\\
0.326171875	0.672824459376823\\
0.328125	0.920656804833032\\
0.330078125	0.907797954786971\\
0.33203125	0.751037455720484\\
0.333984375	1.49525032774785\\
0.3359375	4.75829764250522\\
0.337890625	3.71136002681707\\
0.33984375	2.248518296457\\
0.341796875	0.976191264602246\\
0.34375	0.242727237026799\\
0.345703125	0.0664789767480962\\
0.34765625	0.0258698817953356\\
0.349609375	0.0122779757285617\\
0.3515625	0.00656104554834947\\
0.353515625	0.00432302431334903\\
0.35546875	0.00689221920857971\\
0.357421875	0.0190123008364418\\
0.359375	0.0479416835582149\\
0.361328125	0.136974419027837\\
0.36328125	0.422528064973907\\
0.365234375	0.599523068983055\\
0.3671875	1.58059497340395\\
0.369140625	1.83175577022148\\
0.37109375	1.22436485639152\\
0.373046875	1.49655864766066\\
0.375	1.03671998075164\\
0.376953125	0.43462722162874\\
0.37890625	0.227449267100399\\
0.380859375	0.0566180333710766\\
0.3828125	0.0262454588473143\\
0.384765625	0.0134890171024848\\
0.38671875	0.00891319264740133\\
0.388671875	0.00450313201513133\\
0.390625	0.00209369586077507\\
0.392578125	0.000685563903281063\\
0.39453125	0.000291435306743938\\
0.396484375	0.000121154723949056\\
0.3984375	4.19579164387391e-05\\
0.400390625	2.56641394500789e-05\\
0.40234375	2.42933970356178e-05\\
0.404296875	3.01590848015994e-05\\
0.40625	3.51449295437838e-05\\
0.408203125	5.94826954319316e-05\\
0.41015625	0.000185885969793697\\
0.412109375	0.000734440534162497\\
0.4140625	0.00279417246149847\\
0.416015625	0.00622894270715516\\
0.41796875	0.0143764635373251\\
0.419921875	0.0459151743513555\\
0.421875	0.0995941132118663\\
0.423828125	0.162125387314573\\
0.42578125	0.191571214744156\\
0.427734375	0.174513605151251\\
0.4296875	0.249672282303877\\
0.431640625	0.731152340102697\\
0.43359375	1.41785717709221\\
0.435546875	3.01904052238094\\
0.4375	3.24967054050401\\
0.439453125	3.01019081303112\\
0.44140625	2.88185180815376\\
0.443359375	1.61504795649856\\
0.4453125	1.12856112532701\\
0.447265625	1.10033966426092\\
0.44921875	0.880308000124476\\
0.451171875	0.327956499473223\\
0.453125	0.111505802120864\\
0.455078125	0.0284658499940856\\
0.45703125	0.0122845929014944\\
0.458984375	0.00646478052262784\\
0.4609375	0.00681555633900533\\
0.462890625	0.00357251837166503\\
0.46484375	0.00104902680599508\\
0.466796875	0.000743084719822952\\
0.46875	0.00118956270649101\\
0.470703125	0.00235188191641392\\
0.47265625	0.00290871343563433\\
0.474609375	0.00300927694431038\\
0.4765625	0.00808988013806038\\
0.478515625	0.0215436719420218\\
0.48046875	0.0220283879640583\\
0.482421875	0.0181935214007634\\
0.484375	0.0361073557599178\\
0.486328125	0.0688399052658068\\
0.48828125	0.106791322334072\\
0.490234375	0.244615244238438\\
0.4921875	0.575829622654914\\
0.494140625	1.24674257253804\\
0.49609375	1.10908274233241\\
0.498046875	0.432585084220454\\
0.5	0.15610837580032\\
0.501953125	0.046027219048631\\
0.50390625	0.0144004428111221\\
0.505859375	0.00666381315226051\\
0.5078125	0.00179259868981414\\
0.509765625	0.00088446714276843\\
0.51171875	0.0006017978542193\\
0.513671875	0.000301855510173919\\
0.515625	9.30561294548942e-05\\
0.517578125	3.72244164863978e-05\\
0.51953125	1.42069557664529e-05\\
0.521484375	1.23905913786214e-05\\
0.5234375	1.29066857048315e-05\\
0.525390625	3.15229577957352e-05\\
0.52734375	0.000128718886366951\\
0.529296875	0.000519533862397217\\
0.53125	0.000851864087509797\\
0.533203125	0.00159656966249501\\
0.53515625	0.00372084551485584\\
0.537109375	0.0121344502206569\\
0.5390625	0.039917267112872\\
0.541015625	0.0682929445323844\\
0.54296875	0.194054405532955\\
0.544921875	0.50025665344567\\
0.546875	1.14906294455374\\
0.548828125	1.25116385430808\\
0.55078125	0.628075631203537\\
0.552734375	0.307332768609157\\
0.5546875	0.120104393153136\\
0.556640625	0.0361705229336862\\
0.55859375	0.00920816986544719\\
0.560546875	0.00477067349605611\\
0.5625	0.00444172630580943\\
0.564453125	0.00323673789808033\\
0.56640625	0.00174832702653586\\
0.568359375	0.000641564183801893\\
0.5703125	0.000498369378666842\\
0.572265625	0.00201809365891115\\
0.57421875	0.00438433996678781\\
0.576171875	0.00385605767269749\\
0.578125	0.00487129336546015\\
0.580078125	0.00604283379448046\\
0.58203125	0.0163586930567109\\
0.583984375	0.0319175877713648\\
0.5859375	0.0714184965387655\\
0.587890625	0.270024345962887\\
0.58984375	0.711930317264043\\
0.591796875	0.859026518125887\\
0.59375	1.48032748466675\\
0.595703125	2.47268611919276\\
0.59765625	3.433999253645\\
0.599609375	2.18571088857852\\
0.6015625	1.54033737830059\\
0.603515625	1.72003995046887\\
0.60546875	1.93475446009097\\
0.607421875	2.32267126730978\\
0.609375	1.30796232844473\\
0.611328125	0.391771725090985\\
0.61328125	0.183762674453651\\
0.615234375	0.194476074386875\\
0.6171875	0.15424469917309\\
0.619140625	0.107922984632791\\
0.62109375	0.0591497832555979\\
0.623046875	0.0242508336518487\\
0.625	0.00986829583407005\\
0.626953125	0.0056654924599592\\
0.62890625	0.00379284928311014\\
0.630859375	0.0099230506096856\\
0.6328125	0.0342623058250848\\
0.634765625	0.0662853300112046\\
0.63671875	0.196205343127117\\
0.638671875	0.344758595497035\\
0.640625	0.444908721633639\\
0.642578125	0.982681148164555\\
0.64453125	2.09859192783173\\
0.646484375	0.731763678288764\\
0.6484375	0.429125241431378\\
0.650390625	0.596651101555084\\
0.65234375	0.912526484262962\\
0.654296875	0.706838763609771\\
0.65625	0.436581108475805\\
0.658203125	1.01507485851767\\
0.66015625	0.982271410751767\\
0.662109375	0.954376703352288\\
0.6640625	0.926504493751436\\
0.666015625	0.575427586512505\\
0.66796875	0.284137531397426\\
0.669921875	0.303915469148872\\
0.671875	0.287360525935512\\
0.673828125	0.127229567357663\\
0.67578125	0.0427650628851347\\
0.677734375	0.0104064349826577\\
0.6796875	0.00272856010028797\\
0.681640625	0.000709719546873294\\
0.68359375	0.000175097359854215\\
0.685546875	4.26242113360566e-05\\
0.6875	2.65668258472319e-05\\
0.689453125	8.40261529920686e-06\\
0.69140625	3.82885119062102e-06\\
0.693359375	2.87163937818389e-06\\
0.6953125	6.47274417162426e-06\\
0.697265625	1.78825390225476e-05\\
0.69921875	8.21499449016223e-05\\
0.701171875	0.000167736631123325\\
0.703125	0.000332554966811612\\
0.705078125	0.000248532496089484\\
0.70703125	0.000330491544704475\\
0.708984375	0.000626599989227562\\
0.7109375	0.001704336768553\\
0.712890625	0.00190583145223709\\
0.71484375	0.00458991547642318\\
0.716796875	0.0139963089741669\\
0.71875	0.0193736537367441\\
0.720703125	0.0176225652894995\\
0.72265625	0.0411729888213986\\
0.724609375	0.0683622673378133\\
0.7265625	0.298723792079866\\
0.728515625	0.634773906785304\\
0.73046875	0.617284586585516\\
0.732421875	1.49140536116665\\
0.734375	2.83756353736228\\
0.736328125	2.07973693498168\\
0.73828125	1.12047770947529\\
0.740234375	0.262591200670552\\
0.7421875	0.0879093188332532\\
0.744140625	0.0207665836804339\\
0.74609375	0.0111759272726279\\
0.748046875	0.00448445564068137\\
0.75	0.00179440764072165\\
0.751953125	0.000550495954642633\\
0.75390625	0.000455629822929529\\
0.755859375	0.00138472876342816\\
0.7578125	0.001407429979111\\
0.759765625	0.00338522377188679\\
0.76171875	0.0118794959200906\\
0.763671875	0.0263439925392315\\
0.765625	0.0472617493742197\\
0.767578125	0.151063246384346\\
0.76953125	0.214980604504097\\
0.771484375	0.252130560647043\\
0.7734375	0.596337599734378\\
0.775390625	1.85247769935145\\
0.77734375	2.9165573883176\\
0.779296875	2.06568900780927\\
0.78125	1.81475853962724\\
0.783203125	1.26935609273087\\
0.78515625	0.865908134723338\\
0.787109375	2.49282030333114\\
0.7890625	2.44027998985237\\
0.791015625	0.699248818993676\\
0.79296875	0.265425474388863\\
0.794921875	0.0954473206965691\\
0.796875	0.0357651707668415\\
0.798828125	0.015224792953333\\
0.80078125	0.0038167081633607\\
0.802734375	0.000907917812172609\\
0.8046875	0.000279343088906467\\
0.806640625	0.000105858331529758\\
0.80859375	6.30901238878895e-05\\
0.810546875	2.694563644107e-05\\
0.8125	4.56515429863416e-05\\
0.814453125	6.49264408884836e-05\\
0.81640625	0.000134806511846887\\
0.818359375	0.0001358712186993\\
0.8203125	0.000415055088115491\\
0.822265625	0.000917904872004803\\
0.82421875	0.00189509304349823\\
0.826171875	0.00218376664516142\\
0.828125	0.00506848331694168\\
0.830078125	0.00755725210629279\\
0.83203125	0.0109444972832249\\
0.833984375	0.0272599597121197\\
0.8359375	0.0720917969217098\\
0.837890625	0.120258949808803\\
0.83984375	0.139237955264191\\
0.841796875	0.512423725807854\\
0.84375	1.33139829179708\\
0.845703125	1.47132600588573\\
0.84765625	1.03956235665236\\
0.849609375	1.09117354250951\\
0.8515625	1.34740441828719\\
0.853515625	1.12841697981047\\
0.85546875	1.8456286329368\\
0.857421875	2.39139000147899\\
0.859375	1.49776497345789\\
0.861328125	0.722073296433967\\
0.86328125	0.658071340259046\\
0.865234375	0.526443608047659\\
0.8671875	0.516674491578938\\
0.869140625	0.756741236808357\\
0.87109375	0.556212221682619\\
0.873046875	0.436203372036716\\
0.875	0.182190824266648\\
0.876953125	0.0553942187610787\\
0.87890625	0.0180865914651476\\
0.880859375	0.00372678994634996\\
0.8828125	0.00318562235934201\\
0.884765625	0.00211230253728441\\
0.88671875	0.000580490554615541\\
0.888671875	0.000472905795549196\\
0.890625	0.000364342473972585\\
0.892578125	0.000125217987354442\\
0.89453125	2.98782558687411e-05\\
0.896484375	1.00411808575379e-05\\
0.8984375	5.0595409925536e-06\\
0.900390625	1.96177467540091e-06\\
0.90234375	8.6795038204456e-07\\
0.904296875	2.80349248992444e-07\\
0.90625	1.36913409040245e-07\\
0.908203125	1.09900085777731e-07\\
0.91015625	4.37119472208311e-08\\
0.912109375	1.92474415331074e-08\\
0.9140625	4.42427256475465e-08\\
0.916015625	9.91619996988385e-08\\
0.91796875	2.17114145667915e-07\\
0.919921875	6.06325509205389e-07\\
0.921875	1.36648815239309e-06\\
0.923828125	3.5300462452084e-06\\
0.92578125	1.59075705962349e-05\\
0.927734375	2.85312875584198e-05\\
0.9296875	5.98810226887742e-05\\
0.931640625	0.000115472575370233\\
0.93359375	0.000224107960990305\\
0.935546875	0.000339854286324388\\
0.9375	0.000392367390581695\\
0.939453125	0.00129274032207765\\
0.94140625	0.00721952073830914\\
0.943359375	0.0214771438675815\\
0.9453125	0.0774336282308262\\
0.947265625	0.143034693598239\\
0.94921875	0.426508683189041\\
0.951171875	1.36533783399814\\
0.953125	2.76234231561988\\
0.955078125	1.96313390626068\\
0.95703125	1.67826247376776\\
0.958984375	1.52921001127694\\
0.9609375	0.688782152042591\\
0.962890625	0.422057581645659\\
0.96484375	0.0969596847176935\\
0.966796875	0.0334779860719641\\
0.96875	0.0141932418155865\\
0.970703125	0.00410581332755451\\
0.97265625	0.00159524311045608\\
0.974609375	0.00117466910443586\\
0.9765625	0.00114272240889488\\
0.978515625	0.000855103247837442\\
0.98046875	0.00170056957649576\\
0.982421875	0.00388905373964736\\
0.984375	0.0126170505095642\\
0.986328125	0.0356908335104554\\
0.98828125	0.0926834067428298\\
0.990234375	0.128015559344888\\
0.9921875	0.398370742946043\\
0.994140625	0.807092230863851\\
0.99609375	2.23222476046368\\
0.998046875	2.57828259352615\\
1	0\\
};
\end{axis}
\end{tikzpicture}%

%% file: pics/square/GroundState_1_2_3.tex
%
%
\begin{tikzpicture}

\begin{axis}[%
width=0.9in,
height=0.9in,
scale only axis,
xmin=0,
xmax=1,
xtick={0,0.5,1},
ymin=0,
ymax=10,
yminorticks=true,
yticklabels = {},
axis background/.style={fill=white},
]
\addplot [color=mycolor2, thick]
table[row sep=crcr]{%
0	0\\
0.001953125	0.24065176815199\\
0.00390625	0.952490961590708\\
0.005859375	3.37739298103181\\
0.0078125	3.5766011085699\\
0.009765625	3.04605261479246\\
0.01171875	1.56454217652191\\
0.013671875	4.63423644324448\\
0.015625	3.67232922746813\\
0.017578125	0.960687778109946\\
0.01953125	0.34784078550053\\
0.021484375	0.197321110331032\\
0.0234375	0.0846897463739455\\
0.025390625	0.0514685722175183\\
0.02734375	0.0249107486638791\\
0.029296875	0.010055101085408\\
0.03125	0.0128140448884006\\
0.033203125	0.0335260927959342\\
0.03515625	0.0826126765715397\\
0.037109375	0.150815153127195\\
0.0390625	0.3050513947516\\
0.041015625	0.735455102442042\\
0.04296875	3.56583774852562\\
0.044921875	3.99966530345581\\
0.046875	1.42542017760411\\
0.048828125	1.28460975650698\\
0.05078125	2.32007595667486\\
0.052734375	4.77161573244302\\
0.0546875	6.13190157380657\\
0.056640625	2.84845583236582\\
0.05859375	2.07436161685623\\
0.060546875	2.62638537010394\\
0.0625	2.06771749884454\\
0.064453125	1.45401777508999\\
0.06640625	0.68959605421\\
0.068359375	0.352882730836044\\
0.0703125	0.124755216087605\\
0.072265625	0.0533789986286328\\
0.07421875	0.0232898103712556\\
0.076171875	0.0293928707792996\\
0.078125	0.0162727362673074\\
0.080078125	0.006479206657277\\
0.08203125	0.00190926835573987\\
0.083984375	0.000780633523039888\\
0.0859375	0.00107603831604993\\
0.087890625	0.00287605107489335\\
0.08984375	0.00815348206949903\\
0.091796875	0.0183282428934155\\
0.09375	0.0361860497141796\\
0.095703125	0.057933729442367\\
0.09765625	0.13249093333218\\
0.099609375	0.198896533982814\\
0.1015625	0.456865157073117\\
0.103515625	1.04464750270865\\
0.10546875	1.33695091856372\\
0.107421875	3.70914059962269\\
0.109375	2.12266413391951\\
0.111328125	0.935454011452189\\
0.11328125	0.623135575051255\\
0.115234375	0.177495655479493\\
0.1171875	0.0500361757289699\\
0.119140625	0.0526889859197974\\
0.12109375	0.147639922652006\\
0.123046875	0.415344447559687\\
0.125	1.04389367527179\\
0.126953125	0.619638103033223\\
0.12890625	0.314222197877466\\
0.130859375	0.121956903008575\\
0.1328125	0.0967124360751386\\
0.134765625	0.168074306714866\\
0.13671875	0.329033887106076\\
0.138671875	0.885817435722472\\
0.140625	1.58911856796999\\
0.142578125	0.382300263416436\\
0.14453125	0.086605285761108\\
0.146484375	0.0357610241227064\\
0.1484375	0.0121951984351673\\
0.150390625	0.00621781458317041\\
0.15234375	0.00219690138863462\\
0.154296875	0.00151273490572413\\
0.15625	0.000565955422120612\\
0.158203125	0.000288740728917278\\
0.16015625	9.82348453649767e-05\\
0.162109375	2.64581659964895e-05\\
0.1640625	8.76028169300999e-06\\
0.166015625	3.83120829929987e-06\\
0.16796875	1.31009225871214e-06\\
0.169921875	7.67218054639093e-07\\
0.171875	4.62770811023182e-07\\
0.173828125	4.45296891813908e-07\\
0.17578125	9.15285664299742e-07\\
0.177734375	1.9251451097113e-06\\
0.1796875	6.36305819286797e-06\\
0.181640625	1.7766568170009e-05\\
0.18359375	3.24139669264327e-05\\
0.185546875	0.000139740110226113\\
0.1875	0.000339424825967944\\
0.189453125	0.000667042015084621\\
0.19140625	0.00137588332034763\\
0.193359375	0.00195945703335015\\
0.1953125	0.00245769213754843\\
0.197265625	0.0078766484277905\\
0.19921875	0.0292012490082381\\
0.201171875	0.0792309730763911\\
0.203125	0.130422354366692\\
0.205078125	0.233996838006757\\
0.20703125	0.344484125685933\\
0.208984375	0.482149221361144\\
0.2109375	0.798908174053122\\
0.212890625	1.00115702573225\\
0.21484375	2.31402433606338\\
0.216796875	4.24896008911117\\
0.21875	4.31931876300349\\
0.220703125	2.18226828792956\\
0.22265625	1.20431406923803\\
0.224609375	0.872321251616578\\
0.2265625	0.318131597823434\\
0.228515625	0.18510469049019\\
0.23046875	0.0777104923287276\\
0.232421875	0.0330840391289123\\
0.234375	0.0142643738260299\\
0.236328125	0.00768568877089946\\
0.23828125	0.00567576509405996\\
0.240234375	0.00722382085832481\\
0.2421875	0.00795615738235126\\
0.244140625	0.0208462334813387\\
0.24609375	0.0365362653587844\\
0.248046875	0.127861185062877\\
0.25	0.33196540961013\\
0.251953125	0.331680737759847\\
0.25390625	0.617472633272963\\
0.255859375	1.06516232293946\\
0.2578125	0.374816221441989\\
0.259765625	0.103879626655301\\
0.26171875	0.0302052776275704\\
0.263671875	0.0126154948828061\\
0.265625	0.00618916385836771\\
0.267578125	0.00242815802185641\\
0.26953125	0.00104742599597171\\
0.271484375	0.000325221806085112\\
0.2734375	0.000112402886156827\\
0.275390625	2.67276466341238e-05\\
0.27734375	1.57457621532664e-05\\
0.279296875	8.4389171497416e-06\\
0.28125	7.63825247700212e-06\\
0.283203125	3.41637605911911e-06\\
0.28515625	1.98432808879525e-06\\
0.287109375	9.63431722623142e-07\\
0.2890625	9.36179668873927e-07\\
0.291015625	7.53206362777312e-07\\
0.29296875	1.76756132563597e-07\\
0.294921875	1.87165255402649e-07\\
0.296875	3.20349782427159e-07\\
0.298828125	5.93927779312853e-07\\
0.30078125	1.09674662251975e-06\\
0.302734375	2.22179706716014e-06\\
0.3046875	3.74047932975734e-06\\
0.306640625	1.23854479869802e-05\\
0.30859375	4.06517414115744e-05\\
0.310546875	0.000143153490718564\\
0.3125	0.000498581132055513\\
0.314453125	0.00240401769799766\\
0.31640625	0.00820569578783174\\
0.318359375	0.0133763526348\\
0.3203125	0.0202632576098596\\
0.322265625	0.047667112242664\\
0.32421875	0.109013577097544\\
0.326171875	0.321755838232655\\
0.328125	0.487147726904867\\
0.330078125	0.594029353858869\\
0.33203125	0.685830017267992\\
0.333984375	1.77806514255686\\
0.3359375	6.09449224755393\\
0.337890625	4.6747950306968\\
0.33984375	2.59522792160823\\
0.341796875	1.06704975641803\\
0.34375	0.249092867581185\\
0.345703125	0.0636133511923552\\
0.34765625	0.0223612188323846\\
0.349609375	0.00924785680996947\\
0.3515625	0.00390320338190458\\
0.353515625	0.00131389645815157\\
0.35546875	0.000457596050543271\\
0.357421875	0.000132371725603911\\
0.359375	4.84376748255436e-05\\
0.361328125	1.6612523713552e-05\\
0.36328125	1.29582139121408e-05\\
0.365234375	1.13772742478226e-05\\
0.3671875	2.44852512609928e-05\\
0.369140625	2.7131626203239e-05\\
0.37109375	1.63119570728992e-05\\
0.373046875	1.79616060133666e-05\\
0.375	1.15548687611293e-05\\
0.376953125	4.52002165977451e-06\\
0.37890625	2.166294708005e-06\\
0.380859375	5.10558703983684e-07\\
0.3828125	2.12884637767128e-07\\
0.384765625	9.74066260634662e-08\\
0.38671875	5.85361113259631e-08\\
0.388671875	2.83767377645421e-08\\
0.390625	1.52011915737166e-08\\
0.392578125	1.1936749360642e-08\\
0.39453125	3.06945754148089e-08\\
0.396484375	7.44970771407376e-08\\
0.3984375	2.26288133263411e-07\\
0.400390625	6.63473334875595e-07\\
0.40234375	1.45597256142132e-06\\
0.404296875	3.09137665724703e-06\\
0.40625	5.01386979083683e-06\\
0.408203125	1.10223744932015e-05\\
0.41015625	3.8268019488978e-05\\
0.412109375	0.000162595100569628\\
0.4140625	0.000658218416682253\\
0.416015625	0.00158514047632685\\
0.41796875	0.00405575218416584\\
0.419921875	0.0139478486806686\\
0.421875	0.0332817824894809\\
0.423828125	0.0618020681344327\\
0.42578125	0.0854995621605103\\
0.427734375	0.106999168549836\\
0.4296875	0.207033556063541\\
0.431640625	0.696042016103793\\
0.43359375	1.48189890253128\\
0.435546875	3.39372839577508\\
0.4375	3.72381786511428\\
0.439453125	3.41334961232551\\
0.44140625	3.17061044929623\\
0.443359375	1.57143610818558\\
0.4453125	0.911703545383136\\
0.447265625	0.716138577815392\\
0.44921875	0.523129831102307\\
0.451171875	0.180270077294072\\
0.453125	0.0570857211674362\\
0.455078125	0.0135818651504632\\
0.45703125	0.00515940297116429\\
0.458984375	0.00219753048907636\\
0.4609375	0.00190714684347591\\
0.462890625	0.000904648384811591\\
0.46484375	0.000226123992981514\\
0.466796875	7.99893610511728e-05\\
0.46875	3.8511120135355e-05\\
0.470703125	2.63541523265782e-05\\
0.47265625	1.49499943174269e-05\\
0.474609375	5.67866586989574e-06\\
0.4765625	1.96826541270803e-06\\
0.478515625	2.11547862820977e-06\\
0.48046875	1.2732419275089e-06\\
0.482421875	5.36054009186349e-07\\
0.484375	2.35869466890897e-07\\
0.486328125	1.21051916011227e-07\\
0.48828125	5.81667204615349e-08\\
0.490234375	1.97026247251941e-08\\
0.4921875	1.57892697017224e-08\\
0.494140625	2.53127564449321e-08\\
0.49609375	2.05164072202198e-08\\
0.498046875	7.42655007654377e-09\\
0.5	2.46542359305256e-09\\
0.501953125	6.81464291205158e-10\\
0.50390625	1.97787350015247e-10\\
0.505859375	8.31241401125751e-11\\
0.5078125	2.07555795658423e-11\\
0.509765625	8.81299728163078e-12\\
0.51171875	5.2386314229587e-12\\
0.513671875	2.3972098852042e-12\\
0.515625	6.79168155058967e-13\\
0.517578125	2.507172638564e-13\\
0.51953125	8.56083990623089e-14\\
0.521484375	7.0315562223194e-14\\
0.5234375	8.16730945645216e-14\\
0.525390625	2.2505998698677e-13\\
0.52734375	9.85183633157219e-13\\
0.529296875	4.22484260678139e-12\\
0.53125	7.63847528616793e-12\\
0.533203125	1.6412997943184e-11\\
0.53515625	4.23261209168031e-11\\
0.537109375	1.50824322446655e-10\\
0.5390625	5.29822028571567e-10\\
0.541015625	1.0112361698e-09\\
0.54296875	3.15385150079193e-09\\
0.544921875	8.83683930115249e-09\\
0.546875	2.18149601667732e-08\\
0.548828125	2.63016277114207e-08\\
0.55078125	1.96495929569435e-08\\
0.552734375	2.24095662599032e-08\\
0.5546875	4.25546333632314e-08\\
0.556640625	1.2663864783121e-07\\
0.55859375	4.2672621898224e-07\\
0.560546875	1.83938820556077e-06\\
0.5625	4.46057591026256e-06\\
0.564453125	7.35219615864552e-06\\
0.56640625	9.45081946054089e-06\\
0.568359375	1.79948967993616e-05\\
0.5703125	7.04947827639168e-05\\
0.572265625	0.000373506690193137\\
0.57421875	0.000888736254320739\\
0.576171875	0.000953812390800599\\
0.578125	0.00155327392265331\\
0.580078125	0.00256998848934774\\
0.58203125	0.00818836045273891\\
0.583984375	0.0176751184061377\\
0.5859375	0.044200687518012\\
0.587890625	0.181165656765024\\
0.58984375	0.511338866123522\\
0.591796875	0.700743633355694\\
0.59375	1.42457189181276\\
0.595703125	2.53972942916609\\
0.59765625	3.66816709796045\\
0.599609375	2.20899970116022\\
0.6015625	1.33111073319335\\
0.603515625	1.3032635269907\\
0.60546875	1.3938197540093\\
0.607421875	1.62453971241449\\
0.609375	0.860139725923514\\
0.611328125	0.238816043251585\\
0.61328125	0.0820093280272051\\
0.615234375	0.0658220983114887\\
0.6171875	0.0439911962293546\\
0.619140625	0.0268581398484018\\
0.62109375	0.0131235705469385\\
0.623046875	0.00483181402440032\\
0.625	0.00174090981530066\\
0.626953125	0.000737078893115455\\
0.62890625	0.000177335942566519\\
0.630859375	6.51124134337199e-05\\
0.6328125	2.71079955760562e-05\\
0.634765625	2.35148064060716e-05\\
0.63671875	5.1332081800211e-05\\
0.638671875	9.24104722006089e-05\\
0.640625	0.000133231073556145\\
0.642578125	0.0003268543762495\\
0.64453125	0.000724139184860052\\
0.646484375	0.000236208175633904\\
0.6484375	8.26689634783172e-05\\
0.650390625	6.55367754025632e-05\\
0.65234375	7.1463711592023e-05\\
0.654296875	4.42168413482533e-05\\
0.65625	1.50353379391858e-05\\
0.658203125	2.02807898496467e-05\\
0.66015625	1.60307142894135e-05\\
0.662109375	1.32392981630385e-05\\
0.6640625	1.10751958913477e-05\\
0.666015625	6.21677736030972e-06\\
0.66796875	2.54785435690443e-06\\
0.669921875	1.99224813488211e-06\\
0.671875	1.65506324775053e-06\\
0.673828125	6.86325093391087e-07\\
0.67578125	2.17170225163946e-07\\
0.677734375	5.04906955441936e-08\\
0.6796875	1.25361626508079e-08\\
0.681640625	3.16167780870335e-09\\
0.68359375	9.920506007736e-10\\
0.685546875	1.40229997028071e-09\\
0.6875	5.10063899150325e-09\\
0.689453125	9.87365722835862e-09\\
0.69140625	3.31545672177484e-08\\
0.693359375	9.17399601814066e-08\\
0.6953125	3.15476945364507e-07\\
0.697265625	9.88356019268878e-07\\
0.69921875	4.83184597132839e-06\\
0.701171875	1.07766816734039e-05\\
0.703125	2.35317096750566e-05\\
0.705078125	2.33220019570512e-05\\
0.70703125	4.11781692989103e-05\\
0.708984375	9.6525449975323e-05\\
0.7109375	0.00028923913044389\\
0.712890625	0.000387338518710643\\
0.71484375	0.001084513018469\\
0.716796875	0.0035785837012005\\
0.71875	0.00556500376756529\\
0.720703125	0.00661065550255707\\
0.72265625	0.0194466275285401\\
0.724609375	0.0365350648248803\\
0.7265625	0.174110451248978\\
0.728515625	0.400238312349858\\
0.73046875	0.463207953515001\\
0.732421875	1.30586297813432\\
0.734375	2.62235633367489\\
0.736328125	1.87563464361028\\
0.73828125	0.957793028136578\\
0.740234375	0.211188502531352\\
0.7421875	0.0658240233070012\\
0.744140625	0.0144685534338177\\
0.74609375	0.00699274856617747\\
0.748046875	0.00256028232599827\\
0.75	0.000942099914632245\\
0.751953125	0.000249361744737926\\
0.75390625	0.000140924017037974\\
0.755859375	0.000393074334993441\\
0.7578125	0.000454473838366963\\
0.759765625	0.00124290343082285\\
0.76171875	0.00471803947541207\\
0.763671875	0.0113380514608551\\
0.765625	0.0228823878172801\\
0.767578125	0.0812431888398911\\
0.76953125	0.129446766959709\\
0.771484375	0.182499631924314\\
0.7734375	0.509158424869408\\
0.775390625	1.71050662398563\\
0.77734375	2.7557099458518\\
0.779296875	1.84709381102983\\
0.78125	1.5040888834587\\
0.783203125	1.00433807594762\\
0.78515625	0.724068481297408\\
0.787109375	2.26601454638405\\
0.7890625	2.22435216737016\\
0.791015625	0.607138730625932\\
0.79296875	0.210122906031994\\
0.794921875	0.069869407788836\\
0.796875	0.0239288893404574\\
0.798828125	0.00941695377742702\\
0.80078125	0.00222514961854598\\
0.802734375	0.000498152816606741\\
0.8046875	0.000140894647798302\\
0.806640625	4.73755567252433e-05\\
0.80859375	2.44343118972106e-05\\
0.810546875	6.36199176978287e-06\\
0.8125	3.59466617231806e-06\\
0.814453125	2.15992436298575e-06\\
0.81640625	2.97317544738797e-06\\
0.818359375	2.65563742422575e-06\\
0.8203125	8.34401972865573e-06\\
0.822265625	2.01836819471215e-05\\
0.82421875	4.56614067747765e-05\\
0.826171875	6.26015471398882e-05\\
0.828125	0.000166429223655793\\
0.830078125	0.000283969832065679\\
0.83203125	0.000511412771193214\\
0.833984375	0.00143396105786584\\
0.8359375	0.00414649790954093\\
0.837890625	0.0077097214700102\\
0.83984375	0.0108614201397092\\
0.841796875	0.0446626898426752\\
0.84375	0.124935334744943\\
0.845703125	0.144607582151085\\
0.84765625	0.116341019068656\\
0.849609375	0.157893761941426\\
0.8515625	0.238670572566446\\
0.853515625	0.262820551990726\\
0.85546875	0.529692383740274\\
0.857421875	0.730855237451507\\
0.859375	0.435085594416042\\
0.861328125	0.176438380842146\\
0.86328125	0.118006957521716\\
0.865234375	0.0670253281550562\\
0.8671875	0.0418222534557669\\
0.869140625	0.0462537220770821\\
0.87109375	0.0275262468296616\\
0.873046875	0.0191065012722431\\
0.875	0.00737574002829848\\
0.876953125	0.00208376899708119\\
0.87890625	0.00063756210272918\\
0.880859375	0.000117957916570402\\
0.8828125	7.78386920058285e-05\\
0.884765625	4.6768781769863e-05\\
0.88671875	1.14080525686568e-05\\
0.888671875	7.65097635227165e-06\\
0.890625	5.11837459559873e-06\\
0.892578125	1.64808327953278e-06\\
0.89453125	3.6590634334882e-07\\
0.896484375	1.13529972761581e-07\\
0.8984375	5.15378834094215e-08\\
0.900390625	1.8348157955957e-08\\
0.90234375	7.37254084590626e-09\\
0.904296875	2.22219227050776e-09\\
0.90625	1.02786088660747e-09\\
0.908203125	9.07700240842231e-10\\
0.91015625	6.2900235194148e-10\\
0.912109375	1.31522194180736e-09\\
0.9140625	5.43884242977375e-09\\
0.916015625	1.4061792097396e-08\\
0.91796875	3.42235213129868e-08\\
0.919921875	1.04843967026619e-07\\
0.921875	2.5995743815243e-07\\
0.923828125	7.35341550723244e-07\\
0.92578125	3.54594368641166e-06\\
0.927734375	6.96281621307467e-06\\
0.9296875	1.6241015650552e-05\\
0.931640625	3.59255534026351e-05\\
0.93359375	7.77640871137183e-05\\
0.935546875	0.000136597862116081\\
0.9375	0.000193377337186282\\
0.939453125	0.000725182459240694\\
0.94140625	0.00427425587088481\\
0.943359375	0.0135763247509756\\
0.9453125	0.0525727141373932\\
0.947265625	0.106141823718654\\
0.94921875	0.346649749376576\\
0.951171875	1.19274624231315\\
0.953125	2.55857615962053\\
0.955078125	1.75039931171972\\
0.95703125	1.28305892839513\\
0.958984375	1.05800463376068\\
0.9609375	0.437526727891138\\
0.962890625	0.243273027465583\\
0.96484375	0.0524772473622869\\
0.966796875	0.0167341023583026\\
0.96875	0.00643967900720218\\
0.970703125	0.00173254953291853\\
0.97265625	0.000597130654317804\\
0.974609375	0.000374298341984268\\
0.9765625	0.000354944594077712\\
0.978515625	0.000325077950846076\\
0.98046875	0.000792706243899814\\
0.982421875	0.00203625587303814\\
0.984375	0.00717145027423967\\
0.986328125	0.0218367139918049\\
0.98828125	0.0615492811131685\\
0.990234375	0.0961628126638313\\
0.9921875	0.34009749126246\\
0.994140625	0.754229329896628\\
0.99609375	2.28383994115825\\
0.998046875	2.66176219640194\\
1	0\\
};
\end{axis}
\end{tikzpicture}%

%% file: pics/square/GroundState_1_2_4.tex
%
%
\begin{tikzpicture}

\begin{axis}[%
width=0.9in,
height=0.9in,
scale only axis,
xmin=0,
xmax=1,
xtick={0,0.5,1},
ymin=0,
ymax=10,
yminorticks=true,
ytick={0,4,8},
yticklabel pos=right,
axis background/.style={fill=white},
]
\addplot [color=mycolor2, thick]
table[row sep=crcr]{%
0	0\\
0.001953125	0.21636182624134\\
0.00390625	0.901259828114147\\
0.005859375	3.31756686269635\\
0.0078125	3.59070872135663\\
0.009765625	3.068048667907\\
0.01171875	1.65935506922161\\
0.013671875	5.24195715915105\\
0.015625	4.10785668284237\\
0.017578125	1.02208015789688\\
0.01953125	0.340017634108885\\
0.021484375	0.17257927476763\\
0.0234375	0.0676571883462081\\
0.025390625	0.0354958296032754\\
0.02734375	0.0159918491286547\\
0.029296875	0.00621635430916308\\
0.03125	0.00856731718972598\\
0.033203125	0.024450665824341\\
0.03515625	0.0648554503930124\\
0.037109375	0.133029108036138\\
0.0390625	0.294788445097441\\
0.041015625	0.795934863995685\\
0.04296875	4.0407904932856\\
0.044921875	4.534549374619\\
0.046875	1.52886407527521\\
0.048828125	1.37626997366256\\
0.05078125	2.7392569424792\\
0.052734375	5.99912116149522\\
0.0546875	7.7999452562623\\
0.056640625	3.37143094059505\\
0.05859375	1.97627802644249\\
0.060546875	2.19869187620007\\
0.0625	1.56923781889523\\
0.064453125	1.01829825029268\\
0.06640625	0.436978218208186\\
0.068359375	0.205936534172027\\
0.0703125	0.0672200705480167\\
0.072265625	0.0260371689828396\\
0.07421875	0.0094379382840068\\
0.076171875	0.00938399876377879\\
0.078125	0.00474937002261333\\
0.080078125	0.00176430381571077\\
0.08203125	0.000487834581673244\\
0.083984375	0.000192672590795452\\
0.0859375	0.000286515448635695\\
0.087890625	0.000827801513929268\\
0.08984375	0.00253108596956439\\
0.091796875	0.00622024863861079\\
0.09375	0.0133896952829201\\
0.095703125	0.0244890954351525\\
0.09765625	0.0622600097421096\\
0.099609375	0.106845394420185\\
0.1015625	0.267672345901581\\
0.103515625	0.677813824285096\\
0.10546875	0.969163032671782\\
0.107421875	2.88039069568413\\
0.109375	1.57987281778039\\
0.111328125	0.632186666443019\\
0.11328125	0.382836589168765\\
0.115234375	0.10226574332357\\
0.1171875	0.0216679370283498\\
0.119140625	0.00589507223060134\\
0.12109375	0.00208719272559739\\
0.123046875	0.00147315820210093\\
0.125	0.00247464235168615\\
0.126953125	0.00133489176198935\\
0.12890625	0.000609406502152757\\
0.130859375	0.000182845308353415\\
0.1328125	7.57425070699271e-05\\
0.134765625	2.95468386425418e-05\\
0.13671875	1.47094581983447e-05\\
0.138671875	1.71987486560789e-05\\
0.140625	2.71247015186837e-05\\
0.142578125	6.24380637796945e-06\\
0.14453125	1.33173916203356e-06\\
0.146484375	5.11092782557648e-07\\
0.1484375	1.59493501225878e-07\\
0.150390625	7.48444323750011e-08\\
0.15234375	2.36639904691341e-08\\
0.154296875	1.45493580188821e-08\\
0.15625	5.03090654301588e-09\\
0.158203125	2.34819076195042e-09\\
0.16015625	7.69487328361489e-10\\
0.162109375	2.65632176052187e-10\\
0.1640625	3.13075128325227e-10\\
0.166015625	8.6671255750098e-10\\
0.16796875	2.06783015275543e-09\\
0.169921875	7.3874832575239e-09\\
0.171875	1.37878447695878e-08\\
0.173828125	3.59168364126317e-08\\
0.17578125	1.05758905132101e-07\\
0.177734375	2.50873007745742e-07\\
0.1796875	9.01426818781838e-07\\
0.181640625	2.70218614588295e-06\\
0.18359375	5.36538323805946e-06\\
0.185546875	2.49545994231038e-05\\
0.1875	6.49528878530025e-05\\
0.189453125	0.000140405830259256\\
0.19140625	0.000318284939231591\\
0.193359375	0.000510182157471911\\
0.1953125	0.000779107319466171\\
0.197265625	0.00280769651807231\\
0.19921875	0.0110653541998211\\
0.201171875	0.0320379116186487\\
0.203125	0.0582298939951416\\
0.205078125	0.118384112441791\\
0.20703125	0.199148491944338\\
0.208984375	0.32553718768516\\
0.2109375	0.624185126671617\\
0.212890625	0.913929902446497\\
0.21484375	2.3875971911205\\
0.216796875	4.60557873734423\\
0.21875	4.68279836045538\\
0.220703125	2.23787314324621\\
0.22265625	1.06689017330436\\
0.224609375	0.689925436878061\\
0.2265625	0.234119825398759\\
0.228515625	0.119886703761276\\
0.23046875	0.0463188611224159\\
0.232421875	0.0179837584586097\\
0.234375	0.00692710610297174\\
0.236328125	0.00303103763319815\\
0.23828125	0.00126976099360762\\
0.240234375	0.000888198834706267\\
0.2421875	0.000245397073176144\\
0.244140625	0.0001367235341826\\
0.24609375	3.50167062770075e-05\\
0.248046875	1.45376762084593e-05\\
0.25	1.24793756782837e-05\\
0.251953125	6.41086961890986e-06\\
0.25390625	6.09469367256615e-06\\
0.255859375	8.70280707508077e-06\\
0.2578125	2.88809174491078e-06\\
0.259765625	7.54345630296625e-07\\
0.26171875	2.05606956740697e-07\\
0.263671875	7.72366982899904e-08\\
0.265625	3.50603815648722e-08\\
0.267578125	1.26334270172652e-08\\
0.26953125	5.04509671625731e-09\\
0.271484375	1.45900764351175e-09\\
0.2734375	4.77504138238515e-10\\
0.275390625	1.1226480152121e-10\\
0.27734375	8.18249806617141e-11\\
0.279296875	8.55666182326622e-11\\
0.28125	1.42243740444618e-10\\
0.283203125	2.02456547726927e-10\\
0.28515625	3.28458655802252e-10\\
0.287109375	8.95250373919633e-10\\
0.2890625	2.16570687228888e-09\\
0.291015625	2.88032746837518e-09\\
0.29296875	4.9299917941544e-09\\
0.294921875	2.04359690244374e-08\\
0.296875	5.10540253361635e-08\\
0.298828125	1.12333975375775e-07\\
0.30078125	2.3119636240996e-07\\
0.302734375	5.22097529273417e-07\\
0.3046875	9.93270829775615e-07\\
0.306640625	3.59463799432927e-06\\
0.30859375	1.25979505911302e-05\\
0.310546875	4.69923757174797e-05\\
0.3125	0.000174092248999315\\
0.314453125	0.000883793345118563\\
0.31640625	0.00318490769063733\\
0.318359375	0.00568781106079503\\
0.3203125	0.00992249456246272\\
0.322265625	0.0260165260361951\\
0.32421875	0.0655092848450181\\
0.326171875	0.207629973472225\\
0.328125	0.345020254530465\\
0.330078125	0.502703657174221\\
0.33203125	0.720276072294575\\
0.333984375	2.1492255710994\\
0.3359375	7.74958733935615\\
0.337890625	5.84597374989657\\
0.33984375	2.98624701338421\\
0.341796875	1.16657967375445\\
0.34375	0.257988281771517\\
0.345703125	0.0621785742651546\\
0.34765625	0.0202448035326454\\
0.349609375	0.00770116678170594\\
0.3515625	0.00300335110746285\\
0.353515625	0.000945635442061651\\
0.35546875	0.000308823056787852\\
0.357421875	8.32760900081099e-05\\
0.359375	2.8233370202428e-05\\
0.361328125	8.2627221774164e-06\\
0.36328125	4.01944762298172e-06\\
0.365234375	1.92150651430136e-06\\
0.3671875	1.84111530875091e-06\\
0.369140625	1.57473589983578e-06\\
0.37109375	7.536362240895e-07\\
0.373046875	6.51280296299782e-07\\
0.375	3.68580834099657e-07\\
0.376953125	1.31958901776022e-07\\
0.37890625	5.70310192647146e-08\\
0.380859375	1.26129838642425e-08\\
0.3828125	4.70374115227495e-09\\
0.384765625	1.92507246488563e-09\\
0.38671875	1.09029327535706e-09\\
0.388671875	6.32363989504844e-10\\
0.390625	6.57567591919402e-10\\
0.392578125	1.25785179967589e-09\\
0.39453125	4.5313979295189e-09\\
0.396484375	1.22499500243308e-08\\
0.3984375	4.03156359074444e-08\\
0.400390625	1.26039046376941e-07\\
0.40234375	2.992556108702e-07\\
0.404296875	6.99515827044933e-07\\
0.40625	1.26832409521295e-06\\
0.408203125	3.13030795262774e-06\\
0.41015625	1.16265557753055e-05\\
0.412109375	5.24032005584991e-05\\
0.4140625	0.000223700939594154\\
0.416015625	0.000575185307390057\\
0.41796875	0.00160040086577847\\
0.419921875	0.00585267026430217\\
0.421875	0.0151111185990688\\
0.423828125	0.0311324100064046\\
0.42578125	0.0485175049728356\\
0.427734375	0.0738769382745795\\
0.4296875	0.166853402433139\\
0.431640625	0.609682140843164\\
0.43359375	1.40034650072649\\
0.435546875	3.4143649640058\\
0.4375	3.81249706917517\\
0.439453125	3.45003775991841\\
0.44140625	3.11976044463796\\
0.443359375	1.40513494713846\\
0.4453125	0.706562754947258\\
0.447265625	0.457328415805081\\
0.44921875	0.304190411290312\\
0.451171875	0.0979330660938558\\
0.453125	0.0291891895777791\\
0.455078125	0.00656024031816197\\
0.45703125	0.00227699287012509\\
0.458984375	0.000847351771347421\\
0.4609375	0.000637825299666371\\
0.462890625	0.000280736959300623\\
0.46484375	6.62902630264757e-05\\
0.466796875	2.15065223899139e-05\\
0.46875	9.18335013820295e-06\\
0.470703125	5.41986302233208e-06\\
0.47265625	2.76527104538187e-06\\
0.474609375	9.77502157598139e-07\\
0.4765625	2.88561132653061e-07\\
0.478515625	2.53573905864085e-07\\
0.48046875	1.35571563397315e-07\\
0.482421875	5.30100161789337e-08\\
0.484375	2.10833194124604e-08\\
0.486328125	9.746716977919e-09\\
0.48828125	4.29091878663123e-09\\
0.490234375	1.19089328670911e-09\\
0.4921875	6.14166936394681e-10\\
0.494140625	7.20550347587912e-10\\
0.49609375	5.19840992907255e-10\\
0.498046875	1.75672798773401e-10\\
0.5	5.43358551387258e-11\\
0.501953125	1.42138097283903e-11\\
0.50390625	3.87545870462915e-12\\
0.505859375	1.50440209332964e-12\\
0.5078125	3.53865903118013e-13\\
0.509765625	1.34435698834137e-13\\
0.51171875	7.18658546767284e-14\\
0.513671875	3.04662926818792e-14\\
0.515625	8.05246433471816e-15\\
0.517578125	2.78748134788398e-15\\
0.51953125	8.95727699533625e-16\\
0.521484375	7.58255022928802e-16\\
0.5234375	1.0571077973765e-15\\
0.525390625	3.27418898170424e-15\\
0.52734375	1.52055905195928e-14\\
0.529296875	6.86888069889183e-14\\
0.53125	1.34756006958372e-13\\
0.533203125	3.21967238443259e-13\\
0.53515625	8.99090001379119e-13\\
0.537109375	3.44258799991819e-12\\
0.5390625	1.27881581047664e-11\\
0.541015625	2.67030217752877e-11\\
0.54296875	8.97478322928676e-11\\
0.544921875	2.69804967374119e-10\\
0.546875	7.14777509235153e-10\\
0.548828125	9.8090385955556e-10\\
0.55078125	1.06741732791508e-09\\
0.552734375	1.72945141873937e-09\\
0.5546875	4.0959446633706e-09\\
0.556640625	1.35640074008808e-08\\
0.55859375	4.87889149618823e-08\\
0.560546875	2.21532808321745e-07\\
0.5625	5.74284516609402e-07\\
0.564453125	1.06065060472581e-06\\
0.56640625	1.56786498262072e-06\\
0.568359375	3.46519327146647e-06\\
0.5703125	1.47062147768299e-05\\
0.572265625	8.16594487185474e-05\\
0.57421875	0.000207325927721971\\
0.576171875	0.00025720585217808\\
0.578125	0.000490359096246308\\
0.580078125	0.000952378346879884\\
0.58203125	0.00334146329948595\\
0.583984375	0.00779678980396282\\
0.5859375	0.0212777933629139\\
0.587890625	0.09310359619065\\
0.58984375	0.278854792974981\\
0.591796875	0.42627103383888\\
0.59375	0.984945157361275\\
0.595703125	1.85916871717318\\
0.59765625	2.79234791357257\\
0.599609375	1.58409612391728\\
0.6015625	0.771185303438189\\
0.603515625	0.552739104834642\\
0.60546875	0.477677105479815\\
0.607421875	0.484665345777789\\
0.609375	0.237781377792905\\
0.611328125	0.0623680118709747\\
0.61328125	0.0182265421617805\\
0.615234375	0.0122707066448738\\
0.6171875	0.00723432781526293\\
0.619140625	0.00397392977956649\\
0.62109375	0.00177300844591898\\
0.623046875	0.000601896684005887\\
0.625	0.000202376561090212\\
0.626953125	7.92837085196317e-05\\
0.62890625	1.80241221806485e-05\\
0.630859375	6.03433191585796e-06\\
0.6328125	1.75333917859482e-06\\
0.634765625	7.15880344695127e-07\\
0.63671875	2.97142613102063e-07\\
0.638671875	1.71444341578857e-07\\
0.640625	7.00062770811251e-08\\
0.642578125	4.80496869804573e-08\\
0.64453125	7.60814794214067e-08\\
0.646484375	2.30472921961266e-08\\
0.6484375	6.66038457406237e-09\\
0.650390625	3.945200119389e-09\\
0.65234375	3.33351565528903e-09\\
0.654296875	1.79812633830818e-09\\
0.65625	4.76415942887952e-10\\
0.658203125	4.11135716569179e-10\\
0.66015625	2.53062000334144e-10\\
0.662109375	1.68226810119095e-10\\
0.6640625	1.1537783666675e-10\\
0.666015625	5.76669083369902e-11\\
0.66796875	2.03849377703912e-11\\
0.669921875	1.20365774871443e-11\\
0.671875	8.57861171169843e-12\\
0.673828125	3.27569851895404e-12\\
0.67578125	9.60697770700086e-13\\
0.677734375	2.10926233812752e-13\\
0.6796875	5.01058685222868e-14\\
0.681640625	1.65165417975391e-14\\
0.68359375	2.52539413689197e-14\\
0.685546875	1.11824801186129e-13\\
0.6875	4.69854085802836e-13\\
0.689453125	9.95860029579184e-13\\
0.69140625	3.5506136935916e-12\\
0.693359375	1.04128007561321e-11\\
0.6953125	3.77144580312284e-11\\
0.697265625	1.24830051721293e-10\\
0.69921875	6.36098064713363e-10\\
0.701171875	1.50883150334495e-09\\
0.703125	3.5195888381323e-09\\
0.705078125	4.10962883559499e-09\\
0.70703125	8.29643809993419e-09\\
0.708984375	2.15275587024576e-08\\
0.7109375	6.83834999354304e-08\\
0.712890625	1.02492336395258e-07\\
0.71484375	3.13474188035781e-07\\
0.716796875	1.09058064622502e-06\\
0.71875	1.83487924186991e-06\\
0.720703125	2.53956085233831e-06\\
0.72265625	8.3598368561543e-06\\
0.724609375	1.6927102117567e-05\\
0.7265625	8.53239758116111e-05\\
0.728515625	0.000207539705376684\\
0.73046875	0.000270013927354252\\
0.732421875	0.000835792002379062\\
0.734375	0.00174096314586286\\
0.736328125	0.00121981933007137\\
0.73828125	0.00059729441000441\\
0.740234375	0.000126077889012393\\
0.7421875	3.73593637724855e-05\\
0.744140625	7.81694764851811e-06\\
0.74609375	3.51321356990559e-06\\
0.748046875	1.20897293566555e-06\\
0.75	4.22120410115774e-07\\
0.751953125	1.0902481370099e-07\\
0.75390625	7.04847377531275e-08\\
0.755859375	2.20714136909397e-07\\
0.7578125	2.93710337651475e-07\\
0.759765625	8.84538118821516e-07\\
0.76171875	3.54857508360149e-06\\
0.763671875	9.03154468431593e-06\\
0.765625	1.97564885993529e-05\\
0.767578125	7.52487629248513e-05\\
0.76953125	0.000129689999517065\\
0.771484375	0.000205438843880875\\
0.7734375	0.000631431608191876\\
0.775390625	0.00223463021740555\\
0.77734375	0.00364726304084404\\
0.779296875	0.00229745947521646\\
0.78125	0.00169674782411576\\
0.783203125	0.0010424838156356\\
0.78515625	0.00071019043280029\\
0.787109375	0.00225612978094276\\
0.7890625	0.00220843277753058\\
0.791015625	0.000581325675932536\\
0.79296875	0.000188551612109109\\
0.794921875	5.93045333409726e-05\\
0.796875	1.90683872362298e-05\\
0.798828125	7.09172977369554e-06\\
0.80078125	1.60501128901599e-06\\
0.802734375	3.4392511194861e-07\\
0.8046875	9.18657539817232e-08\\
0.806640625	2.87688188011655e-08\\
0.80859375	1.38888502978894e-08\\
0.810546875	3.35726557378192e-09\\
0.8125	1.55224437941757e-09\\
0.814453125	6.18029712346071e-10\\
0.81640625	4.56394197974855e-10\\
0.818359375	1.18818036236249e-10\\
0.8203125	4.64208823628129e-11\\
0.822265625	1.762955733088e-11\\
0.82421875	1.16352379426683e-11\\
0.826171875	3.90705384159883e-12\\
0.828125	2.25310354023443e-12\\
0.830078125	1.32714143984907e-12\\
0.83203125	3.58820767742263e-13\\
0.833984375	1.22421437422815e-13\\
0.8359375	6.37472922233204e-14\\
0.837890625	3.76069871815443e-14\\
0.83984375	9.63300656085941e-15\\
0.841796875	3.77649497363386e-15\\
0.84375	3.3293322891362e-15\\
0.845703125	2.73780816033546e-15\\
0.84765625	1.32910175469342e-15\\
0.849609375	7.7562420343657e-16\\
0.8515625	6.29515627460635e-16\\
0.853515625	3.45480680374741e-16\\
0.85546875	4.22797374268389e-16\\
0.857421875	5.03128646941567e-16\\
0.859375	2.81878828989686e-16\\
0.861328125	1.07069202746564e-16\\
0.86328125	6.52645203817718e-17\\
0.865234375	3.38368651130551e-17\\
0.8671875	1.85658603398721e-17\\
0.869140625	1.85115108049902e-17\\
0.87109375	1.01671960759248e-17\\
0.873046875	6.71530237018015e-18\\
0.875	2.51342754855796e-18\\
0.876953125	6.92588711169373e-19\\
0.87890625	2.13600179454829e-19\\
0.880859375	7.37969049783055e-20\\
0.8828125	1.98052481370657e-19\\
0.884765625	2.66904110314448e-19\\
0.88671875	6.17059089188812e-19\\
0.888671875	2.13627535963953e-18\\
0.890625	4.41844020364623e-18\\
0.892578125	6.3529834506465e-18\\
0.89453125	2.49166192744229e-17\\
0.896484375	8.38048489402297e-17\\
0.8984375	3.11744352570624e-16\\
0.900390625	7.78660364002412e-16\\
0.90234375	2.57792280464696e-15\\
0.904296875	5.97750593793086e-15\\
0.90625	2.46685182369266e-14\\
0.908203125	6.19269426340258e-14\\
0.91015625	9.91726226523629e-14\\
0.912109375	3.31970367217642e-13\\
0.9140625	1.48339992449195e-12\\
0.916015625	4.06916673357072e-12\\
0.91796875	1.05814965918778e-11\\
0.919921875	3.44975665226529e-11\\
0.921875	9.13654704347289e-11\\
0.923828125	2.75154932911564e-10\\
0.92578125	1.39182902483278e-09\\
0.927734375	2.91439783698049e-09\\
0.9296875	7.30510080302241e-09\\
0.931640625	1.76800853924431e-08\\
0.93359375	4.11447376501211e-08\\
0.935546875	7.94492004562727e-08\\
0.9375	1.27037193423015e-07\\
0.939453125	5.13947972125453e-07\\
0.94140625	3.1462597914502e-06\\
0.943359375	1.04790624382962e-05\\
0.9453125	4.27005581948531e-05\\
0.947265625	9.1816371977641e-05\\
0.94921875	0.000319326764097761\\
0.951171875	0.00115697059465619\\
0.953125	0.00259445027118111\\
0.955078125	0.00172287609634305\\
0.95703125	0.00110380583651324\\
0.958984375	0.000827806565302511\\
0.9609375	0.000320577780337292\\
0.962890625	0.000166071876001393\\
0.96484375	3.42884959427808e-05\\
0.966796875	1.0513640294903e-05\\
0.96875	4.43268359702565e-06\\
0.970703125	2.70269189268738e-06\\
0.97265625	7.12615107479693e-06\\
0.974609375	2.26455963155306e-05\\
0.9765625	4.71368299985179e-05\\
0.978515625	8.84874298404368e-05\\
0.98046875	0.000281188482283534\\
0.982421875	0.000797833295723905\\
0.984375	0.00300644832619414\\
0.986328125	0.00974063485983651\\
0.98828125	0.0294054964005302\\
0.990234375	0.0507474492566737\\
0.9921875	0.197570252820265\\
0.994140625	0.47220267692535\\
0.99609375	1.54088852674541\\
0.998046875	1.80891331652482\\
1	0\\
};
\end{axis}
\end{tikzpicture}%

%% file: pics/square/GroundState_1_3_1.tex
%
%
\begin{tikzpicture}

\begin{axis}[%
width=0.9in,
height=0.9in,
scale only axis,
xmin=0,
xmax=1,
xtick={0,0.5,1},
ymin=0,
ymax=10,
yminorticks=true,
ytick={0,4,8},
ylabel={$c_V = 4$},
ylabel style={below=0.1in},
axis background/.style={fill=white},
]
\addplot [color=mycolor2, thick]
table[row sep=crcr]{%
0	0\\
0.001953125	0.000271356158420013\\
0.00390625	0.0039252702951929\\
0.005859375	0.0287094510629891\\
0.0078125	0.0619645943171287\\
0.009765625	0.0871996505068127\\
0.01171875	0.320222250300181\\
0.013671875	2.76502231432726\\
0.015625	2.16913766121176\\
0.017578125	0.792029307915423\\
0.01953125	0.338363305109077\\
0.021484375	0.222763740403622\\
0.0234375	0.375756929293527\\
0.025390625	0.319125563900582\\
0.02734375	0.439177311251903\\
0.029296875	2.00697737000972\\
0.03125	3.0681622830505\\
0.033203125	1.43536365645197\\
0.03515625	3.60795044263825\\
0.037109375	4.83591533908368\\
0.0390625	1.93709638991761\\
0.041015625	1.41476854644371\\
0.04296875	1.93846673718405\\
0.044921875	0.535560411064743\\
0.046875	0.107725579857658\\
0.048828125	0.0194124234047317\\
0.05078125	0.0159251144008889\\
0.052734375	0.0472164228403041\\
0.0546875	0.195960253968932\\
0.056640625	0.591075352853966\\
0.05859375	0.592159060831045\\
0.060546875	0.108870604863791\\
0.0625	0.0245650575394517\\
0.064453125	0.00915101348308862\\
0.06640625	0.0021233441394874\\
0.068359375	0.00341289196417814\\
0.0703125	0.00761574827293765\\
0.072265625	0.00553787480190872\\
0.07421875	0.0129426103092625\\
0.076171875	0.0298427178049003\\
0.078125	0.0689138500680577\\
0.080078125	0.12919475437597\\
0.08203125	0.326591172247866\\
0.083984375	0.85618492775514\\
0.0859375	2.21697720306762\\
0.087890625	1.43084502559991\\
0.08984375	1.62475672294586\\
0.091796875	0.289033265322295\\
0.09375	0.0506825451664056\\
0.095703125	0.0138171225266511\\
0.09765625	0.00725366048552354\\
0.099609375	0.0012992505149552\\
0.1015625	0.000407097674739135\\
0.103515625	4.42204668697172e-05\\
0.10546875	2.04117244794402e-05\\
0.107421875	8.63652913551042e-06\\
0.109375	1.55269013219009e-06\\
0.111328125	2.89117148166197e-06\\
0.11328125	1.58588314669393e-05\\
0.115234375	3.79903540062162e-05\\
0.1171875	0.000169704244478095\\
0.119140625	0.000338190209037813\\
0.12109375	0.000842690060965986\\
0.123046875	0.00337452460341948\\
0.125	0.0117304951669655\\
0.126953125	0.0539066385841818\\
0.12890625	0.23572154515415\\
0.130859375	1.18193363967358\\
0.1328125	3.76314490187335\\
0.134765625	4.25143084114657\\
0.13671875	2.55094216475037\\
0.138671875	0.82855404594418\\
0.140625	0.199165857678981\\
0.142578125	0.0444366817036679\\
0.14453125	0.0118714230974438\\
0.146484375	0.00587038024506506\\
0.1484375	0.00196637390147524\\
0.150390625	0.00223648640624502\\
0.15234375	0.000821471195991966\\
0.154296875	0.00023186576153876\\
0.15625	3.84781184254819e-05\\
0.158203125	1.0579379756277e-05\\
0.16015625	3.52107106890502e-06\\
0.162109375	6.1326373981237e-06\\
0.1640625	2.67678921073956e-05\\
0.166015625	5.26015438061658e-05\\
0.16796875	0.000104709556356936\\
0.169921875	0.000343630439847438\\
0.171875	0.000797310542481773\\
0.173828125	0.00380230180847485\\
0.17578125	0.00679898033158105\\
0.177734375	0.0108894104621199\\
0.1796875	0.02412421533137\\
0.181640625	0.198220508635904\\
0.18359375	0.973920757817576\\
0.185546875	2.28212387337039\\
0.1875	1.6880551033706\\
0.189453125	0.34360193565008\\
0.19140625	0.0691401634352782\\
0.193359375	0.0146841555064181\\
0.1953125	0.00363330799302709\\
0.197265625	0.000802307423253185\\
0.19921875	0.000268697313187932\\
0.201171875	5.88738736159245e-05\\
0.203125	1.18588404082468e-05\\
0.205078125	2.47420132200925e-06\\
0.20703125	8.64999006306647e-07\\
0.208984375	1.2592196967754e-07\\
0.2109375	2.96806879658793e-08\\
0.212890625	1.36166311762946e-08\\
0.21484375	3.72888894528447e-09\\
0.216796875	1.26883211760034e-09\\
0.21875	4.83709894592895e-10\\
0.220703125	2.60311284148959e-10\\
0.22265625	7.83176565395883e-10\\
0.224609375	2.91551618283884e-09\\
0.2265625	7.2596812435038e-09\\
0.228515625	2.43538130432272e-08\\
0.23046875	1.33268191260303e-07\\
0.232421875	8.24585340550443e-07\\
0.234375	3.25848119739233e-06\\
0.236328125	1.23093733130909e-05\\
0.23828125	1.97532944500884e-05\\
0.240234375	6.51208819719154e-05\\
0.2421875	0.000344072008440579\\
0.244140625	0.0013575037658578\\
0.24609375	0.0122509715703166\\
0.248046875	0.0219653252152388\\
0.25	0.0461390343543626\\
0.251953125	0.120963128551371\\
0.25390625	0.0800041458364101\\
0.255859375	0.122792859138095\\
0.2578125	0.531845594957372\\
0.259765625	3.19395962983083\\
0.26171875	3.52336249444946\\
0.263671875	0.873145903865021\\
0.265625	0.185187540464959\\
0.267578125	0.34407364089488\\
0.26953125	0.0976586057507649\\
0.271484375	0.0332657185251623\\
0.2734375	0.0107464077116579\\
0.275390625	0.00248646434613812\\
0.27734375	0.000656346372483899\\
0.279296875	5.82802538028866e-05\\
0.28125	1.12265388310536e-05\\
0.283203125	3.48796080304222e-06\\
0.28515625	2.71628717464637e-07\\
0.287109375	1.39183083832106e-07\\
0.2890625	5.56988551969726e-07\\
0.291015625	1.43325112625302e-06\\
0.29296875	3.56019631223789e-06\\
0.294921875	2.21249675808553e-05\\
0.296875	4.35995722893229e-05\\
0.298828125	7.46115666599864e-05\\
0.30078125	0.000106184315775474\\
0.302734375	0.000109098581713642\\
0.3046875	0.000532265978511624\\
0.306640625	0.00157477478110728\\
0.30859375	0.00426069067657606\\
0.310546875	0.0144351397467942\\
0.3125	0.203704642302994\\
0.314453125	0.58111101299051\\
0.31640625	1.22536092902235\\
0.318359375	2.71673800445379\\
0.3203125	2.79695918098668\\
0.322265625	3.63278201380488\\
0.32421875	3.89574444750912\\
0.326171875	2.23351774801199\\
0.328125	1.13478308131004\\
0.330078125	0.261052584258101\\
0.33203125	0.0482595300713351\\
0.333984375	0.00621044081484223\\
0.3359375	0.0014726151797099\\
0.337890625	0.000900285981879041\\
0.33984375	0.00584220115785069\\
0.341796875	0.0169155618318681\\
0.34375	0.0885725941868949\\
0.345703125	0.524701906478318\\
0.34765625	1.36682865354044\\
0.349609375	1.19621905013896\\
0.3515625	0.827284376495194\\
0.353515625	1.83745990369728\\
0.35546875	2.55358095993983\\
0.357421875	2.53227644431753\\
0.359375	1.17684261814675\\
0.361328125	0.442551451006142\\
0.36328125	0.489775459359667\\
0.365234375	2.8661011075028\\
0.3671875	3.54059953457303\\
0.369140625	0.570649039639903\\
0.37109375	0.178648193316841\\
0.373046875	0.151493446302652\\
0.375	0.900278239398949\\
0.376953125	0.755556771368446\\
0.37890625	0.230546964345225\\
0.380859375	0.0755739230841365\\
0.3828125	0.0172167032272817\\
0.384765625	0.00474946174448129\\
0.38671875	0.00116981387605089\\
0.388671875	0.000605486865970512\\
0.390625	0.000201113171659125\\
0.392578125	7.68203603247029e-05\\
0.39453125	3.55220673648055e-05\\
0.396484375	2.73595505156373e-05\\
0.3984375	0.000106821925102847\\
0.400390625	0.000736869625933063\\
0.40234375	0.00559139769549268\\
0.404296875	0.0156730267865745\\
0.40625	0.0354034412953623\\
0.408203125	0.0883365323124329\\
0.41015625	0.181229725790376\\
0.412109375	0.907337731406033\\
0.4140625	4.16688977691675\\
0.416015625	3.1973657615599\\
0.41796875	1.66303638650022\\
0.419921875	3.54636149610637\\
0.421875	2.72762795734137\\
0.423828125	1.86981895577335\\
0.42578125	0.900383915607715\\
0.427734375	0.780507722439091\\
0.4296875	1.26216703815422\\
0.431640625	0.444005373751568\\
0.43359375	0.128290424623979\\
0.435546875	0.0301058356546617\\
0.4375	0.0159969293015885\\
0.439453125	0.0481856092356917\\
0.44140625	0.212908006610076\\
0.443359375	0.282596658216304\\
0.4453125	0.467850785366757\\
0.447265625	1.90217413482401\\
0.44921875	2.06991536706082\\
0.451171875	0.67214218708837\\
0.453125	0.171435233337995\\
0.455078125	0.0412150423457744\\
0.45703125	0.00892632874574041\\
0.458984375	0.00321173935200508\\
0.4609375	0.00227899629216523\\
0.462890625	0.00203122528889466\\
0.46484375	0.000596913369753372\\
0.466796875	0.000177096935152495\\
0.46875	9.91239133018924e-05\\
0.470703125	8.90516880680679e-05\\
0.47265625	1.81624776975512e-05\\
0.474609375	1.65385358981714e-06\\
0.4765625	2.83242829540848e-07\\
0.478515625	6.57314403160632e-08\\
0.48046875	1.77478357610179e-08\\
0.482421875	3.17631100986007e-09\\
0.484375	4.23603878184093e-10\\
0.486328125	2.47669962713843e-10\\
0.48828125	7.61606021024689e-10\\
0.490234375	2.07199549803841e-09\\
0.4921875	3.33942899457909e-09\\
0.494140625	4.03896897436327e-09\\
0.49609375	1.62585867616032e-08\\
0.498046875	1.14595281935452e-07\\
0.5	4.00772682518349e-07\\
0.501953125	5.37780516170735e-07\\
0.50390625	1.86602820516753e-06\\
0.505859375	1.16661439794993e-05\\
0.5078125	2.14085516392617e-05\\
0.509765625	2.88743011866477e-05\\
0.51171875	0.000202285428227234\\
0.513671875	0.00140519362145965\\
0.515625	0.00351693832314058\\
0.517578125	0.00976210996997546\\
0.51953125	0.0295715391908919\\
0.521484375	0.120146138619857\\
0.5234375	0.350885211358654\\
0.525390625	0.331724083834095\\
0.52734375	0.709968306445886\\
0.529296875	1.6236012325035\\
0.53125	2.52678451353059\\
0.533203125	1.48553863092758\\
0.53515625	0.192668891082351\\
0.537109375	0.0634115849204598\\
0.5390625	0.00870746597274277\\
0.541015625	0.00109442574056308\\
0.54296875	0.000201034699078923\\
0.544921875	0.000173420598757502\\
0.546875	0.000749743185164012\\
0.548828125	0.00241324241499822\\
0.55078125	0.00732435445001177\\
0.552734375	0.0250960229108318\\
0.5546875	0.163838789323235\\
0.556640625	0.716868090594589\\
0.55859375	2.23098857302211\\
0.560546875	2.88048540769734\\
0.5625	2.35068789233751\\
0.564453125	0.556257001680352\\
0.56640625	0.0938357877205096\\
0.568359375	0.0190533283440792\\
0.5703125	0.0036386009144454\\
0.572265625	0.00149095498664529\\
0.57421875	0.000840013969961367\\
0.576171875	8.86820436995065e-05\\
0.578125	3.5457760723298e-05\\
0.580078125	0.000116829379055611\\
0.58203125	0.000537856018895725\\
0.583984375	0.00266433911420987\\
0.5859375	0.013701915385524\\
0.587890625	0.0440482147273556\\
0.58984375	0.29349014833838\\
0.591796875	0.642669501327778\\
0.59375	0.620276498263596\\
0.595703125	2.22534752166586\\
0.59765625	3.69321522276903\\
0.599609375	3.33911942477503\\
0.6015625	1.62343454793621\\
0.603515625	0.666254038261871\\
0.60546875	0.14268449296418\\
0.607421875	0.0178482095181371\\
0.609375	0.00408103439467621\\
0.611328125	0.00162073746728346\\
0.61328125	0.00020994157502261\\
0.615234375	5.8882248724858e-05\\
0.6171875	1.09286640746506e-05\\
0.619140625	2.02234070581339e-05\\
0.62109375	4.56491643663131e-05\\
0.623046875	9.81550357937861e-05\\
0.625	0.00024422598331565\\
0.626953125	0.00100664040975061\\
0.62890625	0.00368506730818883\\
0.630859375	0.006964731279874\\
0.6328125	0.0282603054020806\\
0.634765625	0.0899937584647928\\
0.63671875	0.269794278720119\\
0.638671875	0.8460822232229\\
0.640625	0.827952012065887\\
0.642578125	0.212586818177014\\
0.64453125	0.0674062882302952\\
0.646484375	0.0359343437249166\\
0.6484375	0.00617518445316021\\
0.650390625	0.00271898208184205\\
0.65234375	0.000960605532595553\\
0.654296875	0.000947234756255451\\
0.65625	0.000591783567500679\\
0.658203125	0.00104549129673438\\
0.66015625	0.000473968747549685\\
0.662109375	0.000638657837446087\\
0.6640625	0.000569489417651769\\
0.666015625	0.000646863704150776\\
0.66796875	0.00558881861967532\\
0.669921875	0.0115974797911031\\
0.671875	0.0286398373528043\\
0.673828125	0.0719432505876209\\
0.67578125	0.111339249160788\\
0.677734375	0.135810645792334\\
0.6796875	0.183802928322632\\
0.681640625	0.661298422242861\\
0.68359375	0.628496602534366\\
0.685546875	0.477162795354311\\
0.6875	0.158854989532551\\
0.689453125	0.0791681785366337\\
0.69140625	0.117599637385532\\
0.693359375	0.367386409680773\\
0.6953125	0.557510264230152\\
0.697265625	0.417631312390067\\
0.69921875	0.691279932522791\\
0.701171875	3.20875707259201\\
0.703125	2.19555144070348\\
0.705078125	0.556070662650876\\
0.70703125	0.052056354736605\\
0.708984375	0.00666054809745612\\
0.7109375	0.00166905759929652\\
0.712890625	0.000499859414031988\\
0.71484375	0.000191091461466257\\
0.716796875	6.56544835375548e-05\\
0.71875	0.000166481932494943\\
0.720703125	0.00142675341236146\\
0.72265625	0.0049712929799933\\
0.724609375	0.0100731435942793\\
0.7265625	0.0163408707266999\\
0.728515625	0.0529083414473096\\
0.73046875	0.248631004799972\\
0.732421875	0.639392362380299\\
0.734375	0.324851834445749\\
0.736328125	0.179902071806194\\
0.73828125	0.0539620165427907\\
0.740234375	0.296490172965881\\
0.7421875	1.06736757883114\\
0.744140625	1.66042397387769\\
0.74609375	3.86916943092441\\
0.748046875	1.18131858783339\\
0.75	0.263975908861692\\
0.751953125	0.0649375954693841\\
0.75390625	0.00895535330501328\\
0.755859375	0.00126893292611404\\
0.7578125	0.000665101932835109\\
0.759765625	0.000192042021248784\\
0.76171875	4.18873379746587e-05\\
0.763671875	6.14042496342092e-05\\
0.765625	0.000190127380529218\\
0.767578125	0.00062325936286026\\
0.76953125	0.00308316901442473\\
0.771484375	0.0223278878395738\\
0.7734375	0.0987766528453704\\
0.775390625	0.252319153768488\\
0.77734375	1.02478754166087\\
0.779296875	3.22053371140287\\
0.78125	1.86536841016157\\
0.783203125	0.879814282214179\\
0.78515625	0.395231281220951\\
0.787109375	0.241042913529282\\
0.7890625	0.0654265034558408\\
0.791015625	0.0216599802773139\\
0.79296875	0.00497214808272797\\
0.794921875	0.00376345455366281\\
0.796875	0.0031927667314263\\
0.798828125	0.000637730110790782\\
0.80078125	0.000157013739267907\\
0.802734375	7.79467317756322e-05\\
0.8046875	2.19986043411961e-05\\
0.806640625	4.89548071082243e-06\\
0.80859375	9.20608651080282e-07\\
0.810546875	3.0915185449655e-07\\
0.8125	6.9955709405282e-08\\
0.814453125	3.97769850848649e-08\\
0.81640625	1.51259202922444e-08\\
0.818359375	4.69535174528801e-09\\
0.8203125	3.26115531920085e-09\\
0.822265625	9.62937672099389e-10\\
0.82421875	1.93018220380093e-09\\
0.826171875	3.81749223564649e-09\\
0.828125	1.36413782204623e-08\\
0.830078125	1.43247540797828e-08\\
0.83203125	8.84762496627467e-09\\
0.833984375	4.13444156126149e-08\\
0.8359375	1.01744782210025e-07\\
0.837890625	4.52907971712379e-08\\
0.83984375	8.94667460749196e-08\\
0.841796875	1.19362897274339e-07\\
0.84375	3.64266907667756e-08\\
0.845703125	6.85771157412419e-08\\
0.84765625	1.74614087534102e-07\\
0.849609375	1.6940197035362e-06\\
0.8515625	7.71555495112471e-06\\
0.853515625	1.80002513659125e-05\\
0.85546875	4.98934202329997e-05\\
0.857421875	0.000189676780606494\\
0.859375	0.000715583544673895\\
0.861328125	0.00131813674317132\\
0.86328125	0.000710803313574006\\
0.865234375	0.00292253145504783\\
0.8671875	0.0051486006949121\\
0.869140625	0.00328281620836863\\
0.87109375	0.00980021358018054\\
0.873046875	0.0125076278363069\\
0.875	0.0277328952313434\\
0.876953125	0.0886631052084286\\
0.87890625	0.426896508168867\\
0.880859375	1.2140100038012\\
0.8828125	2.1768481471312\\
0.884765625	3.52049904835752\\
0.88671875	2.03098634503838\\
0.888671875	0.274038393145679\\
0.890625	0.073035404195222\\
0.892578125	0.0292926260668372\\
0.89453125	0.00579987274640526\\
0.896484375	0.0013629003708413\\
0.8984375	0.00070565006402196\\
0.900390625	0.00559687065991304\\
0.90234375	0.026726743190351\\
0.904296875	0.0980479311207582\\
0.90625	0.445551433412166\\
0.908203125	0.885609653852016\\
0.91015625	0.792463191237342\\
0.912109375	3.11405277188182\\
0.9140625	3.0409928179402\\
0.916015625	1.20986642607625\\
0.91796875	0.504711940684075\\
0.919921875	0.0827585638550052\\
0.921875	0.068193811803955\\
0.923828125	0.0233146101372531\\
0.92578125	0.00316774202658485\\
0.927734375	0.00115534864824494\\
0.9296875	0.000181031892750786\\
0.931640625	0.000158591777812272\\
0.93359375	0.00026047189693435\\
0.935546875	0.00100984438832075\\
0.9375	0.00753897444314995\\
0.939453125	0.0455486609342233\\
0.94140625	0.202464915107684\\
0.943359375	0.913037979363899\\
0.9453125	0.460538865337406\\
0.947265625	0.0845485946873784\\
0.94921875	0.0171841014051871\\
0.951171875	0.00291908538680555\\
0.953125	0.000596531487821907\\
0.955078125	6.62822979259076e-05\\
0.95703125	1.6769192727691e-05\\
0.958984375	3.32848004161889e-06\\
0.9609375	1.31797910181565e-05\\
0.962890625	7.60416193551353e-05\\
0.96484375	0.000766308733158882\\
0.966796875	0.00514125065420016\\
0.96875	0.0134986002699534\\
0.970703125	0.0517992728539711\\
0.97265625	0.188925712394816\\
0.974609375	1.21463273009136\\
0.9765625	3.77981677107054\\
0.978515625	1.02991592355382\\
0.98046875	0.617778995600289\\
0.982421875	0.677781383187855\\
0.984375	0.385803309068485\\
0.986328125	0.104280574862655\\
0.98828125	0.0363975255463327\\
0.990234375	0.0115895229204273\\
0.9921875	0.00159572852128985\\
0.994140625	0.000523084274279186\\
0.99609375	0.000438385650470974\\
0.998046875	7.25991414315722e-05\\
1	0\\
};
\end{axis}
\end{tikzpicture}%

%% file: pics/square/GroundState_1_3_2.tex
%
%
\begin{tikzpicture}

\begin{axis}[%
width=0.9in,
height=0.9in,
scale only axis,
xmin=0,
xmax=1,
xtick={0,0.5,1},
ymin=0,
ymax=10,
yminorticks=true,
yticklabels = {},
axis background/.style={fill=white},
]
\addplot [color=mycolor2, thick]
table[row sep=crcr]{%
0	0\\
0.001953125	8.26984080943957e-05\\
0.00390625	0.00130216150765822\\
0.005859375	0.0105922378590528\\
0.0078125	0.0280857928812145\\
0.009765625	0.056744020346444\\
0.01171875	0.266888871411256\\
0.013671875	2.53527254183571\\
0.015625	1.89476912299642\\
0.017578125	0.556271498275085\\
0.01953125	0.185441349990032\\
0.021484375	0.0746516472370695\\
0.0234375	0.115486093323153\\
0.025390625	0.140898980786928\\
0.02734375	0.314727799139348\\
0.029296875	1.89252808534715\\
0.03125	3.01057557988466\\
0.033203125	1.48063321831923\\
0.03515625	4.25881907035918\\
0.037109375	5.93653188686727\\
0.0390625	2.14560295690043\\
0.041015625	1.05394737029753\\
0.04296875	1.28740597411334\\
0.044921875	0.308483617551031\\
0.046875	0.0542612682170683\\
0.048828125	0.00789758691779984\\
0.05078125	0.00140422373894265\\
0.052734375	0.000348477272505161\\
0.0546875	7.39619617179756e-05\\
0.056640625	7.09913931319142e-05\\
0.05859375	5.02350759978222e-05\\
0.060546875	8.11577162838222e-06\\
0.0625	1.52088479257014e-06\\
0.064453125	4.66114328184003e-07\\
0.06640625	7.36943665147197e-08\\
0.068359375	2.27359294229151e-08\\
0.0703125	1.70183001410623e-08\\
0.072265625	5.74461684071164e-09\\
0.07421875	7.50626784398668e-09\\
0.076171875	1.8010962419745e-08\\
0.078125	5.0669393645981e-08\\
0.080078125	1.15461871237931e-07\\
0.08203125	3.69146105822685e-07\\
0.083984375	1.1198561086864e-06\\
0.0859375	3.30103186932545e-06\\
0.087890625	2.1759865806388e-06\\
0.08984375	2.50283794801169e-06\\
0.091796875	4.01493682892184e-07\\
0.09375	6.28837270662444e-08\\
0.095703125	1.46195139315906e-08\\
0.09765625	6.35348645879616e-09\\
0.099609375	1.03584213995436e-09\\
0.1015625	4.04353422350942e-10\\
0.103515625	7.42355954990519e-10\\
0.10546875	4.92074842870803e-09\\
0.107421875	1.11037667952578e-08\\
0.109375	4.96829384710151e-08\\
0.111328125	4.47995237471273e-07\\
0.11328125	2.94142694632001e-06\\
0.115234375	8.52775689530448e-06\\
0.1171875	4.5383684571933e-05\\
0.119140625	0.000112282509723561\\
0.12109375	0.000385835177799132\\
0.123046875	0.00181194270427322\\
0.125	0.00766848848855812\\
0.126953125	0.0414842683865387\\
0.12890625	0.20771679961523\\
0.130859375	1.21728903560241\\
0.1328125	4.37647195688865\\
0.134765625	5.08516973115689\\
0.13671875	2.79087355601634\\
0.138671875	0.782908187492445\\
0.140625	0.159630062820509\\
0.142578125	0.0310597815875107\\
0.14453125	0.00659676716551448\\
0.146484375	0.00251097655153854\\
0.1484375	0.000628868707997457\\
0.150390625	0.000409851205472001\\
0.15234375	0.000126936962166299\\
0.154296875	2.92549107401611e-05\\
0.15625	4.31016974083218e-06\\
0.158203125	9.69714262469052e-07\\
0.16015625	1.84109322493116e-07\\
0.162109375	1.69167437222277e-07\\
0.1640625	7.26951522693089e-07\\
0.166015625	1.77737536327196e-06\\
0.16796875	4.59129349997643e-06\\
0.169921875	1.95441477664263e-05\\
0.171875	5.77132807274967e-05\\
0.173828125	0.000316982088626339\\
0.17578125	0.000753716659080728\\
0.177734375	0.00160131584345553\\
0.1796875	0.0051198379455546\\
0.181640625	0.0478497519522691\\
0.18359375	0.267521651944425\\
0.185546875	0.724939338996247\\
0.1875	0.538677552864734\\
0.189453125	0.0944118614880896\\
0.19140625	0.0163190162135673\\
0.193359375	0.00301504484622743\\
0.1953125	0.000640297315567525\\
0.197265625	0.000120723835581281\\
0.19921875	3.32770860909788e-05\\
0.201171875	6.36907309440609e-06\\
0.203125	1.11910053861219e-06\\
0.205078125	1.8967291462756e-07\\
0.20703125	5.53846920973424e-08\\
0.208984375	7.16909146098906e-09\\
0.2109375	1.42309040517243e-09\\
0.212890625	4.98383727856341e-10\\
0.21484375	1.15593999933621e-10\\
0.216796875	3.16719027436937e-11\\
0.21875	9.13767087831371e-12\\
0.220703125	3.17320368724195e-12\\
0.22265625	7.2461534055085e-12\\
0.224609375	3.10253118890901e-11\\
0.2265625	9.40954361427565e-11\\
0.228515625	3.98693585915868e-10\\
0.23046875	2.48144945898732e-09\\
0.232421875	1.76535278399634e-08\\
0.234375	8.07916144637904e-08\\
0.236328125	3.60496585033869e-07\\
0.23828125	7.49762742874052e-07\\
0.240234375	3.34040559496942e-06\\
0.2421875	2.02265766140215e-05\\
0.244140625	9.45364527851333e-05\\
0.24609375	0.000950695402864376\\
0.248046875	0.00209337307406585\\
0.25	0.00588818990894075\\
0.251953125	0.019437320431032\\
0.25390625	0.02094750156966\\
0.255859375	0.0990754750693066\\
0.2578125	0.509885349117493\\
0.259765625	3.54819947045915\\
0.26171875	3.90431088385069\\
0.263671875	0.866834526577706\\
0.265625	0.120297785814794\\
0.267578125	0.108671135002363\\
0.26953125	0.0261812519281215\\
0.271484375	0.00737720081568136\\
0.2734375	0.0019505176376784\\
0.275390625	0.00038870187852508\\
0.27734375	8.90722542148806e-05\\
0.279296875	7.15001787672086e-06\\
0.28125	1.19010252630483e-06\\
0.283203125	3.13282996104878e-07\\
0.28515625	2.12694535660581e-08\\
0.287109375	4.69933712476457e-09\\
0.2890625	1.35205447793181e-08\\
0.291015625	4.04443814334649e-08\\
0.29296875	1.28577565876859e-07\\
0.294921875	9.31157635570662e-07\\
0.296875	2.31859791040354e-06\\
0.298828125	5.21895762572258e-06\\
0.30078125	1.07830990895778e-05\\
0.302734375	1.84796069149865e-05\\
0.3046875	0.000128756889936342\\
0.306640625	0.000451853623232769\\
0.30859375	0.00151733677844026\\
0.310546875	0.00620899767499198\\
0.3125	0.0976913471450634\\
0.314453125	0.321324131220418\\
0.31640625	0.96233429256342\\
0.318359375	2.59241488142282\\
0.3203125	2.87908001115874\\
0.322265625	4.29883850191081\\
0.32421875	4.605326813763\\
0.326171875	2.28323540380716\\
0.328125	0.964782706711086\\
0.330078125	0.193583381441853\\
0.33203125	0.0315726386000308\\
0.333984375	0.00357010201351712\\
0.3359375	0.000680928948675028\\
0.337890625	0.000150254975373761\\
0.33984375	0.000442583431153039\\
0.341796875	0.00148137200145228\\
0.34375	0.00895315154595619\\
0.345703125	0.0602322844242594\\
0.34765625	0.185201422984755\\
0.349609375	0.208399041825734\\
0.3515625	0.32584264973801\\
0.353515625	1.0243740701554\\
0.35546875	1.74694681017484\\
0.357421875	1.74550386573256\\
0.359375	0.735799058990617\\
0.361328125	0.271076292960945\\
0.36328125	0.483026836790986\\
0.365234375	3.29344050125464\\
0.3671875	4.11199019165822\\
0.369140625	0.588508837914294\\
0.37109375	0.119887467029909\\
0.373046875	0.0171220585879522\\
0.375	0.0153110957143971\\
0.376953125	0.0102230122004828\\
0.37890625	0.00256794893906858\\
0.380859375	0.000712423010094934\\
0.3828125	0.000136232449126831\\
0.384765625	3.14198027620765e-05\\
0.38671875	6.47887811208946e-06\\
0.388671875	2.70759747672634e-06\\
0.390625	8.44875675417673e-07\\
0.392578125	5.91723489429059e-07\\
0.39453125	1.66191422950729e-06\\
0.396484375	5.20037570645355e-06\\
0.3984375	3.01874959400412e-05\\
0.400390625	0.000238069112569455\\
0.40234375	0.00200931506966267\\
0.404296875	0.00671069252931041\\
0.40625	0.0193558571695569\\
0.408203125	0.06064492439212\\
0.41015625	0.166408006038539\\
0.412109375	0.978112003045145\\
0.4140625	5.05587691934761\\
0.416015625	3.72663864849073\\
0.41796875	1.65459122934283\\
0.419921875	3.61098055585078\\
0.421875	2.56101755516283\\
0.423828125	1.43767741225714\\
0.42578125	0.485434373475981\\
0.427734375	0.203945511526494\\
0.4296875	0.213913863669853\\
0.431640625	0.0644621947847262\\
0.43359375	0.015982579963591\\
0.435546875	0.00328450582282299\\
0.4375	0.00226180784596536\\
0.439453125	0.00907256396859978\\
0.44140625	0.0481741232963847\\
0.443359375	0.0830961694513643\\
0.4453125	0.20643805199132\\
0.447265625	1.02936007125119\\
0.44921875	1.11660196931668\\
0.451171875	0.308463244214474\\
0.453125	0.0682191556694698\\
0.455078125	0.0138370642852712\\
0.45703125	0.0025947297646588\\
0.458984375	0.000641782058816009\\
0.4609375	0.00029785836552278\\
0.462890625	0.000187538172046675\\
0.46484375	4.59335485761883e-05\\
0.466796875	1.1176537979883e-05\\
0.46875	3.91661038075924e-06\\
0.470703125	2.57758099889238e-06\\
0.47265625	4.64563172984335e-07\\
0.474609375	3.79428283242344e-08\\
0.4765625	5.60333511085449e-09\\
0.478515625	1.12247810341776e-09\\
0.48046875	2.58263340302252e-10\\
0.482421875	4.102293075047e-11\\
0.484375	4.31987077102232e-12\\
0.486328125	1.18627854772392e-12\\
0.48828125	9.77213391937533e-13\\
0.490234375	2.64506530908162e-12\\
0.4921875	5.76722666696471e-12\\
0.494140625	1.09201360102052e-11\\
0.49609375	5.65013215920258e-11\\
0.498046875	4.57552476624404e-10\\
0.5	1.84877421011633e-09\\
0.501953125	3.42561157335524e-09\\
0.50390625	1.53338788017756e-08\\
0.505859375	1.10132984273099e-07\\
0.5078125	2.45617546472642e-07\\
0.509765625	5.56530887968256e-07\\
0.51171875	4.76760108334894e-06\\
0.513671875	3.70356128797869e-05\\
0.515625	0.000113676656641918\\
0.517578125	0.000387855303158166\\
0.51953125	0.0014251292017921\\
0.521484375	0.0068965060099255\\
0.5234375	0.024083820193135\\
0.525390625	0.0322674456819153\\
0.52734375	0.116368474642091\\
0.529296875	0.344716410714878\\
0.53125	0.605492185679453\\
0.533203125	0.340719650676228\\
0.53515625	0.0390788501264532\\
0.537109375	0.0109304517928668\\
0.5390625	0.00133946380647031\\
0.541015625	0.000147516012313149\\
0.54296875	2.71980007927615e-05\\
0.544921875	3.89947473066348e-05\\
0.546875	0.000229444226504385\\
0.548828125	0.000889661318305401\\
0.55078125	0.00320421252248437\\
0.552734375	0.0134895085447037\\
0.5546875	0.0998382987911922\\
0.556640625	0.503365883021367\\
0.55859375	1.79003993754467\\
0.560546875	2.70953988403144\\
0.5625	2.19908779112814\\
0.564453125	0.460288393156766\\
0.56640625	0.0684190450980127\\
0.568359375	0.0119291856322434\\
0.5703125	0.00199044194609794\\
0.572265625	0.000558836596729419\\
0.57421875	0.000254682248386478\\
0.576171875	2.39355893438259e-05\\
0.578125	7.13266063036443e-06\\
0.580078125	2.12509173440374e-05\\
0.58203125	0.000112124005649477\\
0.583984375	0.000644574420815184\\
0.5859375	0.00374791649184247\\
0.587890625	0.0143471017525547\\
0.58984375	0.109542200668311\\
0.591796875	0.28903160013291\\
0.59375	0.419497910158016\\
0.595703125	2.24534072357952\\
0.59765625	4.04142058962981\\
0.599609375	3.60117869034018\\
0.6015625	1.55050257651661\\
0.603515625	0.522971961896074\\
0.60546875	0.0986172021156119\\
0.607421875	0.0107605388067756\\
0.609375	0.00211032266866878\\
0.611328125	0.000697667981270356\\
0.61328125	7.95147508774885e-05\\
0.615234375	1.897911334311e-05\\
0.6171875	2.20742800674296e-06\\
0.619140625	4.63979922593953e-07\\
0.62109375	1.78827751864341e-07\\
0.623046875	7.51420899129892e-08\\
0.625	1.49263877898434e-08\\
0.626953125	4.27608022001249e-09\\
0.62890625	1.2962083932212e-09\\
0.630859375	3.66827104728006e-10\\
0.6328125	4.97779187398044e-11\\
0.634765625	2.36876303999474e-11\\
0.63671875	5.77632873204351e-12\\
0.638671875	5.00146795944638e-12\\
0.640625	3.6327945373384e-12\\
0.642578125	8.43013187196219e-13\\
0.64453125	4.42172933834899e-13\\
0.646484375	7.59222816538614e-13\\
0.6484375	2.03895752023026e-12\\
0.650390625	2.03988620837084e-11\\
0.65234375	3.92601363015896e-11\\
0.654296875	1.99977546834252e-10\\
0.65625	3.05299064125477e-10\\
0.658203125	1.03960735707851e-09\\
0.66015625	1.21743907365203e-09\\
0.662109375	3.9139210783951e-09\\
0.6640625	6.69996582234114e-09\\
0.666015625	1.87844828290209e-08\\
0.66796875	1.97818887754611e-07\\
0.669921875	4.91189771745868e-07\\
0.671875	1.75298956863214e-06\\
0.673828125	5.36704012575281e-06\\
0.67578125	1.14128988356241e-05\\
0.677734375	2.06119750247271e-05\\
0.6796875	5.68825931029193e-05\\
0.681640625	0.000256841324928897\\
0.68359375	0.000321300314736508\\
0.685546875	0.000530911394178932\\
0.6875	0.00114649869814451\\
0.689453125	0.00400475224776232\\
0.69140625	0.0125600066683169\\
0.693359375	0.0580632261440208\\
0.6953125	0.109911162249658\\
0.697265625	0.207784938587719\\
0.69921875	0.655672147644045\\
0.701171875	3.55965967793652\\
0.703125	2.36727360265801\\
0.705078125	0.526203518317023\\
0.70703125	0.0447970558337916\\
0.708984375	0.00496812576715979\\
0.7109375	0.00106964128921008\\
0.712890625	0.000256821705343883\\
0.71484375	7.8498949826393e-05\\
0.716796875	1.55104170894643e-05\\
0.71875	1.69648631450442e-06\\
0.720703125	1.71827923864504e-06\\
0.72265625	6.08073828255162e-06\\
0.724609375	1.60168563097014e-05\\
0.7265625	3.35430178712641e-05\\
0.728515625	0.000161763670977974\\
0.73046875	0.000869907844363062\\
0.732421875	0.00264871020090379\\
0.734375	0.00241252787370467\\
0.736328125	0.00531476877989631\\
0.73828125	0.0194051475226208\\
0.740234375	0.193087778660136\\
0.7421875	0.809831693009777\\
0.744140625	1.61598294827669\\
0.74609375	4.16505091271833\\
0.748046875	1.11288062746732\\
0.75	0.213858812939592\\
0.751953125	0.0456572726802784\\
0.75390625	0.00560022427996579\\
0.755859375	0.000676879716668094\\
0.7578125	0.000278167767003314\\
0.759765625	6.78480780594365e-05\\
0.76171875	1.25160840194646e-05\\
0.763671875	1.60762258092302e-05\\
0.765625	5.77417477534579e-05\\
0.767578125	0.000238031010356422\\
0.76953125	0.00136434885938349\\
0.771484375	0.0110684195285286\\
0.7734375	0.057452415607218\\
0.775390625	0.176016694325704\\
0.77734375	0.863432887544664\\
0.779296875	3.1165325578875\\
0.78125	1.6895600364882\\
0.783203125	0.619562063752494\\
0.78515625	0.214961695219316\\
0.787109375	0.0989190952003807\\
0.7890625	0.0228957080110028\\
0.791015625	0.00615462109876027\\
0.79296875	0.00113904177759173\\
0.794921875	0.000389852013679795\\
0.796875	0.000247762043062958\\
0.798828125	4.28940676098644e-05\\
0.80078125	8.75230688352157e-06\\
0.802734375	3.32839793740043e-06\\
0.8046875	8.01560042229744e-07\\
0.806640625	1.55423507072496e-07\\
0.80859375	2.47716440786637e-08\\
0.810546875	6.96170765689812e-09\\
0.8125	1.24334403637971e-09\\
0.814453125	5.45838941940519e-10\\
0.81640625	1.66125556870284e-10\\
0.818359375	4.23350324957802e-11\\
0.8203125	1.89789810921722e-11\\
0.822265625	3.49543963101654e-12\\
0.82421875	1.10479277626971e-12\\
0.826171875	1.81718896070805e-13\\
0.828125	1.21787391736488e-13\\
0.830078125	8.90163285657501e-14\\
0.83203125	7.79292489097298e-14\\
0.833984375	4.81614568419791e-13\\
0.8359375	1.53928134921797e-12\\
0.837890625	1.90723692539814e-12\\
0.83984375	9.65547768862332e-12\\
0.841796875	1.95739868610554e-11\\
0.84375	2.72561325847935e-11\\
0.845703125	2.1665045467554e-10\\
0.84765625	7.59852523544388e-10\\
0.849609375	8.32244629461935e-09\\
0.8515625	4.30604116594978e-08\\
0.853515625	1.25390455015035e-07\\
0.85546875	4.30603047162742e-07\\
0.857421875	1.94997558714985e-06\\
0.859375	8.62958010112916e-06\\
0.861328125	2.00597076030945e-05\\
0.86328125	2.55579406728102e-05\\
0.865234375	0.000162955334304171\\
0.8671875	0.000370297886667862\\
0.869140625	0.000434239961642684\\
0.87109375	0.0022082933163345\\
0.873046875	0.00376039223369757\\
0.875	0.0121670129134937\\
0.876953125	0.0487311882659156\\
0.87890625	0.270579926479341\\
0.880859375	0.918971775882024\\
0.8828125	2.09662086081552\\
0.884765625	3.71075680264378\\
0.88671875	2.0577998041982\\
0.888671875	0.24765897505149\\
0.890625	0.0536633777727986\\
0.892578125	0.0177968053615456\\
0.89453125	0.00303679747327873\\
0.896484375	0.000607374385441159\\
0.8984375	0.000186742590350606\\
0.900390625	0.0013913738359297\\
0.90234375	0.00757471410246208\\
0.904296875	0.0324140418270397\\
0.90625	0.170763963915502\\
0.908203125	0.41508638235105\\
0.91015625	0.682498424780488\\
0.912109375	3.12636336507797\\
0.9140625	3.05297663506918\\
0.916015625	0.963841898114423\\
0.91796875	0.341696386658053\\
0.919921875	0.0456455875388557\\
0.921875	0.0269747848879429\\
0.923828125	0.00758323273555884\\
0.92578125	0.000897880204584004\\
0.927734375	0.000271505520707652\\
0.9296875	3.3959053398962e-05\\
0.931640625	8.51162714278337e-06\\
0.93359375	2.38530422273343e-06\\
0.935546875	4.556646415799e-07\\
0.9375	3.75286265928824e-08\\
0.939453125	9.13636445005546e-09\\
0.94140625	1.96094383379753e-09\\
0.943359375	2.0749198107795e-09\\
0.9453125	8.84841842055299e-10\\
0.947265625	1.47813393120938e-10\\
0.94921875	5.93256141168663e-11\\
0.951171875	2.68992921592296e-10\\
0.953125	1.22079545044765e-09\\
0.955078125	1.14024258358323e-08\\
0.95703125	8.14335926143294e-08\\
0.958984375	5.36111333051501e-07\\
0.9609375	4.79301887973959e-06\\
0.962890625	3.2205193448382e-05\\
0.96484375	0.000361261384199947\\
0.966796875	0.00269823031287403\\
0.96875	0.00857072865949903\\
0.970703125	0.0400638831464445\\
0.97265625	0.17287866315098\\
0.974609375	1.25924433106411\\
0.9765625	4.1067689326641\\
0.978515625	0.966869022449948\\
0.98046875	0.331456514018152\\
0.982421875	0.213765810729157\\
0.984375	0.0964102477540222\\
0.986328125	0.0224014188210898\\
0.98828125	0.0058154839384733\\
0.990234375	0.00158673391628626\\
0.9921875	0.000190024553715787\\
0.994140625	4.71780614729852e-05\\
0.99609375	2.94750018756173e-05\\
0.998046875	4.31672660036975e-06\\
1	0\\
};
\end{axis}
\end{tikzpicture}%

%% file: pics/square/GroundState_1_3_3.tex
%
%
\begin{tikzpicture}

\begin{axis}[%
width=0.9in,
height=0.9in,
scale only axis,
xmin=0,
xmax=1,
xtick={0,0.5,1},
ymin=0,
ymax=10,
yminorticks=true,
yticklabels = {},
axis background/.style={fill=white},
]
\addplot [color=mycolor2, thick]
table[row sep=crcr]{%
0	0\\
0.001953125	4.55510901875391e-07\\
0.00390625	7.63394689241389e-06\\
0.005859375	6.70908857030707e-05\\
0.0078125	0.000204563962821343\\
0.009765625	0.000509338483840784\\
0.01171875	0.00276045970107056\\
0.013671875	0.0280638416929005\\
0.015625	0.0203813575660325\\
0.017578125	0.00585617285348718\\
0.01953125	0.00343758021375442\\
0.021484375	0.0103098542263527\\
0.0234375	0.0347777289045287\\
0.025390625	0.0677382811942304\\
0.02734375	0.200372333022613\\
0.029296875	1.38992639997893\\
0.03125	2.32310678008794\\
0.033203125	1.45952435913635\\
0.03515625	4.97761001082207\\
0.037109375	7.22971796382836\\
0.0390625	2.42249590330842\\
0.041015625	0.765736052086144\\
0.04296875	0.73229174278092\\
0.044921875	0.156624579019846\\
0.046875	0.0250351941472526\\
0.048828125	0.00333186994063573\\
0.05078125	0.000529441486546427\\
0.052734375	0.000116727109493516\\
0.0546875	1.97616680296848e-05\\
0.056640625	1.26581947322581e-05\\
0.05859375	7.11673588311179e-06\\
0.060546875	1.04709036526757e-06\\
0.0625	1.73406450524451e-07\\
0.064453125	4.67747245740301e-08\\
0.06640625	6.70083495118896e-09\\
0.068359375	1.6530096069057e-09\\
0.0703125	8.88668184675073e-10\\
0.072265625	1.98418901608613e-10\\
0.07421875	6.17726957774616e-11\\
0.076171875	1.93853420904863e-11\\
0.078125	4.90743942293993e-12\\
0.080078125	1.67956126315189e-12\\
0.08203125	3.39751479873524e-13\\
0.083984375	1.37959357535931e-13\\
0.0859375	1.25192527645501e-13\\
0.087890625	4.64003922172681e-14\\
0.08984375	3.9786098422478e-14\\
0.091796875	1.0675977317982e-14\\
0.09375	3.22758804115454e-14\\
0.095703125	1.71754232151145e-13\\
0.09765625	7.87696738928399e-13\\
0.099609375	4.5941611865424e-12\\
0.1015625	2.15135220463846e-11\\
0.103515625	1.34291553244928e-10\\
0.10546875	1.00604577290043e-09\\
0.107421875	2.65523544228165e-09\\
0.109375	1.38804546832789e-08\\
0.111328125	1.3609273309609e-07\\
0.11328125	9.77476865827683e-07\\
0.115234375	3.23080362103152e-06\\
0.1171875	1.93451353982982e-05\\
0.119140625	5.54276149384651e-05\\
0.12109375	0.000230074996493281\\
0.123046875	0.00120234750988854\\
0.125	0.00581290059214831\\
0.126953125	0.0351996677348645\\
0.12890625	0.19433637927901\\
0.130859375	1.27044871097132\\
0.1328125	5.01227876799887\\
0.134765625	5.96794290914\\
0.13671875	3.0492844851801\\
0.138671875	0.766455376153312\\
0.140625	0.139235339327669\\
0.142578125	0.0245863736210441\\
0.14453125	0.00452839869064591\\
0.146484375	0.00146678184567433\\
0.1484375	0.000317384810941313\\
0.150390625	0.000155141350338596\\
0.15234375	4.26360663301802e-05\\
0.154296875	8.56421358982634e-06\\
0.15625	1.15804685536539e-06\\
0.158203125	2.29848851218033e-07\\
0.16015625	3.26849122114521e-08\\
0.162109375	8.86606524240276e-09\\
0.1640625	1.56017837772571e-09\\
0.166015625	5.13480417949947e-10\\
0.16796875	1.74919926046961e-10\\
0.169921875	4.33004527098176e-11\\
0.171875	5.56774460546346e-12\\
0.173828125	2.31514627028604e-12\\
0.17578125	6.40781534694713e-13\\
0.177734375	2.37149323926823e-13\\
0.1796875	3.1053152919729e-14\\
0.181640625	4.00885613612018e-15\\
0.18359375	1.53605299157835e-15\\
0.185546875	1.52502888894701e-15\\
0.1875	9.47812512781484e-16\\
0.189453125	1.48715226800675e-16\\
0.19140625	2.30836338513011e-17\\
0.193359375	3.8606194501527e-18\\
0.1953125	7.36066091340639e-19\\
0.197265625	1.24690823730447e-19\\
0.19921875	3.15852741053476e-20\\
0.201171875	1.18937613183739e-20\\
0.203125	4.5996777974369e-20\\
0.205078125	4.3307551225714e-19\\
0.20703125	1.34864705284307e-18\\
0.208984375	6.35416037623366e-18\\
0.2109375	5.66639535497521e-17\\
0.212890625	3.25646661255616e-16\\
0.21484375	9.87352680932499e-16\\
0.216796875	4.55581234961981e-15\\
0.21875	2.06586349987624e-14\\
0.220703125	7.19491479430137e-14\\
0.22265625	4.93570827688102e-13\\
0.224609375	2.47407815437955e-12\\
0.2265625	8.61387949956128e-12\\
0.228515625	4.23806986021644e-11\\
0.23046875	2.89020937799907e-10\\
0.232421875	2.27004645028211e-09\\
0.234375	1.15382759602846e-08\\
0.236328125	5.78377428130374e-08\\
0.23828125	1.42536952584221e-07\\
0.240234375	7.55951567985243e-07\\
0.2421875	5.04079939261086e-06\\
0.244140625	2.65108708444155e-05\\
0.24609375	0.000288325718888501\\
0.248046875	0.000731954937568464\\
0.25	0.00246229068692775\\
0.251953125	0.00938881353616059\\
0.25390625	0.0132277235635101\\
0.255859375	0.084490021250138\\
0.2578125	0.478674989261484\\
0.259765625	3.68781171263361\\
0.26171875	4.0435769226904\\
0.263671875	0.830026181617236\\
0.265625	0.0989734482474454\\
0.267578125	0.0631174224576974\\
0.26953125	0.0135633750094435\\
0.271484375	0.0033600023080779\\
0.2734375	0.000775308684206788\\
0.275390625	0.000138983596281641\\
0.27734375	2.87584561253803e-05\\
0.279296875	2.144484811041e-06\\
0.28125	3.22075861284555e-07\\
0.283203125	7.53920526039107e-08\\
0.28515625	4.72584645362454e-09\\
0.287109375	6.98323048942221e-10\\
0.2890625	1.28191214818588e-09\\
0.291015625	4.16182229472921e-09\\
0.29296875	1.54816336713159e-08\\
0.294921875	1.24306404375967e-07\\
0.296875	3.62303974934436e-07\\
0.298828125	9.66997764412089e-07\\
0.30078125	2.46496598090323e-06\\
0.302734375	5.38995903486457e-06\\
0.3046875	4.38534823997526e-05\\
0.306640625	0.000173055573500482\\
0.30859375	0.000670967411023449\\
0.310546875	0.00311803636963194\\
0.3125	0.0529956383177834\\
0.314453125	0.194046366787297\\
0.31640625	0.721568969475944\\
0.318359375	2.23621674295503\\
0.3203125	2.75883564909717\\
0.322265625	4.99284322755725\\
0.32421875	5.37125814535069\\
0.326171875	2.35831789781187\\
0.328125	0.86414289367575\\
0.330078125	0.156828207918941\\
0.33203125	0.0233478192810639\\
0.333984375	0.00241242710263895\\
0.3359375	0.000411924484925885\\
0.337890625	6.01611531750546e-05\\
0.33984375	9.91035901553198e-06\\
0.341796875	8.8773060377523e-06\\
0.34375	4.92115269918035e-05\\
0.345703125	0.000361641757343136\\
0.34765625	0.00126790318085084\\
0.349609375	0.00184077823590182\\
0.3515625	0.00453293570436914\\
0.353515625	0.0168129274145882\\
0.35546875	0.0349431055395561\\
0.357421875	0.0434227485300726\\
0.359375	0.0348224298000369\\
0.361328125	0.0856590356427657\\
0.36328125	0.482578332908085\\
0.365234375	3.71493463778651\\
0.3671875	4.67413821123702\\
0.369140625	0.618453682152174\\
0.37109375	0.110974362619513\\
0.373046875	0.0138024628565745\\
0.375	0.00706035337999729\\
0.376953125	0.00389745871249089\\
0.37890625	0.000855762023809801\\
0.380859375	0.000211033288917512\\
0.3828125	3.57885322483108e-05\\
0.384765625	7.29935145901233e-06\\
0.38671875	1.33694618996672e-06\\
0.388671875	4.82995904784871e-07\\
0.390625	1.46764585750489e-07\\
0.392578125	1.46272848322322e-07\\
0.39453125	5.77189515853774e-07\\
0.396484375	2.13336630478742e-06\\
0.3984375	1.38028935840006e-05\\
0.400390625	0.000119133332053318\\
0.40234375	0.00108714368121847\\
0.404296875	0.00410705524331669\\
0.40625	0.0139072270104519\\
0.408203125	0.0504153484682903\\
0.41015625	0.165531963885256\\
0.412109375	1.08149673827947\\
0.4140625	6.09129390971307\\
0.416015625	4.31226985996141\\
0.41796875	1.52402054556012\\
0.419921875	2.93732940395732\\
0.421875	1.90128809392148\\
0.423828125	0.911432505689444\\
0.42578125	0.261604323130648\\
0.427734375	0.0802338268636699\\
0.4296875	0.060151715080788\\
0.431640625	0.015969490695545\\
0.43359375	0.00353330356449789\\
0.435546875	0.000565204123703467\\
0.4375	8.62457192227591e-05\\
0.439453125	2.02319935182413e-05\\
0.44140625	5.88674983744153e-06\\
0.443359375	2.36868209659941e-06\\
0.4453125	5.86527676235852e-07\\
0.447265625	4.68625398445055e-07\\
0.44921875	4.11040052722459e-07\\
0.451171875	9.96416848351595e-08\\
0.453125	1.98473846232671e-08\\
0.455078125	3.57772724821795e-09\\
0.45703125	6.08042252192204e-10\\
0.458984375	1.25006068177083e-10\\
0.4609375	4.66941887594085e-11\\
0.462890625	2.38141562147314e-11\\
0.46484375	5.15574833979784e-12\\
0.466796875	1.11128701171187e-12\\
0.46875	3.09020452650559e-13\\
0.470703125	1.66961108671578e-13\\
0.47265625	2.74917681320706e-14\\
0.474609375	2.07474024753477e-15\\
0.4765625	2.76104665161011e-16\\
0.478515625	4.98341549405809e-17\\
0.48046875	1.02399406902633e-17\\
0.482421875	1.49167561131404e-18\\
0.484375	1.41081783225907e-19\\
0.486328125	3.12788368556177e-20\\
0.48828125	5.17510957677593e-21\\
0.490234375	2.47157427673657e-21\\
0.4921875	1.01632950804325e-21\\
0.494140625	3.40714144599151e-22\\
0.49609375	8.86794340918518e-22\\
0.498046875	7.60485064762289e-21\\
0.5	3.4078353896068e-20\\
0.501953125	7.72207038447543e-20\\
0.50390625	4.0169473109164e-19\\
0.505859375	3.17490592413483e-18\\
0.5078125	8.11180837482371e-18\\
0.509765625	2.39616268143254e-17\\
0.51171875	2.29694178069749e-16\\
0.513671875	1.93592891438178e-15\\
0.515625	6.83867179301297e-15\\
0.517578125	2.67775077196473e-14\\
0.51953125	1.12079493343846e-13\\
0.521484375	6.10920861210146e-13\\
0.5234375	2.41759538378938e-12\\
0.525390625	4.03404684878018e-12\\
0.52734375	1.84368961332093e-11\\
0.529296875	6.39921246397767e-11\\
0.53125	1.35829962960463e-10\\
0.533203125	1.2517987806276e-10\\
0.53515625	2.78064350970317e-10\\
0.537109375	1.53618105840342e-09\\
0.5390625	1.03913081441622e-08\\
0.541015625	1.53496998914503e-07\\
0.54296875	9.6330154198039e-07\\
0.544921875	5.22635343107869e-06\\
0.546875	3.83492407817179e-05\\
0.548828125	0.00016879853197896\\
0.55078125	0.000684774418225521\\
0.552734375	0.00330459924827486\\
0.5546875	0.0267549023790474\\
0.556640625	0.149492862906544\\
0.55859375	0.592862275811695\\
0.560546875	1.06189462876746\\
0.5625	0.863993588543101\\
0.564453125	0.164967220411838\\
0.56640625	0.0223789475439321\\
0.568359375	0.00350554903986509\\
0.5703125	0.000532832023005969\\
0.572265625	0.000121087794912879\\
0.57421875	4.7689818900389e-05\\
0.576171875	4.18852318900983e-06\\
0.578125	1.44634813062797e-06\\
0.580078125	6.46056097144978e-06\\
0.58203125	3.83336067795659e-05\\
0.583984375	0.000245030618754048\\
0.5859375	0.00155767017591111\\
0.587890625	0.00673700189866834\\
0.58984375	0.0566312444217193\\
0.591796875	0.170779614655016\\
0.59375	0.3180216316153\\
0.595703125	2.05441495673618\\
0.59765625	3.94836636191055\\
0.599609375	3.47119887320209\\
0.6015625	1.3611975197827\\
0.603515625	0.398825608322324\\
0.60546875	0.0685072808034845\\
0.607421875	0.00678967375601184\\
0.609375	0.00119625014193724\\
0.611328125	0.000347889913829521\\
0.61328125	3.62095345937931e-05\\
0.615234375	7.75233504028701e-06\\
0.6171875	8.23509806104038e-07\\
0.619140625	1.49394742275482e-07\\
0.62109375	4.87478516938174e-08\\
0.623046875	1.75465135120736e-08\\
0.625	3.12403063837267e-09\\
0.626953125	7.81421294762158e-10\\
0.62890625	2.01144838363324e-10\\
0.630859375	5.05478331369551e-11\\
0.6328125	6.06217029510112e-12\\
0.634765625	2.46421257764499e-12\\
0.63671875	4.81901934006521e-13\\
0.638671875	2.6322497616227e-13\\
0.640625	1.5205958028507e-13\\
0.642578125	3.06138073139536e-14\\
0.64453125	7.45147899575722e-15\\
0.646484375	5.46665597988422e-15\\
0.6484375	1.16151959572528e-14\\
0.650390625	1.2429158474791e-13\\
0.65234375	2.84284048641055e-13\\
0.654296875	1.65137545390401e-12\\
0.65625	3.08917037643322e-12\\
0.658203125	1.2371354445233e-11\\
0.66015625	1.96824660300007e-11\\
0.662109375	7.68572094646193e-11\\
0.6640625	1.6294199431528e-10\\
0.666015625	5.77991224480859e-10\\
0.66796875	6.67077658051286e-09\\
0.669921875	1.87805482442031e-08\\
0.671875	8.2725540534123e-08\\
0.673828125	2.88956791202143e-07\\
0.67578125	7.44911703221695e-07\\
0.677734375	1.66600590622825e-06\\
0.6796875	6.05744007191916e-06\\
0.681640625	3.09026377258744e-05\\
0.68359375	4.75120955843844e-05\\
0.685546875	0.000111794035823745\\
0.6875	0.000335967526284718\\
0.689453125	0.00137881735180068\\
0.69140625	0.0049348073113031\\
0.693359375	0.0264581211226091\\
0.6953125	0.058618227906614\\
0.697265625	0.162163554190297\\
0.69921875	0.626052214942172\\
0.701171875	3.72090633720696\\
0.703125	2.41940734209905\\
0.705078125	0.488086429958387\\
0.70703125	0.0387310029258635\\
0.708984375	0.00388701386122732\\
0.7109375	0.000752536398973559\\
0.712890625	0.000156938875484959\\
0.71484375	4.20860689693413e-05\\
0.716796875	7.55051400894717e-06\\
0.71875	7.12281930673808e-07\\
0.720703125	2.56147901625598e-07\\
0.72265625	7.50552510477524e-07\\
0.724609375	2.28323007965742e-06\\
0.7265625	5.61190623362303e-06\\
0.728515625	3.33772366704263e-05\\
0.73046875	0.00019724450899425\\
0.732421875	0.000679682932689524\\
0.734375	0.000855900446080523\\
0.736328125	0.00262333041829964\\
0.73828125	0.0119079775679771\\
0.740234375	0.128598793903061\\
0.7421875	0.599233309723318\\
0.744140625	1.43191113933156\\
0.74609375	3.95591131741065\\
0.748046875	0.958335533569042\\
0.75	0.165435787446809\\
0.751953125	0.0318817540014623\\
0.75390625	0.00359136941963805\\
0.755859375	0.000390819917266998\\
0.7578125	0.000136787766767531\\
0.759765625	2.96210258850841e-05\\
0.76171875	4.66392206469523e-06\\
0.763671875	3.73819895304025e-06\\
0.765625	1.42357258511797e-05\\
0.767578125	6.78160484954928e-05\\
0.76953125	0.000430758107447135\\
0.771484375	0.00379565110493301\\
0.7734375	0.0220546684032916\\
0.775390625	0.0767402283847726\\
0.77734375	0.426951226022717\\
0.779296875	1.70815775731648\\
0.78125	0.883671189597662\\
0.783203125	0.274439431335523\\
0.78515625	0.0810360326452256\\
0.787109375	0.0311084690961171\\
0.7890625	0.00644538111492914\\
0.791015625	0.0015083071934466\\
0.79296875	0.000250035628203105\\
0.794921875	6.28015690110241e-05\\
0.796875	3.30618458682126e-05\\
0.798828125	5.17547123088324e-06\\
0.80078125	9.34334217909475e-07\\
0.802734375	2.99539379650158e-07\\
0.8046875	6.44605668274872e-08\\
0.806640625	1.13221521624398e-08\\
0.80859375	1.61185176092228e-09\\
0.810546875	4.00527217753757e-10\\
0.8125	6.19451902185641e-11\\
0.814453125	2.30395310329017e-11\\
0.81640625	6.06916687374798e-12\\
0.818359375	1.37068772935539e-12\\
0.8203125	4.98154559797303e-13\\
0.822265625	8.30021051597025e-14\\
0.82421875	2.25209087265642e-14\\
0.826171875	3.20249828449066e-15\\
0.828125	1.51372652762664e-15\\
0.830078125	9.02190700956672e-16\\
0.83203125	8.04314207084648e-16\\
0.833984375	5.49595687848434e-15\\
0.8359375	1.98909570555441e-14\\
0.837890625	3.34603325820461e-14\\
0.83984375	2.00376621952117e-13\\
0.841796875	4.74882114078478e-13\\
0.84375	9.13387639498852e-13\\
0.845703125	8.33746729645689e-12\\
0.84765625	3.36177254994485e-11\\
0.849609375	3.99268358980285e-10\\
0.8515625	2.26615146258187e-09\\
0.853515625	7.66186437392717e-09\\
0.85546875	3.02880053586085e-08\\
0.857421875	1.54310105890478e-07\\
0.859375	7.63357002595496e-07\\
0.861328125	2.07249291904439e-06\\
0.86328125	3.78876117951117e-06\\
0.865234375	2.786684085297e-05\\
0.8671875	7.40439117523862e-05\\
0.869140625	0.000116969347912924\\
0.87109375	0.000717415491751897\\
0.873046875	0.00144801079677558\\
0.875	0.00572093344773744\\
0.876953125	0.0263460067648415\\
0.87890625	0.161697891152945\\
0.880859375	0.62335477167905\\
0.8828125	1.69290261174015\\
0.884765625	3.21091067885585\\
0.88671875	1.72575942082059\\
0.888671875	0.191221825204518\\
0.890625	0.0361901448934583\\
0.892578125	0.0105187281459474\\
0.89453125	0.00161745070855301\\
0.896484375	0.000290800211199875\\
0.8984375	6.48253497503584e-05\\
0.900390625	0.000445157511839471\\
0.90234375	0.00266427575430738\\
0.904296875	0.0127212017230212\\
0.90625	0.0744450515722832\\
0.908203125	0.209381003221045\\
0.91015625	0.475929473810476\\
0.912109375	2.37095857852391\\
0.9140625	2.32363140333815\\
0.916015625	0.625119097539053\\
0.91796875	0.19639113215263\\
0.919921875	0.023243083648238\\
0.921875	0.0111802415598806\\
0.923828125	0.0027466984838949\\
0.92578125	0.000295681707725194\\
0.927734375	7.86329604968758e-05\\
0.9296875	8.99853957554807e-06\\
0.931640625	1.95654048347381e-06\\
0.93359375	4.86676472735085e-07\\
0.935546875	8.45992766920917e-08\\
0.9375	6.43932492129331e-09\\
0.939453125	1.38413787529261e-09\\
0.94140625	2.43387772405245e-10\\
0.943359375	1.5238073585285e-10\\
0.9453125	5.70275066566191e-11\\
0.947265625	9.80273527145423e-12\\
0.94921875	1.04044884041431e-11\\
0.951171875	7.72940489088058e-11\\
0.953125	3.93463588592089e-10\\
0.955078125	4.01527707113038e-09\\
0.95703125	3.11330472386407e-08\\
0.958984375	2.27865192580146e-07\\
0.9609375	2.1994515377224e-06\\
0.962890625	1.62459227159797e-05\\
0.96484375	0.000197009333606477\\
0.966796875	0.00159262257286115\\
0.96875	0.00577425422589739\\
0.970703125	0.0307707316227363\\
0.97265625	0.149268994650674\\
0.974609375	1.18880954922607\\
0.9765625	4.01920246099787\\
0.978515625	0.868313887756397\\
0.98046875	0.236642642528593\\
0.982421875	0.114512137487414\\
0.984375	0.0438463839080508\\
0.986328125	0.00916759915319258\\
0.98828125	0.00198557043278233\\
0.990234375	0.000485145034913691\\
0.9921875	5.27963711477501e-05\\
0.994140625	1.11504825179702e-05\\
0.99609375	5.76325641626273e-06\\
0.998046875	7.71763718051233e-07\\
1	0\\
};
\end{axis}
\end{tikzpicture}%

%% file: pics/square/GroundState_1_3_4.tex
%
%
\begin{tikzpicture}

\begin{axis}[%
width=0.9in,
height=0.9in,
scale only axis,
xmin=0,
xmax=1,
xtick={0,0.5,1},
ymin=0,
ymax=10,
yminorticks=true,
ytick={0,4,8},
yticklabel pos=right,
axis background/.style={fill=white},
]
\addplot [color=mycolor2, thick]
table[row sep=crcr]{%
0	0\\
0.001953125	1.94661086899007e-09\\
0.00390625	3.41956308762346e-08\\
0.005859375	3.18459418545286e-07\\
0.0078125	1.07325208978764e-06\\
0.009765625	3.0681392054387e-06\\
0.01171875	1.82973068060775e-05\\
0.013671875	0.000197033318020975\\
0.015625	0.000181636601116551\\
0.017578125	0.000238531686679924\\
0.01953125	0.000655768953655954\\
0.021484375	0.00360718824360059\\
0.0234375	0.0137036735396404\\
0.025390625	0.0310437200447496\\
0.02734375	0.104409315093082\\
0.029296875	0.79027434689627\\
0.03125	1.43708168246556\\
0.033203125	1.45279247681411\\
0.03515625	5.82113487323936\\
0.037109375	8.77542680377015\\
0.0390625	2.80940377696779\\
0.041015625	0.700639630612037\\
0.04296875	0.551936932407467\\
0.044921875	0.108676445739099\\
0.046875	0.0162208215371292\\
0.048828125	0.00202496604776748\\
0.05078125	0.000297505794709196\\
0.052734375	6.0498379407734e-05\\
0.0546875	9.15675431219376e-06\\
0.056640625	4.73629776413859e-06\\
0.05859375	2.29385291057281e-06\\
0.060546875	3.15786879798212e-07\\
0.0625	4.804457605526e-08\\
0.064453125	1.18545464508785e-08\\
0.06640625	1.58707628053049e-09\\
0.068359375	3.47125376484829e-10\\
0.0703125	1.58651028169605e-10\\
0.072265625	3.24647645527508e-11\\
0.07421875	9.1014797919523e-12\\
0.076171875	2.56607755998106e-12\\
0.078125	5.82282098304478e-13\\
0.080078125	1.80184617348976e-13\\
0.08203125	3.29382939196815e-14\\
0.083984375	1.11350998587026e-14\\
0.0859375	7.65189948460087e-15\\
0.087890625	2.39615438745722e-15\\
0.08984375	1.9299292539992e-15\\
0.091796875	1.25609436443301e-15\\
0.09375	6.73968798571226e-15\\
0.095703125	3.92237207977403e-14\\
0.09765625	1.95166478042224e-13\\
0.099609375	1.26288325949645e-12\\
0.1015625	6.37870443285536e-12\\
0.103515625	4.33020256313571e-11\\
0.10546875	3.45417101414283e-10\\
0.107421875	1.0140070523305e-09\\
0.109375	5.86829674607247e-09\\
0.111328125	6.10666832411239e-08\\
0.11328125	4.67799993016861e-07\\
0.115234375	1.69429423114402e-06\\
0.1171875	1.10021083386264e-05\\
0.119140625	3.48755153810331e-05\\
0.12109375	0.000163408208672802\\
0.123046875	0.000919711466090288\\
0.125	0.00487290329016919\\
0.126953125	0.0319221534192148\\
0.12890625	0.188969263044216\\
0.130859375	1.33377731266126\\
0.1328125	5.63648220800445\\
0.134765625	6.84775992582787\\
0.13671875	3.31448636201824\\
0.138671875	0.769321873549959\\
0.140625	0.128892139749624\\
0.142578125	0.0212457247725225\\
0.14453125	0.00355888992053623\\
0.146484375	0.00103480902939641\\
0.1484375	0.0002040800165921\\
0.150390625	8.38250118416233e-05\\
0.15234375	2.11824945352863e-05\\
0.154296875	3.87122594695729e-06\\
0.15625	4.92270265015438e-07\\
0.158203125	8.99110542476479e-08\\
0.16015625	1.18593939618293e-08\\
0.162109375	2.94387486676333e-09\\
0.1640625	4.63292887993819e-10\\
0.166015625	1.36822116656486e-10\\
0.16796875	4.23118597550334e-11\\
0.169921875	9.48278255733359e-12\\
0.171875	1.13488218406871e-12\\
0.173828125	4.13594070357702e-13\\
0.17578125	1.01538350475493e-13\\
0.177734375	3.38362187657362e-14\\
0.1796875	4.17384079744878e-15\\
0.181640625	4.95998860087088e-16\\
0.18359375	1.48793789660536e-16\\
0.185546875	9.98916254317301e-17\\
0.1875	5.45366876437165e-17\\
0.189453125	7.9162187828377e-18\\
0.19140625	1.13911759260262e-18\\
0.193359375	1.77505349522076e-19\\
0.1953125	3.13641148427478e-20\\
0.197265625	4.92568260688835e-21\\
0.19921875	1.16563858478343e-21\\
0.201171875	5.72760312631437e-22\\
0.203125	2.91365898992087e-21\\
0.205078125	2.95574757814412e-20\\
0.20703125	1.00688889684987e-19\\
0.208984375	5.22923805853423e-19\\
0.2109375	4.9460307870958e-18\\
0.212890625	3.05798701692472e-17\\
0.21484375	1.03852616191748e-16\\
0.216796875	5.21371817214615e-16\\
0.21875	2.59198996468842e-15\\
0.220703125	9.92934511571851e-15\\
0.22265625	7.34767083163822e-14\\
0.224609375	3.98599085739489e-13\\
0.2265625	1.52425777860469e-12\\
0.228515625	8.28480629212611e-12\\
0.23046875	6.02896562471919e-11\\
0.232421875	5.07830953454106e-10\\
0.234375	2.78080797763722e-09\\
0.236328125	1.51227512474509e-08\\
0.23828125	4.17888704934364e-08\\
0.240234375	2.47768603185208e-07\\
0.2421875	1.76921192294253e-06\\
0.244140625	1.01046026494117e-05\\
0.24609375	0.000116240829294444\\
0.248046875	0.000325795238396227\\
0.25	0.00123196048111354\\
0.251953125	0.0051724739363771\\
0.25390625	0.00861275116540828\\
0.255859375	0.0641569396697964\\
0.2578125	0.388815216513752\\
0.259765625	3.21742917430082\\
0.26171875	3.51453443588606\\
0.263671875	0.681043559985657\\
0.265625	0.0743878742016133\\
0.267578125	0.0386807381586558\\
0.26953125	0.00767429445370624\\
0.271484375	0.00173948936329192\\
0.2734375	0.000365460943401919\\
0.275390625	6.07926921702227e-05\\
0.27734375	1.1695656346179e-05\\
0.279296875	8.27034384016599e-07\\
0.28125	1.15541118365809e-07\\
0.283203125	2.48954951423531e-08\\
0.28515625	1.47870772641985e-09\\
0.287109375	1.8425036484638e-10\\
0.2890625	2.50466369314072e-10\\
0.291015625	8.56021494992666e-10\\
0.29296875	3.53841077526549e-09\\
0.294921875	3.0515495203954e-08\\
0.296875	9.90634307682963e-08\\
0.298828125	2.95876137008196e-07\\
0.30078125	8.61764390699171e-07\\
0.302734375	2.18271388307212e-06\\
0.3046875	1.95165252479568e-05\\
0.306640625	8.36293928979425e-05\\
0.30859375	0.000357656840636404\\
0.310546875	0.00181409547439755\\
0.3125	0.032569214010103\\
0.314453125	0.128783879788287\\
0.31640625	0.549805875724668\\
0.318359375	1.88601937036744\\
0.3203125	2.61574712759739\\
0.322265625	5.73439556879759\\
0.32421875	6.197257808658\\
0.326171875	2.48018984589885\\
0.328125	0.819854088560598\\
0.330078125	0.138465060509811\\
0.33203125	0.0193075486553405\\
0.333984375	0.00187158479667195\\
0.3359375	0.000297595541987926\\
0.337890625	4.07194200488218e-05\\
0.33984375	5.2686975719769e-06\\
0.341796875	1.98647175175418e-06\\
0.34375	7.85216552821692e-06\\
0.345703125	6.09963232347581e-05\\
0.34765625	0.000234322535993991\\
0.349609375	0.00039983621940489\\
0.3515625	0.00122540852078023\\
0.353515625	0.00502237228957842\\
0.35546875	0.0120919564713486\\
0.357421875	0.0187285178289385\\
0.359375	0.0228315506813709\\
0.361328125	0.079597632797761\\
0.36328125	0.500552537285268\\
0.365234375	4.07726916256041\\
0.3671875	5.15346637711271\\
0.369140625	0.644342943533058\\
0.37109375	0.105958126812192\\
0.373046875	0.0122267379291141\\
0.375	0.00481852945524466\\
0.376953125	0.00234377797305209\\
0.37890625	0.000469292732686475\\
0.380859375	0.000106532137686909\\
0.3828125	1.66229683084127e-05\\
0.384765625	3.11308125706268e-06\\
0.38671875	5.25717931433636e-07\\
0.388671875	1.71762017543862e-07\\
0.390625	5.15733489944709e-08\\
0.392578125	6.40896343583112e-08\\
0.39453125	2.97657789546678e-07\\
0.396484375	1.21927995149765e-06\\
0.3984375	8.50160262150042e-06\\
0.400390625	7.82581525777829e-05\\
0.40234375	0.000755581080473641\\
0.404296875	0.00311252920085882\\
0.40625	0.0117461580672796\\
0.408203125	0.0469625904969249\\
0.41015625	0.173525554276649\\
0.412109375	1.21922585190044\\
0.4140625	7.30351461395978\\
0.416015625	4.99122702629057\\
0.41796875	1.33415998343179\\
0.419921875	1.70059984320116\\
0.421875	0.975522568925541\\
0.423828125	0.412019124154565\\
0.42578125	0.106693816467208\\
0.427734375	0.0279219606322265\\
0.4296875	0.0171527954403384\\
0.431640625	0.00416872310452835\\
0.43359375	0.000853413952386396\\
0.435546875	0.000125946470692411\\
0.4375	1.78873898759073e-05\\
0.439453125	3.83708796288259e-06\\
0.44140625	9.64908855922106e-07\\
0.443359375	3.4730635024755e-07\\
0.4453125	7.45529965512236e-08\\
0.447265625	3.51295445886067e-08\\
0.44921875	2.58025573105304e-08\\
0.451171875	5.69781369589036e-09\\
0.453125	1.05345835217922e-09\\
0.455078125	1.74890112288543e-10\\
0.45703125	2.7732143885454e-11\\
0.458984375	5.08283170455509e-12\\
0.4609375	1.66142860748359e-12\\
0.462890625	7.38915856376438e-13\\
0.46484375	1.46873167718732e-13\\
0.466796875	2.91589600136904e-14\\
0.46875	7.05609404291299e-15\\
0.470703125	3.34530789522122e-15\\
0.47265625	5.16159794041301e-16\\
0.474609375	3.68053733403615e-17\\
0.4765625	4.55246740612962e-18\\
0.478515625	7.63317055436395e-19\\
0.48046875	1.44831161544447e-19\\
0.482421875	1.98237892556078e-20\\
0.484375	1.74838368833913e-21\\
0.486328125	3.58103418834306e-22\\
0.48828125	5.37522372305804e-23\\
0.490234375	2.20866745658903e-23\\
0.4921875	7.60286968659252e-24\\
0.494140625	1.44648655420028e-24\\
0.49609375	2.33136973363453e-25\\
0.498046875	2.43231848165346e-26\\
0.5	9.90649548246342e-27\\
0.501953125	2.09684172429862e-27\\
0.50390625	3.44865958414526e-28\\
0.505859375	5.68445307835305e-29\\
0.5078125	3.2839491914886e-29\\
0.509765625	2.16333729908438e-29\\
0.51171875	1.69297709403864e-28\\
0.513671875	1.43696341471742e-27\\
0.515625	5.15483955201557e-27\\
0.517578125	2.04948242564333e-26\\
0.51953125	8.70643240053578e-26\\
0.521484375	4.81196546936894e-25\\
0.5234375	1.93299395996321e-24\\
0.525390625	3.30910835542756e-24\\
0.52734375	1.55210608890884e-23\\
0.529296875	5.4959155303904e-23\\
0.53125	1.19649245277534e-22\\
0.533203125	1.16847041428167e-22\\
0.53515625	2.86873625143173e-22\\
0.537109375	1.6205782188381e-21\\
0.5390625	1.1190522550524e-20\\
0.541015625	1.67827111536367e-19\\
0.54296875	1.07411346378495e-18\\
0.544921875	5.94419985953007e-18\\
0.546875	4.45800279377706e-17\\
0.548828125	2.01453654586185e-16\\
0.55078125	8.39092271707225e-16\\
0.552734375	4.1754523459283e-15\\
0.5546875	3.45484222828662e-14\\
0.556640625	1.98067168783452e-13\\
0.55859375	8.10149820506215e-13\\
0.560546875	1.54554290607507e-12\\
0.5625	1.35870159863951e-12\\
0.564453125	7.0372858795475e-13\\
0.56640625	3.35342256269711e-12\\
0.568359375	2.58569071052225e-11\\
0.5703125	1.35714562823888e-10\\
0.572265625	1.5470786232229e-09\\
0.57421875	4.12106189778982e-09\\
0.576171875	2.22647368608193e-08\\
0.578125	1.86025159105111e-07\\
0.580078125	1.71622378226929e-06\\
0.58203125	1.1100695600675e-05\\
0.583984375	7.64688698134e-05\\
0.5859375	0.000518233789413737\\
0.587890625	0.00244131270819576\\
0.58984375	0.0219617020417434\\
0.591796875	0.0727115357023234\\
0.59375	0.158482109081394\\
0.595703125	1.14326383342636\\
0.59765625	2.31763852449198\\
0.599609375	2.01498393204627\\
0.6015625	0.736985130229254\\
0.603515625	0.195598319792625\\
0.60546875	0.0314189996829446\\
0.607421875	0.00291003479395329\\
0.609375	0.00047559247403854\\
0.611328125	0.000126440917313386\\
0.61328125	1.23402919898099e-05\\
0.615234375	2.44587422260447e-06\\
0.6171875	2.43703832175322e-07\\
0.619140625	4.01316277420519e-08\\
0.62109375	1.17296241532418e-08\\
0.623046875	3.79761721284113e-09\\
0.625	6.26374150170423e-10\\
0.626953125	1.42807058817324e-10\\
0.62890625	3.29522917629681e-11\\
0.630859375	7.61828976074221e-12\\
0.6328125	8.4163932219039e-13\\
0.634765625	3.07927607082086e-13\\
0.63671875	5.3889045230384e-14\\
0.638671875	2.34253441940449e-14\\
0.640625	1.16985588356505e-14\\
0.642578125	2.18170252548206e-15\\
0.64453125	4.46148408619155e-16\\
0.646484375	2.35149725911805e-16\\
0.6484375	4.23265814695063e-16\\
0.650390625	4.75146141754169e-15\\
0.65234375	1.22161437776563e-14\\
0.654296875	7.74511144220587e-14\\
0.65625	1.65581953125172e-13\\
0.658203125	7.35349659967035e-13\\
0.66015625	1.40562731170424e-12\\
0.662109375	6.15024262185598e-12\\
0.6640625	1.49216250221173e-11\\
0.666015625	6.07402340840288e-11\\
0.66796875	7.46481628405602e-10\\
0.669921875	2.29558419937917e-09\\
0.671875	1.1547652474677e-08\\
0.673828125	4.4169040561511e-08\\
0.67578125	1.29058887875765e-07\\
0.677734375	3.30107528073698e-07\\
0.6796875	1.40285788213893e-06\\
0.681640625	7.76628640814308e-06\\
0.68359375	1.36865346560452e-05\\
0.685546875	3.88779916444525e-05\\
0.6875	0.000137980037794579\\
0.689453125	0.000625279509976218\\
0.69140625	0.0024501012757222\\
0.693359375	0.0145033340422973\\
0.6953125	0.0357937454080712\\
0.697265625	0.121453621664503\\
0.69921875	0.527342417696613\\
0.701171875	3.34117336553553\\
0.703125	2.13666140078657\\
0.705078125	0.401782669557966\\
0.70703125	0.0302961257531575\\
0.708984375	0.00283483286566072\\
0.7109375	0.000509279535950114\\
0.712890625	9.65464612464974e-05\\
0.71484375	2.36399033544449e-05\\
0.716796875	3.95988919531804e-06\\
0.71875	3.47646869436619e-07\\
0.720703125	6.33861226348144e-08\\
0.72265625	1.05273104842729e-07\\
0.724609375	3.22443868086399e-07\\
0.7265625	8.77135359732921e-07\\
0.728515625	5.9296583919539e-06\\
0.73046875	3.74644219587535e-05\\
0.732421875	0.000140812440203448\\
0.734375	0.000215258075367588\\
0.736328125	0.000778225495560286\\
0.73828125	0.00399442428156051\\
0.740234375	0.0456871674961772\\
0.7421875	0.229585769580538\\
0.744140625	0.62143731765246\\
0.74609375	1.80038455899672\\
0.748046875	0.406511932698821\\
0.75	0.0650601410513651\\
0.751953125	0.0116546423261981\\
0.75390625	0.0012350189928932\\
0.755859375	0.000125086306227393\\
0.7578125	3.92673909858781e-05\\
0.759765625	7.82083893181489e-06\\
0.76171875	1.06953043741851e-06\\
0.763671875	1.17455302408865e-07\\
0.765625	3.74241322513327e-08\\
0.767578125	5.13497781829514e-09\\
0.76953125	5.3403483478873e-10\\
0.771484375	9.40844393325616e-11\\
0.7734375	1.49331700024268e-11\\
0.775390625	6.98648607727581e-12\\
0.77734375	1.87488371346493e-11\\
0.779296875	7.37540961070189e-11\\
0.78125	3.77641693625001e-11\\
0.783203125	1.13852726310836e-11\\
0.78515625	3.26860678992911e-12\\
0.787109375	1.21566130582439e-12\\
0.7890625	2.46889237711177e-13\\
0.791015625	5.63693744663272e-14\\
0.79296875	9.16871274870217e-15\\
0.794921875	2.1991484418945e-15\\
0.796875	1.12066104583006e-15\\
0.798828125	1.72300815970348e-16\\
0.80078125	3.0452327619344e-17\\
0.802734375	9.48230804176735e-18\\
0.8046875	2.00023013566058e-18\\
0.806640625	3.4518867609432e-19\\
0.80859375	4.81832395001178e-20\\
0.810546875	1.17174541342897e-20\\
0.8125	1.76916091426837e-21\\
0.814453125	6.39935005393192e-22\\
0.81640625	1.65143797771178e-22\\
0.818359375	3.87192600102005e-23\\
0.8203125	3.03656035558546e-23\\
0.822265625	4.34492685599647e-23\\
0.82421875	3.23569852357459e-22\\
0.826171875	1.28008389123847e-21\\
0.828125	6.83402830277115e-21\\
0.830078125	1.54163131089772e-20\\
0.83203125	5.87375579726044e-20\\
0.833984375	4.68974186083338e-19\\
0.8359375	1.82347326400918e-18\\
0.837890625	3.55738141636409e-18\\
0.83984375	2.30805750297997e-17\\
0.841796875	5.95349733014843e-17\\
0.84375	1.33074628213715e-16\\
0.845703125	1.29848328163141e-15\\
0.84765625	5.65222610469651e-15\\
0.849609375	7.02792569521227e-14\\
0.8515625	4.20544514981572e-13\\
0.853515625	1.54302413014392e-12\\
0.85546875	6.58399722068418e-12\\
0.857421875	3.57921467924745e-11\\
0.859375	1.88443741205696e-10\\
0.861328125	5.56818871830318e-10\\
0.86328125	1.20107070887107e-09\\
0.865234375	9.45292089414424e-09\\
0.8671875	2.73373956895328e-08\\
0.869140625	4.98860893770267e-08\\
0.87109375	3.33476672065016e-07\\
0.873046875	7.37294385306437e-07\\
0.875	3.22057794352425e-06\\
0.876953125	1.59868951657144e-05\\
0.87890625	0.000103803487413857\\
0.880859375	0.000429533579705772\\
0.8828125	0.00128423614514273\\
0.884765625	0.00253670438123553\\
0.88671875	0.00133780475611861\\
0.888671875	0.000141399287743684\\
0.890625	2.48739005787814e-05\\
0.892578125	6.71852406398916e-06\\
0.89453125	9.74261016081248e-07\\
0.896484375	1.63772158039595e-07\\
0.8984375	1.54991888839092e-08\\
0.900390625	1.71965845200145e-09\\
0.90234375	1.74952742358019e-09\\
0.904296875	7.68000503798971e-09\\
0.90625	4.66558875857792e-08\\
0.908203125	1.38968320043862e-07\\
0.91015625	3.52927983877513e-07\\
0.912109375	1.81581316048792e-06\\
0.9140625	1.78493286477288e-06\\
0.916015625	4.51282132253896e-07\\
0.91796875	1.34995997955966e-07\\
0.919921875	1.52660671092192e-08\\
0.921875	6.79492333604316e-09\\
0.923828125	1.58267783851076e-09\\
0.92578125	1.64038160425099e-10\\
0.927734375	4.14498104727496e-11\\
0.9296875	4.5767496077495e-12\\
0.931640625	9.41617930234766e-13\\
0.93359375	2.23327422709947e-13\\
0.935546875	3.73797763433288e-14\\
0.9375	2.95488005734165e-15\\
0.939453125	2.13788099816015e-15\\
0.94140625	1.00576057561665e-14\\
0.943359375	6.08108472313477e-14\\
0.9453125	8.09119598320314e-14\\
0.947265625	3.71955601072433e-13\\
0.94921875	2.25630994866456e-12\\
0.951171875	2.04538764832441e-11\\
0.953125	1.12340320936002e-10\\
0.955078125	1.22055484792222e-09\\
0.95703125	1.0042598466478e-08\\
0.958984375	7.92004299653385e-08\\
0.9609375	8.08005842773346e-07\\
0.962890625	6.3856892388643e-06\\
0.96484375	8.19073321384894e-05\\
0.966796875	0.000701189633118304\\
0.96875	0.00278768289466559\\
0.970703125	0.0162433077862809\\
0.97265625	0.0855245571855247\\
0.974609375	0.725810987779026\\
0.9765625	2.5197688077226\\
0.978515625	0.512562248104873\\
0.98046875	0.122425773691931\\
0.982421875	0.0499478247904636\\
0.984375	0.0171095153885748\\
0.986328125	0.00332196797405648\\
0.98828125	0.000639463454164969\\
0.990234375	0.000144477770313954\\
0.9921875	1.47027986681928e-05\\
0.994140625	2.79797670142127e-06\\
0.99609375	1.2739556391644e-06\\
0.998046875	1.59974530403625e-07\\
1	0\\
};
\end{axis}
\end{tikzpicture}%

%% file: pics/square/GroundState_1_4_1.tex
%
%
\begin{tikzpicture}

\begin{axis}[%
width=0.9in,
height=0.9in,
scale only axis,
xmin=0,
xmax=1,
xtick={0,0.5,1},
ymin=0,
ymax=10,
yminorticks=true,
ytick={0,4,8},
ylabel={$c_V = 8$},
ylabel style={below=0.1in},
axis background/.style={fill=white},
]
\addplot [color=mycolor2, thick]
table[row sep=crcr]{%
0	0\\
0.001953125	0.000155979756478956\\
0.00390625	0.000598641064918063\\
0.005859375	0.005466119692193\\
0.0078125	0.0125650320605111\\
0.009765625	0.135310094993105\\
0.01171875	0.838895130355829\\
0.013671875	4.05531427800189\\
0.015625	1.88160477740067\\
0.017578125	0.046773068816753\\
0.01953125	0.00295779501982992\\
0.021484375	0.00109608759943441\\
0.0234375	6.92167230187353e-05\\
0.025390625	2.75708123520641e-05\\
0.02734375	0.000118097379472193\\
0.029296875	0.00122312926228317\\
0.03125	0.00779897151105486\\
0.033203125	0.0220315522730515\\
0.03515625	0.11362168025024\\
0.037109375	1.8595707307672\\
0.0390625	2.07365916669032\\
0.041015625	3.71339172631233\\
0.04296875	4.1389460867598\\
0.044921875	2.63616411749121\\
0.046875	4.10364980676081\\
0.048828125	2.01415327207951\\
0.05078125	0.568960396489454\\
0.052734375	0.141548800754962\\
0.0546875	0.654258291086718\\
0.056640625	5.61338690803513\\
0.05859375	0.784873965690806\\
0.060546875	0.0821814996706928\\
0.0625	0.0114605169557184\\
0.064453125	0.00177627591783949\\
0.06640625	9.5936512817966e-05\\
0.068359375	3.72001962198021e-06\\
0.0703125	7.43343772564449e-06\\
0.072265625	8.02045852913124e-05\\
0.07421875	0.00137354612369749\\
0.076171875	0.0247146241787303\\
0.078125	0.203892751451077\\
0.080078125	4.65328909462629\\
0.08203125	1.56995131409804\\
0.083984375	4.19980603430828\\
0.0859375	2.15189723399826\\
0.087890625	0.199640754154122\\
0.08984375	0.00507490299258554\\
0.091796875	0.000140424159200572\\
0.09375	4.21639019373284e-06\\
0.095703125	2.55395629716652e-07\\
0.09765625	3.79048433841517e-06\\
0.099609375	1.5588170879871e-05\\
0.1015625	5.3036281470945e-05\\
0.103515625	0.000187237784262373\\
0.10546875	0.000845714024784326\\
0.107421875	0.00302587834347547\\
0.109375	0.072933802196064\\
0.111328125	0.666030378174205\\
0.11328125	3.338720397759\\
0.115234375	0.600864912630652\\
0.1171875	0.851644573713042\\
0.119140625	0.178557699958192\\
0.12109375	0.0297490100060686\\
0.123046875	0.00148505897799275\\
0.125	0.000725872877757756\\
0.126953125	0.000170616659648128\\
0.12890625	3.34524320778451e-05\\
0.130859375	2.58900658158805e-05\\
0.1328125	1.9993605603084e-06\\
0.134765625	2.9613335450382e-07\\
0.13671875	1.62233809293941e-08\\
0.138671875	1.91603964794136e-09\\
0.140625	1.49180013260332e-10\\
0.142578125	1.78574997217026e-11\\
0.14453125	7.93747812518601e-13\\
0.146484375	3.02972161530784e-13\\
0.1484375	1.84131126043751e-13\\
0.150390625	8.63196845714968e-13\\
0.15234375	1.01858632652002e-11\\
0.154296875	1.24692159568829e-11\\
0.15625	1.25225819653026e-11\\
0.158203125	5.19751001745655e-11\\
0.16015625	3.36983964738867e-10\\
0.162109375	1.81858294179916e-09\\
0.1640625	2.80557908643383e-08\\
0.166015625	2.67834936256341e-07\\
0.16796875	2.40914846164608e-06\\
0.169921875	6.82087474171766e-06\\
0.171875	0.00010106103140334\\
0.173828125	0.000316075568545623\\
0.17578125	0.000261972821385701\\
0.177734375	0.00253044534246152\\
0.1796875	0.0257240712138182\\
0.181640625	0.24867559184038\\
0.18359375	1.14301117603481\\
0.185546875	4.26667522103656\\
0.1875	2.42044675852755\\
0.189453125	0.215730962159471\\
0.19140625	0.00973965155722175\\
0.193359375	0.000711845790496627\\
0.1953125	0.000109511536732082\\
0.197265625	1.37211996096238e-05\\
0.19921875	2.30624482752435e-06\\
0.201171875	3.06071966625344e-07\\
0.203125	3.06502334907375e-08\\
0.205078125	3.9830009249542e-09\\
0.20703125	2.76315917553086e-10\\
0.208984375	1.10832422447417e-11\\
0.2109375	2.54791500704781e-12\\
0.212890625	1.70126128459284e-12\\
0.21484375	1.38810199302324e-11\\
0.216796875	8.6516713012939e-11\\
0.21875	1.95983704921644e-10\\
0.220703125	9.42587167499149e-10\\
0.22265625	1.06446550873788e-08\\
0.224609375	1.16307223900904e-08\\
0.2265625	8.75203942872458e-08\\
0.228515625	5.20643825102012e-07\\
0.23046875	5.92891811758787e-07\\
0.232421875	1.32677766619682e-05\\
0.234375	6.81837311018099e-05\\
0.236328125	0.00116809376470586\\
0.23828125	0.012803164587871\\
0.240234375	0.0382474095464207\\
0.2421875	0.710755347334668\\
0.244140625	4.03185397346143\\
0.24609375	1.221535749156\\
0.248046875	0.299089573067805\\
0.25	0.0917788319373618\\
0.251953125	0.0802037275216364\\
0.25390625	0.0188203612490414\\
0.255859375	0.00136082401641166\\
0.2578125	0.000228372649335107\\
0.259765625	4.69436626737479e-05\\
0.26171875	2.98224200394096e-06\\
0.263671875	5.39715557182079e-07\\
0.265625	4.94636800039232e-08\\
0.267578125	2.89529684242683e-09\\
0.26953125	1.64767772482673e-09\\
0.271484375	1.23228993161766e-10\\
0.2734375	1.61524394527938e-11\\
0.275390625	3.87515662651947e-12\\
0.27734375	2.68486736496975e-12\\
0.279296875	6.56608956717204e-12\\
0.28125	8.36279936431085e-12\\
0.283203125	8.71027460708249e-11\\
0.28515625	4.78997645429951e-10\\
0.287109375	1.81127202590391e-09\\
0.2890625	1.30253691124415e-08\\
0.291015625	3.72890356049372e-07\\
0.29296875	3.54027889905573e-06\\
0.294921875	2.66209986188127e-05\\
0.296875	0.000247549020673628\\
0.298828125	0.00354257729306868\\
0.30078125	0.0348188557788003\\
0.302734375	0.148521953372285\\
0.3046875	3.44359277118253\\
0.306640625	3.19313964354979\\
0.30859375	0.786856160260719\\
0.310546875	2.79151131412894\\
0.3125	0.561662674079344\\
0.314453125	0.091410905254298\\
0.31640625	0.0098959834812374\\
0.318359375	0.00195690139273489\\
0.3203125	0.000297756571316547\\
0.322265625	1.80055644028983e-05\\
0.32421875	7.69328394750337e-07\\
0.326171875	1.55904330328713e-07\\
0.328125	1.10530007995023e-08\\
0.330078125	4.3921395620147e-10\\
0.33203125	9.39557239145394e-11\\
0.333984375	1.31440490658339e-11\\
0.3359375	4.54703041715416e-12\\
0.337890625	9.87225058704706e-13\\
0.33984375	6.34661642276593e-14\\
0.341796875	2.28104243513731e-14\\
0.34375	3.0298738967736e-15\\
0.345703125	5.0217689133079e-16\\
0.34765625	2.96069655555276e-17\\
0.349609375	2.4848210666651e-18\\
0.3515625	3.01081264752014e-17\\
0.353515625	6.61576086835495e-16\\
0.35546875	1.7691590795973e-14\\
0.357421875	1.78439194037042e-13\\
0.359375	6.46533464733924e-13\\
0.361328125	7.58895839338251e-12\\
0.36328125	4.98913701783789e-12\\
0.365234375	7.89749763784871e-12\\
0.3671875	8.8885164498988e-11\\
0.369140625	7.33111738247928e-10\\
0.37109375	2.85857280955621e-09\\
0.373046875	1.44978548591704e-08\\
0.375	4.51729963943135e-08\\
0.376953125	9.67135132712855e-08\\
0.37890625	9.1777258102586e-08\\
0.380859375	4.11442787522747e-07\\
0.3828125	3.28250395057252e-06\\
0.384765625	1.97927288252368e-05\\
0.38671875	1.20376101994682e-05\\
0.388671875	0.000148863338448072\\
0.390625	0.00126736588063433\\
0.392578125	0.00447248554542265\\
0.39453125	0.137471450262683\\
0.396484375	1.85720873549747\\
0.3984375	2.57582948420823\\
0.400390625	0.111195665579792\\
0.40234375	0.00338013197695656\\
0.404296875	0.000711599657060807\\
0.40625	0.000559172219657823\\
0.408203125	0.00980111452229216\\
0.41015625	0.0513743083224436\\
0.412109375	0.294562783148509\\
0.4140625	1.79793672957863\\
0.416015625	2.72636636227255\\
0.41796875	0.475061897054697\\
0.419921875	0.021050756371063\\
0.421875	0.00451682292051356\\
0.423828125	0.000476517820401301\\
0.42578125	4.10498630800489e-05\\
0.427734375	8.41644506465626e-06\\
0.4296875	4.51742413231959e-06\\
0.431640625	1.18429835838825e-06\\
0.43359375	9.21427798596914e-08\\
0.435546875	1.2709417190233e-08\\
0.4375	1.34491408834068e-09\\
0.439453125	1.92663842604515e-10\\
0.44140625	2.14432953076429e-11\\
0.443359375	3.01520113080477e-12\\
0.4453125	3.0631132422685e-13\\
0.447265625	2.55804137614244e-14\\
0.44921875	1.56775506944019e-15\\
0.451171875	1.66140767341809e-16\\
0.453125	8.8170449646008e-17\\
0.455078125	8.30379460120447e-18\\
0.45703125	2.43099086776615e-19\\
0.458984375	1.21363173330193e-20\\
0.4609375	1.3807712935947e-21\\
0.462890625	3.83210735753478e-21\\
0.46484375	1.90366457731443e-20\\
0.466796875	6.51776929567294e-20\\
0.46875	5.00323422762511e-19\\
0.470703125	4.45232781860305e-18\\
0.47265625	3.74166263258978e-17\\
0.474609375	6.89285212972628e-16\\
0.4765625	2.23983162774572e-14\\
0.478515625	7.79745069759702e-13\\
0.48046875	1.40394796769777e-11\\
0.482421875	8.88788022866835e-11\\
0.484375	1.32008752788761e-10\\
0.486328125	7.71336510446943e-10\\
0.48828125	9.92339222196205e-09\\
0.490234375	2.20521606042389e-07\\
0.4921875	4.09675980902187e-07\\
0.494140625	4.63385444219502e-06\\
0.49609375	1.81404454767131e-05\\
0.498046875	8.15553365886629e-05\\
0.5	0.00102703447159048\\
0.501953125	0.00546238092408896\\
0.50390625	0.0588996095497107\\
0.505859375	0.1380458602581\\
0.5078125	0.933566053351663\\
0.509765625	1.0610505788734\\
0.51171875	1.21746291571641\\
0.513671875	0.0646689380533114\\
0.515625	0.00725640079983223\\
0.517578125	0.00108137746221315\\
0.51953125	0.000133544753176871\\
0.521484375	1.4519957485659e-05\\
0.5234375	4.08394427509564e-06\\
0.525390625	9.55966748355377e-07\\
0.52734375	2.87227530688695e-06\\
0.529296875	8.12747927067715e-07\\
0.53125	7.97169387222732e-08\\
0.533203125	3.8212694374588e-08\\
0.53515625	1.09239542262568e-08\\
0.537109375	1.21528101612294e-09\\
0.5390625	5.55070356268383e-11\\
0.541015625	4.04816102651143e-12\\
0.54296875	1.12710379999385e-12\\
0.544921875	1.49954780203038e-11\\
0.546875	3.24811344867966e-10\\
0.548828125	6.14993315681339e-09\\
0.55078125	1.20404342281144e-07\\
0.552734375	5.20817088935924e-07\\
0.5546875	6.78588875531773e-06\\
0.556640625	9.89032638050913e-06\\
0.55859375	6.15231378031466e-05\\
0.560546875	0.000581146659337713\\
0.5625	0.00919943527006132\\
0.564453125	0.177011003350642\\
0.56640625	2.93129515240778\\
0.568359375	6.61605491998895\\
0.5703125	1.21912891621288\\
0.572265625	0.615708848351122\\
0.57421875	0.0924825574832452\\
0.576171875	0.271255181386217\\
0.578125	2.43929833342336\\
0.580078125	2.29178088546278\\
0.58203125	1.95727171636717\\
0.583984375	0.495108922804661\\
0.5859375	0.0514422062327287\\
0.587890625	0.00535488255490968\\
0.58984375	0.000312511745867568\\
0.591796875	2.88919710607639e-05\\
0.59375	1.28205514244968e-06\\
0.595703125	4.01474820359258e-08\\
0.59765625	4.72165564282119e-09\\
0.599609375	2.87651620236279e-09\\
0.6015625	4.37844913955303e-10\\
0.603515625	7.49392624391294e-11\\
0.60546875	1.05364961404842e-11\\
0.607421875	3.22960975454922e-12\\
0.609375	1.86468237779218e-12\\
0.611328125	3.13072318361188e-13\\
0.61328125	4.86200342356261e-14\\
0.615234375	8.19348103643525e-15\\
0.6171875	1.72813073881441e-16\\
0.619140625	2.10707391214989e-17\\
0.62109375	2.66585401275923e-18\\
0.623046875	1.59972361349802e-19\\
0.625	1.40338194045684e-20\\
0.626953125	1.45411168715533e-21\\
0.62890625	2.82133990203463e-21\\
0.630859375	5.62943339783278e-21\\
0.6328125	2.08727199270746e-20\\
0.634765625	4.19896582591996e-20\\
0.63671875	9.08413555107085e-20\\
0.638671875	4.35170964694367e-19\\
0.640625	6.66948355526744e-18\\
0.642578125	2.37213745924926e-17\\
0.64453125	1.10849354037961e-15\\
0.646484375	7.97940164915213e-15\\
0.6484375	1.33940227337811e-13\\
0.650390625	2.2160208819719e-12\\
0.65234375	2.7301392289097e-11\\
0.654296875	2.85191309182929e-10\\
0.65625	3.19017926321559e-09\\
0.658203125	1.77076109896789e-08\\
0.66015625	2.84507267356538e-07\\
0.662109375	4.27362218968388e-07\\
0.6640625	9.15413269234599e-06\\
0.666015625	7.13708838938552e-05\\
0.66796875	0.00104058363252372\\
0.669921875	0.00557891908435786\\
0.671875	0.000728420439778775\\
0.673828125	0.00022695082484258\\
0.67578125	1.20047492550751e-05\\
0.677734375	1.12142203192828e-06\\
0.6796875	2.08285233424072e-07\\
0.681640625	4.20952645877149e-08\\
0.68359375	6.86421648482314e-08\\
0.685546875	1.57893088057031e-06\\
0.6875	1.37053389642824e-05\\
0.689453125	6.66762841210899e-05\\
0.69140625	0.000778594116352608\\
0.693359375	0.0105338356284264\\
0.6953125	0.125371405725333\\
0.697265625	0.448497696439855\\
0.69921875	3.40422723899562\\
0.701171875	1.09391627210291\\
0.703125	0.0796868555380536\\
0.705078125	0.00341972004653412\\
0.70703125	0.000359340353961464\\
0.708984375	6.07297109331703e-05\\
0.7109375	3.91956343736349e-05\\
0.712890625	5.22100519631842e-06\\
0.71484375	1.4424147944992e-07\\
0.716796875	2.65516661018142e-08\\
0.71875	5.67531898519089e-09\\
0.720703125	7.47990552420229e-10\\
0.72265625	2.72749576486388e-11\\
0.724609375	2.45513259370238e-12\\
0.7265625	2.26271984985592e-13\\
0.728515625	1.05783558585738e-14\\
0.73046875	6.70727728643325e-16\\
0.732421875	4.85292604156598e-17\\
0.734375	1.29363959432953e-17\\
0.736328125	2.79417014963252e-18\\
0.73828125	4.10973466414445e-19\\
0.740234375	1.37326747233415e-20\\
0.7421875	1.07112455793443e-21\\
0.744140625	3.7342582435318e-23\\
0.74609375	3.43776602602952e-24\\
0.748046875	5.01488015223765e-25\\
0.75	6.62960154111078e-26\\
0.751953125	1.05928947900023e-25\\
0.75390625	5.33412114012397e-25\\
0.755859375	6.06062779435863e-24\\
0.7578125	9.1856301231985e-23\\
0.759765625	2.9982071687962e-22\\
0.76171875	1.20576731117779e-21\\
0.763671875	2.08961288945643e-20\\
0.765625	1.22563328780863e-19\\
0.767578125	1.0311949803308e-18\\
0.76953125	2.04098960640507e-17\\
0.771484375	1.46903790503806e-16\\
0.7734375	8.11006815706771e-16\\
0.775390625	1.09781966183179e-15\\
0.77734375	3.98393512810747e-15\\
0.779296875	2.60656179450397e-14\\
0.78125	1.13935824995781e-13\\
0.783203125	4.85044334929561e-13\\
0.78515625	2.08064030546897e-12\\
0.787109375	8.36679956453641e-12\\
0.7890625	7.48121701676487e-11\\
0.791015625	8.68101439065353e-10\\
0.79296875	5.49365263711273e-09\\
0.794921875	6.72734679799235e-08\\
0.796875	2.07677858989288e-07\\
0.798828125	3.08981901805122e-06\\
0.80078125	5.07649202223649e-05\\
0.802734375	0.00152840485103075\\
0.8046875	0.0165780556695798\\
0.806640625	0.371987922940212\\
0.80859375	3.72399618346411\\
0.810546875	3.14579712717496\\
0.8125	0.367979429073492\\
0.814453125	0.0185099264197844\\
0.81640625	0.00343087106748398\\
0.818359375	0.000252417920753392\\
0.8203125	2.77807525625331e-05\\
0.822265625	7.90519358031582e-06\\
0.82421875	2.73181977476396e-06\\
0.826171875	2.536143320989e-07\\
0.828125	1.65347650093725e-08\\
0.830078125	1.91774407289914e-09\\
0.83203125	4.90815790962278e-10\\
0.833984375	2.23148919644909e-09\\
0.8359375	1.67320921362178e-08\\
0.837890625	1.31386805508524e-07\\
0.83984375	2.90731083806806e-06\\
0.841796875	7.50229465313583e-05\\
0.84375	0.00107063222960261\\
0.845703125	0.0093068615154505\\
0.84765625	0.0554311032041752\\
0.849609375	0.696734150158326\\
0.8515625	1.30659154214829\\
0.853515625	5.38685889040994\\
0.85546875	3.39280573860976\\
0.857421875	0.45780246074872\\
0.859375	0.030180083999128\\
0.861328125	0.0224971975719564\\
0.86328125	0.1280064417837\\
0.865234375	2.64726883258568\\
0.8671875	1.02467973510711\\
0.869140625	0.0837063481238228\\
0.87109375	0.00915880748642944\\
0.873046875	0.00113629250763512\\
0.875	0.00012245172720467\\
0.876953125	1.34193757880386e-05\\
0.87890625	4.5053521749584e-06\\
0.880859375	2.48039195099111e-05\\
0.8828125	0.00016177209339465\\
0.884765625	0.000385417275391492\\
0.88671875	0.00706132085178525\\
0.888671875	0.0817190032659682\\
0.890625	0.38926685059684\\
0.892578125	1.85690130252354\\
0.89453125	1.4166809333736\\
0.896484375	0.158383816862378\\
0.8984375	0.0215732188531353\\
0.900390625	0.00141182983916358\\
0.90234375	0.000152637823409414\\
0.904296875	9.94072461659703e-06\\
0.90625	4.51807398524289e-07\\
0.908203125	3.34571859873206e-08\\
0.91015625	3.35022068426403e-08\\
0.912109375	6.22194998368404e-08\\
0.9140625	1.20469013693197e-06\\
0.916015625	1.14882570017654e-05\\
0.91796875	0.000116518058248753\\
0.919921875	0.00217870768878299\\
0.921875	0.0296885980059839\\
0.923828125	0.315276237370446\\
0.92578125	3.21614734975907\\
0.927734375	0.611681656538483\\
0.9296875	0.0993118940508245\\
0.931640625	0.556599730665451\\
0.93359375	1.45087022773142\\
0.935546875	0.770571092441201\\
0.9375	0.0548605399006317\\
0.939453125	0.0116782995358853\\
0.94140625	0.00595664905711419\\
0.943359375	0.000714782786027321\\
0.9453125	0.000133878938935671\\
0.947265625	4.35790588211668e-06\\
0.94921875	1.34718730148081e-05\\
0.951171875	9.22078593401276e-05\\
0.953125	0.000861716482747381\\
0.955078125	0.00805302786153867\\
0.95703125	0.0331914399765475\\
0.958984375	0.171041075424298\\
0.9609375	2.14409222556018\\
0.962890625	5.39773054272233\\
0.96484375	0.642292033259078\\
0.966796875	0.0174697821598056\\
0.96875	0.00391280907781628\\
0.970703125	0.000630890918702659\\
0.97265625	0.000209490967242557\\
0.974609375	1.66107863934515e-05\\
0.9765625	1.67013500756686e-05\\
0.978515625	1.3102434258105e-06\\
0.98046875	2.40175887624836e-07\\
0.982421875	1.6989390944505e-06\\
0.984375	2.26842412173458e-06\\
0.986328125	5.11720373471957e-05\\
0.98828125	0.000849537514838215\\
0.990234375	0.00290580219294733\\
0.9921875	0.0209648269112677\\
0.994140625	0.165260925578627\\
0.99609375	1.66330453664682\\
0.998046875	2.35229596440281\\
1	0\\
};
\end{axis}
\end{tikzpicture}%

%% file: pics/square/GroundState_1_4_2.tex
%
%
\begin{tikzpicture}

\begin{axis}[%
width=0.9in,
height=0.9in,
scale only axis,
xmin=0,
xmax=1,
xtick={0,0.5,1},
ymin=0,
ymax=10,
yminorticks=true,
yticklabels = {},
axis background/.style={fill=white},
]
\addplot [color=mycolor2, thick]
table[row sep=crcr]{%
0	0\\
0.001953125	3.21104471170423e-05\\
0.00390625	0.000176682046045353\\
0.005859375	0.00194377177276112\\
0.0078125	0.0062533017789461\\
0.009765625	0.0872370741790148\\
0.01171875	0.638216402227768\\
0.013671875	4.83445053934249\\
0.015625	2.13069151145397\\
0.017578125	0.047448288657101\\
0.01953125	0.00245213158287788\\
0.021484375	0.000691611269032843\\
0.0234375	3.71556516217277e-05\\
0.025390625	5.4712357606024e-06\\
0.02734375	2.0906378910028e-05\\
0.029296875	0.000274620783583135\\
0.03125	0.00211200539864526\\
0.033203125	0.0081093277982157\\
0.03515625	0.0530353226982713\\
0.037109375	1.02611689506793\\
0.0390625	1.52082613665257\\
0.041015625	3.38319569321271\\
0.04296875	4.5147354966453\\
0.044921875	2.66862132428944\\
0.046875	3.57787209271653\\
0.048828125	1.33911748999565\\
0.05078125	0.303593540724277\\
0.052734375	0.0887773725909905\\
0.0546875	0.685630236577035\\
0.056640625	6.49750977986688\\
0.05859375	0.822788893637181\\
0.060546875	0.0701782474497453\\
0.0625	0.00810237074474323\\
0.064453125	0.00102177809192421\\
0.06640625	4.86684984062933e-05\\
0.068359375	1.67665552304488e-06\\
0.0703125	3.92345893803713e-06\\
0.072265625	5.08722929392366e-05\\
0.07421875	0.00100396869455984\\
0.076171875	0.0205713457505432\\
0.078125	0.203467653088528\\
0.080078125	5.17854855192481\\
0.08203125	1.57962242238621\\
0.083984375	3.52485409920652\\
0.0859375	1.53538416305618\\
0.087890625	0.124292439429\\
0.08984375	0.00283272686654725\\
0.091796875	7.01785054065951e-05\\
0.09375	1.84778715358507e-06\\
0.095703125	5.12778436016265e-08\\
0.09765625	1.637825978179e-08\\
0.099609375	7.19526598091699e-08\\
0.1015625	3.29497996908896e-07\\
0.103515625	2.09618196368632e-06\\
0.10546875	1.15791425973805e-05\\
0.107421875	7.43112769305379e-05\\
0.109375	0.00204019588384397\\
0.111328125	0.021888181506061\\
0.11328125	0.124041926072078\\
0.115234375	0.0156349731026386\\
0.1171875	0.0112070929858684\\
0.119140625	0.001822538354046\\
0.12109375	0.000251376853562037\\
0.123046875	1.03872044975188e-05\\
0.125	3.56497455591481e-06\\
0.126953125	6.39746966839196e-07\\
0.12890625	9.03505926681636e-08\\
0.130859375	4.40433683084359e-08\\
0.1328125	2.93117439079312e-09\\
0.134765625	3.50340253947381e-10\\
0.13671875	1.65601567232167e-11\\
0.138671875	1.6608523962378e-12\\
0.140625	1.07687637738877e-13\\
0.142578125	1.08131306705124e-14\\
0.14453125	4.17017220369895e-16\\
0.146484375	1.41572057262449e-16\\
0.1484375	1.87624274901895e-16\\
0.150390625	1.75593031597207e-15\\
0.15234375	2.41864885909195e-14\\
0.154296875	4.80264514800008e-14\\
0.15625	1.15826283457763e-13\\
0.158203125	1.14615647588492e-12\\
0.16015625	8.93396323001109e-12\\
0.162109375	6.04203821282668e-11\\
0.1640625	1.11665077101596e-09\\
0.166015625	1.27609267224056e-08\\
0.16796875	1.39622622493992e-07\\
0.169921875	5.42990527886742e-07\\
0.171875	9.61384178865768e-06\\
0.173828125	3.90549385194151e-05\\
0.17578125	8.83418153801551e-05\\
0.177734375	0.00109830756190975\\
0.1796875	0.0136487764730807\\
0.181640625	0.154561949981044\\
0.18359375	0.881621985891777\\
0.185546875	4.06384201469859\\
0.1875	2.19849448028796\\
0.189453125	0.168115952235063\\
0.19140625	0.00648325258074278\\
0.193359375	0.000404242345139666\\
0.1953125	5.05283289110118e-05\\
0.197265625	5.22201138386687e-06\\
0.19921875	6.87272848205819e-07\\
0.201171875	7.66927929553875e-08\\
0.203125	6.47909304713451e-09\\
0.205078125	6.99176782257595e-10\\
0.20703125	4.21542499297864e-11\\
0.208984375	1.45391516275274e-12\\
0.2109375	2.18838627669711e-13\\
0.212890625	5.17376072599866e-14\\
0.21484375	2.72633588763072e-13\\
0.216796875	2.11544027023621e-12\\
0.21875	6.5600125738707e-12\\
0.220703125	4.46051954516631e-11\\
0.22265625	5.89006900119216e-10\\
0.224609375	9.97352437882798e-10\\
0.2265625	1.18539080470659e-08\\
0.228515625	8.66305476409316e-08\\
0.23046875	1.64593851015591e-07\\
0.232421875	4.69101659370435e-06\\
0.234375	2.95278537830114e-05\\
0.236328125	0.000587129308311242\\
0.23828125	0.00755940402088769\\
0.240234375	0.0302621746095111\\
0.2421875	0.669454701740861\\
0.244140625	3.98260596127563\\
0.24609375	0.937750329680896\\
0.248046875	0.179580889189286\\
0.25	0.0323552502471842\\
0.251953125	0.0176593175195353\\
0.25390625	0.00298448578926636\\
0.255859375	0.000187801128431619\\
0.2578125	2.27638986499538e-05\\
0.259765625	3.79995397290274e-06\\
0.26171875	2.07410341945548e-07\\
0.263671875	3.02167027525722e-08\\
0.265625	2.36264169267898e-09\\
0.267578125	1.15426262838539e-10\\
0.26953125	4.15281566438963e-11\\
0.271484375	2.67598569086001e-12\\
0.2734375	2.77131427612803e-13\\
0.275390625	5.46938248263592e-14\\
0.27734375	8.45742564703243e-14\\
0.279296875	2.74305754395744e-13\\
0.28125	9.28476232589784e-13\\
0.283203125	1.18343389808537e-11\\
0.28515625	8.33942542812094e-11\\
0.287109375	4.10137924073244e-10\\
0.2890625	3.84802084030927e-09\\
0.291015625	1.24151551083831e-07\\
0.29296875	1.39157031612735e-06\\
0.294921875	1.26151890106748e-05\\
0.296875	0.000138286355319706\\
0.298828125	0.00236721685423942\\
0.30078125	0.0276484975903714\\
0.302734375	0.154456073481602\\
0.3046875	4.08817924837701\\
0.306640625	3.83543698063252\\
0.30859375	0.376306524051452\\
0.310546875	1.14363465586306\\
0.3125	0.187883050061126\\
0.314453125	0.025457274439995\\
0.31640625	0.0022667950337201\\
0.318359375	0.000347062612514966\\
0.3203125	4.44530844026595e-05\\
0.322265625	2.32700675659334e-06\\
0.32421875	8.72410890522021e-08\\
0.326171875	1.39091954322463e-08\\
0.328125	8.56748447376873e-10\\
0.330078125	2.99066396391849e-11\\
0.33203125	4.90694583216389e-12\\
0.333984375	5.26154674062878e-13\\
0.3359375	1.3079945270779e-13\\
0.337890625	2.25407626681205e-14\\
0.33984375	1.23380252901831e-15\\
0.341796875	3.23812908110236e-16\\
0.34375	3.52088985962893e-17\\
0.345703125	4.64676275796774e-18\\
0.34765625	2.40435899857535e-19\\
0.349609375	8.74183815070188e-21\\
0.3515625	9.66547005847546e-21\\
0.353515625	2.33686964220683e-19\\
0.35546875	7.03306181114404e-18\\
0.357421875	8.28813226817852e-17\\
0.359375	3.89263930384747e-16\\
0.361328125	5.48660497620315e-15\\
0.36328125	8.46703646098574e-15\\
0.365234375	2.9006153459165e-14\\
0.3671875	3.97783114086055e-13\\
0.369140625	3.93065916291502e-12\\
0.37109375	2.08006449846097e-11\\
0.373046875	1.3307316698016e-10\\
0.375	5.53554887351642e-10\\
0.376953125	1.64161659713165e-09\\
0.37890625	4.97624977073095e-09\\
0.380859375	3.09083513231044e-08\\
0.3828125	3.21579414856045e-07\\
0.384765625	2.3202416800479e-06\\
0.38671875	4.79945666410043e-06\\
0.388671875	7.46937823034715e-05\\
0.390625	0.000761092938749686\\
0.392578125	0.00353881738519145\\
0.39453125	0.124334622610339\\
0.396484375	1.98359089680423\\
0.3984375	2.81965184084546\\
0.400390625	0.107852954831615\\
0.40234375	0.00282620267582502\\
0.404296875	0.000463090052091918\\
0.40625	1.63540055677848e-05\\
0.408203125	8.6364112693553e-07\\
0.41015625	2.49332352407142e-07\\
0.412109375	4.31284812013937e-08\\
0.4140625	4.6771779578443e-08\\
0.416015625	6.82449524759318e-08\\
0.41796875	1.08731939977063e-08\\
0.419921875	4.19721885035224e-10\\
0.421875	7.08299675816759e-11\\
0.423828125	6.35948549938673e-12\\
0.42578125	4.35625844539571e-13\\
0.427734375	6.99777795502225e-14\\
0.4296875	2.07463584109918e-14\\
0.431640625	4.32272062952931e-15\\
0.43359375	2.85778468472235e-16\\
0.435546875	3.27648417494256e-17\\
0.4375	2.93251424121702e-18\\
0.439453125	3.31568641116923e-19\\
0.44140625	3.06203731661867e-20\\
0.443359375	3.55820973114862e-21\\
0.4453125	2.96739753778776e-22\\
0.447265625	2.13563487986487e-23\\
0.44921875	1.13634707596768e-24\\
0.451171875	8.74467458454824e-26\\
0.453125	3.25410475122377e-26\\
0.455078125	2.64043151989733e-27\\
0.45703125	6.82899431758441e-29\\
0.458984375	2.96684360440562e-30\\
0.4609375	2.30803841523693e-31\\
0.462890625	2.49153902405492e-32\\
0.46484375	3.81546181918973e-33\\
0.466796875	1.23023226382255e-33\\
0.46875	8.12137801152504e-35\\
0.470703125	1.31119672392639e-35\\
0.47265625	3.81408842147252e-37\\
0.474609375	2.39929623890469e-38\\
0.4765625	5.2410717255833e-40\\
0.478515625	1.6476440472725e-41\\
0.48046875	5.91767851740692e-41\\
0.482421875	4.49130368835203e-40\\
0.484375	1.09674126027503e-39\\
0.486328125	8.67835155282432e-39\\
0.48828125	1.34318599204668e-37\\
0.490234375	3.3740476548455e-36\\
0.4921875	9.87929612154289e-36\\
0.494140625	1.36851734835508e-34\\
0.49609375	6.77259234846521e-34\\
0.498046875	4.7710698011866e-33\\
0.5	6.9302954412875e-32\\
0.501953125	4.96326113712371e-31\\
0.50390625	6.27888957377283e-30\\
0.505859375	2.06707808273495e-29\\
0.5078125	1.90986047043593e-28\\
0.509765625	3.68173721148751e-28\\
0.51171875	5.90337410317894e-28\\
0.513671875	3.35132144708654e-27\\
0.515625	3.54872574853514e-26\\
0.517578125	3.55701368101296e-25\\
0.51953125	3.60110653901504e-24\\
0.521484375	2.33125632842512e-23\\
0.5234375	2.11984049636593e-22\\
0.525390625	9.02313550822668e-22\\
0.52734375	1.07510037635715e-20\\
0.529296875	9.80212532714258e-21\\
0.53125	1.25489252384542e-19\\
0.533203125	6.4470021839265e-19\\
0.53515625	2.26701141760752e-18\\
0.537109375	1.04757206806094e-17\\
0.5390625	3.44651516181623e-16\\
0.541015625	4.47762894483904e-15\\
0.54296875	9.58781134587783e-14\\
0.544921875	1.83878941786298e-12\\
0.546875	4.56917014655964e-11\\
0.548828125	1.0085074502155e-09\\
0.55078125	2.24545133977544e-08\\
0.552734375	1.24050227083177e-07\\
0.5546875	1.89639038200546e-06\\
0.556640625	4.48009323713222e-06\\
0.55859375	3.97234445924586e-05\\
0.560546875	0.000439111551355777\\
0.5625	0.0085050762042314\\
0.564453125	0.187634958803277\\
0.56640625	3.54071576051916\\
0.568359375	8.36071612272344\\
0.5703125	1.28304817307997\\
0.572265625	0.427153943266979\\
0.57421875	0.0501033586713064\\
0.576171875	0.0859775206237332\\
0.578125	0.880202650669279\\
0.580078125	0.797843817350621\\
0.58203125	0.470299168504092\\
0.583984375	0.106945339065144\\
0.5859375	0.00920633813941094\\
0.587890625	0.000811671815504515\\
0.58984375	4.04800757659661e-05\\
0.591796875	3.20329059862387e-06\\
0.59375	1.2592491011698e-07\\
0.595703125	3.47606322626642e-09\\
0.59765625	3.16551939046192e-10\\
0.599609375	1.28795934909224e-10\\
0.6015625	1.56037341544228e-11\\
0.603515625	2.19292459954518e-12\\
0.60546875	2.4972839995357e-13\\
0.607421875	3.62463782898061e-14\\
0.609375	1.38628572865066e-14\\
0.611328125	1.64369521791609e-15\\
0.61328125	2.06736494277921e-16\\
0.615234375	2.82324938137557e-17\\
0.6171875	5.34782114679133e-19\\
0.619140625	5.16256988806065e-20\\
0.62109375	5.5141675865774e-21\\
0.623046875	2.89045904743243e-22\\
0.625	2.14505684410638e-23\\
0.626953125	1.46879391260235e-24\\
0.62890625	2.17950981006842e-25\\
0.630859375	6.25876105651927e-26\\
0.6328125	1.95201077442533e-26\\
0.634765625	5.37504390582641e-27\\
0.63671875	1.09861497127182e-27\\
0.638671875	1.09248792633962e-28\\
0.640625	9.04864320823482e-30\\
0.642578125	1.91140593462353e-30\\
0.64453125	5.08479147596176e-29\\
0.646484375	4.42021473622976e-28\\
0.6484375	8.70323337717813e-27\\
0.650390625	1.68348206866785e-25\\
0.65234375	2.42232633285746e-24\\
0.654296875	3.02639344563289e-23\\
0.65625	3.98727131805893e-22\\
0.658203125	2.75081615106781e-21\\
0.66015625	5.23355770952701e-20\\
0.662109375	1.16331090029942e-19\\
0.6640625	3.29602116983718e-18\\
0.666015625	3.06735812343092e-17\\
0.66796875	5.33382702247335e-16\\
0.669921875	3.5389836553958e-15\\
0.671875	3.44629004139282e-15\\
0.673828125	2.13215286961501e-14\\
0.67578125	5.02540543734791e-13\\
0.677734375	4.69391924367125e-12\\
0.6796875	9.52426361551674e-11\\
0.681640625	3.03766747498286e-10\\
0.68359375	3.11926917273965e-09\\
0.685546875	8.52745112742269e-08\\
0.6875	8.91333889595455e-07\\
0.689453125	5.41981329579165e-06\\
0.69140625	7.55095012979605e-05\\
0.693359375	0.00118599003796547\\
0.6953125	0.016431171113086\\
0.697265625	0.0815010903874631\\
0.69921875	0.678939305200628\\
0.701171875	0.18060534940668\\
0.703125	0.0113743149695105\\
0.705078125	0.000431171266848989\\
0.70703125	3.75175769572024e-05\\
0.708984375	4.42263227080411e-06\\
0.7109375	1.92311794213636e-06\\
0.712890625	2.12587591562097e-07\\
0.71484375	5.19957646718888e-09\\
0.716796875	7.2457508275012e-10\\
0.71875	1.24433160334988e-10\\
0.720703125	1.37278898199436e-11\\
0.72265625	4.4528200828051e-13\\
0.724609375	3.29908694038361e-14\\
0.7265625	2.61506041641252e-15\\
0.728515625	1.07292721453326e-16\\
0.73046875	5.84302859804742e-18\\
0.732421875	3.59128080219107e-19\\
0.734375	7.19175255989653e-20\\
0.736328125	1.24094760328539e-20\\
0.73828125	1.52372015733439e-21\\
0.740234375	4.53856171258913e-23\\
0.7421875	2.99257558128629e-24\\
0.744140625	9.27816513565816e-26\\
0.74609375	7.07890416361697e-27\\
0.748046875	8.3681788608744e-28\\
0.75	9.328553499031e-29\\
0.751953125	1.59907747726251e-28\\
0.75390625	1.10497516073751e-27\\
0.755859375	1.45934174542819e-26\\
0.7578125	2.65774880185668e-25\\
0.759765625	1.11837228895905e-24\\
0.76171875	7.2072916019609e-24\\
0.763671875	1.43754478738213e-22\\
0.765625	1.06213463599882e-21\\
0.767578125	1.07232217688744e-20\\
0.76953125	2.46880538235164e-19\\
0.771484375	2.11282778786939e-18\\
0.7734375	1.43200252031063e-17\\
0.775390625	3.47935849668132e-17\\
0.77734375	1.8233023885281e-16\\
0.779296875	1.45705422259571e-15\\
0.78125	8.48154818518021e-15\\
0.783203125	4.67069383382671e-14\\
0.78515625	2.62272600423733e-13\\
0.787109375	1.43526565272985e-12\\
0.7890625	1.53721409177287e-11\\
0.791015625	2.0814604259022e-10\\
0.79296875	1.65310441671095e-09\\
0.794921875	2.37730345166568e-08\\
0.796875	9.57199869573532e-08\\
0.798828125	1.70345349236178e-06\\
0.80078125	3.2393360015918e-05\\
0.802734375	0.00109488868517435\\
0.8046875	0.0144088944859976\\
0.806640625	0.365617202182328\\
0.80859375	4.23498538165534\\
0.810546875	3.47832078172243\\
0.8125	0.35911954233344\\
0.814453125	0.0156497610544645\\
0.81640625	0.00232261696678408\\
0.818359375	0.000146289806997976\\
0.8203125	1.31755236269403e-05\\
0.822265625	2.51395427620945e-06\\
0.82421875	6.71222612514424e-07\\
0.826171875	5.30718534306737e-08\\
0.828125	2.94737015636507e-09\\
0.830078125	2.67448907368619e-10\\
0.83203125	5.83833495626157e-11\\
0.833984375	3.24854210426034e-10\\
0.8359375	3.19567414048157e-09\\
0.837890625	3.00795801451015e-08\\
0.83984375	7.60787652873894e-07\\
0.841796875	2.23069934556936e-05\\
0.84375	0.000367049136986798\\
0.845703125	0.00387846441060769\\
0.84765625	0.0292614092855922\\
0.849609375	0.424961757898883\\
0.8515625	1.02022234633302\\
0.853515625	6.1582910756378\\
0.85546875	3.68324047637131\\
0.857421875	0.416885987563922\\
0.859375	0.0230279944986146\\
0.861328125	0.00107569716317656\\
0.86328125	0.000127106711839742\\
0.865234375	0.000157169382689516\\
0.8671875	5.33276978758086e-05\\
0.869140625	3.64649391469259e-06\\
0.87109375	3.36274324985455e-07\\
0.873046875	3.38276673748698e-08\\
0.875	3.05635747634299e-09\\
0.876953125	2.79983451067336e-10\\
0.87890625	1.43154831533647e-11\\
0.880859375	2.04109536240508e-12\\
0.8828125	7.56888598189102e-13\\
0.884765625	7.10780377091562e-14\\
0.88671875	2.08412215975124e-15\\
0.888671875	1.83550771733477e-16\\
0.890625	4.32916702582035e-17\\
0.892578125	4.08647884467955e-17\\
0.89453125	2.44441052356094e-17\\
0.896484375	2.31507089185507e-18\\
0.8984375	2.63433883729931e-19\\
0.900390625	1.47887122943964e-20\\
0.90234375	1.34757112133229e-21\\
0.904296875	7.6276025322611e-23\\
0.90625	2.11178499865657e-23\\
0.908203125	3.64465456692701e-22\\
0.91015625	2.0615179897352e-21\\
0.912109375	1.07665903600574e-20\\
0.9140625	2.42777889700344e-19\\
0.916015625	2.76943519727291e-18\\
0.91796875	3.26570107994e-17\\
0.919921875	6.95311302832879e-16\\
0.921875	1.10308217596017e-14\\
0.923828125	1.34520930886893e-13\\
0.92578125	1.57195032505644e-12\\
0.927734375	2.71116246623778e-13\\
0.9296875	3.79613536825871e-14\\
0.931640625	2.45640261455901e-13\\
0.93359375	9.04434776726456e-13\\
0.935546875	8.72080571631765e-13\\
0.9375	4.18094655804654e-12\\
0.939453125	5.32461386455503e-11\\
0.94140625	2.17001985461926e-10\\
0.943359375	7.8588068938519e-10\\
0.9453125	1.58102688142e-08\\
0.947265625	1.6741491492041e-07\\
0.94921875	3.99742253753967e-06\\
0.951171875	3.44942548942715e-05\\
0.953125	0.000385683878717898\\
0.955078125	0.00435758111925614\\
0.95703125	0.0221569676476644\\
0.958984375	0.167800211647346\\
0.9609375	2.43027488294705\\
0.962890625	6.41922075859807\\
0.96484375	0.689040740649647\\
0.966796875	0.01651848459298\\
0.96875	0.00288662279178573\\
0.970703125	0.000368692530855399\\
0.97265625	8.82242803680883e-05\\
0.974609375	5.80939941137533e-06\\
0.9765625	2.40328852116533e-06\\
0.978515625	1.63389752478611e-07\\
0.98046875	3.10123831000736e-08\\
0.982421875	2.81783178490606e-07\\
0.984375	5.90060444419393e-07\\
0.986328125	1.65635344525899e-05\\
0.98828125	0.000318845423439299\\
0.990234375	0.00143044457321792\\
0.9921875	0.0128556292541212\\
0.994140625	0.121732896994981\\
0.99609375	1.44640726441044\\
0.998046875	2.03696641957704\\
1	0\\
};
\end{axis}
\end{tikzpicture}%

%% file: pics/square/GroundState_1_4_3.tex
%
%
\begin{tikzpicture}

\begin{axis}[%
width=0.9in,
height=0.9in,
scale only axis,
xmin=0,
xmax=1,
xtick={0,0.5,1},
ymin=0,
ymax=10,
yminorticks=true,
yticklabels = {},
axis background/.style={fill=white},
]
\addplot [color=mycolor2, thick]
table[row sep=crcr]{%
0	0\\
0.001953125	9.99996013241289e-06\\
0.00390625	7.11370722924576e-05\\
0.005859375	0.000912108278948508\\
0.0078125	0.00376983001896518\\
0.009765625	0.0635881886965735\\
0.01171875	0.536803349734199\\
0.013671875	5.51952329383834\\
0.015625	2.31810765827752\\
0.017578125	0.0469017291445571\\
0.01953125	0.00207385429575643\\
0.021484375	0.000471182448525363\\
0.0234375	2.23210697033337e-05\\
0.025390625	1.56760653530016e-06\\
0.02734375	2.94808497097602e-06\\
0.029296875	4.47578625611074e-05\\
0.03125	0.000402319293700733\\
0.033203125	0.00195291179917368\\
0.03515625	0.0153948222233082\\
0.037109375	0.34222218101686\\
0.0390625	0.706600219186749\\
0.041015625	2.02415263643951\\
0.04296875	3.91580685787974\\
0.044921875	2.09381520975593\\
0.046875	1.84750036293762\\
0.048828125	0.529943383251958\\
0.05078125	0.101186308252399\\
0.052734375	0.0538585223931245\\
0.0546875	0.661890996620034\\
0.056640625	6.83274394248575\\
0.05859375	0.792329182218004\\
0.060546875	0.0571853468696211\\
0.0625	0.00564712724607774\\
0.064453125	0.000600596250675149\\
0.06640625	2.56563646433751e-05\\
0.068359375	7.89183496568801e-07\\
0.0703125	1.9554574109191e-06\\
0.072265625	2.94593521720852e-05\\
0.07421875	0.000656227862883468\\
0.076171875	0.0150438892777226\\
0.078125	0.17322715088097\\
0.080078125	4.84559599738809\\
0.08203125	1.27679590980294\\
0.083984375	0.967200767044688\\
0.0859375	0.338952464463995\\
0.087890625	0.0242278440296028\\
0.08984375	0.000502044928384368\\
0.091796875	1.12951178863024e-05\\
0.09375	2.66040031599384e-07\\
0.095703125	6.66707126220301e-09\\
0.09765625	2.74899835313184e-10\\
0.099609375	6.67510024345778e-11\\
0.1015625	1.97191262243087e-11\\
0.103515625	1.04817157834288e-12\\
0.10546875	4.08295562103794e-13\\
0.107421875	1.14463614482559e-14\\
0.109375	4.21982892628317e-16\\
0.111328125	5.00156278200516e-17\\
0.11328125	6.48598477429678e-17\\
0.115234375	6.58738591407948e-18\\
0.1171875	3.05243974378314e-18\\
0.119140625	4.05702404989942e-19\\
0.12109375	4.77209828151611e-20\\
0.123046875	1.7054446715651e-21\\
0.125	4.50106043070596e-22\\
0.126953125	6.57142258398692e-23\\
0.12890625	7.48335901598174e-24\\
0.130859375	2.63801629065096e-24\\
0.1328125	1.54586048969473e-25\\
0.134765625	1.55579486350968e-26\\
0.13671875	1.45541942935838e-27\\
0.138671875	1.05978678754211e-26\\
0.140625	3.46769917290492e-25\\
0.142578125	2.85222084680741e-24\\
0.14453125	6.49021067692867e-23\\
0.146484375	7.57007867628337e-22\\
0.1484375	3.15115748090422e-21\\
0.150390625	3.99330084237873e-20\\
0.15234375	6.19274070486613e-19\\
0.154296875	1.67424167989443e-18\\
0.15625	5.80283537859238e-18\\
0.158203125	7.58791915675135e-17\\
0.16015625	6.80392462335486e-16\\
0.162109375	5.45774585124762e-15\\
0.1640625	1.15397087222541e-13\\
0.166015625	1.51644137804094e-12\\
0.16796875	1.92699091887925e-11\\
0.169921875	9.38012571302452e-11\\
0.171875	1.8992167052426e-09\\
0.173828125	9.35997836767524e-09\\
0.17578125	3.2746544846416e-08\\
0.177734375	4.68263548720767e-07\\
0.1796875	6.76997383598027e-06\\
0.181640625	8.68693995444365e-05\\
0.18359375	0.000585434635743901\\
0.185546875	0.00317903336363173\\
0.1875	0.00164371546549209\\
0.189453125	0.000111197497540829\\
0.19140625	3.79271017741648e-06\\
0.193359375	2.08867108935747e-07\\
0.1953125	2.23031460015954e-08\\
0.197265625	1.98966697495838e-09\\
0.19921875	2.18771522553749e-10\\
0.201171875	2.13079959901337e-11\\
0.203125	1.5769740389402e-12\\
0.205078125	1.47355527337982e-13\\
0.20703125	7.9428492826192e-15\\
0.208984375	2.51354709939624e-16\\
0.2109375	1.24028448845748e-16\\
0.212890625	5.08478677294833e-16\\
0.21484375	7.17730275779872e-15\\
0.216796875	6.78603538917632e-14\\
0.21875	2.6784880777815e-13\\
0.220703125	2.31416131319076e-12\\
0.22265625	3.48170188831467e-11\\
0.224609375	8.05467533933813e-11\\
0.2265625	1.23453622740384e-09\\
0.228515625	1.06772153056358e-08\\
0.23046875	2.84799507611473e-08\\
0.232421875	9.46885276034822e-07\\
0.234375	7.05132398592811e-06\\
0.236328125	0.000158950006336823\\
0.23828125	0.00234552607683987\\
0.240234375	0.0117884676330057\\
0.2421875	0.299746611612305\\
0.244140625	1.86018488711334\\
0.24609375	0.356788390378153\\
0.248046875	0.0563401759648582\\
0.25	0.00760247300778587\\
0.251953125	0.00299672058783\\
0.25390625	0.00039727168659899\\
0.255859375	2.21949943835863e-05\\
0.2578125	2.1197240257592e-06\\
0.259765625	2.97845935786287e-07\\
0.26171875	1.43155939733114e-08\\
0.263671875	1.74825647154254e-09\\
0.265625	1.19445492580986e-10\\
0.267578125	5.0861796160705e-12\\
0.26953125	1.33322133691618e-12\\
0.271484375	7.56603194420115e-14\\
0.2734375	6.54489931304472e-15\\
0.275390625	1.27141191626732e-15\\
0.27734375	7.52701487371683e-15\\
0.279296875	3.03028369810482e-14\\
0.28125	1.59613818493635e-13\\
0.283203125	2.35243694773198e-12\\
0.28515625	2.01474638271124e-11\\
0.287109375	1.22075207248299e-10\\
0.2890625	1.39859345032554e-09\\
0.291015625	5.00285211649315e-08\\
0.29296875	6.45357914211174e-07\\
0.294921875	6.83547272310765e-06\\
0.296875	8.61102645025167e-05\\
0.298828125	0.0017083235967863\\
0.30078125	0.0230693076116897\\
0.302734375	0.159266070341995\\
0.3046875	4.70820072777661\\
0.306640625	4.47833894742134\\
0.30859375	0.142299299407441\\
0.310546875	0.159588685691093\\
0.3125	0.0213812148044171\\
0.314453125	0.00247946134833756\\
0.31640625	0.000188434966882565\\
0.318359375	2.35958571520882e-05\\
0.3203125	2.61233199934776e-06\\
0.322265625	1.2085629396918e-07\\
0.32421875	4.05418328901678e-09\\
0.326171875	5.32981651873489e-10\\
0.328125	2.9088542317862e-11\\
0.330078125	9.096812430089e-13\\
0.33203125	1.21387757610862e-13\\
0.333984375	1.07381995892296e-14\\
0.3359375	2.08183274592033e-15\\
0.337890625	2.97416963741485e-16\\
0.33984375	1.42831567339373e-17\\
0.341796875	2.95079746207649e-18\\
0.34375	2.72694413079201e-19\\
0.345703125	2.99183008964777e-20\\
0.34765625	1.38435887589337e-21\\
0.349609375	4.53434621564202e-23\\
0.3515625	6.26829308758768e-23\\
0.353515625	1.716664316366e-21\\
0.35546875	5.72503603809113e-20\\
0.357421875	7.70573721842126e-19\\
0.359375	4.45048558975439e-18\\
0.361328125	7.26766108976121e-17\\
0.36328125	1.77467863101172e-16\\
0.365234375	8.33545849478953e-16\\
0.3671875	1.31485232572854e-14\\
0.369140625	1.51085430731161e-13\\
0.37109375	1.00783039087359e-12\\
0.373046875	7.7522841853007e-12\\
0.375	4.03625708241536e-11\\
0.376953125	1.53004647756563e-10\\
0.37890625	7.37710947715984e-10\\
0.380859375	5.60160185364268e-09\\
0.3828125	7.08892043144462e-08\\
0.384765625	5.94864419439264e-07\\
0.38671875	2.08216779529419e-06\\
0.388671875	3.70793763530507e-05\\
0.390625	0.00043876027112557\\
0.392578125	0.00253422402748194\\
0.39453125	0.0994768501789982\\
0.396484375	1.82528628267942\\
0.3984375	2.65154118912443\\
0.400390625	0.0912992455283376\\
0.40234375	0.00212060227143742\\
0.404296875	0.000290466577227525\\
0.40625	9.18621164239379e-06\\
0.408203125	4.20598120982422e-07\\
0.41015625	9.57404781166428e-08\\
0.412109375	1.1935758298795e-08\\
0.4140625	2.13240436602384e-09\\
0.416015625	1.57139469218431e-09\\
0.41796875	2.11338917577743e-10\\
0.419921875	7.2597122025802e-12\\
0.421875	1.0100786822311e-12\\
0.423828125	7.91212170913123e-14\\
0.42578125	4.53875519636202e-15\\
0.427734375	6.07882930427487e-16\\
0.4296875	1.26472042719034e-16\\
0.431640625	2.18449274821587e-17\\
0.43359375	1.25969894673585e-18\\
0.435546875	1.23777648033184e-19\\
0.4375	9.62329027309093e-21\\
0.439453125	9.01198807502984e-22\\
0.44140625	7.13659227105709e-23\\
0.443359375	7.07756186170225e-24\\
0.4453125	5.01722073071972e-25\\
0.447265625	3.17926010418517e-26\\
0.44921875	1.49969872597251e-27\\
0.451171875	9.36345915125373e-29\\
0.453125	2.6682596839078e-29\\
0.455078125	1.90533358974117e-30\\
0.45703125	4.42926076007765e-32\\
0.458984375	1.71171882615325e-33\\
0.4609375	1.16029500094756e-34\\
0.462890625	1.04304556100906e-35\\
0.46484375	1.30728587896213e-36\\
0.466796875	3.37600896016299e-37\\
0.46875	1.93086462298306e-38\\
0.470703125	2.58695252667165e-39\\
0.47265625	6.80831669055703e-41\\
0.474609375	3.77918018645728e-42\\
0.4765625	7.52779212613013e-44\\
0.478515625	2.05064012288065e-45\\
0.48046875	6.48125080063213e-45\\
0.482421875	5.79185106688513e-44\\
0.484375	1.9764091532572e-43\\
0.486328125	1.90315331839885e-42\\
0.48828125	3.42358954781585e-41\\
0.490234375	9.56226511940892e-40\\
0.4921875	3.83207923862468e-39\\
0.494140625	6.20827403474285e-38\\
0.49609375	3.7139354015976e-37\\
0.498046875	3.48874318932717e-36\\
0.5	5.72157724408227e-35\\
0.501953125	5.13735032380927e-34\\
0.50390625	7.43413738606027e-33\\
0.505859375	3.16726298613555e-32\\
0.5078125	3.61027544980564e-31\\
0.509765625	1.02280003509548e-30\\
0.51171875	2.33111376400272e-30\\
0.513671875	3.02523048946653e-29\\
0.515625	3.72243249238804e-28\\
0.517578125	4.36334991869316e-27\\
0.51953125	5.14630228747873e-26\\
0.521484375	3.99862712280829e-25\\
0.5234375	4.24726318579096e-24\\
0.525390625	2.28755499264405e-23\\
0.52734375	3.15868026384192e-22\\
0.529296875	4.62957601914383e-22\\
0.53125	7.97168311378705e-21\\
0.533203125	4.99594957523256e-20\\
0.53515625	2.32198043346062e-19\\
0.537109375	1.33907282407996e-18\\
0.5390625	4.94006578874457e-17\\
0.541015625	7.3033644212276e-16\\
0.54296875	1.76413132651192e-14\\
0.544921875	3.79993359288347e-13\\
0.546875	1.06137111241845e-11\\
0.548828125	2.66725933345185e-10\\
0.55078125	6.63640746890298e-09\\
0.552734375	4.46243209851598e-08\\
0.5546875	7.79170448726573e-07\\
0.556640625	2.55783408596822e-06\\
0.55859375	2.81297833170606e-05\\
0.560546875	0.000354996163734696\\
0.5625	0.00809124428924679\\
0.564453125	0.200721671543524\\
0.56640625	4.24236996184383\\
0.568359375	10.4433752065298\\
0.5703125	1.38079264863919\\
0.572265625	0.340949618587504\\
0.57421875	0.0321778681629369\\
0.576171875	0.00127987427792933\\
0.578125	0.000545535405006271\\
0.580078125	0.000309107893588354\\
0.58203125	9.24494518133812e-05\\
0.583984375	1.750723264655e-05\\
0.5859375	1.28983592826444e-06\\
0.587890625	9.87845759539526e-08\\
0.58984375	4.32095661872999e-09\\
0.591796875	2.9945560791754e-10\\
0.59375	1.05963330468575e-11\\
0.595703125	2.62620612932842e-13\\
0.59765625	1.98849131571946e-14\\
0.599609375	6.0553615747885e-15\\
0.6015625	6.12233979326151e-16\\
0.603515625	7.31248740692897e-17\\
0.60546875	7.04359912429218e-18\\
0.607421875	7.27484357934192e-19\\
0.609375	2.07501065647022e-19\\
0.611328125	1.91786785674011e-20\\
0.61328125	2.0369426920871e-21\\
0.615234375	2.34098535583157e-22\\
0.6171875	4.04007665421573e-24\\
0.619140625	3.23837970035649e-25\\
0.62109375	2.99679701974079e-26\\
0.623046875	1.39843584303221e-27\\
0.625	9.02611679606321e-29\\
0.626953125	5.33963475893397e-30\\
0.62890625	6.28075124736025e-31\\
0.630859375	1.40683823278228e-31\\
0.6328125	3.27112128711713e-32\\
0.634765625	6.92610819656524e-33\\
0.63671875	1.16948888484542e-33\\
0.638671875	1.00607210424045e-34\\
0.640625	6.80278515311442e-36\\
0.642578125	6.3201653256722e-37\\
0.64453125	1.5341245730208e-38\\
0.646484375	1.3854580791555e-39\\
0.6484375	7.56856661977617e-41\\
0.650390625	2.51305131186509e-42\\
0.65234375	2.20006048783973e-43\\
0.654296875	1.42721056614158e-44\\
0.65625	1.00305624544502e-45\\
0.658203125	9.11627754243336e-47\\
0.66015625	9.8620635941748e-48\\
0.662109375	1.99974969962138e-48\\
0.6640625	5.75190960592085e-50\\
0.666015625	6.47188613702228e-51\\
0.66796875	1.2330471691042e-50\\
0.669921875	8.07469117001939e-50\\
0.671875	7.99227910972479e-50\\
0.673828125	4.97753849155363e-49\\
0.67578125	1.17859530460153e-47\\
0.677734375	1.10497160756387e-46\\
0.6796875	2.24804508102786e-45\\
0.681640625	7.21125088218544e-45\\
0.68359375	7.43486788236214e-44\\
0.685546875	2.03780450818115e-42\\
0.6875	2.13777404046807e-41\\
0.689453125	1.30550439177459e-40\\
0.69140625	1.82508264358879e-39\\
0.693359375	2.87507350839083e-38\\
0.6953125	3.99525888299832e-37\\
0.697265625	1.99362348407715e-36\\
0.69921875	1.66380260318829e-35\\
0.701171875	4.40843536811531e-36\\
0.703125	2.76831521961929e-37\\
0.705078125	1.04680446400521e-38\\
0.70703125	9.0753135086034e-40\\
0.708984375	1.06351386827771e-40\\
0.7109375	4.5918205684258e-41\\
0.712890625	5.05743460225699e-42\\
0.71484375	1.2339895564351e-43\\
0.716796875	1.7107816023922e-44\\
0.71875	2.92552167991981e-45\\
0.720703125	3.21675503941336e-46\\
0.72265625	1.12904282870469e-47\\
0.724609375	3.13441084332108e-47\\
0.7265625	2.43198326924173e-46\\
0.728515625	6.7589981008704e-45\\
0.73046875	1.93412790220959e-43\\
0.732421875	2.8631860838351e-42\\
0.734375	3.37558073580334e-41\\
0.736328125	1.9747019692354e-40\\
0.73828125	1.84766728090403e-39\\
0.740234375	3.27337744651279e-38\\
0.7421875	9.03034393302242e-37\\
0.744140625	2.219924351318e-35\\
0.74609375	5.98567534692175e-34\\
0.748046875	6.32912338952749e-33\\
0.75	4.22164585571611e-32\\
0.751953125	9.45981450313934e-31\\
0.75390625	8.8485289098772e-30\\
0.755859375	1.32846089027486e-28\\
0.7578125	2.8132431542941e-27\\
0.759765625	1.45044911002254e-26\\
0.76171875	1.2517997787904e-25\\
0.763671875	2.81132730853368e-24\\
0.765625	2.50370225137446e-23\\
0.767578125	2.93517841772942e-22\\
0.76953125	7.67640448204184e-21\\
0.771484375	7.60264703875411e-20\\
0.7734375	6.10020031136557e-19\\
0.775390625	2.13483435856697e-18\\
0.77734375	1.42668304357843e-17\\
0.779296875	1.33569088101682e-16\\
0.78125	9.6798469877063e-16\\
0.783203125	6.52014950784639e-15\\
0.78515625	4.50876069009958e-14\\
0.787109375	3.10584924105658e-13\\
0.7890625	3.86377128238798e-12\\
0.791015625	5.96460985955971e-11\\
0.79296875	5.69394962665558e-10\\
0.794921875	9.36412779604938e-09\\
0.796875	4.6556361353542e-08\\
0.798828125	9.5786780036688e-07\\
0.80078125	2.06264679535294e-05\\
0.802734375	0.00077085146848175\\
0.8046875	0.0118966450400494\\
0.806640625	0.335785564621827\\
0.80859375	4.41216133287736\\
0.810546875	3.5196321699542\\
0.8125	0.323671694057847\\
0.814453125	0.0124947190639084\\
0.81640625	0.00154592462024787\\
0.818359375	8.53617466913136e-05\\
0.8203125	6.55519852698288e-06\\
0.822265625	9.56135383850721e-07\\
0.82421875	2.07777398703869e-07\\
0.826171875	1.43343973060323e-08\\
0.828125	6.96018013885787e-10\\
0.830078125	5.22568092134583e-11\\
0.83203125	9.59491820768157e-12\\
0.833984375	5.94545101554581e-11\\
0.8359375	7.18537626305298e-10\\
0.837890625	7.8685340656353e-09\\
0.83984375	2.23092961926267e-07\\
0.841796875	7.3016745881547e-06\\
0.84375	0.00013576475228626\\
0.845703125	0.00168437849874537\\
0.84765625	0.0153556917461111\\
0.849609375	0.252424343543384\\
0.8515625	0.783758123587307\\
0.853515625	6.29311673115184\\
0.85546875	3.55647769151801\\
0.857421875	0.346706271754349\\
0.859375	0.0168924875057826\\
0.861328125	0.00069491245958332\\
0.86328125	6.73659952829836e-05\\
0.865234375	2.24055073536576e-05\\
0.8671875	5.99976421168064e-06\\
0.869140625	3.53523683254037e-07\\
0.87109375	2.82416168039548e-08\\
0.873046875	2.39472694368764e-09\\
0.875	1.86805319567589e-10\\
0.876953125	1.49124334338628e-11\\
0.87890625	6.71579401470337e-13\\
0.880859375	7.86046156509882e-14\\
0.8828125	2.12929511511506e-14\\
0.884765625	1.73981726453696e-15\\
0.88671875	4.57024416525154e-17\\
0.888671875	3.40265197174173e-18\\
0.890625	5.92124567295943e-19\\
0.892578125	2.43668052726954e-19\\
0.89453125	1.05354612970485e-19\\
0.896484375	8.6260728890749e-21\\
0.8984375	8.4395591696898e-22\\
0.900390625	4.16164682097755e-23\\
0.90234375	3.28282743826952e-24\\
0.904296875	1.61560375932548e-25\\
0.90625	5.91876178892734e-27\\
0.908203125	3.36944368191283e-27\\
0.91015625	2.14541278187616e-26\\
0.912109375	1.51711420213537e-25\\
0.9140625	3.90004322389287e-24\\
0.916015625	5.31848683599321e-23\\
0.91796875	7.3163131642761e-22\\
0.919921875	1.77984874554966e-20\\
0.921875	3.29475885017691e-19\\
0.923828125	4.63751322316674e-18\\
0.92578125	6.39290876294873e-17\\
0.927734375	4.91762044601271e-17\\
0.9296875	6.72377974113801e-16\\
0.931640625	8.58568180214602e-15\\
0.93359375	4.18277509837249e-14\\
0.935546875	6.63227865446733e-14\\
0.9375	5.44351018205007e-13\\
0.939453125	7.94762252039767e-12\\
0.94140625	3.95119881880226e-11\\
0.943359375	1.90625164370366e-10\\
0.9453125	4.43228034505446e-09\\
0.947265625	5.77971050466317e-08\\
0.94921875	1.53279747483341e-06\\
0.951171875	1.56420412187053e-05\\
0.953125	0.000203075161374793\\
0.955078125	0.00268236246687083\\
0.95703125	0.0162244020969174\\
0.958984375	0.159975947932465\\
0.9609375	2.62265693858544\\
0.962890625	7.21622376053954\\
0.96484375	0.705406273999061\\
0.966796875	0.0152063531989661\\
0.96875	0.0021815607048953\\
0.970703125	0.000232833592739526\\
0.97265625	4.35943336546653e-05\\
0.974609375	2.49875198591714e-06\\
0.9765625	6.50814520394232e-07\\
0.978515625	3.88413806281019e-08\\
0.98046875	2.14158225199128e-09\\
0.982421875	3.69748202698079e-10\\
0.984375	4.96202632054526e-11\\
0.986328125	1.82473018269345e-12\\
0.98828125	1.54638683800077e-12\\
0.990234375	7.6879651340849e-12\\
0.9921875	7.91723443657418e-11\\
0.994140625	8.44951186203691e-10\\
0.99609375	1.121106550943e-08\\
0.998046875	1.5714981556089e-08\\
1	0\\
};
\end{axis}
\end{tikzpicture}%

%% file: pics/square/GroundState_1_4_4.tex
%
%
\begin{tikzpicture}

\begin{axis}[%
width=0.9in,
height=0.9in,
scale only axis,
xmin=0,
xmax=1,
xtick={0,0.5,1},
ymin=0,
ymax=10,
yminorticks=true,
ytick={0,4,8},
yticklabel pos=right,
axis background/.style={fill=white},
]
\addplot [color=mycolor2, thick]
table[row sep=crcr]{%
0	0\\
0.001953125	4.60536266386309e-06\\
0.00390625	3.87850965039932e-05\\
0.005859375	0.00055489792164142\\
0.0078125	0.00271455396934531\\
0.009765625	0.052055375515253\\
0.01171875	0.487930958112965\\
0.013671875	6.08358674601113\\
0.015625	2.46119214797436\\
0.017578125	0.0463738654319463\\
0.01953125	0.00184062744500664\\
0.021484375	0.000359760901409788\\
0.0234375	1.55826995288673e-05\\
0.025390625	7.66988473961168e-07\\
0.02734375	9.78728846831671e-08\\
0.029296875	1.88681353861151e-08\\
0.03125	1.3454172178148e-07\\
0.033203125	7.53748039273697e-07\\
0.03515625	6.71301772227517e-06\\
0.037109375	0.000163752783295083\\
0.0390625	0.000420161101012557\\
0.041015625	0.0014600573364729\\
0.04296875	0.00392309476341234\\
0.044921875	0.00205069849736747\\
0.046875	0.00177758674970048\\
0.048828125	0.00146677499300618\\
0.05078125	0.00387431796193503\\
0.052734375	0.0322788537187583\\
0.0546875	0.541695526744008\\
0.056640625	5.9576020446735\\
0.05859375	0.646515096169974\\
0.060546875	0.0414518996426222\\
0.0625	0.0036596638799796\\
0.064453125	0.000344604564524302\\
0.06640625	1.35826172525184e-05\\
0.068359375	3.60014552676278e-07\\
0.0703125	4.34281667102098e-07\\
0.072265625	7.10153177045996e-06\\
0.07421875	0.000172836700138688\\
0.076171875	0.00430441350100795\\
0.078125	0.0553117064831767\\
0.080078125	1.66010553623453\\
0.08203125	0.409053734536678\\
0.083984375	0.149240928140564\\
0.0859375	0.0430314827642482\\
0.087890625	0.00280035621625561\\
0.08984375	5.40698524566783e-05\\
0.091796875	1.13248213985172e-06\\
0.09375	2.45795779659952e-08\\
0.095703125	5.74452102948267e-10\\
0.09765625	2.12249458365003e-11\\
0.099609375	4.35897761027199e-12\\
0.1015625	1.03664634144237e-12\\
0.103515625	4.90723997684895e-14\\
0.10546875	1.53945419378336e-14\\
0.107421875	3.99146206808462e-16\\
0.109375	1.34944574709192e-17\\
0.111328125	1.27849717743797e-18\\
0.11328125	1.00586749397043e-18\\
0.115234375	8.94019294000914e-20\\
0.1171875	3.1966230892399e-20\\
0.119140625	3.69764807739973e-21\\
0.12109375	3.87939513561042e-22\\
0.123046875	1.2542024012843e-23\\
0.125	2.77713446902296e-24\\
0.126953125	3.52118885749805e-25\\
0.12890625	3.48792495808397e-26\\
0.130859375	9.97103202872876e-27\\
0.1328125	5.32615897734029e-28\\
0.134765625	4.7188398117479e-29\\
0.13671875	1.79532717399997e-30\\
0.138671875	1.41472167046772e-31\\
0.140625	7.07542508594565e-33\\
0.142578125	5.50223888337966e-34\\
0.14453125	1.7045404124127e-35\\
0.146484375	2.38076746894381e-36\\
0.1484375	3.99735283312069e-37\\
0.150390625	1.95637798849754e-38\\
0.15234375	1.96310099470657e-39\\
0.154296875	6.02068195115226e-40\\
0.15625	1.42321418757671e-40\\
0.158203125	4.55432306096788e-42\\
0.16015625	4.61360804147906e-43\\
0.162109375	1.12638988448642e-43\\
0.1640625	1.42069166294791e-42\\
0.166015625	1.86324405021875e-41\\
0.16796875	2.36765799182296e-40\\
0.169921875	1.15251863317032e-39\\
0.171875	2.33353241970378e-38\\
0.173828125	1.15004329943058e-37\\
0.17578125	4.02350816709988e-37\\
0.177734375	5.75346897076785e-36\\
0.1796875	8.31814400094684e-35\\
0.181640625	1.06734859086663e-33\\
0.18359375	7.19312978915474e-33\\
0.185546875	3.90602115480717e-32\\
0.1875	2.0196036186208e-32\\
0.189453125	1.36626359433676e-33\\
0.19140625	4.66003440481885e-35\\
0.193359375	2.56631229394561e-36\\
0.1953125	2.74034696963232e-37\\
0.197265625	2.44466749752111e-38\\
0.19921875	2.6880060640179e-39\\
0.201171875	2.61810560798978e-40\\
0.203125	1.94360600210165e-41\\
0.205078125	2.48233343798388e-42\\
0.20703125	9.76198346844881e-42\\
0.208984375	3.21065538290414e-40\\
0.2109375	4.3848572056498e-39\\
0.212890625	2.40690952033382e-38\\
0.21484375	3.48108533611167e-37\\
0.216796875	3.36472569963636e-36\\
0.21875	1.36619150905566e-35\\
0.220703125	1.21294238557582e-34\\
0.22265625	1.85420526396573e-33\\
0.224609375	4.44514464196302e-33\\
0.2265625	7.00120997272757e-32\\
0.228515625	6.17974530862824e-31\\
0.23046875	1.71186565514882e-30\\
0.232421875	5.79167757854977e-29\\
0.234375	4.40181730593363e-28\\
0.236328125	1.00759142998386e-26\\
0.23828125	1.51188322319988e-25\\
0.240234375	7.80535138235428e-25\\
0.2421875	2.01811125090645e-23\\
0.244140625	1.25905767569343e-22\\
0.24609375	2.35650131351152e-23\\
0.248046875	3.63621498841262e-24\\
0.25	4.75734031381243e-25\\
0.251953125	1.80750460058732e-25\\
0.25390625	2.33384565141881e-26\\
0.255859375	1.56111894028284e-27\\
0.2578125	9.24632366992576e-27\\
0.259765625	4.92686726629127e-26\\
0.26171875	8.21539096636886e-25\\
0.263671875	9.49903170615776e-24\\
0.265625	1.01618282885866e-22\\
0.267578125	2.7745925632019e-21\\
0.26953125	3.37200194358751e-20\\
0.271484375	2.39818015261606e-19\\
0.2734375	7.3814939040555e-18\\
0.275390625	4.89401840215946e-17\\
0.27734375	1.51023540420491e-15\\
0.279296875	7.02825715229809e-15\\
0.28125	4.73074511706639e-14\\
0.283203125	7.7274716884906e-13\\
0.28515625	7.57842585411197e-12\\
0.287109375	5.31119857896705e-11\\
0.2890625	6.96864151839049e-10\\
0.291015625	2.69000889289122e-08\\
0.29296875	3.84318520948653e-07\\
0.294921875	4.55233574345804e-06\\
0.296875	6.34454173770282e-05\\
0.298828125	0.0013984787519687\\
0.30078125	0.020976660995076\\
0.302734375	0.167675535231427\\
0.3046875	5.37095763250999\\
0.306640625	5.17255687844818\\
0.30859375	0.133849265749109\\
0.310546875	0.0942303615724474\\
0.3125	0.0109233290138622\\
0.314453125	0.00113188908849213\\
0.31640625	7.68859602833551e-05\\
0.318359375	8.37485876995257e-06\\
0.3203125	8.34015276476267e-07\\
0.322265625	3.52446354666718e-08\\
0.32421875	1.08996095787458e-09\\
0.326171875	1.25131860426418e-10\\
0.328125	6.2480794123202e-12\\
0.330078125	1.80304691601827e-13\\
0.33203125	2.08561885135373e-14\\
0.333984375	1.62134870044203e-15\\
0.3359375	2.65723896745355e-16\\
0.337890625	3.32499590331791e-17\\
0.33984375	1.45419253185667e-18\\
0.341796875	2.55279013939923e-19\\
0.34375	2.10219261648769e-20\\
0.345703125	2.02429723541812e-21\\
0.34765625	8.62383197774058e-23\\
0.349609375	2.51826422431785e-24\\
0.3515625	1.15204484502481e-24\\
0.353515625	3.25508755443781e-23\\
0.35546875	1.17121564472851e-21\\
0.357421875	1.7373478960439e-20\\
0.359375	1.15935095692323e-19\\
0.361328125	2.10095286085767e-18\\
0.36328125	6.71275668037619e-18\\
0.365234375	3.80102012856841e-17\\
0.3671875	6.6203225886348e-16\\
0.369140625	8.48154490191473e-15\\
0.37109375	6.62028866733653e-14\\
0.373046875	5.79382283023281e-13\\
0.375	3.52118882054981e-12\\
0.376953125	1.57654000879608e-11\\
0.37890625	9.76440838718796e-11\\
0.380859375	8.50177218429224e-10\\
0.3828125	1.22923990937756e-08\\
0.384765625	1.15080279086895e-07\\
0.38671875	5.37709669644663e-07\\
0.388671875	1.05154886765351e-05\\
0.390625	0.000138579735186339\\
0.392578125	0.000930271055897791\\
0.39453125	0.039558896969231\\
0.396484375	0.803200772844151\\
0.3984375	1.18583525371664\\
0.400390625	0.0377711446753099\\
0.40234375	0.000804815406597121\\
0.404296875	9.70700366837305e-05\\
0.40625	2.83031717787662e-06\\
0.408203125	1.17272768492766e-07\\
0.41015625	2.27307154183474e-08\\
0.412109375	2.44669308536917e-09\\
0.4140625	3.10419905155655e-10\\
0.416015625	1.61806538840834e-10\\
0.41796875	1.92712207061178e-11\\
0.419921875	6.08107312078202e-13\\
0.421875	7.38818267362241e-14\\
0.423828125	5.24068665132693e-15\\
0.42578125	2.65925032348394e-16\\
0.427734375	3.13948263011704e-17\\
0.4296875	5.26796178090075e-18\\
0.431640625	7.96572627724053e-19\\
0.43359375	4.16108803980872e-20\\
0.435546875	3.65878574603692e-21\\
0.4375	2.56878273553311e-22\\
0.439453125	2.108779090834e-23\\
0.44140625	1.49619161169789e-24\\
0.443359375	1.3242331532839e-25\\
0.4453125	8.36005580848186e-27\\
0.447265625	4.82575165406455e-28\\
0.44921875	2.0842692396251e-29\\
0.451171875	1.13812682266766e-30\\
0.453125	2.70868014686958e-31\\
0.455078125	1.76112413549974e-32\\
0.45703125	3.78503075347454e-34\\
0.458984375	1.34260659691714e-35\\
0.4609375	8.23638104654039e-37\\
0.462890625	6.51527043900583e-38\\
0.46484375	7.12509988734271e-39\\
0.466796875	1.57743165909397e-39\\
0.46875	8.1400129415087e-41\\
0.470703125	9.55680345814411e-42\\
0.47265625	2.33564615578753e-43\\
0.474609375	1.18291034094237e-44\\
0.4765625	2.20004557938766e-46\\
0.478515625	5.58620385356248e-48\\
0.48046875	1.95392093276669e-47\\
0.482421875	1.97002533471639e-46\\
0.484375	8.32239577620594e-46\\
0.486328125	9.11233945974832e-45\\
0.48828125	1.82453400973166e-43\\
0.490234375	5.51142729631698e-42\\
0.4921875	2.70489133384046e-41\\
0.494140625	4.88988523966551e-40\\
0.49609375	3.3439487611999e-39\\
0.498046875	3.76106325648076e-38\\
0.5	6.74169587509551e-37\\
0.501953125	7.05841180914295e-36\\
0.50390625	1.12607604320741e-34\\
0.505859375	5.71276063096261e-34\\
0.5078125	7.45493749412844e-33\\
0.509765625	2.68001485316282e-32\\
0.51171875	7.58002525423322e-32\\
0.513671875	1.3406115199499e-30\\
0.515625	1.83438919333614e-29\\
0.517578125	2.4055212429635e-28\\
0.51953125	3.16610267476172e-27\\
0.521484375	2.79885027516209e-26\\
0.5234375	3.32394522430061e-25\\
0.525390625	2.10085392624101e-24\\
0.52734375	3.22035292867008e-23\\
0.529296875	6.2421646032763e-23\\
0.53125	1.25459489083434e-21\\
0.533203125	9.03711145622405e-21\\
0.53515625	5.03750974457149e-20\\
0.537109375	3.37378849294757e-19\\
0.5390625	1.35185275430874e-17\\
0.541015625	2.19697001216284e-16\\
0.54296875	5.79567584621293e-15\\
0.544921875	1.35972060075201e-13\\
0.546875	4.13797413407266e-12\\
0.548828125	1.14272790513499e-10\\
0.55078125	3.08594374718821e-09\\
0.552734375	2.38198816886967e-08\\
0.5546875	4.57751453800247e-07\\
0.556640625	1.85611344723083e-06\\
0.55859375	2.34054367438178e-05\\
0.560546875	0.000325308531391667\\
0.5625	0.00831537405814916\\
0.564453125	0.224824911374339\\
0.56640625	5.16758244949681\\
0.568359375	13.1370590701128\\
0.5703125	1.55972909924407\\
0.572265625	0.317190540222414\\
0.57421875	0.0269167931012048\\
0.576171875	0.000965944052312092\\
0.578125	0.000255222363106143\\
0.580078125	0.000109907745674183\\
0.58203125	2.34824949270452e-05\\
0.583984375	3.89449176824297e-06\\
0.5859375	2.5669871976218e-07\\
0.587890625	1.77462750355876e-08\\
0.58984375	7.06204619174267e-10\\
0.591796875	4.44229473233776e-11\\
0.59375	1.45412325916387e-12\\
0.595703125	3.33004899005026e-14\\
0.59765625	2.22568688885952e-15\\
0.599609375	5.59867171136524e-16\\
0.6015625	4.98905064492628e-17\\
0.603515625	5.30426688131453e-18\\
0.60546875	4.53712144824061e-19\\
0.607421875	3.86992646338836e-20\\
0.609375	9.10414681779902e-21\\
0.611328125	7.14214649345965e-22\\
0.61328125	6.73425421142108e-23\\
0.615234375	6.84551753549744e-24\\
0.6171875	1.10272065422778e-25\\
0.619140625	7.76228003484291e-27\\
0.62109375	6.47204286318316e-28\\
0.623046875	2.7726117771795e-29\\
0.625	1.61740352227494e-30\\
0.626953125	8.61739974202464e-32\\
0.62890625	8.72086010078886e-33\\
0.630859375	1.65506713138805e-33\\
0.6328125	3.18684258099534e-34\\
0.634765625	5.66888410123257e-35\\
0.63671875	8.36645022856988e-36\\
0.638671875	6.48050301582664e-37\\
0.640625	3.86899518557334e-38\\
0.642578125	3.22458065070315e-39\\
0.64453125	7.21354813428501e-41\\
0.646484375	5.82881440044449e-42\\
0.6484375	2.89740793152722e-43\\
0.650390625	8.86557167090877e-45\\
0.65234375	6.94567421543444e-46\\
0.654296875	4.10430770255242e-47\\
0.65625	2.57748452288579e-48\\
0.658203125	2.09697050673357e-49\\
0.66015625	1.89941582020302e-50\\
0.662109375	3.37610480512296e-51\\
0.6640625	8.88095217422572e-53\\
0.666015625	7.99963935325794e-54\\
0.66796875	2.82684841950364e-55\\
0.669921875	1.26609405156659e-55\\
0.671875	1.01296385069578e-56\\
0.673828125	1.52938981896748e-57\\
0.67578125	5.27319647205129e-59\\
0.677734375	3.33927655668842e-60\\
0.6796875	3.12154656944878e-61\\
0.681640625	3.15131300604044e-62\\
0.68359375	1.79417524528559e-63\\
0.685546875	9.45926733307305e-65\\
0.6875	6.94803590394597e-66\\
0.689453125	6.73887447596291e-67\\
0.69140625	6.7855353413884e-68\\
0.693359375	7.51458029605167e-67\\
0.6953125	1.29970688985289e-65\\
0.697265625	9.7095812689996e-65\\
0.69921875	1.03950969488199e-63\\
0.701171875	1.60216708718737e-63\\
0.703125	1.49565737769397e-62\\
0.705078125	3.95895137715622e-61\\
0.70703125	8.12605709010393e-60\\
0.708984375	1.34788906341452e-58\\
0.7109375	6.80162055982864e-58\\
0.712890625	5.21865775815845e-57\\
0.71484375	1.66742616251186e-55\\
0.716796875	3.37039660362196e-54\\
0.71875	2.14292216592117e-53\\
0.720703125	2.5258640048002e-52\\
0.72265625	4.97396043031929e-51\\
0.724609375	1.85755244603638e-49\\
0.7265625	1.67427605030935e-48\\
0.728515625	5.05951761275101e-47\\
0.73046875	1.57272506602974e-45\\
0.732421875	2.56493817041367e-44\\
0.734375	3.36107153496571e-43\\
0.736328125	2.28239536398164e-42\\
0.73828125	2.41495791280237e-41\\
0.740234375	4.74668762511283e-40\\
0.7421875	1.42266587843613e-38\\
0.744140625	3.83732568415579e-37\\
0.74609375	1.12065639697022e-35\\
0.748046875	1.32907184726296e-34\\
0.75	1.00877757707964e-33\\
0.751953125	2.47380924057245e-32\\
0.75390625	2.67950915376419e-31\\
0.755859375	4.41649319360498e-30\\
0.7578125	1.04139660042017e-28\\
0.759765625	6.18980207850365e-28\\
0.76171875	6.39847670160591e-27\\
0.763671875	1.56719302061062e-25\\
0.765625	1.59214374857321e-24\\
0.767578125	2.0767680499195e-23\\
0.76953125	5.95780710539984e-22\\
0.771484375	6.55974641281285e-21\\
0.7734375	5.93694433764495e-20\\
0.775390625	2.60313126326906e-19\\
0.77734375	2.0410629644089e-18\\
0.779296875	2.13716806725917e-17\\
0.78125	1.79888987425702e-16\\
0.783203125	1.39344996128327e-15\\
0.78515625	1.11213715075281e-14\\
0.787109375	8.94931754825279e-14\\
0.7890625	1.2399079769451e-12\\
0.791015625	2.10602427575679e-11\\
0.79296875	2.28918241557434e-10\\
0.794921875	4.14766403425045e-09\\
0.796875	2.38867251296865e-08\\
0.798828125	5.44797930550967e-07\\
0.80078125	1.28453060288386e-05\\
0.802734375	0.000517134034404493\\
0.8046875	0.00894212084156643\\
0.806640625	0.273055469361332\\
0.80859375	3.9376184390347\\
0.810546875	3.06721103306636\\
0.8125	0.257899841874493\\
0.814453125	0.0091154042548801\\
0.81640625	0.000991093879867551\\
0.818359375	4.9726884417943e-05\\
0.8203125	3.4122267043302e-06\\
0.822265625	4.18945095635905e-07\\
0.82421875	7.88289013258035e-08\\
0.826171875	4.92411472976405e-09\\
0.828125	2.16977918635436e-10\\
0.830078125	1.43014274878283e-11\\
0.83203125	2.22340787781543e-12\\
0.833984375	1.38537953325965e-11\\
0.8359375	1.9248473691837e-10\\
0.837890625	2.35088456309384e-09\\
0.83984375	7.24692211003793e-08\\
0.841796875	2.57177302739795e-06\\
0.84375	5.23071409879188e-05\\
0.845703125	0.0007276584729926\\
0.84765625	0.00757452819557568\\
0.849609375	0.136338358286839\\
0.8515625	0.508964400480097\\
0.853515625	4.82183797248141\\
0.85546875	2.58832467811915\\
0.857421875	0.226274452698923\\
0.859375	0.0100553701953838\\
0.861328125	0.00037766727397822\\
0.86328125	3.25235098094118e-05\\
0.865234375	6.96392604651249e-06\\
0.8671875	1.58534693506019e-06\\
0.869140625	8.39557799975834e-08\\
0.87109375	6.04542107541749e-09\\
0.873046875	4.54140240863068e-10\\
0.875	3.18687274063039e-11\\
0.876953125	2.30144541729418e-12\\
0.87890625	9.45591682796e-14\\
0.880859375	9.671929256298e-15\\
0.8828125	2.13864133927988e-15\\
0.884765625	1.57884281661025e-16\\
0.88671875	3.82758359141908e-18\\
0.888671875	2.53967651770879e-19\\
0.890625	3.76990549925033e-20\\
0.892578125	1.08871571485336e-20\\
0.89453125	3.8112448962655e-21\\
0.896484375	2.8097746866662e-22\\
0.8984375	2.46474721404768e-23\\
0.900390625	1.10609928383652e-24\\
0.90234375	7.85931306035325e-26\\
0.904296875	3.52059018675898e-27\\
0.90625	1.18692393327045e-28\\
0.908203125	6.77316235230805e-29\\
0.91015625	4.93203723482682e-28\\
0.912109375	4.16483451901822e-27\\
0.9140625	1.16709818966086e-25\\
0.916015625	1.78615699183643e-24\\
0.91796875	2.71926871060296e-23\\
0.919921875	7.22402397939287e-22\\
0.921875	1.47937759617932e-20\\
0.923828125	2.29003908580869e-19\\
0.92578125	3.51381369707341e-18\\
0.927734375	4.36201098321603e-18\\
0.9296875	7.26730020394519e-17\\
0.931640625	1.02368396216425e-15\\
0.93359375	5.94092068906473e-15\\
0.935546875	1.25050789305491e-14\\
0.9375	1.28964582606081e-13\\
0.939453125	2.07488899882863e-12\\
0.94140625	1.18643417972532e-11\\
0.943359375	6.89823587075637e-11\\
0.9453125	1.7748106170513e-09\\
0.947265625	2.67037598379712e-08\\
0.94921875	7.65390137190318e-07\\
0.951171875	8.80263087478471e-06\\
0.953125	0.000127275851486833\\
0.955078125	0.00187949192479973\\
0.95703125	0.0128708960807337\\
0.958984375	0.150444653510518\\
0.9609375	2.7017591247785\\
0.962890625	7.65602217298902\\
0.96484375	0.696513280843893\\
0.966796875	0.0138978374567696\\
0.96875	0.00173714379455607\\
0.970703125	0.000163694952660932\\
0.97265625	2.59777111335546e-05\\
0.974609375	1.3492617121161e-06\\
0.9765625	2.68930650948372e-07\\
0.978515625	1.46678289245582e-08\\
0.98046875	7.35455743728798e-10\\
0.982421875	1.04793133344121e-10\\
0.984375	1.24740812413193e-11\\
0.986328125	4.04443517494998e-13\\
0.98828125	2.07107054788057e-14\\
0.990234375	2.31671653401382e-15\\
0.9921875	2.53720836317323e-16\\
0.994140625	1.89705148455202e-17\\
0.99609375	2.40536710748011e-18\\
0.998046875	1.85771958277694e-18\\
1	0\\
};
\end{axis}
\end{tikzpicture}%

%% file: pics/fig_L1norm_maxV1D.tex
%
%
%
\begin{tikzpicture}

\begin{axis}[%
width=4.6in,
height=1.9in,
at={(1.29in,0.481in)},
scale only axis,
xmode=log,
xmin=0.1,
xmax=10000,
xminorticks=true,
xlabel={$\max V_{1D}$},
ymin=0.,
ymax=1.05,
yminorticks=true,
ylabel={$L^1$-norm of ground state},
axis background/.style={fill=white},
title style={font=\bfseries},
legend style={at={(0.97,0.95)}, legend cell align=left, align=left, draw=white!15!black}
]

\addplot [color=mycolor1, mark=*, very thick]
  table[row sep=crcr]{%
0.1523	0.750459005875456\\
0.3046	0.596191423685869\\
0.6092	0.446909268874178\\
1.2184	0.311511853651896\\
2.4368	0.212286858575143\\
4.8736	0.170920259597358\\
9.7472	0.130774981102662\\
19.4944	0.102215688979688\\
38.9888	0.0837007647819049\\
77.9776	0.0688973968715804\\
155.9552	0.0639681393408589\\
311.9104	0.0636633801091856\\
623.8208	0.0549978360369636\\
1247.6416	0.0493357991611269\\
2495.2832	0.0409414676607647\\
4990.5664	0.0396764257700482\\
9981.1328	0.0338176959133359\\
};
\addlegendentry{disorder potential} 

\addplot [color=mycolor2, mark=square*, very thick]
table[row sep=crcr]{%
0.1523	0.99667485901409\\
0.3046	0.996394886000826\\
0.6092	0.995284285659181\\
1.2184	0.990968927001776\\
2.4368	0.975494197699726\\
4.8736	0.932165092454443\\
9.7472	0.858584142002615\\
19.4944	0.786587196865502\\
38.9888	0.735134924579325\\
77.9776	0.701579261307285\\
155.9552	0.679769044474597\\
311.9104	0.665410266252318\\
623.8208	0.655793323179808\\
1247.6416	0.649222769847125\\
2495.2832	0.644657018414904\\
4990.5664	0.641590479810963\\
9981.1328	0.639484501971978\\
};
\addlegendentry{periodic potential} 

\end{axis}
\end{tikzpicture}%

%% file: pics/fig_L1norm_kappaV_left.tex
%
%
%
\begin{tikzpicture}

\begin{axis}[%
width=2.35in,
height=2.2in,
scale only axis,
xmode=log,
xmin=0.045,
xmax=8,
xminorticks=true,
xlabel={$\kappa_{1D}$},
xlabel style={below=0.04in},
ymode=log,
ymin=0.012,
ymax=1,
yminorticks=true,
ylabel={$L^1$-norm of ground state},
ylabel style={below=2.9in},
axis background/.style={fill=white},
legend style={at={(0.965,0.40)}, legend cell align=left, align=left, draw=white!15!black}
]

\addplot [color=mycolor1, very thick, mark=*]
  table[row sep=crcr]{%
6.824	0.750460635468644\\
3.412	0.634768204348943\\
1.706	0.517188228959751\\
0.853	0.423139471532081\\
0.4265	0.348013065215615\\
0.21325	0.249473370804324\\
0.106625	0.183171747350098\\
0.0533125	0.143583511297947\\
};
\addlegendentry{$c_V=2^0$}

\addplot [color=mycolor2, very thick, mark=*]
  table[row sep=crcr]{%
6.824	0.311512736329639\\
3.412	0.210214549788687\\
1.706	0.162756775204815\\
0.853	0.118549855627395\\
0.4265	0.106609672426493\\
0.21325	0.0797559310927967\\
0.106625	0.0552529727587049\\
0.0533125	0.0547035649434822\\
};
\addlegendentry{$c_V=2^3$}

\addplot [color=mycolor3, very thick, mark=*]
  table[row sep=crcr]{%
6.824	0.132860227477827\\
3.412	0.0890200877479873\\
1.706	0.0663156090001146\\
0.853	0.0526072629432014\\
0.4265	0.0414149095688623\\
0.21325	0.041174837232508\\
0.106625	0.0410499812561942\\
0.0533125	0.0409862828979995\\
};
\addlegendentry{$c_V=2^6$}


\addplot [color=black, dashed]
table[row sep=crcr]{%
	6.824	0.640021164758475\\
	3.412	0.508089995560753\\
	1.706	0.403188068928719\\
	0.853	0.320010582379238\\
	0.4265	0.253992567554879\\
	0.21325	0.201594034464360\\
	0.106625	0.160005291189619\\
	0.0533125	0.126996283777440\\
};
\addlegendentry{order~$\frac 13$}


\end{axis}
\end{tikzpicture}%

%% file: pics/fig_L1norm_kappaV_right.tex
%
%
%
\begin{tikzpicture}

\begin{axis}[%
width=2.35in,
height=2.2in,
scale only axis,
xmode=log,
xmin=0.12,
xmax=25,
xminorticks=true,
xlabel={$\max V_{1D}$},
ymode=log,
ymin=0.012,
ymax=1,
yminorticks=true,
yticklabel pos=right,
axis background/.style={fill=white},
legend style={at={(0.46,0.40)}, legend cell align=left, align=left, draw=white!15!black}
]
\addplot [color=mycolor1, very thick, mark=*]
  table[row sep=crcr]{%
  	0.1523	0.750460635468644\\
  	0.3046	0.596191996569447\\
  	0.6092	0.446911902914291\\
  	1.2184	0.311512737614114\\
  	2.4368	0.212287795537656\\
  	4.8736	0.170920135519283\\
  	9.7472	0.130775055585033\\
  	19.4944	0.102215701974682\\
};
\addlegendentry{$c_\kappa = 2^0$}

\addplot [color=mycolor2, very thick, mark=*]
  table[row sep=crcr]{%
0.1523	0.423139447590847\\
0.3046	0.260870740753995\\
0.6092	0.182125313933377\\
1.2184	0.118549855572084\\
2.4368	0.0839475129171451\\
4.8736	0.0683380021301337\\
9.7472	0.0526072620877142\\
19.4944	0.0395343920758554\\
};
\addlegendentry{$c_\kappa = 2^{-3}$}

\addplot [color=mycolor3, very thick, mark=*]
  table[row sep=crcr]{%
0.1523	0.183171462881745\\
0.3046	0.136095765303331\\
0.6092	0.0969894622392084\\
1.2184	0.0552529708475029\\
2.4368	0.0476206364128641\\
4.8736	0.0437072717414933\\
9.7472	0.0410499811910665\\
19.4944	0.0388503742018907\\
};
\addlegendentry{$c_\kappa = 2^{-6}$}



\addplot [color=black, dashed]
table[row sep=crcr]{%
	0.1523	0.640604606076423\\
	0.3046	0.452975861015976\\
	0.6092	0.320302303038212\\
	1.2184	0.226487930507988\\
	2.4368	0.160151151519106\\
	4.8736	0.113243965253994\\
	9.7472	0.080075575759553\\
	19.4944	0.056621982626997\\
};
\addlegendentry{order~$-\frac 12$}

\end{axis}
\end{tikzpicture}%